\documentclass[pra,floatfix,twocolumn,amsfonts,amsmath,amssymb,nofootinbib,reprint,superscriptaddress]{revtex4-1}
\pdfoutput=1

\usepackage[utf8]{inputenc}
\usepackage[T1]{fontenc}
\usepackage{accents}
\usepackage[colorlinks]{hyperref}

\usepackage[normalem]{ulem}
\usepackage{amsthm}
\usepackage{mathtools}
\usepackage{bbm}
\usepackage{bm}

\usepackage{graphicx}
\usepackage{tikz}
\usetikzlibrary{arrows}

\newtheorem{theorem}{Theorem}
\newtheorem{proposition}[theorem]{Proposition}

\newtheorem{lemma}[theorem]{Lemma}
\newtheorem{observation}[theorem]{Observation}

\newtheorem{innercustomprop}{Proposition}
\newenvironment{customprop}[1]
  {\renewcommand\theinnercustomprop{#1}\innercustomprop}
  {\endinnercustomprop}

\def\bra#1{\mathinner{\langle{#1}|}}
\def\ket#1{\mathinner{|{#1}\rangle}}
\def\braket#1#2{\mathinner{\langle{#1|#2}\rangle}}
\def\ketbra#1#2{\mathinner{|{#1}\rangle\!\langle{#2}|}}

\def\BraVert{\egroup\,\mid@vertical\,\bgroup}

\def\dbra#1{\mathinner{\langle\!\langle{#1}|}}
\def\dket#1{\mathinner{|{#1}\rangle\!\rangle}}
\def\dketbra#1#2{\mathinner{|{#1}\rangle\!\rangle\!\langle\!\langle{#2}|}}

\DeclareMathOperator{\Tr}{Tr}

\DeclareMathOperator{\rank}{rank}
\DeclareMathOperator{\range}{range}

\renewcommand\L{\mathcal{L}}

\newcommand{\M}{\mathcal{M}}
\newcommand{\HS}{\mathcal{H}}
\newcommand{\N}{\mathcal{N}}
\newcommand{\A}{\mathcal{A}}

\newcommand{\K}{\mathcal{K}}

\newcommand{\id}{\mathbbm{1}}

\newcommand{\bigpi}{\scalebox{1.6}{$\pi$}}

\begin{document}

\title{Quantum circuits with classical versus quantum control of causal order}

\author{Julian Wechs}
%\email{julian.wechs@neel.cnrs.fr}
\affiliation{Univ.\ Grenoble Alpes, CNRS, Grenoble INP\footnote{Institute of Engineering Univ. Grenoble Alpes}, Institut N\'eel, 38000 Grenoble, France}
\affiliation{Centre for Quantum Information and Communication (QuIC), \'Ecole Polytechnique de Bruxelles, C.P.\ 165, Universit\'e libre de Bruxelles, 1050 Bruxelles, Belgium}

\author{Hippolyte Dourdent}
%\email{hippolyte.lazourenko-dourdent@neel.cnrs.fr}
\affiliation{Univ.\ Grenoble Alpes, CNRS, Grenoble INP\footnote{Institute of Engineering Univ. Grenoble Alpes}, Institut N\'eel, 38000 Grenoble, France}

\author{Alastair A. Abbott}
%\email{alastair.abbott@unige.ch}
\affiliation{D\'epartement de Physique Appliqu\'ee, Universit\'e de Gen\`eve, 1211 Gen\`eve, Switzerland}
\affiliation{Univ.\ Grenoble Alpes, Inria, 38000 Grenoble, France}

\author{Cyril Branciard}
%\email{cyril.branciard@neel.cnrs.fr}
\affiliation{Univ.\ Grenoble Alpes, CNRS, Grenoble INP\footnote{Institute of Engineering Univ. Grenoble Alpes}, Institut N\'eel, 38000 Grenoble, France}

\date{\today}

\begin{abstract}

Quantum supermaps are transformations that map quantum operations to quantum operations.
It is known that quantum supermaps which respect a definite, predefined causal order between their input operations correspond to fixed-order quantum circuits, also called quantum combs. 
A systematic understanding of the physical interpretation of more general types of quantum supermaps---in particular, those incompatible with a definite causal structure---is however lacking. 
In this paper, we identify two new types of circuits that naturally generalise the fixed-order case and that likewise correspond to distinct classes of quantum supermaps, which we fully characterise. We first introduce ``\emph{quantum circuits with classical control of causal order}'', in which the order of operations is still well-defined, but not necessarily fixed in advance: it can in particular be established dynamically, in a classically-controlled manner, as the circuit is being used.
We then consider ``\emph{quantum circuits with quantum control of causal order}'', in which the order of operations is controlled coherently. 
The supermaps described by these classes of circuits are physically realisable, and the latter encompasses all known examples of physically realisable processes with indefinite causal order, including the celebrated ``\emph{quantum switch}''.
Interestingly, it also contains new examples arising from the combination of dynamical and coherent control of causal order, and we detail explicitly one such process.
Nevertheless, we show that quantum circuits with quantum control of causal order can only generate ``causal'' correlations, compatible with a well-defined causal order.
We furthermore extend our considerations to probabilistic circuits that produce also classical outcomes, and we demonstrate by an example how the characterisations derived in this work allow us to identify new advantages for quantum information processing tasks that could be demonstrated in practice.
\end{abstract}

\maketitle

\makeatletter
\def\l@subsubsection#1#2{}
\makeatother

\tableofcontents

\section{Introduction}

The standard paradigm used in quantum information theory is that of quantum circuits. 
In this framework, quantum computations are performed through the application of quantum operations on some quantum system in a given, definite order. 
An approach that is relevant in many situations is to consider quantum circuits with open slots, into which arbitrary input operations can be inserted~\cite{chiribella08,chiribella09}. 
Such circuits can be understood as higher-order transformations that map quantum operations to quantum operations. 
Mathematically, they can be described as \emph{quantum supermaps}, i.e., maps which take completely positive (CP) maps to other CP maps~\cite{chiribella08a}. 
Quantum supermaps corresponding to circuits in which the input operations are performed in a definite, fixed order are also called \emph{quantum combs}~\cite{chiribella08}.

More generally however, quantum supermaps do not need to presuppose a fixed causal order of the different operations. 
The investigation of quantum structures that go beyond the quantum circuit framework and that are incompatible with a global causal order between the operations has begun to receive significant attention, motivated not only by foundational questions~\cite{hardy05,oreshkov12,zych19}, but also by the possibility of obtaining advantages in quantum information processing~\cite{chiribella13}. 
A useful description of quantum supermaps, encompassing those that are incompatible with any definite causal structure, is given by the process matrix framework~\cite{oreshkov12}.

Beyond quantum circuits with a fixed causal order (those represented by quantum combs), one can consider situations in which the causal order depends on how the circuit is being used. The order can in particular be established dynamically, i.e., the order of future operations can depend on previous ones~\cite{hardy05,baumeler14,oreshkov16,wechs19}. 
If the causal order is thus controlled in a classical manner, then as the circuit is being used the operations are still realised in a well-defined causal order, established on the fly.
However, one can also consider situations in which the causal order is indefinite, for instance subject to quantum superpositions.
Indeed, a realisable example of a quantum process with indefinite causal order---or, in a more technical jargon, of a ``causally nonseparable'' quantum process~\cite{oreshkov12,araujo15,oreshkov16,wechs19}---is the so-called ``quantum switch''~\cite{chiribella13}. 
In this process, two operations are applied to a target system in an order that is coherently controlled by another quantum system. 
If the control system is prepared in a superposition state, the two operations are applied in a ``superposition of orders''. 
A generalisation to $N$ operations applied in a superposition of different orders has also been proposed~\cite{colnaghi12,araujo14,facchini15}. 
Notably, the quantum switch can provide advantages in various quantum information processing tasks over standard, causally ordered quantum circuits~\cite{chiribella13,chiribella12,colnaghi12,araujo14,facchini15,feix15,
guerin16,ebler18,salek18,chiribella18,mukhopadhyay18,
mukhopadhyay20,procopio19,frey19,loizeau20,caleffi20,gupta19,procopio20,
zhao20,taddei20,felce20,guha20,sazim20,wilson20a,simonov20}, and has now been demonstrated in several experiments~\cite{procopio15,rubino17,rubino17a,goswami18,goswami20,wei19,guo20,taddei20}.

In light of such possibilities, it is notable that more general, constructive formulations of classes of quantum supermaps encompassing dynamical and coherent control of causal order have not been forthcoming.
In contrast, significant progress has been made in classifying quantum supermaps using the process matrix framework, notably by studying their causal structure~\cite{araujo15,oreshkov16,wechs19} and reversibility~\cite{araujo17,barrett20,yokojima20}.
This framework, however, adopts an inherently top-down approach, and it remains unclear whether generic quantum supermaps can be given faithful physical realisations.
In this paper we instead adopt a bottom-up approach, presenting two new, general classes of quantum supermaps that are realisable by construction.
These classes can be described as types of generalised quantum circuits, naturally extending the notion of quantum circuits with fixed causal order (``QC-FOs'', Sec.~\ref{sec:QCFOs}).

We first describe ``quantum circuits with classical control of causal order'' (``QC-CCs'', Sec.~\ref{sec:QCCCs}) in which the  causal order between $N$ operations can be classically controlled and thereby established dynamically, while ensuring that each operation is applied once and only once---a crucial assumption to ensure one obtains a quantum supermap.
Our study thus formalises the description of ``classically controlled quantum circuits'' proposed in Ref.~\cite{oreshkov16}.
The classical nature of the control in QC-CCs means that the causal order remains well-defined (if not fixed), so the corresponding processes are causally separable.
It is then natural to consider quantum circuits in which the causal order is controlled coherently, which leads us to formulate the new class of ``quantum circuits with quantum control of causal order'' (``QC-QCs'', Sec.~\ref{sec:QCQCs}), which contains the quantum switch as a particular example.
This class, however, also contains more general types of causally nonseparable quantum processes, a fact we illustrate with a novel example that qualitatively differs from the quantum switch.
Nevertheless, not all quantum supermaps can be realised as QC-QCs. 
In particular, we show that the correlations generated by QC-QCs are always compatible with a well-defined causal order, which means that processes that can violate so-called \emph{causal inequalities}~\cite{oreshkov12,abbott16,branciard16,baumeler16} cannot be realised as QC-QCs.
The relation between these different classes of quantum supermaps is summarised in Fig.~\ref{fig:venn_diagram}.

For each of these classes of generalised quantum circuits we show how they can be described as process matrices, characterise the classes of process matrices they define, and show how, given a process matrix from one of these classes, one can construct the corresponding circuit.

\begin{figure}[t]
	\begin{center}
	\includegraphics[width=\columnwidth]{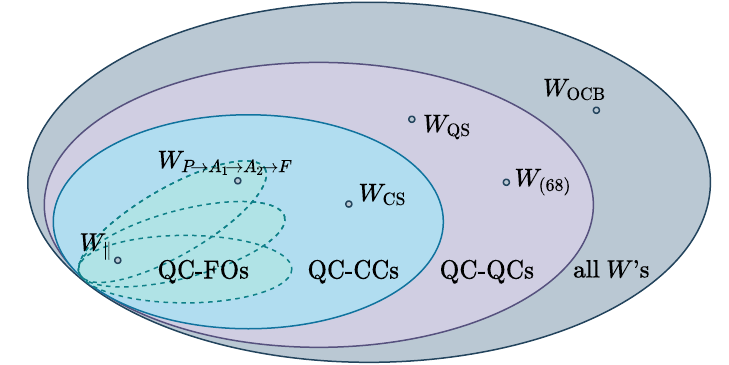}
	\end{center}
	\caption{
	Venn diagram illustrating the relation between the classes of quantum supermaps studied in this paper. 
	QC-FOs are quantum circuits compatible with a single, fixed, causal order (Sec.~\ref{sec:QCFOs}), such as the process $W_{P \text{\scalebox{.4}[1]{$\to$}} A_1 \text{\scalebox{.4}[1]{$\to$}} A_2 \text{\scalebox{.4}[1]{$\to$}} F}$ described in Eq.~\eqref{eq:W_PA1A2F}.
	QC-FOs form a non-convex set since a mixture of QC-FOs compatible with different orders is, in general, not compatible with any single order, while some processes, such as the parallel circuit $W_{\parallel}$ of Eq.~\eqref{eq:W_parallel} are compatible with any causal order.
	QC-CCs are quantum circuits with classical control of causal order (Sec.~\ref{sec:QCCCs}), such as the ``classical switch'' $W_\text{CS}$ (Eq.~\eqref{eq:W_CS}); all QC-CCs are causally separable processes.
	QC-QCs are quantum circuits with quantum control of causal order (Sec.~\ref{sec:QCQCs}), such as the quantum switch $W_\text{QS}$ (Eq.~\eqref{eq:W_QS}) and the new quantum process $W_\text{(68)}$ we describe in Eq.~\eqref{eq:processvec_newQCQC}, both of which are causally nonseparable.
	QC-QCs are a strict subset of all quantum supermaps: those violating causal inequalities, such as the $W_\text{OCB}$ of Ref.~\cite{oreshkov12} cannot be described as QC-QCs.
	}
	\label{fig:venn_diagram}
\end{figure}

In Sec.~\ref{sec:probQCs}, we then generalise our analysis to probabilistic (post-selected) quantum circuits.
We characterise the classes of probabilistic quantum supermaps, or ``quantum superinstruments'', that can be realised in terms of probabilistic QC-FOs, QC-CCs and QC-QCs.

The perspective of higher-order quantum transformations has turned out to be very useful for the investigation of quantum information processing tasks that involve the processing of unknown operations. 
For instance, the description of quantum combs (i.e., QC-FOs) in terms of quantum supermaps has been used to formulate and study various such tasks as semidefinite optimisation problems~\cite{chiribella16}. 
This approach can be extended to the more general classes of QC-CCs and QC-QCs, based on the characterisations that we provide in this work. 
In particular, our characterisation of QC-QCs and their corresponding probabilistic counterparts allows one to investigate possible quantum information processing applications of quantum processes which go beyond quantum circuits with a well-defined causal order, but for which a concrete realisation scheme exists. 
We illustrate this in Sec.~\ref{sec:applications}, where we consider a generalisation of a recently studied black-box discrimination task~\cite{shimbo18} and show that QC-QCs can provide a higher probability of success than any QC-FO or QC-CC.

Our work thus paves the way for a more systematic study of possible quantum processes with indefinite causal order, beyond the quantum switch, that are realisable in practice with current technologies, and of their applications for quantum information processing.

\section{Quantum circuits as quantum supermaps}
\label{sec:QC_supermaps}

Before proceeding further, let us first introduce the mathematical tools we shall use to manipulate and study quantum supermaps, and recall how a quantum circuit can be described by a so-called \emph{process matrix}~\cite{oreshkov12}.

\subsection{Preliminaries: mathematical tools}
\label{subsec:preliminaries}

In this paper we generically use the notation $\HS^X$ (for various different superscripts $X$) to denote a Hilbert space. 
$\L(\HS^X)$ is then defined as the space of linear operators on $\HS^X$ (operators $\HS^X \to \HS^X$); in particular, the identity operator is written $\id^X \in \L(\HS^X)$.%
\footnote{In general superscripts are used to indicate, when necessary, the relevant Hilbert spaces (we may omit them when the situation is clear enough). All Hilbert spaces throughout the paper are taken to be finite-dimensional. 
A generalisation of the process matrix framework to infinite-dimensional Hilbert spaces has been proposed in Ref.~\cite{giacomini16}.}
For two Hilbert spaces $\HS^X$ and $\HS^Y$, we use the short-hand notation $\HS^{XY} \coloneqq \HS^X \otimes \HS^Y$ to denote their tensor product (the order in which we write the factors being irrelevant, as long as we keep track of which space each of them corresponds to). 
$\Tr_X$ ($\Tr_Y$) then denotes the partial trace over $\HS^X$ (over $\HS^Y$), while $\Tr$ denotes the full trace.

\subsubsection{The Choi isomorphism}

Linear operators and maps are conveniently expressed using the Choi isomorphism~\cite{choi75}, which allows one to write them in the form of vectors or matrices. 
To define this we choose, for each Hilbert space $\HS^X$ under consideration, a fixed orthonormal basis $\{\ket{i}^X\}_i$---the \emph{computational basis} of $\HS^X$.
For a Hilbert space $\HS^{XY}$ obtained as the tensor product of two Hilbert spaces $\HS^X$ and $\HS^Y$ with computational bases $\{\ket{i}^X\}_i$ and $\{\ket{j}^Y\}_j$, respectively, the computational basis is naturally taken to be $\{\ket{i,j}^{XY} \coloneqq \ket{i}^X \otimes \ket{j}^Y\}_{i,j}$.

The choice of fixed computational bases is used in particular to define, for any pair of isomorphic Hilbert spaces $\HS^X$ and $\HS^{X'}$ with computational basis states $\ket{i}^X$ and $\ket{i}^{X'}$ in one-to-one correspondence,\footnote{Throughout, whenever we refer to isomorphic Hilbert spaces we will always implicitly assume that their computational basis states are in one-to-one correspondence.} the unnormalised maximally entangled state---written as a ``double-ket vector''
\begin{align}
\dket{\id}^{XX'} \coloneqq \sum_i \ket{i}^X \otimes \ket{i}^{X'} \quad \in \HS^X \otimes \HS^{X'}. \label{eq:def_ketId}
\end{align}
The computational basis is also used to define transposition of operators in $\L(\HS^X)$, denoted ${}^T$, or ${}^{T_X}$ for the partial transpose over $\HS^X$ only, in the case of an operator over a composite system in $\L(\HS^{XY})$.

\medskip

In this paper we shall make use of two (directly related) versions of the Choi isomorphism: the ``pure case''  and the ``mixed case'' versions.

For the first case we define, for any linear operator $V: \HS^X \to \HS^Y$, its \emph{Choi vector} as%
\footnote{Note that this double-ket notation is consistent with the definition of Eq.~\eqref{eq:def_ketId} when $V = \id^{X\to X'}: \HS^X \to \HS^{X'}, \ket{i}^X \mapsto \ket{i}^{X'}$ is the ``identity'' operator that defines the  one-to-one correspondence between the computational basis states of $\HS^X$ and $\HS^{X'}$.
In Eq.~\eqref{eq:def_Choi_vector}, $\dket{\id}^{XX}$ is defined as in Eq.~\eqref{eq:def_ketId} by taking $\HS^{X'}$ to just be a copy of $\HS^X$ (without introducing any ambiguity in the notations). \label{footnote:isomorphism_identity}}
\begin{align}
\dket{V} \coloneqq & \, \big( \id^X \otimes V \big) \dket{\id}^{XX} \notag \\
= & \, \sum_i \ket{i}^X \otimes V \ket{i}^X \quad \in \HS^{XY}. \label{eq:def_Choi_vector}
\end{align}
For the second case, for any linear map $\M: \L(\HS^X) \to \L(\HS^Y)$ we define its \emph{Choi matrix} as%
\footnote{To relate the two definitions of the Choi isomorphism, note that if $\M: \L(\HS^X) \to \L(\HS^Y)$ is obtained in terms of its Kraus operators $V_k: \HS^X \to \HS^Y$ as $\M(\rho) = \sum_k V_k \rho V_k^\dagger$, then its Choi matrix is obtained in terms of the Choi vectors $\dket{V_k}$ as $M = \sum_k \dketbra{V_k}{V_k}$ (with $\dbra{V_k} \coloneqq \dket{V_k}^\dagger$).}
\begin{align}
M \coloneqq & \, \big({\cal I}^X \otimes \M \big) \big( \dketbra{\id}{\id}^{XX} \big) \notag \\
= & \, \sum_{i,i'} \ketbra{i}{i'}^X \otimes \M\big( \ketbra{i}{i'}^X \big) \quad \in \L\big(\HS^{XY}\big) \label{eq:def_Choi_matrix}
\end{align}
(where ${\cal I}^X$ denotes the identity map on $\L(\HS^X)$). A fundamental property is that a linear map $\M$ is completely positive if and only if its Choi matrix is positive semidefinite~\cite{choi75}.

The inverse Choi isomorphism is easily obtained, in the two cases, as
\begin{align}
V = \dket{V}^{T_X} = \sum_i  \bra{i}^X \!\otimes\! \id^Y \dket{V} \bra{i}^X
\label{eq:inverse_pure_Choi}
\end{align}
and
\begin{align}
\M(\rho) = \Tr_X \big[ (\rho^T \otimes \id^Y) M \big] \label{eq:inverse_mixed_Choi}
\end{align}
for any $\rho \in \L(\HS^X)$.
This implies in particular that $\Tr[\M(\rho)] = \Tr [ \rho^T (\Tr_Y M) ]$, from which one can see that $\M$ is trace-preserving if and only if $\Tr_Y M = \id^X$.

\subsubsection{The link product}
\label{subsubsec:link_product}

We now introduce a special kind of product for vectors and matrices---the so-called \emph{link product}~\cite{chiribella08,chiribella09}---which will prove useful in describing the composition of quantum operations in terms of their Choi representations.

\medskip

Let $\HS^{XY} = \HS^X \otimes \HS^Y$ and $\HS^{YZ} = \HS^Y \otimes \HS^Z$ be two tensor product Hilbert spaces sharing the same (possibly trivial) space factor $\HS^Y$, and with non-overlapping $\HS^X, \HS^Z$.

The link product of any two vectors $\ket{a} \in \HS^{XY}$ and $\ket{b} \in \HS^{YZ}$ is defined (with respect to the computational basis $\{\ket{i}^Y\}_i$ of $\HS^Y$) as
\begin{align}
\ket{a} * \ket{b} \coloneqq & \, \big( \id^{XZ} \otimes \dbra{\id}^{YY} \big) (\ket{a} \otimes \ket{b}) \notag \\
= & \, \big(\ket{a}^{T_Y} \otimes \id^Z \big) \ket{b} \notag \\
= & \, \sum_i \ket{a_i}^X \otimes \ket{b_i}^Z \quad \in \HS^{XZ} \label{eq:def_pure_link_product}
\end{align}
with $\ket{a_i}^X \coloneqq (\id^X \otimes \bra{i}^Y) \ket{a} \in \HS^X$ and $\ket{b_i}^Z \coloneqq (\bra{i}^Y \otimes \id^Z) \ket{b} \in \HS^Z$ (so that $\ket{a} = \sum_i \ket{a_i}^X \otimes \ket{i}^Y$ and $\ket{b} = \sum_i \ket{i}^Y \otimes \ket{b_i}^Z$).

Similarly, the link product of any two operators $A \in \L(\HS^{XY})$ and $B \in \L(\HS^{YZ})$ is defined as~\cite{chiribella08,chiribella09}%
\footnote{Note, to check consistency with the pure case, that if $A$ and $B$ are for instance of the form $A = \ketbra{a}{a}$ and $B = \ketbra{b}{b}$, then $A * B = \ketbra{a}{a} * \ketbra{b}{b} = (\ket{a} * \ket{b})(\bra{a} * \bra{b}) = (\ket{a} * \ket{b})(\ket{a} * \ket{b})^\dagger$.}
\begin{align}
A * B \coloneqq & \, \big( \id^{XZ} \otimes \dbra{\id}^{YY} \big) (A \otimes B) \big( \id^{XZ} \otimes \dket{\id}^{YY} \big) \notag \\
= & \, \Tr_{Y} \big[ \big( A^{T_Y} \otimes \id^Z \big) \big( \id^X \otimes B \big) \big] \notag \\
= & \, \sum_{ii'} A_{ii'}^X \otimes B_{ii'}^Z \quad \in \L\big(\HS^{XZ}\big) \label{eq:def_mixed_link_product}
\end{align}
with $A_{ii'}^X \coloneqq (\id^X \otimes \bra{i}^Y) A (\id^X \otimes \ket{i'}^Y) \in \L(\HS^X)$ and $B_{ii'}^Z \coloneqq (\bra{i}^Y \otimes \id^Z) A (\ket{i'}^Y \otimes \id^Z) \in \L(\HS^Z)$ (so that $A = \sum_{i,i'} A_{ii'}^X \otimes \ketbra{i}{i'}^Y$ and $B = \sum_{i,i'} \ketbra{i}{i'}^Y \otimes B_{ii'}^Z$).

Let us state some properties of these link products that will be useful.
Firstly, note that they are commutative (up to a re-ordering of the tensor products). 
For a trivial 1-dimensional space $\HS^Y$---i.e., for $\ket{a} \in \HS^X$ and $\ket{b} \in \HS^Z$, or $A \in \L(\HS^X)$ and $B \in \L(\HS^Z)$ in distinct, non-overlapping spaces%
\footnote{Here, with a minor abuse of notation, we formally identify $\ket{a} \in \HS^{XY}$ with $(\id^X \otimes \bra{0}^Y) \ket{a} \in \HS^X$, and $A \in \L(\HS^{XY})$ with $(\id^X \otimes \bra{0}^Y) A (\id^X \otimes \ket{0}^Y) \in \HS^X$, where $\{\ket{0}^Y\}$ denotes the computational basis of the trivial 1-dimensional space $\HS^Y$ (and similarly for the other cases with trivial space factors).}---they reduce to tensor products ($\ket{a} * \ket{b} = \ket{a} \otimes \ket{b}$ or $A * B = A \otimes B$). 
For trivial spaces $\HS^X$ and $\HS^Z$ on the other hand---i.e., for $\ket{a}, \ket{b} \in \HS^Y$, or $A , B \in \L(\HS^Y)$ in the same spaces---they reduce to scalar products ($\ket{a} * \ket{b} = \sum_i \braket{i}{a} \braket{i}{b} = \ket{a}^T \ket{b}$ or $A * B = \Tr[A^T B]$). 
Note also that the link product of two positive semidefinite matrices is positive semidefinite (or a nonnegative real number for trivial spaces $\HS^X$ and $\HS^Z$).

We will often consider link products of vectors $\ket{a}, \ket{b}$ or matrices $A,B$ in (or acting on) some Hilbert spaces given as $\bigotimes_{j \in \mathsf{A}} \HS^j$ and $\bigotimes_{j \in \mathsf{B}} \HS^j$, for some (non-overlapping) tensor factors $\HS^j$ and some sets of indices $\mathsf{A}, \mathsf{B}$. 
The definitions above are then used by taking $\HS^X = \bigotimes_{j \in \mathsf{A} \backslash \mathsf{B}} \HS^j$, $\HS^Y = \bigotimes_{j \in \mathsf{A} \cap \mathsf{B}} \HS^j$ and $\HS^Z = \bigotimes_{j \in \mathsf{B} \backslash \mathsf{A}} \HS^j$.
The 2-fold products can also be extended to define $n$-fold link products of $n$ vectors $\ket{a_k} \in \HS^{\mathsf{A}_k} \coloneqq \bigotimes_{j \in \mathsf{A}_k} \HS^j$ or $n$ matrices $A_k \in \L(\HS^{\mathsf{A}_k})$, for $n$ sets of indices $\mathsf{A}_k$. 
Provided (as will be the case for all $n$-fold link products written in this paper) that each constituent Hilbert space $\HS^j$ appears at most twice in all $\HS^{\mathsf{A}_k}$'s---i.e., that $\mathsf{A}_{k_1} \cap \mathsf{A}_{k_2} \cap \mathsf{A}_{k_3} = \emptyset$ for all $k_1 \neq k_2 \neq k_3$---the $n$-fold link products thus defined are associative (in addition to being commutative)~\cite{chiribella08,chiribella09}, and can unambiguously be written without parentheses as $\ket{a_1} * \ket{a_2} * \cdots * \ket{a_n}$ or $A_1 * A_2 * \cdots * A_n$.

\medskip

\begin{figure}[t]
	\begin{center}
	\includegraphics[scale=1]{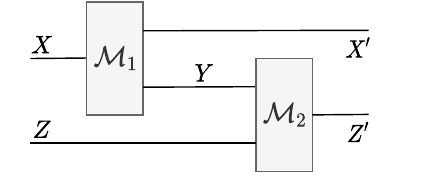}
	\end{center}
	\caption{Composition of two linear maps $\M_1: \L(\HS^X) \to \L(\HS^{X'Y})$ and $\M_2: \L(\HS^{YZ}) \to \L(\HS^{Z'})$ (as indicated by the labels on the wires, to be read from left to right). 
	The Choi matrix of the composed map $\M \coloneqq ({\cal I}^{X'} \otimes \M_2)\circ(\M_1 \otimes {\cal I}^Z)$ is obtained as the \emph{link product} of the Choi matrices of $\M_1$ and $\M_2$, as in Eq.~\eqref{eq:compose_link_prod_mixed}---and similarly for the ``pure case'' of two linear operators $V_1: \HS^{X} \to \HS^{X'Y}$ and $V_2: \HS^{YZ} \to \HS^{Z'}$, as in Eq.~\eqref{eq:compose_link_prod_pure}.}
	\label{fig:link_prod_composition}
\end{figure}

The initial motivation for introducing the link product (originally for matrices)~\cite{chiribella08,chiribella09} was to give a convenient way to write the Choi representation of a quantum operation obtained as the composition of two operations in sequence. 
To illustrate this, consider two linear operators $V_1: \HS^{X} \to \HS^{X'Y}$ and $V_2: \HS^{YZ} \to \HS^{Z'}$, with the output space of $V_1$ overlapping (through the tensor factor $\HS^Y$) with the input space of $V_2$: see Fig.~\ref{fig:link_prod_composition}.
It can easily be verified that the Choi vector of the composed operator $V \coloneqq (\id^{X'} \otimes V_2)(V_1 \otimes \id^Z): \HS^{XZ} \to \HS^{X'Z'}$ is obtained, in terms of the Choi vectors $\dket{V_1} \in \HS^{XX'Y}$ and $\dket{V_2} \in \HS^{YZZ'}$ of $V_1$ and $V_2$, as
\begin{align}
\dket{V} = \dket{V_1} * \dket{V_2} \quad \in \HS^{XX'ZZ'}. \label{eq:compose_link_prod_pure}
\end{align}
Similarly, for two linear maps $\M_1: \L(\HS^{X}) \to \L(\HS^{X'Y})$ and $\M_2: \L(\HS^{YZ}) \to \L(\HS^{Z'})$
the Choi matrix of the composition $\M \coloneqq ({\cal I}^{X'} \otimes \M_2) \circ (\M_1 \otimes {\cal I}^Z): \L(\HS^{XZ}) \to \L(\HS^{X'Z'})$ is obtained, in terms of the Choi matrices $M_1 \in \L(\HS^{XX'Y})$ and $M_2 \in \L(\HS^{YZZ'})$ of $\M_1$ and $\M_2$, as
\begin{align}
M = M_1 * M_2 \quad \in \L\big(\HS^{XX'ZZ'}\big). \label{eq:compose_link_prod_mixed}
\end{align}

Finally, we note that the link product allows one to write the inverse Choi isomorphism in a simple way. 
Indeed, the Choi matrix of the operation that consists in preparing some state (or density matrix) $\rho \in \L(\HS^X)$---i.e., of the map $1 \to \rho$, with a trivial input space---is $\rho$ itself. 
The Choi matrix that represents the preparation of $\M(\rho)$---i.e., the composition of the preparation of $\rho$ with the map $\M$---is also $\M(\rho)$ itself, and is obtained by link-multiplying the Choi matrices $M$ (of $\M$) and $\rho$:
\begin{align}
\M(\rho) = \rho * M, \label{eq:inverse_mixed_Choi_link_prod}
\end{align}
which is indeed equivalent to Eq.~\eqref{eq:inverse_mixed_Choi}.%
\footnote{Similarly for the pure case, one has $V \ket{\psi} = \ket{\psi} * \dket{V}$, so that $V = \sum_i V \ketbra{i}{i}^X = \sum_i ( \ket{i}^X * \dket{V} ) \bra{i}^X$, equivalently to Eq.~\eqref{eq:inverse_pure_Choi}.
\\
Note also that this extends to operators acting on just a subpart of a composite system: e.g., for $V: \HS^X \to \HS^Y$ and $\ket{\psi} \in\HS^{XX'}$, one still has $(V \otimes \id^{X'}) \ket{\psi} = \ket{\psi} * \dket{V} (\in \HS^{YX'})$; and analogously for the mixed case of Eq.~\eqref{eq:inverse_mixed_Choi_link_prod}.}

\subsection{Process matrices}
\label{subsec:PM_supermaps}

The sequential composition of two linear maps is an example of a  \emph{quantum supermap}~\cite{chiribella08a,chiribella08,chiribella09}: a process that takes any two ``freely chosen'' maps (say, $\A_1, \A_2$) to some new map (namely, $\M = \A_2 \circ \A_1$).
The \emph{process matrix framework} allows one to describe all possible ways to combine some ``free'' maps and define a new map (or originally, a probability distribution%
\footnote{The original process matrix framework of Ref.~\cite{oreshkov12} fits in the description given here (which follows that of Ref.~\cite{araujo17}, rather), by considering any probability distribution $P(\A_1,\ldots,\A_N)$ as the map $\M: 1 \mapsto P(\A_1,\ldots,\A_N)$, from and to some trivial 1-dimensional Hilbert spaces $\HS^P$ and $\HS^F$. In Appendix~\ref{app:W_equiv_descr} we elaborate on this, proving that the descriptions of Refs.~\cite{oreshkov12} and~\cite{araujo17} are indeed equivalent.\label{footnote:pm_formalisms}}%
) in a consistent manner~\cite{oreshkov12,araujo17}.

Let us make this more precise. 
Throughout the paper we will consider scenarios with $N \ge 1$ free quantum operations $\A_k$ ($k \in \N \coloneqq \{1, \ldots, N\}$), from some input to some output Hilbert spaces $\HS^{A_k^I}$ and $\HS^{A_k^O}$, of (finite, possibly different) dimensions $d_k^I$ and $d_k^O$, respectively. 
That is, the $N$ operations are any completely positive (CP) linear maps $\A_k: \L(\HS^{A_k^I}) \to \L(\HS^{A_k^O})$.
We will use the short-hand notations $\HS^{A_k^{IO}} \coloneqq \HS^{A_k^I} \otimes \HS^{A_k^O}$ and $\HS^{A_\N^{IO}} \coloneqq \bigotimes_{k \in \N} \HS^{A_k^{IO}}$, or more generally $\HS^{A_\K^{IO}} \coloneqq \bigotimes_{k \in \K} \HS^{A_k^{IO}}$ for any subset $\K \subseteq \N$. 

We are interested in how one can combine these $N$ operations so as to define a new quantum operation $\M: \L(\HS^P) \to \L(\HS^F)$, from some $d_P$-dimensional Hilbert space $\HS^P$ to some $d_F$-dimensional Hilbert space $\HS^F$, which can be thought of as embedding quantum systems in a ``global past'' and a ``global future'' of all $N$ operations, respectively; see Fig.~\ref{fig:general_W}. 
That is, how to define a function
\begin{align}
f: (\A_1,\ldots,\A_N) \mapsto \M. \label{eq:M_f_A1_An}
\end{align}
For consistency with a probabilistic interpretation, we impose that $f$ must be $N$-linear---so that if a given operation $\A_k$ is obtained as a probabilistic mixture of some operations $\A_k^{(j)}$, then the resulting map $\M$ should also be obtained as the corresponding probabilistic mixture: $f(\A_1,\ldots,\sum_j p^{(j)}\A_k^{(j)},\ldots,\A_N) = \sum_j p^{(j)} f(\A_1,\ldots,\A_k^{(j)},\ldots,\A_N)$. 
Furthermore, we require that $f$ must not only transform any set of $N$ CP maps $\A_k$ into another valid CP map, but that it can also be applied locally to extended maps $\A_k': \L(\HS^{A_k^IA_k^{I'}}) \to \L(\HS^{A_k^OA_k^{O'}})$ involving some ancillary Hilbert spaces $\HS^{A_k^{I'}}$ and $\HS^{A_k^{O'}}$ and still gives valid CP maps in such cases. 
Functions $f$ that satisfy these constraints define so-called completely CP-preserving (CCP) quantum supermaps~\cite{chiribella08a,chiribella08}. 

\begin{figure}[t]
	\begin{center}
	\includegraphics[scale=.8]{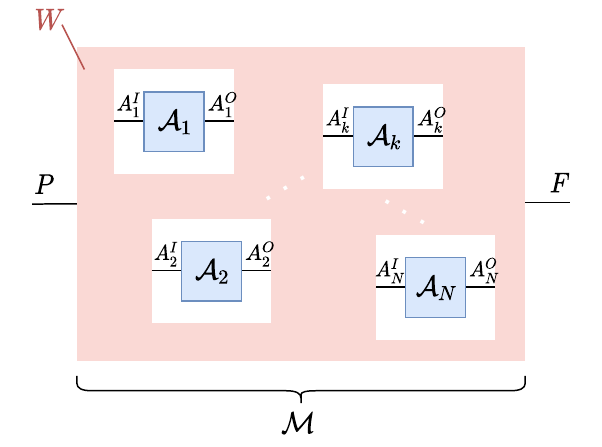}
	\end{center}
	\caption{A completely CP-preserving quantum supermap takes $N$ quantum operations---i.e.,\ CP maps---$\A_k$ (for $k=1,\dots,N$) with input and output Hilbert spaces $\HS^{A_k^I}$ and $\HS^{A_k^O}$, resp., to a new CP map $\M$ with an input Hilbert space $\HS^P$ in the ``global past'' of all operations $\A_k$ and an output Hilbert space $\HS^F$ in their ``global future''~\cite{oreshkov12,araujo17}.
	The Choi representation $M$ of the global map $\M$ is obtained from the Choi representations $A_k$ of the maps $\A_k$ according to Eq.~\eqref{eq:choi_map_M}, in terms of the \emph{process matrix} $W$ (represented by the salmon-coloured area), which describes how the $N$ operations are combined together to define the induced map $\M$. 
	Note that how exactly the $N$ operations are connected---i.e., their causal relations---need not be specified \emph{a priori}.}
	\label{fig:general_W}
\end{figure}

The ``supermapping'' of Eq.~\eqref{eq:M_f_A1_An} can be written at the level of the Choi matrices $M \in \L(\HS^{PF})$ of $\M$ and $A_k \in \L(\HS^{A_k^{IO}})$ of the operations $\A_k$, as $(A_1,\ldots,A_N) \mapsto M$.
Translating the previous constraints on $f$, it can be shown that the dependency on the Choi matrices can be written in terms of a Hermitian operator---a so-called \emph{process matrix}~\cite{oreshkov12,araujo17}
\begin{align}
W \in \L(\HS^{PA_\N^{IO}F}),
\end{align}
in the form
\begin{align}
M & = \Tr_{A_\N^{IO}} \big[ (A_1^T \otimes \cdots \otimes A_N^T \otimes \id^{PF}) W \big] \notag \\
& = (A_1 \otimes \cdots \otimes A_N) * W \quad \in \ \L(\HS^{PF}), \label{eq:choi_map_M}
\end{align}
where in the second line we used the link product notation defined previously, see Eq.~\eqref{eq:def_mixed_link_product}.
The requirement that $f$ above must be completely CP-preserving is equivalent here to $W$ being positive semidefinite, $W \succeq 0$.

Process matrices were originally introduced to describe \emph{deterministic} supermaps, such that if all CP maps $\A_k$ are trace-preserving (TP), then so must be the induced map $\M$ (they have thus sometimes been called \emph{superchannels}~\cite{gour19,quintino19a}).
This condition imposes some ``validity constraints'' on the allowed process matrices $W$---namely, that they must belong to some particular subspace $\L_V$ of $\L(\HS^{PA_\N^{IO}F})$, and be normalised such that $\Tr W = d_P ( \Pi_{k\in\N} \, d_k^O)$~\cite{oreshkov12,araujo15,oreshkov16,araujo17}; see Appendix~\ref{app:W_valid_matrices}.
By default, by ``process matrices'' we will refer to such deterministic ones---as will be considered in Secs.~\ref{sec:QCFOs}--\ref{sec:QCQCs} below. 
One may, however, also relax these constraints and consider \emph{probabilistic} process matrices which turn TP maps into a trace-nonincreasing induced map, and which may be part of a so-called \emph{quantum superinstrument}~\cite{chiribella13a}---namely, sets of probabilistic process matrices summing up to a deterministic one.
We will consider this possibility further in Sec.~\ref{sec:probQCs}.

\medskip

We emphasise that in the general construction of the process matrix framework, one does not specify \emph{a priori} \emph{how} the $N$ variable operations $\A_k$ are to be connected, and how these are causally related. 
In fact, while certain process matrices describe some clear causal connections, the framework also allows for process matrices which are incompatible with any well-defined causal structure between the $N$ operations~\cite{oreshkov12}. 
Some of these process matrices (like, e.g., that of the ``quantum switch'' mentioned in the introduction~\cite{chiribella13}) can be understood as exhibiting some kind of quantum superposition, or quantum coherent control, of causal orders.
In general, however, it has proven unclear how to interpret causally indefinite process matrices or, indeed, to determine which such processes can be given an interpretation of this (or any other) kind.

In the present paper, we study several different classes of process matrices for which one can give a clear interpretation for the underlying causal relations.
These classes can be described as types of generalised quantum circuits defining CCP quantum supermaps, into which the free, ``external'' operations $\A_k$ can be ``plugged in'' in either a fixed, a classically-controlled, or a coherently-controlled causal order.
This latter possibility can notably lead to causally indefinite process matrices, defining a broad, new class of such supermaps which, by construction, can be meaningfully interpreted.
For each type of circuit, we will calculate the induced global map $M$ as a function of the operations $A_k$ (in their Choi representations), and write their dependency in the form of Eq.~\eqref{eq:choi_map_M}, so as to identify the process matrix $W$ that describes them---noting that as the $\A_k$'s can be any CP maps, then $A_1 \otimes \cdots \otimes A_N$ spans the whole space of Hermitian matrices in $\L(\HS^{A_\N^{IO}})$, so that the Hermitian matrix $W$ that gives the correct induced map $\M$, or $M$ for all possible $A_1, \ldots, A_N$ via Eq.~\eqref{eq:choi_map_M}, is unique.

\section{Quantum circuits with fixed causal order}
\label{sec:QCFOs}

Quantum circuits with fixed causal order (QC-FOs) have been studied in detail before, often under the name of ``quantum combs''~\cite{chiribella08,chiribella09}.
Here we simply recall their description (Proposition~\ref{prop:descr_W_QCFO}) and characterisation (Proposition~\ref{prop:charact_W_QCFO}) in terms of process matrices so as to make the paper self-contained and to set the stage for the study of quantum circuits without a fixed causal order.

\subsection{Description}

\begin{figure*}[t]
	\begin{center}
	\includegraphics[scale=.75]{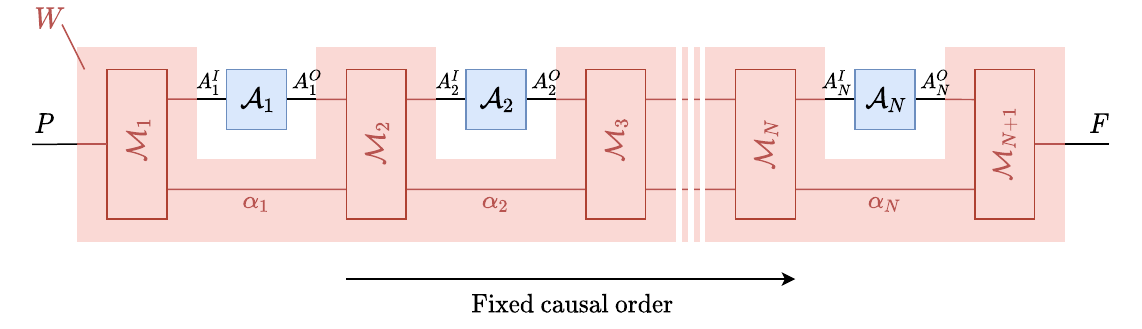}
	\end{center}
	\caption{A quantum circuit with fixed causal order (or, equivalently, a ``quantum comb''~\cite{chiribella08,chiribella09}), here shown with the order of operations $(\A_1, \A_2, \ldots, \A_N)$. 
	Its process matrix representation is given by $W = M_1 * M_2 * \cdots * M_{N+1}$, in terms of the Choi matrices $M_n$ of the internal circuit operations $\M_n$, as in Proposition~\ref{prop:descr_W_QCFO}.}
	\label{fig:QCFO}
\end{figure*}

We thus consider a quantum circuit with $N$ ``open slots'' into which the CP maps $\A_1, \ldots, \A_N$ are placed in a fixed order (so as to define the global map $\M$, as described above). 
We will denote, for example, the ordering in which $\A_1$ is applied first, then $\A_2$, etc., as $(\A_1, \A_2, \ldots, \A_N)$.
A QC-FO connects these ``external'' CP maps through ``internal'' quantum operations $\M_1,\dots,\M_{N+1}$ that take the output of each external map to the input of the subsequent one, as shown in Fig.~\ref{fig:QCFO}.
These internal circuit operations may involve additional ancillary systems or ``memories'' that are entangled with the ``target systems'' that the external CP maps act upon. 
For the moment, we consider deterministic circuits that do not themselves produce random transformations. 
The internal circuit operations $\M_n$ must therefore preserve the trace of their input states, i.e., they must be CPTP maps. 

More specifically, the circuit initially applies a CPTP map $\M_1: \L(\HS^P) \to \L(\HS^{A_1^I\alpha_1})$ which takes the circuit's input in the global past $\HS^P$ and outputs a state in the input Hilbert space $\HS^{A_1^I}$ of the first operation $\A_1$ (the target system), which in general may be entangled with an ancillary system in some Hilbert space $\HS^{\alpha_1}$.
Then, for $1 \le n \le N-1$, the output state of each external CP map $\A_n$ in the Hilbert space $\HS^{A^O_n}$ and the ancillary system in $\HS^{\alpha_n}$ are jointly mapped to the input Hilbert space $\HS^{A_{n+1}^I}$ of $\A_{n+1}$ and an ancillary system in some Hilbert space $\HS^{\alpha_{n+1}}$ by a CPTP map $\M_{n+1}: \L(\HS^{A_n^O\alpha_n}) \to \L(\HS^{A_{n+1}^I\alpha_{n+1}})$.
Finally, after the last operation $\A_{N}$, a CPTP map $\M_{N+1}: \L(\HS^{A_N^O\alpha_N}) \to \L(\HS^F)$ takes the output state of $\A_N$ in $\HS^{A_{N}^O}$, together with the ancillary state in $\HS^{\alpha_N}$, to the global output state of the full circuit in the global future $\HS^F$.
The maps $\M_1$, $\M_{n+1}$, and $\M_{N+1}$ above have Choi representations  $M_1 \in \L(\HS^{PA_1^I\alpha_1})$, $M_{n+1} \in \L(\HS^{A_n^O\alpha_nA_{n+1}^I\alpha_{n+1}})$, and $M_{N+1} \in \L(\HS^{A_N^O\alpha_NF})$, respectively.

\medskip

Let us elaborate further on the trace-preservation conditions we impose on the internal circuit operations $\M_n$.
As mentioned, these should preserve the trace of their input states; note however that we only require this for their \emph{possible} input states---i.e., not necessarily for their full input spaces $\L(\HS^{A_{n-1}^O\alpha_{n-1}})$, but only for its subspace that can actually be populated following the internal and external circuit operations previously applied.
Indeed if, for instance, a subspace of $\HS^{\alpha_{n-1}}$ is never populated by the previous internal operation $\M_{n-1}$, then we do not care about how $\M_n$ acts on that subspace.

It is in this relaxed sense, restricted to the possibly populated input spaces---which we shall call the \emph{effective input spaces}---that the TP conditions are to be understood throughout the paper.%
\footnote{Note that any map that is TP in this relaxed sense can always be artificially extended to a map that is TP on its full input space. 
Hence, it would be equivalent to impose here that the internal operations $\M_n$ are ``fully TP'' (as is usually done~\cite{chiribella08,chiribella09}). 
It will simplify matters here, however, to only require this ``effective TP'' condition: for the cases of quantum circuits with classical or quantum control of causal orders, imposing that the internal operations are fully TP would introduce unnecessary complications.} 
In the present case of QC-FOs, we show in Appendix~\ref{app:proof_charact_QCFOs} that these TP conditions can be expressed as the following constraints on the operations' Choi matrices:
\begin{align}
& \Tr_{A_1^I\alpha_1} M_1 = \id^P, \label{eq:TP_constr_QCFO_1} \\[1mm]
& \forall \, n = 1, \ldots, N{-}1, \ \Tr_{A_{n+1}^I\alpha_{n+1}} (M_1 * \cdots * M_n * M_{n+1}) \notag \\[-1mm]
& \hspace{29mm} = \Tr_{\alpha_n} (M_1 * \cdots * M_n) \otimes \id^{A_n^O}, \label{eq:TP_constr_QCFO_n} \\[1mm]
& \textup{and} \quad \Tr_F (M_1 * \cdots * M_N * M_{N+1}) \notag \\[-1mm]
& \hspace{25mm} = \Tr_{\alpha_N} (M_1 * \cdots * M_N) \otimes \id^{A_N^O}, \label{eq:TP_constr_QCFO_N}
\end{align}
which are in general weaker than (and indeed implied by) the TP assumptions applied to the full input spaces of the operations $\M_n$ (which can be written as $\Tr_{A_{n+1}^I\alpha_{n+1}} M_{n+1} = \id^{A_n^O\alpha_n}$ for $n = 1, \ldots, N{-}1$, and $\Tr_F M_{N+1} = \id^{A_N^O\alpha_N}$).

\medskip

The previous description of the process represented in Fig.~\ref{fig:QCFO}, with the internal circuit operations satisfying the TP constraints of Eqs.~\eqref{eq:TP_constr_QCFO_1}--\eqref{eq:TP_constr_QCFO_N}, formally defines what we call a \emph{Quantum Circuit with Fixed causal Order} (QC-FO).
These processes are indeed ``standard'' quantum circuits and, as shown in Refs.~\cite{chiribella08,chiribella09}, are the most general CCP quantum supermaps (obtained with an  ``axiomatic approach'') that respect the fixed causal order $(\A_1, \A_2, \ldots, \A_N)$, i.e., that do not allow for any signalling ``from the future to the past''~\cite{kretschmann05}. 
More precisely, this means that for any $n$, the output state following $\A_n$ (i.e., the target system in $\HS^{A_n^O}$) does not depend on the external operations $\A_{n+1},\dots,\A_N$ applied ``later'' in the circuit.

\medskip

Let us now consider how to obtain the description of a QC-FO as a process matrix.
Recall firstly that the Choi matrix of the sequential composition of quantum operations is obtained by link-multiplying the composite operations.
Here, the Choi matrix of the induced global map $\M: \L(\HS^P) \to \L(\HS^F)$ is thus
\begin{align}
M & = M_1 * A_1 * M_2 * \cdots * M_N * A_N * M_{N+1} \notag \\
& = (A_1 \otimes \cdots \otimes A_N) * ( M_1 * M_2 * \cdots * M_N * M_{N+1} ) \notag \\
& \hspace{20mm} \in \L(\HS^{PF}), \label{eq:induced_M_QCFO}
\end{align}
where in the second line we used the commutativity and associativity of the link product, and the fact that it reduces to tensor products for non-overlapping Hilbert spaces, to write it in the form of Eq.~\eqref{eq:choi_map_M}.
This allows us to identify the process matrix $W$ as the second term in parentheses above, and which, as noted at the end of Sec.~\ref{sec:QC_supermaps}, is moreover unique.
This thus proves the following:

\begin{proposition}[Process matrix description of QC-FOs] \label{prop:descr_W_QCFO}
The process matrix corresponding to the quantum circuit of Fig.~\ref{fig:QCFO}, with the fixed causal order $(\A_1, \A_2, \ldots, \A_N)$, is
\begin{align}
W = M_1 * M_2 * \cdots * M_N * M_{N+1} \quad \in \ \L(\HS^{PA_\N^{IO} F}). \label{eq:W_QCFO}
\end{align}
\end{proposition}

We note that this coincides precisely with the description of quantum combs given in Refs.~\cite{chiribella08,chiribella09}.

\subsection{Characterisation}

This description of QC-FOs allows us to obtain the following characterisation of their process matrices.

\begin{proposition}[Characterisation of QC-FOs] \label{prop:charact_W_QCFO}
For a given matrix $W \in \L(\HS^{PA_\N^{IO} F})$, let us define the reduced matrices (for $1 \le n \le N$, and relative to the fixed order $(\A_1, \A_2, \ldots, \A_N)$) $W_{(n)} \coloneqq \frac{1}{d_n^O d_{n+1}^O \cdots d_N^O}\Tr_{A_n^OA_{\{n+1,\ldots,N\}}^{IO}F} W \in \L(\HS^{PA_{\{1,\ldots,n-1\}}^{IO}A_n^I})$.

The process matrix $W \in \L(\HS^{PA_\N^{IO} F})$ of a quantum circuit with the fixed causal order $(\A_1, \A_2, \ldots, \A_N)$ is a positive semidefinite matrix such that its reduced matrices $W_{(n)}$ just defined satisfy
\begin{align}
& \Tr_{A_1^I} W_{(1)} = \id^P, \notag \\
& \forall \, n = 1, \ldots, N-1, \quad \Tr_{A_{n+1}^I} W_{(n+1)} = W_{(n)} \otimes \id^{A_n^O}, \notag \\
& \textup{and} \quad \Tr_F W = W_{(N)} \otimes \id^{A_N^O}. \label{eq:charact_W_QCFO}
\end{align}

Conversely, any positive semidefinite matrix $W \in \L(\HS^{PA_\N^{IO} F})$ whose reduced matrices $W_{(n)}$ satisfy the constraints of Eq.~\eqref{eq:charact_W_QCFO} is the process matrix of a quantum circuit with the fixed causal order $(\A_1, \A_2, \ldots, \A_N)$.
\end{proposition}

Equivalent results were already proved in Refs.~\cite{gutoski06,chiribella09}.
We give a self-contained proof in Appendix~\ref{app:proof_charact_QCFOs}, and here simply outline the proof approach.

To prove the first direction (the necessary condition), one needs simply to note that, for a QC-FO as described above, the reduced matrices defined above are of the form $W_{(n)} = \Tr_{\alpha_n} (M_1 * \cdots * M_n)$ and, according to Eqs.~\eqref{eq:TP_constr_QCFO_1}--\eqref{eq:TP_constr_QCFO_N} indeed satisfy Eq.~\eqref{eq:charact_W_QCFO}.
Note that Eq.~\eqref{eq:charact_W_QCFO} implies that $W$ satisfies the validity constraints for process matrices (cf.\ Appendix~\ref{app:W_valid_matrices}).

For the second direction (the sufficient condition), we provide an explicit construction: 
for a given $W$ whose reduced matrices $W_{(n)}$ satisfy Eq.~\eqref{eq:charact_W_QCFO}, we construct CPTP maps $\M_n$ (with Choi matrices $M_n$ obtained from the reduced matrices) which, for $1 \le n \le N$, act as isometries on their effective input spaces, and whose link product gives $W$ as in Eq.~\eqref{eq:W_QCFO}. 
That is, given such a $W$, we provide a way to explicitly construct the corresponding QC-FO. 
Note that this realisation is not unique, and different circuits may be described by the same process matrix. 
Moreover, a process matrix of this class may be compatible with different fixed causal orders.

\medskip

The description we gave of QC-FOs includes, as a specific case, the situation where the CP maps $\A_n$ (or just some of them) are used in parallel. 
The parallel composition of CP maps is equivalent to their composition in an arbitrary fixed order, with internal circuit operations in between that send the different input systems to the respective CP maps one at a time, while passing on the outputs of the preceding CP maps, as well as the inputs of the subsequent ones, via some ancillary systems; see the second explicit example below, and Appendix~\ref{app:QCPARs} for further details.
For completeness and ease of reference, let us state here how the process matrix characterisation of Proposition~\ref{prop:charact_W_QCFO} simplifies for such quantum circuits with operations used in parallel (QC-PARs).

\begin{proposition}[Characterisation of QC-PARs] \label{prop:charact_W_QCPAR}
The process matrix of a quantum circuit with operations used in parallel is a positive semidefinite matrix $W \in \L(\HS^{PA_\N^{IO} F})$ such that
\begin{align}
\Tr_F W = W_{(I)} \otimes \id^{A_\N^O} \quad \textup{with} \quad \Tr_{A_\N^I} W_{(I)} = \id^P \label{eq:charact_W_QCPAR}
\end{align}
for some matrix $W_{(I)} \in \L(\HS^{PA_\N^I})$.

Conversely, any positive semidefinite matrix $W \in \L(\HS^{PA_\N^{IO} F})$ satisfying Eq.~\eqref{eq:charact_W_QCPAR} is the process matrix of a quantum circuit with operations used in parallel.
\end{proposition}

A proof of this proposition, as well as a more detailed exposition of QC-PARs, are given in Appendix~\ref{app:QCPARs}.

\subsection{Examples}
\label{subsec:QCFO_examples}

As a simple example of a QC-FO, consider a process in which two CP maps $\A_1$ and $\A_2$ are applied successively to the input state from the global past, and then the output is sent to the global future; see Fig.~\ref{fig:W_PA1A2F}. 
This scenario corresponds to a QC-FO with the order $(\A_1, \A_2)$, with internal circuit operations that are (clearly TP) identity channels (between isomorphic Hilbert spaces $\HS^P$ and $\HS^{A_1^I}$, $\HS^{A_1^O}$ and $\HS^{A_2^I}$, and $\HS^{A_2^O}$ and $\HS^F$, with Choi matrices of the form $\dketbra{\id}{\id}^{XY}$), and that do not involve additional ancillary systems. 
The corresponding process matrix, as per Proposition~\ref{prop:descr_W_QCFO}, is
\begin{equation}
    \label{eq:W_PA1A2F}
    W_{P \text{\scalebox{.4}[1]{$\to$}} A_1 \text{\scalebox{.4}[1]{$\to$}} A_2 \text{\scalebox{.4}[1]{$\to$}} F} = \dketbra{\id}{\id}^{P A^I_1} \otimes \dketbra{\id}{\id}^{A^O_1 A^I_2} \otimes \dketbra{\id}{\id}^{A^O_2 F} ,
\end{equation}
and it is straightforward to verify that it satisfies the characterisation of Proposition~\ref{prop:charact_W_QCFO}.

\begin{figure}[t]
\begin{center}
\includegraphics[width=\columnwidth]{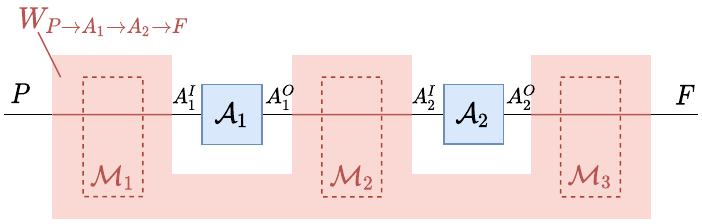}
\caption{
A QC-FO applying the CP maps $\A_1$ and $\A_2$ successively to a system initially provided in the global past $\HS^P$. The internal operations $\M_1,\M_2,\M_3$ are simply identity channels between the respective Hilbert spaces (cf.\ Fig.~\ref{fig:QCFO}).
}
\label{fig:W_PA1A2F}
\end{center}
\end{figure}

Another example is a scenario where a bipartite state is prepared in the global past and sent (via identity channels) in parallel to $\A_1$ and $\A_2$, whose outputs are then sent (again via identity channels) to the global future; see Fig.~\ref{fig:W_parallel}. 
Here the past Hilbert space decomposes as $\HS^P = \HS^{P_1} \otimes \HS^{P_2}$, with each $\HS^{P_k}$ isomorphic to $\HS^{A_k^I}$, and the future Hilbert space decomposes as $\HS^F = \HS^{F_1} \otimes \HS^{F_2}$, with each $\HS^{F_k}$ isomorphic to $\HS^{A_k^O}$. 
The corresponding ``parallel'' process matrix is
\begin{equation}
     W_{\parallel} = \dketbra{\id}{\id}^{P_1 \!A^I_1} \!\otimes\! \dketbra{\id}{\id}^{P_2 \!A^I_2} \!\otimes\!  \dketbra{\id}{\id}^{A^O_1 \!F_1} \!\otimes\! \dketbra{\id}{\id}^{A^O_2 \!F_2}\!\!. \label{eq:W_parallel}
\end{equation}
$W_{\parallel}$ is the process matrix of a QC-PAR, as can be verified from Proposition~\ref{prop:charact_W_QCPAR}.
It is thus also a QC-FO, compatible with both orders $(\A_1, \A_2)$ and $(\A_2, \A_1)$ (and satisfies Proposition~\ref{prop:charact_W_QCFO} for both orders). 
Indeed, a realisation of $W_{\parallel}$ as a QC-FO conforming to the description above with the causal order $(\A_1, \A_2)$ is given through the circuit operations (in their Choi representation) $M_1 = \dketbra{\id}{\id}^{P_1 A^I_1} \otimes \dketbra{\id}{\id}^{P_2 \alpha_1}$, $M_2 = \dketbra{\id}{\id}^{A^O_1 \alpha_2} \otimes \dketbra{\id}{\id}^{\alpha_1 A^I_2}$ and $M_3 = \dketbra{\id}{\id}^{A^O_2 F_2} \otimes \dketbra{\id}{\id}^{\alpha_2 F_1}$, by introducing some ancillary Hilbert spaces $\HS^{\alpha_1}$ isomorphic to $\HS^{P_2}$ and $\HS^{A_2^I}$, and $\HS^{\alpha_2}$ isomorphic to $\HS^{A_1^O}$ and $\HS^{F_1}$ (see Fig.~\ref{fig:W_parallel} and Appendix~\ref{app:QCPARs}). 
A realisation of $W_{\parallel}$ in terms of a QC-FO with the order $(\A_2, \A_1)$ is similarly given by the operations $M'_1 = \dketbra{\id}{\id}^{P_1 \alpha_1'} \otimes \dketbra{\id}{\id}^{P_2 A^I_2}$, $M'_2 = \dketbra{\id}{\id}^{A^O_2\alpha_2'} \otimes \dketbra{\id}{\id}^{\alpha_1' A^I_1}$ and $M'_3 = \dketbra{\id}{\id}^{A^O_1 F_1} \otimes \dketbra{\id}{\id}^{\alpha_2' F_2}$, with now $\HS^{\alpha_1'}$ isomorphic to $\HS^{P_1}$ and $\HS^{A_1^I}$, and $\HS^{\alpha_2'}$ isomorphic to $\HS^{A_2^O}$ and $\HS^{F_2}$. It can easily be checked that $M_1 * M_2 * M_3 = M'_1 * M'_2 * M'_3 = W_{\parallel}$.

\begin{figure}[t]
\begin{center}
\includegraphics[width=\columnwidth]{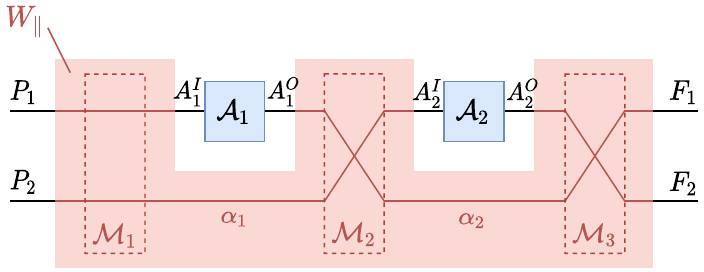}
\caption{A QC-FO which applies the CP maps $\A_1$ and $\A_2$ to the respective parts of a bipartite system prepared in $\HS^{P_1}\otimes \HS^{P_2}$, and then sends the outputs to $\HS^{F_1}\otimes \HS^{F_2}$. The process matrix describing this QC-FO, $W_{\parallel}$, could also be implemented as a QC-FO compatible with the order $\A_2\prec \A_1$, or directly as a QC-PAR (cf.\ Appendix~\ref{app:QCPARs}).
}
\label{fig:W_parallel}
\end{center}
\end{figure}

This example illustrates the fact that a given process matrix may have different realisations, and, more particularly, that process matrices described in the class of QC-FO may be compatible with different causal orders, or even with a parallel composition of the external operations.
Note also that the class of QC-FOs (i.e., quantum circuits compatible with some fixed order) is not convex (in contrast to those compatible with a single fixed order): a convex mixture of process matrices compatible with two different orders may not be compatible with any single fixed order, and thus not describe a QC-FO.

\section{Quantum circuits with classical control of causal order}
\label{sec:QCCCs}

While QC-FOs form an important and well-studied class of quantum supermaps, it is nonetheless a rather restrictive class.
Indeed, there are supermaps which are compatible with a well-defined causal structure (i.e., are causally separable~\cite{oreshkov12,oreshkov16,wechs19}) but which cannot be described as QC-FOs.
This is the case, for instance, of many supermaps representing probabilistic mixtures of QC-FOs with different causal orders, or of processes in which the causal order is established dynamically~\cite{hardy05,baumeler14,oreshkov16,wechs19}.
Here, motivated by a preliminary formulation in Ref.~\cite{oreshkov16}, we present a circuit model encompassing such possibilities, in which the causal order between the $N$ quantum operations $\A_k$ is still well-defined, but not fixed from the outset. 
Instead, in these quantum circuits with classical control of causal order (QC-CCs) it can be established dynamically, with the operations in the past determining the causal order of the operations in the future.
We will show below how to describe QC-CCs in terms of process matrices (Proposition~\ref{prop:descr_W_QCCC}), and characterise the set of process matrices they define (Proposition~\ref{prop:charact_W_QCCC}).

As recalled in Sec.~\ref{subsec:PM_supermaps}, in order for such circuits to define valid quantum supermaps they must be linear in the operations $\A_k$.
It is thus necessary to require that QC-CCs always apply each operation exactly once.
This excludes scenarios, for instance, where certain operations may or may not be applied, depending on the state of some control system~\cite{abbott20,chiribella19,kristjansson19}. 
Thus, only the order, and not the use, of the operations can be controlled classically within the framework considered here.

\subsection{Description}
\label{subsec:QCCC_description}

\begin{figure*}[t]
	\begin{center}
	\includegraphics[scale=.75]{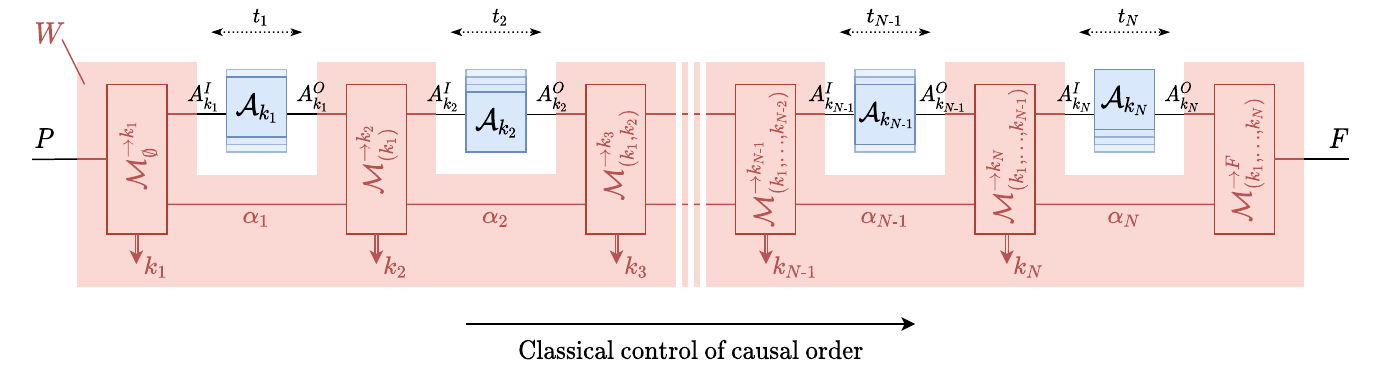}
	\end{center}
	\caption{Quantum circuit with classical control of causal order (QC-CC). The causal order is controlled, and established dynamically, by the outcomes $k_n$ of the internal circuit operations $\M_{(k_1,\ldots,k_{n-1})}^{\to k_n}$, represented by the double-stroke arrows.
	The superimposed boxes $\A_{k_n}$ at each time slot $t_n$ indicate that any of the $N$ external operations $\A_k$ can \emph{a priori} be applied at any time slot; we illustrate here the case where the causal order of operations ends up being $(k_1, k_2, \ldots, k_{N-1}, k_N)$.
	The process matrix $W$ that represents the circuit above is a (classical) combination of the different contributions corresponding to the different (dynamically established) orders $(k_1, \ldots, k_N)$. It is obtained from the Choi matrices $M_{(k_1,\ldots,k_{n-1})}^{\to k_n}$ of the internal circuit operations $\M_{(k_1,\ldots,k_{n-1})}^{\to k_n}$ according to Proposition~\ref{prop:descr_W_QCCC}.}
	\label{fig:QCCCv1}
\end{figure*}

We consider a generalised quantum circuit as represented schematically in Fig.~\ref{fig:QCCCv1}, with $N$ ``open slots'' at different time slots $t_n$ ($1 \le n \le N$). 
At each time slot, one (and only one) operation $\A_k$ will be applied (and each operation $\A_k$ can \emph{a priori} be applied at any time slot $t_n$).%
\footnote{Note that one could also consider $M>N$ time slots, and not apply an external operation at all of them; this would however just amount to introducing some trivial operations (identity channels) to fill the ``empty'' slots. \label{footnote:M_slots}}
Compared to the previous case of QC-FOs, however, precisely which operation is applied at each time slot $t_n$ is not predefined in a QC-CC.
Instead, before the first time slot $t_1$, and between each pair of consecutive time slots $t_n, t_{n+1}$ (for $1 \le n \le N-1$), the circuit applies an internal quantum operation which determines, in particular, which (thus far unused) operation $\A_k$ shall be applied next (while also transforming its input state and, potentially, additional ancillary systems). 
A final internal operation is then applied, taking the output of the operation applied at the last time slot $t_N$ to the output of the circuit in $\HS^F$.

The internal operations thus not only map the output state of the preceding operation to the input state of some subsequent one (together with potential ancillary systems): they now also produce a classical outcome, indicating which is the subsequent external operation to be applied. 
Such operations that keep track of both the classical and the quantum output are called \emph{quantum instruments}~\cite{davies70}. 
Mathematically, a quantum instrument is a collection of CP maps (associated to the different classical outputs), which sum up to a CPTP map. 

\medskip

More precisely, before the first time slot $t_1$, the circuit applies some internal quantum instrument $\{\M_\emptyset^{\to k_1}\}_{k_1 \in \N}$, where each operation $\M_\emptyset^{\to k_1} : \L(\HS^P) \to \L(\HS^{A_{k_1}^I\alpha_1})$, attached to the classical output $k_1$ that ``controls'' which external operation shall be applied first, maps the circuit's input in $\HS^P$ to the incoming space $\HS^{A_{k_1}^I}$ of the operation $\A_{k_1}$ and (possibly) also to some ancillary system in some Hilbert space $\HS^{\alpha_1}$.%
\footnote{Note that the various operations $\M_\emptyset^{\to k_1}$ that form the instrument $\{\M_\emptyset^{\to k_1}\}_{k_1 \in \N}$ (and similarly for the operations $\M_{(k_1,\ldots,k_n)}^{\to k_{n+1}}$ that form the instruments $\{\M_{(k_1,\ldots,k_n)}^{\to k_{n+1}}\}_{k_{n+1} \in \N \backslash \{k_1, \ldots, k_n\}}$ considered subsequently) do not have the same output spaces. This is, however, not a problem as we can formally extend the CP maps to have a common output space (e.g., $\L(\oplus_{k_1\in\N}\HS^{A_{k_1}^I}\otimes\HS^{\alpha_1})$). \label{footnote:different_output_spaces}}
Between the time slots $t_n$ and $t_{n+1}$, for $1 \le n \le N-1$, the circuit applies a quantum instrument $\{\M_{(k_1,\ldots,k_n)}^{\to k_{n+1}}\}_{k_{n+1} \in \N \backslash \{k_1, \ldots, k_n\}}$ conditioned on the sequence $(k_1, \ldots, k_n)$ of operations that have already been performed.%
\footnote{In accordance with the assumption that each operation can only be applied once, all sequences $(k_1, \ldots, k_n)$ we shall write assume that all $k_i$'s ($1 \le i \le n$) are different; for $n=N$, a sequence $(k_1, \ldots, k_N)$ shall thus contain each operation label $k_i \in \N$ once and only once. When we write the $k_i$'s within parentheses as in $(k_1, \ldots, k_n)$, their order matters (as opposed to $\{k_1, \ldots, k_n\}$ which denotes an unordered set).}
Each operation $\M_{(k_1,\ldots,k_n)}^{\to k_{n+1}} : \L(\HS^{A_{k_n}^O \alpha_n}) \to \L(\HS^{A_{k_{n+1}}^I \alpha_{n+1}})$, attached to the classical output $k_{n+1}$ indicating the next operation to apply, takes the output system of the last performed operation $\A_{k_n}$, together with the ancillary system in $\HS^{\alpha_n}$, to the incoming space $\HS^{A_{k_{n+1}}^I}$ of some yet unperformed operation $\A_{k_{n+1}}$ (hence with $k_{n+1} \in \N \backslash \{k_1, \ldots, k_n\}$) and an ancillary system in some Hilbert space $\HS^{\alpha_{n+1}}$. 
Before the time slot $t_N$ only one operation $\A_{k_N}$ is left to be performed (so that the instruments $\{\M_{(k_1,\ldots,k_{N-1})}^{\to k_N}\}$ only have one possible outcome $k_N$), and after $t_N$ all operations $\A_k$ have been performed exactly once. 
The circuit then applies a CPTP map $\M_{(k_1,\ldots,k_N)}^{\to F} : \L(\HS^{A_{k_N}^O \alpha_N}) \to \L(\HS^{F})$ that takes the output system of $\A_{k_N}$, together with the ancillary state in $\HS^{\alpha_N}$, to the output of the circuit in $\HS^F$.

\medskip

Let us elaborate further on the constraints required for the internal circuit operations to be valid quantum instruments (and thus for the circuit to be deterministic).
While each individual CP map of an instrument, say $\{\M_{(k_1,\ldots,k_n)}^{\to k_{n+1}}\}_{k_{n+1} \in \N \backslash \{k_1, \ldots, k_n\}}$, need not be TP, the trace should be preserved once all outcomes are summed over (i.e., the quantity $\sum_{k_{n+1}} \Tr \M_{(k_1,\ldots,k_n)}^{\to k_{n+1}}(\cdot)$) for any state in the effective input space of the operation.
As we show in Appendix~\ref{app:proof_charact_QCCCs}, analogously to Eqs.~\eqref{eq:TP_constr_QCFO_1}--\eqref{eq:TP_constr_QCFO_N}, these (effective) TP conditions translate here into the following constraints on the operations' Choi matrices:%
\footnote{From here on in, we will often omit the range of the sum when it is clear from context. E.g., $\sum_{k_1}$ is to be understood as $\sum_{k_1 \in \N}$ in Eq.~\eqref{eq:TP_constr_QCCC_1}, while $\sum_{k_{n+1}}$ means $\sum_{k_{n+1} \in \N\backslash\{k_1,\ldots,k_n\}}$ in Eq.~\eqref{eq:TP_constr_QCCC_n}.}
\begin{align}
& \hspace{2mm}\sum_{k_1} \Tr_{A_{k_1}^I\alpha_1} M_\emptyset^{\to k_1} = \id^P, \label{eq:TP_constr_QCCC_1} \\[1mm]
& \forall \, n = 1, \ldots, N{-}1, \ \forall \, (k_1, \ldots, k_n), \notag \\
& \hspace{2mm} \sum_{k_{n+1}} \! \Tr_{A_{k_{n+1}}^I\!\alpha_{n+1}} \!\! \big( M_\emptyset^{\to k_1} \!*\! \cdots \!*\! M_{(k_1,\ldots,k_{n-1})}^{\to k_n} \!*\! M_{(k_1,\ldots,k_n)}^{\to k_{n+1}} \big) \notag \\[-2mm]
& \hspace{10mm} = \Tr_{\alpha_n} \!\big(M_\emptyset^{\to k_1} * \cdots * M_{(k_1,\ldots,k_{n-1})}^{\to k_n} \big) \!\otimes\! \id^{A_{k_n}^O}\!, \label{eq:TP_constr_QCCC_n} \\[1mm]
& \textup{and} \ \forall \, (k_1, \ldots, k_N), \notag \\
& \hspace{2mm} \Tr_F \big( M_\emptyset^{\to k_1} \!* \cdots *\! M_{(k_1,\ldots,k_{N-1})}^{\to k_N} \!*\! M_{(k_1,\ldots,k_N)}^{\to F} \big) \notag \\%[-2mm]
& \hspace{10mm} = \Tr_{\alpha_N} \!\big(M_\emptyset^{\to k_1} \!* \cdots *\! M_{(k_1,\ldots,k_{N-1})}^{\to k_N} \big) \!\otimes\! \id^{A_{k_N}^O}\!. \label{eq:TP_constr_QCCC_N}
\end{align}

The previous description of the process under consideration, as represented in Fig.~\ref{fig:QCCCv1} and with the internal circuit operations $\M_\emptyset^{\to k_1}, \M_{(k_1,\ldots,k_n)}^{\to k_{n+1}}, \M_{(k_1,\dots,k_N)}^{\to F}$ satisfying the TP constraints of Eqs.~\eqref{eq:TP_constr_QCCC_1}--\eqref{eq:TP_constr_QCCC_N}, formally defines what we call a \emph{Quantum Circuit with Classical Control of causal order} (QC-CC).
Note that QC-FOs are a special case of QC-CCs as the internal CPTP maps of a QC-FO can be seen as instruments with only one non-trivial classical output.

\medskip

Let us now see how to obtain the description of a QC-CC as a process matrix.
As for QC-FOs (cf.\ Eq.~\eqref{eq:induced_M_QCFO}), in the case where the operations $\M_\emptyset^{\to k_1}$, $\M_{(k_1)}^{\to k_2}$, $\M_{(k_1,k_2)}^{\to k_3}$, \ldots, $\M_{(k_1,\ldots,k_{N-1})}^{\to k_N}$ and $\M_{(k_1,\ldots,k_N)}^{\to F}$ are applied in between the external operations $\A_k$---which thus end up being applied in the dynamically established order $(k_1, k_2, \ldots, k_N)$---the Choi matrix of the global CP map induced by the circuit is obtained as the link product
\begin{align}
& M_\emptyset^{\to k_1} * A_{k_1} * M_{(k_1)}^{\to k_2} * A_{k_2} * M_{(k_1,k_2)}^{\to k_3} * \cdots \notag \\
& \qquad \cdots * M_{(k_1,\ldots,k_{N-1})}^{\to k_N} * A_{k_N} * M_{(k_1,\ldots,k_N)}^{\to F} \notag \\[1mm]
& = (A_1 \otimes \cdots \otimes A_N) * \big( M_\emptyset^{\to k_1} * M_{(k_1)}^{\to k_2} * M_{(k_1,k_2)}^{\to k_3} * \cdots \notag \\
& \hspace{26mm} \cdots * M_{(k_1,\ldots,k_{N-1})}^{\to k_N} * M_{(k_1,\ldots,k_N)}^{\to F} \big), \label{eq:induced_Mk1_kN_QCCC}
\end{align}
where we used in particular the fact that each operation $A_k$ appears once and only each in $A_{k_1} * A_{k_2} * \cdots * A_{k_N}$ to reorder these terms.

As just stated, this induced map is conditioned on the causal order ending up being $(k_1, k_2, \ldots, k_N)$.%
\footnote{Note that this induced map is not TP; instead, the trace of its output equals the trace of its input, multiplied by the probability that the causal order of operations indeed ends up being $(k_1, \ldots, k_N)$.}
However, we want to describe the deterministic map that does not ``post-select'' on this order; indeed, the outcomes of the internal quantum instruments are \emph{internal} to the process.
We thus need to sum Eq.~\eqref{eq:induced_Mk1_kN_QCCC} above over all possible orders $(k_1, k_2, \ldots, k_N)$ to obtain the induced global map:
\begin{align}
M & = \sum_{(k_1,\ldots,k_N)} M_\emptyset^{\to k_1} * A_{k_1} * M_{(k_1)}^{\to k_2} * \cdots \notag \\[-4mm]
& \hspace{21mm} \cdots * M_{(k_1,\ldots,k_{N-1})}^{\to k_N} * A_{k_N} * M_{(k_1,\ldots,k_N)}^{\to F} \notag \\[1mm]
& \hspace{35mm} \in \L(\HS^{PF}). \label{eq:induced_M_QCCC}
\end{align}
Noting that the sum can be applied only to the second term in parentheses in Eq.~\eqref{eq:induced_Mk1_kN_QCCC} (which, for each $(k_1, \ldots, k_N)$, belongs to the same space $\L(\HS^{PA_\N^{IO} F})$), and that the induced map is then written in the form of Eq.~\eqref{eq:choi_map_M}, we can directly identify the process matrix $W$ and obtain the following:

\begin{proposition}[Process matrix description of QC-CCs] \label{prop:descr_W_QCCC}
The process matrix corresponding to the quantum circuit with classical control of causal order depicted in Fig.~\ref{fig:QCCCv1} is
\begin{align}
W = \sum_{(k_1,\ldots,k_N)} W_{(k_1,\ldots,k_N,F)} \label{eq:W_QCCC}
\end{align}
where
\begin{align}
& W_{(k_1,\ldots,k_N,F)} \coloneqq M_\emptyset^{\to k_1} * M_{(k_1)}^{\to k_2} * M_{(k_1,k_2)}^{\to k_3} * \cdots \notag \\
& \hspace{30mm} \cdots * M_{(k_1,\ldots,k_{N-1})}^{\to k_N} * M_{(k_1,\ldots,k_N)}^{\to F} \notag \\
& \hspace{40mm} \in \ \L(\HS^{PA_\N^{IO} F}). \label{eq:W_k1_kN_QCCC}
\end{align}
\end{proposition}

\subsection{Characterisation}

The above description of QC-CCs allows us to obtain the following characterisation of their process matrices.

\begin{proposition}[Characterisation of QC-CCs] \label{prop:charact_W_QCCC}
The process matrix $W \in \L(\HS^{PA_\N^{IO} F})$ of a quantum circuit with classical control of causal order can be decomposed in terms of positive semidefinite matrices $W_{(k_1,\ldots,k_n)} \in \L(\HS^{PA_{\{k_1,\ldots,k_{n-1}\}}^{IO} A_{k_n}^I})$ and $W_{(k_1,\ldots,k_N,F)} \in \L(\HS^{PA_\N^{IO} F})$, for all nonempty ordered subsets $(k_1,\ldots,k_n)$ of $\N$ (with $1 \le n \le N$, $k_i \neq k_j$ for $i \neq j$), in such a way that
\begin{align}
W = \sum_{(k_1,\ldots,k_N)} W_{(k_1,\ldots,k_{N},F)} \label{eq:charact_W_QCCC_decomp_sum}
\end{align}
and
\begin{align}
& \sum_{k_1} \Tr_{A_{k_1}^I} W_{(k_1)} = \id^P, \notag \\[1mm]
& \forall \, n = 1, \ldots, N{-}1, \ \forall \, (k_1, \ldots, k_n), \notag \\
& \hspace{3mm} \sum_{k_{n+1}} \Tr_{A_{k_{n+1}}^I} \! W_{(k_1,\ldots,k_n,k_{n+1})} = W_{(k_1,\ldots,k_n)} \otimes \id^{A_{k_n}^O}, \notag \\[1mm]
& \textup{and} \ \forall \, (k_1, \ldots, k_N), \notag \\
& \hspace{3mm} \Tr_F W_{(k_1,\ldots,k_N,F)} = W_{(k_1,\ldots,k_N)} \otimes \id^{A_{k_N}^O}. \label{eq:charact_W_QCCC_decomp_constr}
\end{align}

Conversely, any Hermitian matrix $W \in \L(\HS^{PA_\N^{IO} F})$ that admits a decomposition in terms of positive semidefinite matrices $W_{(k_1,\ldots,k_n)} \in \L(\HS^{PA_{\{k_1,\ldots,k_{n-1}\}}^{IO} A_{k_n}^I})$ and $W_{(k_1,\ldots,k_N,F)} \in \L(\HS^{PA_\N^{IO} F})$ satisfying Eqs.~\eqref{eq:charact_W_QCCC_decomp_sum}--\eqref{eq:charact_W_QCCC_decomp_constr} above is the process matrix of a quantum circuit with classical control of causal order.
\end{proposition}

The full proof is given in Appendix~\ref{app:proof_charact_QCCCs};
here, we simply outline briefly the proof approach.

As was the case of QC-FOs, the necessary condition follows from the form of Eqs.~\eqref{eq:W_QCCC}--\eqref{eq:W_k1_kN_QCCC}, and the TP constraints of Eqs.~\eqref{eq:TP_constr_QCCC_1}--\eqref{eq:TP_constr_QCCC_N}, with $W_{(k_1,\ldots,k_n)} \equiv \Tr_{\alpha_n} \big(M_\emptyset^{\to k_1} * \cdots * M_{(k_1,\ldots,k_{n-1})}^{\to k_n} \big)$.

To prove the sufficient condition, we again provide an explicit construction of a QC-CC: given a matrix $W$ with a decomposition satisfying Eqs.~\eqref{eq:charact_W_QCCC_decomp_sum}--\eqref{eq:charact_W_QCCC_decomp_constr}, we construct the operations $\M_\emptyset^{\to k_1}$, $\M_{(k_1,\ldots,k_n)}^{\to k_{n+1}}$ and $\M_{(k_1,\ldots,k_N)}^{\to F}$ (which, except in general for the last one, can each be taken to have a single Kraus operator) whose induced process matrix is precisely $W$.
As was the case for QC-FOs, this construction is not unique and different QC-CCs may be described by the same process matrix.

\medskip

It can be verified that Eqs.~\eqref{eq:charact_W_QCCC_decomp_sum}--\eqref{eq:charact_W_QCCC_decomp_constr} imply that $W$ satisfies the validity constraints for process matrices (cf.\ Appendix~\ref{app:W_valid_matrices}).
Note, however, that the individual matrices $W_{(k_1,\ldots,k_N,F)}$ in Proposition~\ref{prop:charact_W_QCCC} may or may not be valid (deterministic) process matrices.

If the $W_{(k_1,\ldots,k_N,F)}$'s are valid process matrices (up to normalisation), each compatible with the fixed causal order $(k_1,\ldots,k_N)$, then $W$ is simply a probabilistic mixture of quantum circuits with different fixed causal orders.
We recover the case of QC-FOs when there is only one term in the sum of Eq.~\eqref{eq:charact_W_QCCC_decomp_sum}; if that single term corresponds to the order $(k_1,\ldots,k_N) = (1,\ldots,N)$, the constraints of Eq.~\eqref{eq:charact_W_QCCC_decomp_constr} simply reduce to those of Eq.~\eqref{eq:charact_W_QCFO} (with $W_{(1,\ldots,n)} \equiv W_{(n)}$ and $W_{(1,\ldots,N,F)} \equiv W$).

If the $W_{(k_1,\ldots,k_N,F)}$'s are not valid process matrices, then the causal order depends, at least in part, on the input state of the circuit (in the global past space $\HS^P$) and on the external operations $\A_n$ inserted in the slots of the QC-CC.
The $W_{(k_1,\ldots,k_N,F)}$'s can, in that case, be interpreted as probabilistic process matrices which are post-selected on the order $(k_1,\ldots,k_N,F)$ being realised (see Sec.~\ref{sec:probQCs}).

\medskip

We finish by noting that if we consider the case with trivial one-dimensional global past and global future Hilbert spaces $\HS^P$ and $\HS^F$---i.e., the ``original'' version of process matrices as supermaps that take linear CP maps to probabilities~\cite{oreshkov12}---then the characterisation of Proposition~\ref{prop:charact_W_QCCC} (given more explicitly for this case in Appendix~\ref{app:charact_dP_dF_1}) coincides precisely with the sufficient condition for the causal separability of general $N$-partite process matrices obtained in Ref.~\cite{wechs19}.
Hence, unsurprisingly, QC-CCs define causally separable processes.%
\footnote{This is the case even if $\HS^P$ and $\HS^F$ are nontrivial, as the two versions of the process matrix framework are equivalent; see Appendix~\ref{app:W_equiv_descr}.}

\subsection{Example}
\label{subsec:QCCCs_example}

The simplest example of a QC-CC without a predetermined (even probabilistic) causal order is the ``classical switch''~\cite{chiribella13}, in which a classical ``control system'' is used to incoherently control the order in which two CP maps, $\A_1$ and $\A_2$, are applied to some ``target system''; see Fig.~\ref{fig:W_CS}.
These two systems are initially provided in the global past $\HS^P = \HS^{P_\textup{t}} \otimes \HS^{P_\textup{c}}$ and, after the operations are applied, are sent to the global future $\HS^F = \HS^{F_\textup{t}} \otimes \HS^{F_\textup{c}}$.
Here, $\HS^{P_\textup{t}}$ and $\HS^{F_\textup{t}}$ are $d_\textup{t}$-dimensional Hilbert spaces for the target system and $\HS^{P_\textup{c}}$ and $\HS^{F_\textup{c}}$ are $2$-dimensional Hilbert spaces (with computational bases denoted here $\{\ket{1},\ket{2}\}$) in which the classical control bit is encoded.
The operations $\A_1$ and $\A_2$ thus also act on $d_\textup{t}$-dimensional spaces $\HS^{A^k_I},\HS^{A^k_O}$.
The circuit begins by performing a measurement on the control system, and depending on the (classical) measurement outcome, the target system is sent (via identity channels) first to $\A_1$ and then to $\A_2$ (outcome `1'), or vice versa (outcome `2'). 
The order is thus not fixed \emph{a priori}, but is established through the preparation of the control system in the global past. 

\begin{figure}[t]
\begin{center}
\includegraphics[width=\columnwidth]{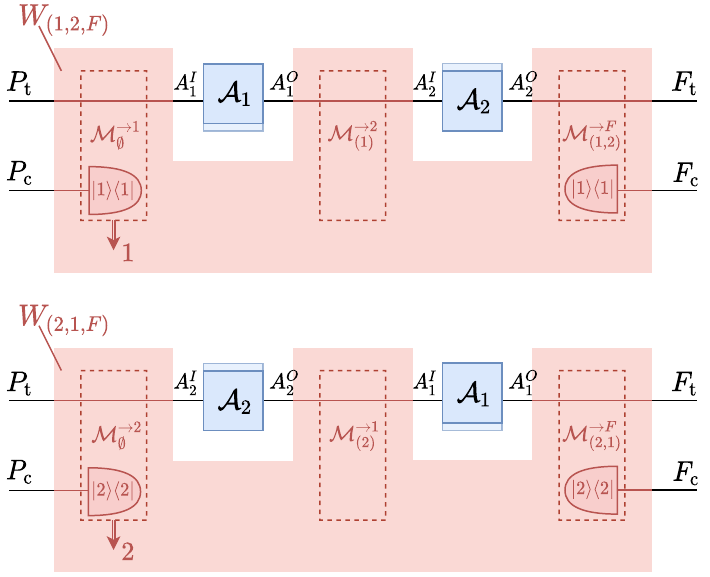}
\caption{The two possible realisations of the classical switch.
In this QC-CC, the order of the two CP maps $\A_1$ and $\A_2$ is controlled incoherently through the ``control system'' in $\HS^{P_\textup{c}}$, which is measured as part of the first internal circuit operation.
The process matrix $W_\text{CS}$ is obtained as the sum of the two corresponding probabilistic process matrices $W_{(1,2,F)}$ and $W_{(2,1,F)}$ (cf. Eq.~\eqref{eq:W_CS}).
}
\label{fig:W_CS}
\end{center}
\end{figure}

To see that the classical switch can be described as a QC-CC, we can take the internal circuit operations with Choi matrices
\begin{align}
M_\emptyset^{\to k_1} & = \dketbra{\id}{\id}^{P_\textup{t} A_{k_1}^I} \otimes \ketbra{k_1}{k_1}^{P_\textup{c}}, \notag \\
M_{(k_1)}^{\to k_2} & = \dketbra{\id}{\id}^{A_{k_1}^O A_{k_2}^I}, \notag \\
M_{(k_1,k_2)}^{\to F} & = \dketbra{\id}{\id}^{A_{k_2}^O F_\textup{t}} \otimes \ketbra{k_1}{k_1}^{F_\textup{c}}. \label{eq:Ms_CSwitch}
\end{align}
These operations can be interpreted intuitively:
$\M_\emptyset^{\to k_1}$ is an identity channel sending the initial target system in $\HS^{P_\textup{t}}$ to the input space of the first operation $\A_{k_1}$, post-selected on the outcome $k_1$ of the measurements on $\HS^{P_\textup{c}}$;
$\M_{(k_1)}^{\to k_2}$ is an identity channel sending the target from the output of $\A_{k_1}$ to the input of $\A_{k_2}$;
and $\M_{(k_1,k_2)}^{\to F}$ sends the output of the second operation to the global future, while preparing the control system in $\HS^{F_\textup{c}}$ in the appropriate state, $\ketbra{k_1}{k_1}$.
It is easy to verify that these operations indeed satisfy the TP conditions of Eqs.~\eqref{eq:TP_constr_QCCC_1}--\eqref{eq:TP_constr_QCCC_N}.

The process matrix describing the classical switch defined by the operations~\eqref{eq:Ms_CSwitch} is thus
\begin{align}
W_\textup{CS} = & M_\emptyset^{\to 1} * M_{(1)}^{\to 2} * M_{(1,2)}^{\to F} + M_\emptyset^{\to 2} * M_{(2)}^{\to 1} * M_{(2,1)}^{\to F} \notag\\
    = & \ketbra{1}{1}^{P_\textup{c}} \dketbra{\id}{\id}^{P_\textup{t}A_1^I} \dketbra{\id}{\id}^{A_1^OA_2^I} \dketbra{\id}{\id}^{A_2^OF_\textup{t}} \ketbra{1}{1}^{F_\textup{c}} \notag \\
& + \! \ketbra{2}{2}^{P_\textup{c}} \! \dketbra{\id}{\id}^{P_\textup{t}A_2^I} \! \dketbra{\id}{\id}^{A_2^O\!A_1^I} \! \dketbra{\id}{\id}^{A_1^O\!F_\textup{t}} \! \ketbra{2}{2}^{F_\textup{c}} \notag \\[1mm]
& \hspace{27mm} \in \L(\HS^{P_\textup{c}P_\textup{t}A_1^{IO}A_2^{IO}F_\textup{t}F_\textup{c}}) \label{eq:W_CS}
\end{align}
(where the tensor products are implicit). 
One can readily check that $W_\textup{CS}$ indeed satisfies the characterisation of Proposition~\ref{prop:charact_W_QCCC}, with $W_{(k_1)} = M_\emptyset^{\to k_1}$, $W_{(k_1,k_2)} = M_\emptyset^{\to k_1} \otimes M_{(k_1)}^{\to k_2}$, and $W_{(k_1,k_2,F)} = M_\emptyset^{\to k_1} \otimes M_{(k_1)}^{\to k_2} \otimes M_{(k_1,k_2)}^{\to F}$.

Note that this process goes beyond a probabilistic mixture of two fixed-order quantum circuits. 
Indeed, the two individual summands in Eq.~\eqref{eq:W_CS} do not satisfy the validity constraints for process matrices, and only their sum does. 
This reflects the fact that the first internal operation applied by the circuit, $\{\M_\emptyset^{\to k_1}\}_{k_1\in\N}$, is probabilistic, and if we post-select on one of the two outcomes, we do not end up with a valid (deterministic) supermap. 
(Indeed, as we will see later in Sec.\ref{sec:probQCs}, the individual terms are probabilistic process matrices.) 
To obtain a valid process, we thus need to combine the terms corresponding to the different outcomes.
This also proves (as was already shown in Ref.~\cite{chiribella13}), that such a classical switch cannot be realised by a standard QC-FO.

Lastly, let us observe that if one traces out $F$ from the process matrix of the classical switch, the resulting matrix $\Tr_F W_\textup{CS}$ is also a valid QC-CC (with now a trivial global future) with a still well-defined, but not predefined, causal order.
(This is also the case if one only traces out $F_\textup{t}$ or $F_\textup{c}$.)
Indeed, taking $M_\emptyset^{\to k_1}$ and $M_{(k_1)}^{\to k_2}$ as in Eq.~\eqref{eq:Ms_CSwitch} and $M_{(k_1,k_2)}^{\to F} = \id^{A^O_{k_2}}$, one recovers the corresponding process matrix.

\section{Quantum circuits with quantum control of causal order}
\label{sec:QCQCs}

In this section we go one step further, defining a class of circuits in which the causal order is controlled not classically, as in QC-CCs, but coherently in a quantum manner.
Such circuits may no longer always combine the operations $\A_k$ in a well-defined causal manner, but instead they do so in an indefinite causal order.
As for the classes above, we show how to describe these quantum circuits with quantum control of causal order (QC-QCs) as process matrices (Proposition~\ref{prop:descr_W_QCQC}) and characterise the set of process matrices they define (Proposition~\ref{prop:charact_W_QCQC}).
Before we present these QC-QCs, however, we will revisit QC-CCs from a slightly different angle.
In particular, we will first present a different, but equivalent, description of QC-CCs that will lead more naturally to this new class of QC-QCs.

\subsection{Revisiting the description of quantum circuits with classical control of causal order}
\label{subsec:QCCC_revisit}

\subsubsection{Introducing explicit control systems}

In the previous section we said that each internal operation $\M_{(k_1,\ldots,k_n)}^{\to k_{n+1}}$ applied by the circuit between the time slots $t_n$ and $t_{n+1}$ was conditioned on which operations $\A_k$ had already been performed (thereby allowing us to ensure that each external operation is applied once and only once, as required), and their order $(k_1,\ldots,k_n)$.
This conditioning can, in fact, be included in the description of the operation applied between $t_n$ and $t_{n+1}$ by introducing a physical ``control'' system that explicitly encodes the outcomes $k_n$ of the instruments $\{\M_{(k_1,\ldots,k_{n-1})}^{\to k_n}\}_{k_n}$, and stores on the fly the dynamically established causal order.

To this end, we add an explicit control system to the circuit, in which we encode the full order of the preceding (and currently applied) external operations in the computational basis states $\ket{(k_1,\ldots,k_n)}^{C_n^{(\prime)}}$ of some Hilbert space $\HS^{C_n^{(\prime)}}$ (for $1 \le n \le N$).
Here $C_n$ denotes the control system just before the external operation $\A_{k_n}$ (at time $t_n$) is applied, while $C_n'$ denotes the control system just after (see below).
As these control systems will, for now, act ``classically'', it will be useful to use the following notation:
\begin{align}
[\![(k_1,\ldots,k_n)]\!]^{C_n^{(\prime)}} \coloneqq \ketbra{(k_1,\ldots,k_n)}{(k_1,\ldots,k_n)}^{C_n^{(\prime)}}\!\!.
\end{align}
Note that while the example of the classical switch in Sec.~\ref{subsec:QCCCs_example} utilised a control qubit in the global past $\HS^P$ and future $\HS^F$, the role of the explicit control system we introduce here is more precise.
In that example, it would be used, e.g., to propagate the control qubit in $\HS^{P_\textup{c}}$ through the circuit to $\HS^{F_\textup{c}}$ and apply the correct external operation at each time slot.

This control system will be used to control both the choice of external operation $\A_{k_n}$ and the internal operations $\M_{(k_1,\ldots,k_n)}^{\to k_{n+1}}$, as illustrated in Fig.~\ref{fig:QCCCv2}.
To formally achieve this, we need to embed the input and output Hilbert spaces at each time slot $t_n$ within a common Hilbert space, before introducing global controlled operations acting in these spaces.
To simplify this, we will henceforth, and without loss of generality,%
\footnote{Indeed, if the input and output spaces of the operations $\A_k$ are not all the same, we can introduce additional ancillary input spaces $\HS^{A^{I'}_k}$ (of dimension $d^{I'}_k$) and output spaces $\HS^{A^{O'}_k}$ (of dimension $d^{O'}_k$), in such a way that $d^I_k d^{I'}_k = d^I$ and $d^O_k d^{O'}_k = d^O$ for all $k$, and upon which the extended operations act trivially. 
(Such $d^{I'}_k$ and $d^{O'}_k$ can always be found: one can simply choose for $d^I$ the least common multiple of all $d^{I}_k$ (and similarly for $d^O$), and then take $d^{I'}_k = d^I/d^{I}_k$, and $d^{O'}_k = d^O/d^{O}_k$.) 
The original scenario is then recovered when the additional input and output spaces are traced out.}
assume that all the external operations $\A_k$ have the same input space dimension ($d_k^I = d^I \ \forall \,k$), and the same output space dimension ($d_k^O = d^O \ \forall \,k$). 
All their input spaces are thus isomorphic to each other, and likewise for their output spaces.
As a result, the ``target'' system at each time slot is always of the same dimension, regardless of which external operation is applied to it (although the input and output dimensions may still differ, i.e., if $d^I\neq d^O$).

\begin{figure*}[t]
	\begin{center}
	\includegraphics[scale=.75]{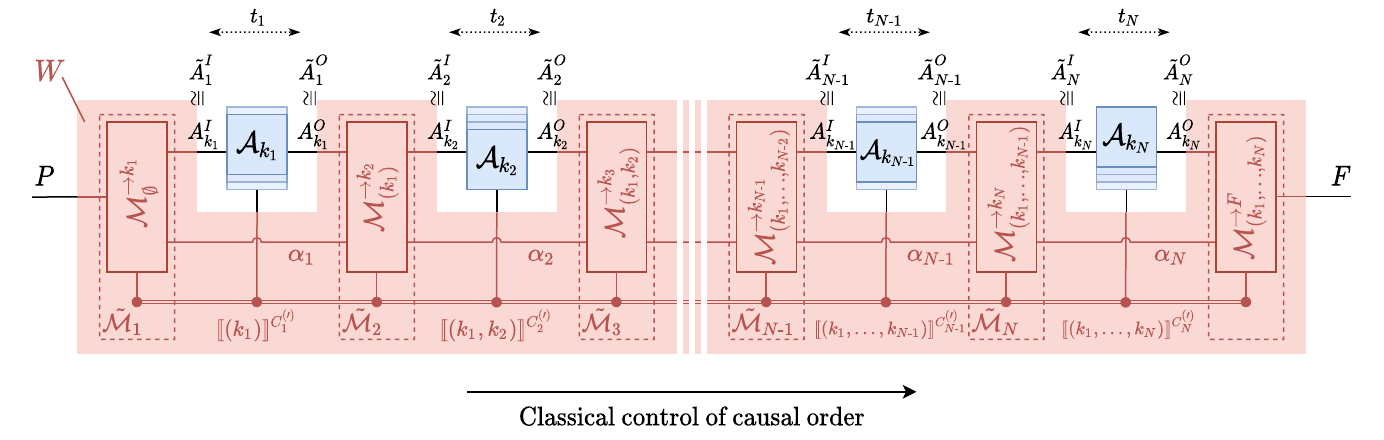}
	\end{center}
	\caption{Another possible representation of a QC-CC, equivalent to Fig.~\ref{fig:QCCCv1}. Here we show explicitly the transmission of the information about the causal order, established dynamically and stored on the fly in the states $[\![(k_1,\ldots,k_n)]\!]^{C_n^{(\prime)}}$ of some control system. The double-stroke lines indicate that this information is classical. This information is used to control which external operation $\A_{k_n}$ is to be applied at each time slot $t_n$, thus defining the joint operation $\tilde \A_n$ of Eq.~\eqref{eq:tilde_An_QCCC} on the target and control systems. It is also used to control the internal circuit operations $\M_{\emptyset}^{\to k_{1}}$, $\M_{(k_1,\ldots,k_n)}^{\to k_{n+1}}$, and $\M_{(k_1,\ldots,k_N)}^{\to F}$, defining the joint operations $\tilde \M_{1}$, $\tilde \M_{n+1}$, and $\tilde \M_{N}$ of Eqs.~\eqref{eq:tilde_Mn1_QCCC}--\eqref{eq:tilde_MN_QCCC} on the target, ancillary and control systems.}
	\label{fig:QCCCv2}
\end{figure*}

At each time slot $t_n$, we first introduce the ``generic'' input and output spaces $\HS^{\tilde A_n^I}$ and $\HS^{\tilde A_n^O}$ (with tildes), isomorphic to the $\HS^{A_{k_n}^I}$ and $\HS^{A_{k_n}^O}$ spaces, respectively.
We can then formally ``identify'' each $\HS^{A_{k_n}^I}$ with $\HS^{\tilde A_n^I}$ and each $\HS^{A_{k_n}^O}$ with $\HS^{\tilde A_n^O}$, and write the external operations $\A_{k_n}: \L(\HS^{A_{k_n}^I}) \to \L(\HS^{A_{k_n}^O})$ as operations of the form $\tilde\A_{k_n}: \L(\HS^{\tilde A_n^I}) \to \L(\HS^{\tilde A_n^O})$, with Choi matrices $\tilde A_{k_n} \in \L(\HS^{\tilde A_n^I\tilde A_n^O})$, and the internal circuit operations $\M_{(k_1,\ldots,k_n)}^{\to k_{n+1}}: \L(\HS^{A_{k_n}^O \alpha_n}) \to \L(\HS^{A_{k_{n+1}}^I \alpha_{n+1}})$ of a QC-CC as $\tilde\M_{(k_1,\ldots,k_n)}^{\to k_{n+1}}: \L(\HS^{\tilde A_n^O \alpha_n}) \to \L(\HS^{\tilde A_{n+1}^I \alpha_{n+1}})$, with Choi matrices $\tilde M_{(k_1,\ldots,k_n)}^{\to k_{n+1}}\in \L(\HS^{\tilde A_n^O \alpha_n\tilde A_{n+1}^I \alpha_{n+1}})$.%
\footnote{More formally, $\tilde\A_{k_n} \coloneqq {\cal I}^{A_{k_n}^O \to \tilde A_n^O} \circ \A_{k_n} \circ {\cal I}^{\tilde A_n^I \to A_{k_n}^I}$ and $\tilde\M_{(k_1,\ldots,k_n)}^{\to k_{n+1}} \coloneqq {\cal I}^{A_{k_{n+1}}^I \alpha_{n+1} \to \tilde A_{n+1}^I \alpha_{n+1}} \circ \M_{(k_1,\ldots,k_n)}^{\to k_{n+1}} \circ {\cal I}^{\tilde A_n^O \alpha_n \to A_{k_n}^O \alpha_n}$, where ${\cal I}^{X\to X'}$ is the ``identity'' map that relates the computational basis states of two isomorphic Hilbert spaces $\HS^X$ and $\HS^{X'}$, cf.\ Footnote~\ref{footnote:isomorphism_identity}.
In terms of Choi matrices, $\tilde A_{k_n} = \big( \dketbra{\id}{\id}^{A_{k_n}^I\tilde A_n^I} \otimes \dketbra{\id}{\id}^{A_{k_n}^O\tilde A_n^O} \big) * A_{k_n}$ and $\tilde M_{(k_1,\ldots,k_n)}^{\to k_{n+1}} = \big( \dketbra{\id}{\id}^{A_{k_n}^O\tilde A_n^O} \otimes \dketbra{\id}{\id}^{A_{k_{n+1}}^I\tilde A_{n+1}^I} \big) * M_{(k_1,\ldots,k_n)}^{\to k_{n+1}}$. \label{footnote:generic_IO_spaces}}
Similarly, we write $\M_\emptyset^{\to k_1}$ and $\M_{(k_1,\ldots,k_N)}^{\to F}$ as $\tilde\M_\emptyset^{\to k_1}: \L(\HS^P) \to \L(\HS^{\tilde A_1^I \alpha_1})$ and $\tilde\M_{(k_1,\ldots,k_N)}^{\to F}: \L(\HS^{\tilde A_N^O \alpha_N}) \to \L(\HS^F)$, with respective Choi matrices $\tilde M_\emptyset^{\to k_1}\in \L(\HS^{P\tilde A_1^I \alpha_1})$ and $\tilde M_{(k_1,\ldots,k_N)}^{\to F}\in \L(\HS^{\tilde A_N^O \alpha_NF})$.

\medskip

This allows us, at each time slot $t_n$ (for $1 \le n \le N$), to then embed the external operations $\tilde\A_{k_n}$ into some ``larger'' conditional operations $\tilde\A_n$ which use the control system to apply the correct $\tilde\A_{k_n}$:
\begin{align}
\tilde \A_n \coloneqq & \sum_{(k_1,\ldots,k_n)} \tilde\A_{k_n} \otimes \bigpi_{(k_1,\ldots,k_n)}^{C_n \to C_n'}: \notag \\[-1mm]
& \hspace{20mm} \L(\HS^{\tilde A_n^I C_n}) \to \L(\HS^{\tilde A_n^O C_n'}), \label{eq:tilde_An_QCCC}
\end{align}
where $\bigpi_{(k_1,\ldots,k_n)}^{C_n \to C_n'}$ is the (classical) map that projects the control system onto the state $[\![(k_1,\ldots,k_n)]\!]^{C_n}$, while re-labelling the control system $C_n$ to $C_n'$.
The corresponding Choi matrix of $\tilde\A_n$ is
\begin{align}
\tilde A_n = & \sum_{(k_1,\ldots,k_n)} \tilde A_{k_n} \otimes [\![(k_1,\ldots,k_n)]\!]^{C_n} \otimes [\![(k_1,\ldots,k_n)]\!]^{C_n'} \notag \\[-1mm]
& \hspace{25mm} \in \L(\HS^{\tilde A_n^I C_n\tilde A_n^O C_n'}). \label{eq:Choi_An_QCCC}
\end{align}

Similarly, we can embed the internal circuit operations $\tilde\M_{(k_1,\ldots,k_n)}^{\to k_{n+1}}$ into some ``larger'' operations that also involve the control systems, as shown in Fig.~\ref{fig:QCCCv2}.
These enlarged operations, unlike the $\tilde\M_{(k_1,\ldots,k_n)}^{\to k_{n+1}}$, are deterministic (i.e., CPTP) operations since the probabilistic choice of outcome $k_{n+1}$ is now encoded in the (classical) correlations between the control system and the joint target-ancilla system.
More precisely, we now have (for $1 \le n \le N-1$)
\begin{align}
\tilde \M_{n+1} \coloneqq & \sum_{(k_1,\ldots,k_n,k_{n+1})} \tilde\M_{(k_1,\ldots,k_n)}^{\to k_{n+1}} \otimes \bigpi_{(k_1,\ldots,k_n),k_{n+1}}^{C_n' \to C_{n+1}}: \notag \\
& \hspace{5mm} \L(\HS^{\tilde A_n^O \alpha_nC_n'}) \to \L(\HS^{\tilde A_{n+1}^I \alpha_{n+1}C_{n+1}}), \label{eq:tilde_Mn1_QCCC}
\end{align}
where $\bigpi_{(k_1,\ldots,k_n),k_{n+1}}^{C_n' \to C_{n+1}}: \L(\HS^{C_n'}) \to \L(\HS^{C_{n+1}})$ is the (classical) map that projects the control system onto $[\![(k_1,\ldots,k_n)]\!]^{C_n'}$ (the state just after the conditional operation $\tilde \A_n$, with a prime) and updates it to $[\![(k_1,\ldots,k_n,k_{n+1})]\!]^{C_{n+1}}$ (the state of the control system just before the next conditional operation $\tilde \A_{n+1}$, with no prime).
Likewise, the edge cases of the first and last operations are now
\begin{align}
\tilde \M_1 \coloneqq  & \sum_{k_1} \tilde\M_\emptyset^{\to k_1} \otimes \bigpi_{\emptyset,k_1}^{\emptyset \to C_1}: \notag \\
& \hspace{15mm} \L(\HS^P) \to \L(\HS^{\tilde A_1^I \alpha_1C_1}) \label{eq:tilde_M0_QCCC}
\end{align}
and
\begin{align}
\tilde \M_{N+1} \coloneqq & \sum_{(k_1,\ldots,k_N)} \tilde\M_{(k_1,\ldots,k_N)}^{\to F} \otimes \bigpi_{(k_1,\ldots,k_N),F}^{C_N' \to \emptyset}: \notag \\
& \hspace{10mm} \L(\HS^{\tilde A_N^O \alpha_NC_N'}) \to \L(\HS^F), \label{eq:tilde_MN_QCCC}
\end{align}
where $\bigpi_{\emptyset,k_1}^{\emptyset \to C_1}$ and $\bigpi_{(k_1,\ldots,k_N),F}^{C_N' \to \emptyset}$ are the maps that create the initial control states $[\![(k_1)]\!]^{C_1}$ and that project onto the final control states $[\![(k_1,\ldots,k_N)]\!]^{C_N'}$, respectively.

We can thus see explicitly that the control system $C_n$ controls which external operation $\A_{k_n}$ is applied at time slot $t_n$ (and hence the causal order) as well as the internal operations $\M_{(k_1,\ldots,k_n)}^{\to k_{n+1}}$, and that it does so in a classical manner.
Indeed, the internal operations cannot create any entanglement between the control system and the target or ancillary systems; instead, there is only ever classical correlation between the (classical) state of the control system and the other systems.
This justifies, in particular, the terminology of QC-CC.

The Choi matrices of the internal operations, for completeness, are
\begin{align}
\tilde M_1 = & \sum_{k_1} \tilde M_\emptyset^{\to k_1} \otimes [\![(k_1)]\!]^{C_1} \quad \in \L(\HS^{P\tilde A_1^I \alpha_1C_1}), \label{eq:def_tilde_M1_QCCC}\\[1mm]
\tilde M_{n+1} = & \sum_{(k_1,\ldots,k_n,k_{n+1})} \tilde M_{(k_1,\ldots,k_n)}^{\to k_{n+1}} \otimes [\![(k_1,\ldots,k_n)]\!]^{C_n'} \notag \\[-3mm]
& \hspace{30mm} \otimes [\![(k_1,\ldots,k_n,k_{n+1})]\!]^{C_{n+1}} \notag \\[1mm]
& \hspace{14mm} \in \L(\HS^{\tilde A_n^O \alpha_nC_n'\tilde A_{n+1}^I \alpha_{n+1}C_{n+1}}), \label{eq:Choi_Mn1_QCCC}\\[2mm]
\tilde M_{N+1} =& \sum_{(k_1,\ldots,k_N)} \tilde M_{(k_1,\ldots,k_N)}^{\to F} \otimes [\![(k_1,\ldots,k_N)]\!]^{C_N'} \notag \\[-2mm]
& \hspace{32mm} \in \L(\HS^{\tilde A_N^O \alpha_NC_N'F}). \label{eq:def_tilde_MN1_QCCC}
\end{align}

\medskip

The TP conditions for the internal operations of a QC-CC previously given in Eqs.~\eqref{eq:TP_constr_QCCC_1}--\eqref{eq:TP_constr_QCCC_N} (in terms of the Choi matrices $M_\emptyset^{\to k_1}$, $M_{(k_1,\ldots,k_n)}^{\to k_{n+1}}$ and $M_{(k_1,\ldots,k_N)}^{\to F}$ of the corresponding maps) are readily recovered in this alternative formulation by imposing that the enlarged operations  $\tilde \M_n$ are TP (on their effective input spaces) and that they preserve the probabilities for a given order of the thus-far applied external operations to be realised; see Appendix~\ref{app:proof_charact_QCCCs}.

\medskip

Finally, let us check that this formulation of QC-CCs is indeed equivalent to that given in the previous section.
Note first that the operations $\tilde \A_n$ and $\tilde \M_n$ described above are applied in a well-defined order. 
The global induced map (in its Choi version) is then obtained, similarly to Eq.~\eqref{eq:induced_M_QCFO} for the QC-FO case, by link-multiplying all these operations:
\begin{align}
M & = \tilde M_1 * \tilde A_1 * \tilde M_2 * \cdots * \tilde M_N * \tilde A_N * \tilde M_{N+1} \notag \\[1mm]
& = \sum_{(k_1,\ldots,k_N)} \tilde M_\emptyset^{\to k_1} * \tilde A_{k_1} * \tilde M_{(k_1)}^{\to k_2} * \cdots \notag \\[-4mm]
& \hspace{25mm} \cdots * \tilde M_{(k_1,\ldots,k_{N-1})}^{\to k_N} * \tilde A_{k_N} * \tilde M_{(k_1,\ldots,k_N)}^{\to F} \notag \\[1mm]
& = \sum_{(k_1,\ldots,k_N)} M_\emptyset^{\to k_1} * A_{k_1} * M_{(k_1)}^{\to k_2} * \cdots \notag \\[-4mm]
& \hspace{25mm} \cdots * M_{(k_1,\ldots,k_{N-1})}^{\to k_N} * A_{k_N} * M_{(k_1,\ldots,k_N)}^{\to F} \notag \\[1mm]
& \hspace{35mm} \in \L(\HS^{PF}) \label{eq:Choi_M_QCCC_v2}
\end{align}
where the second equality is obtained by ``contracting'' all control systems in the link products (in particular, by exploiting that $[\![(k_1,\ldots,k_n)]\!]^{C_n} * [\![(k_1',\ldots,k_n')]\!]^{C_n} = \delta_{k_1,k_1'}\cdots\delta_{k_n,k_n'}$, with $\delta$ the Kronecker delta), and where our formal identification (via the appropriate isomorphism, see Footnote~\ref{footnote:generic_IO_spaces}) of the external operations' input and output spaces $\HS^{A_{k_n}^I}$ and $\HS^{A_{k_n}^O}$ with the generic spaces $\HS^{\tilde A_n^I}$ and $\HS^{\tilde A_n^O}$ at each time slot $t_n$ allowed us, in the last line, to remove the tildes and obtain the third equality.

We thus recover Eq.~\eqref{eq:induced_M_QCCC} from the previous description of QC-CCs, and consequently also the same process matrix description of our QC-CC as in Proposition~\ref{prop:descr_W_QCCC}, and the same characterisation of QC-CC process matrices as in Proposition~\ref{prop:charact_W_QCCC}.

\subsubsection{``Purifying'' the internal circuit operations}

With the goal of progressing towards circuits with quantum, rather than classical, control of causal order, we make here one further simplification.
We will show that it suffices to consider only ``pure'' QC-CCs, in which all the internal circuit operations are isometries, and to consider the action of such QC-CCs when pure external operations are inserted in them.
This will make it significantly easier to describe coherence between the control and target/ancillary systems, which will be a crucial aspect of the shift to quantum control.

To this end, let us note that since we do not make any particular assumption about the ancillary Hilbert spaces $\HS^{\alpha_n}$ (e.g., about their dimension), they can be used to ``purify''%
\footnote{In particular, according to Stinespring's dilation theorem~\cite{stinespring55}, for any CP map $\M:\L(\HS^X)\to\L(\HS^Y)$ there exists an ancillary Hilbert space $\HS^\alpha$ and a linear operator $V:\HS^X \to \HS^{Y\alpha}$ such that $\M(\rho) = \Tr_\alpha(V\rho V^\dagger) \ \forall \rho$.
In the case of the generalised quantum circuits we consider here, the ancillary ``purifying'' systems can be carried through the circuit via the ancillary systems $\HS^{\alpha_n}$ before being traced out at the very end.}
the operations $\M_{(k_1,\ldots,k_{n-1})}^{\to k_n}$ for $1 \le n \le N$.
Without loss of generality, we can thus assume they consist of the application of just one Kraus operator, which we shall denote $V_{(k_1,\ldots,k_{n-1})}^{\to k_n}: \HS^{A_{k_{n-1}}^O\alpha_{n-1}} \to \HS^{A_{k_n}^I\alpha_n}$ (so that $\M_{(k_1,\ldots,k_{n-1})}^{\to k_n}(\varrho) = V_{(k_1,\ldots,k_{n-1})}^{\to k_n} \varrho V_{(k_1,\ldots,k_{n-1})}^{\to k_n\, \dagger}$); the Choi representations of the operations are then simply
\begin{align}
M_{(k_1,\ldots,k_{n-1})}^{\to k_n} = \dketbra{V_{(k_1,\ldots,k_{n-1})}^{\to k_n}}{V_{(k_1,\ldots,k_{n-1})}^{\to k_n}},
\end{align}
where $\dket{V_{(k_1,\ldots,k_{n-1})}^{\to k_n}} \in \HS^{A_{k_{n-1}}^O\alpha_{n-1}A_{k_n}^I\alpha_n}$ (or $\dket{V_\emptyset^{\to k_1}} \in \HS^{PA_{k_1}^I\alpha_1}$ for $n=1$) is the Choi vector representation of $V_{(k_1,\ldots,k_{n-1})}^{\to k_n}$, as introduced in Sec.~\ref{subsec:preliminaries}.
Similarly, for the final operations $\M_{(k_1,\ldots,k_N)}^{\to F}$, one can introduce an ancillary Hilbert space $\HS^{\alpha_{F}}$ so as to purify these operations and write them in terms of only one Kraus operator $V_{(k_1,\ldots,k_N)}^{\to F}$, before tracing out the ancillary system in $\HS^{\alpha_{F}}$. 
Without loss of generality we can thus write
\begin{align}
M_{(k_1,\ldots,k_N)}^{\to F} = \Tr_{\alpha_{F}} \dketbra{V_{(k_1,\ldots,k_N)}^{\to F}}{V_{(k_1,\ldots,k_N)}^{\to F}},
\end{align}
with $\dket{V_{(k_1,\ldots,k_N)}^{\to F}} \in \HS^{A_{k_N}^O\alpha_NF\alpha_{F}}$.

It will similarly be convenient to assume that the external operations $\A_k$ correspond to the application of a single Kraus operator. 
In a slight, but generally unambiguous, conflict of notation we will reuse the notation $A_k$ for this Kraus operator, with Choi vector representation $\dket{A_k} \in \HS^{A_k^IA_k^O}$ (so that the Choi matrix of the map $\A_k$ is now $\dketbra{A_k}{A_k} \in \L(\HS^{A_k^IA_k^O})$). 
The general case of multiple Kraus operators can then easily be recovered by summing what we would get for different combinations of Kraus operators for each $\A_k$.

With these simplifications, the calculation of the induced map $\M$ following Eq.~\eqref{eq:induced_M_QCCC} is made significantly easier.
More importantly, when we consider a quantum control system it will allow us to directly study a pure global map $V: \HS^P \to \HS^{F\alpha_F}$, with Choi vector $\dket{V} \in \HS^{PF\alpha_F}$ as a function of all pure external and internal operations (and only trace out the $\HS^{\alpha_F}$ ancillary system at the very end).

\subsection{Turning the classical control into a coherent control of causal order}
\label{subsec:QCQC_descr}

The reformulation of QC-CCs above provides a clearer view of how to proceed towards quantum control of causal order, namely by turning the classical control system into a quantum one which can be used to coherently control the internal circuit operations.
In order to capture the most general form of quantum control, however, it is necessary to make one crucial adjustment to the control system.
Recall that in the case of a classical control, the state $[\![(k_1,\ldots,k_n)]\!]^{C_n^{(\prime)}}$ of the control system was used to keep track of the whole history of which operations had been applied so far. 
For a quantum control we will instead use the control system to record only which operations have already been applied and to encode which operation should be applied at a given time slot, but, importantly, we will not require that it keep track of the order in which the previous operations were applied.	

For these circuits to define valid supermaps, recall that we need to ensure that each external operation is applied once and only once.
The unordered set $\{k_1,\dots,k_{n-1}\}$ of operations already applied is thus the minimal information needed to ensure that, at each time slot $t_n$ and in each coherent ``branch'' of the computation, an operation is applied that has not previously been used in that branch.
This relaxed control system will notably allow, for example, for different orders $(k_1,\ldots,k_{n-1})$ and $(k_1',\ldots,k_{n-1}')$ corresponding to the same set $\K_{n-1} = \{k_1,\ldots,k_{n-1}\} = \{k_1',\ldots,k_{n-1}'\}$ to ``interfere'' and thus make the causal order indefinite.

\medskip

In what follows, it will be useful to adopt the following notation.
We generically denote by $\K_n$ a subset of $\N$ with $n$ elements (with $0 \le n \le N$), so that in particular $K_0 = \emptyset$ and $\K_N = \N$.
We identify singletons with their single element, so as to write, for instance $\K \backslash k = \K \backslash \{k\}$, $k_N = \{k_N\} = \N \backslash \K_{N-1}$, or $\K_{n-1}\cup k_n = \K_{n-1}\cup\{k_n\} = \K_n$.

\subsubsection{General description}

In order to define quantum circuits with quantum control of causal order (QC-QCs), we thereby consider generalised quantum circuits of the form represented in Fig.~\ref{fig:QCQC}.
As anticipated by the above discussions, we exploit a quantum control system in the Hilbert spaces $\HS^{C_n^{(\prime)}}$, which now have computational basis states of the form $\ket{\K_{n-1},k_n}^{C_n^{(\prime)}}$, where $\K_{n-1}$ specifies the (unordered) set of $n-1$ operations that have already been applied before the time slot $t_n$, and $k_n \notin \K_{n-1}$ labels the operation to be applied at time slot $t_n$.
This control system thus coherently controls coherently both the application of the external operations $A_{k_n}$ (which, recall, we now identify with a single Kraus operator) as well as the pure operations $V_{\K_{n-1},k_n}^{\to k_{n+1}}$ within the internal circuit operations.

\begin{figure*}[t]
	\begin{center}
	\includegraphics[scale=.75]{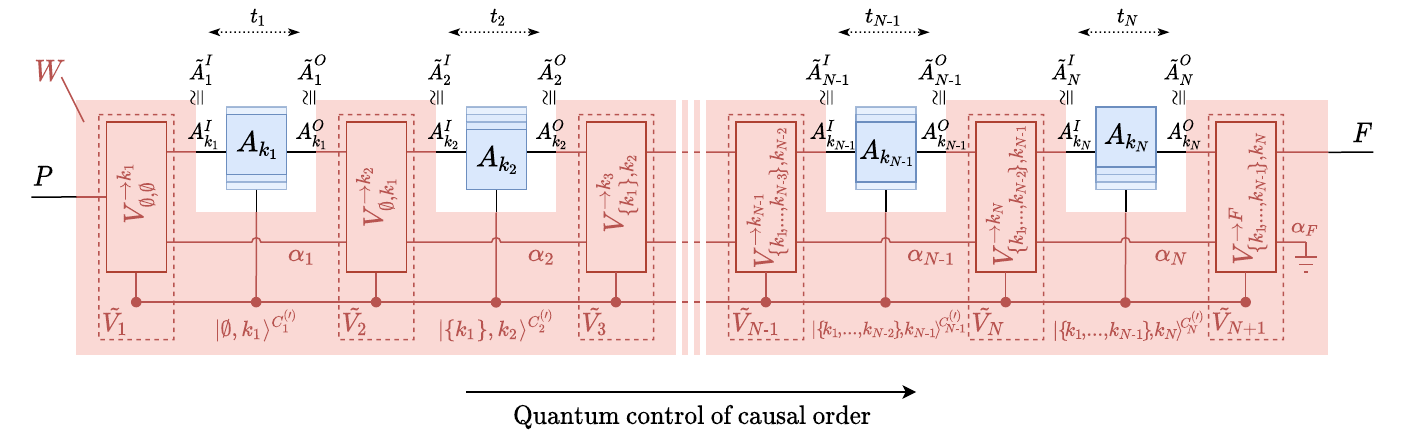}
	\end{center}
	\caption{Quantum circuit with quantum control of causal order (QC-QC). 
	We replaced the classical control system of Fig.~\ref{fig:QCCCv2} by a quantum control system with basis states $\ket{\{k_1,\ldots,k_{n-1}\},k_n}^{C_n}$, which only store information about which operations ($\{k_1,\ldots,k_{n-1}\}$) have already been applied (but not about their order) and the currently performed operation ($k_n$).
	(Note that in contrast to the previous figures, the ``boxes'' are labelled by linear operators, rather than linear CP maps).
	We illustrate here the component $\ket{w_{(k_1, \ldots, k_N,F)}}$ of the process, corresponding to the order $(k_1, \ldots, k_N)$---which is coherently superposed with other components, corresponding to different orders, in order to obtain the process matrix $W$ from the internal operations $V_{\K_{n-1},k_n}^{\to k_{n+1}}$ of the circuit; see Proposition~\ref{prop:descr_W_QCQC}.
	}
	\label{fig:QCQC}
\end{figure*}

To achieve this, we work, as in the previous subsection, with the ``generic'' input and output spaces $\HS^{\tilde A_n^I}$ and $\HS^{\tilde A_n^O}$, isomorphic to $\HS^{A_{k_n}^I}$ and $\HS^{A_{k_n}^O}$, respectively.
The external operations $A_{k_n}: \HS^{A_{k_n}^I} \to \HS^{A_{k_n}^O}$  can then be rewritten as operations on these spaces as $\tilde A_{k_n}: \HS^{\tilde A_n^I} \to \HS^{\tilde A_n^O}$ with Choi vectors $\dket{\tilde A_{k_n}}\in\HS^{\tilde A_n^I \tilde A_n^O}$.%
\footnote{\label{fn:QCQCtilde_translation}Analogously to Footnote~\ref{footnote:generic_IO_spaces}, one more formally has, in terms of the corresponding Choi vectors, $\dket{\tilde A_{k_n}} = \big( \dket{\id}^{A_{k_n}^I\tilde A_n^I} \otimes \dket{\id}^{A_{k_n}^O\tilde A_n^O} \big) * \dket{A_{k_n}}$. Similarly, for the internal operations, one has $\dket{\tilde V_{\K_{n-1},k_n}^{\to k_{n+1}}} = \big( \dket{\id}^{A_{k_n}^O\tilde A_n^O} \otimes \dket{\id}^{A_{k_{n+1}}^I\tilde A_{n+1}^I} \big) * \dket{V_{\K_{n-1},k_n}^{\to k_{n+1}}}$.}
These are then embedded into larger conditional operations $\tilde A_n$ (for $1\le n \le N$) which use the control system to apply the correct $\tilde A_{k_n}$ at time slot $t_n$ (cf.\ Eq.~\eqref{eq:tilde_An_QCCC}):%
\footnote{As in Footnote~\ref{footnote:M_slots}, one could also consider $M>N$ time slots, and fill the ``empty'' slots by trivial identity operations; this would in particular allow one to more faithfully describe the ``asymmetric'' version of the quantum switch (see Sec.~\ref{subsubsec:QSwitch}) considered, e.g., in Ref.~\cite{oreshkov19}.}
\begin{align}
\tilde A_n \coloneqq & \sum_{\K_{n-1},k_n} \tilde A_{k_n} \otimes \ket{\K_{n-1},k_n}^{C_n'}\bra{\K_{n-1},k_n}^{C_n}: \notag \\[-1mm]
& \hspace{20mm} \HS^{\tilde A_n^I C_n} \to \HS^{\tilde A_n^O C_n'}, \label{eq:tilde_An_QCQC}
\end{align}
where here, as in the remainder of what follows, summations of this form assume $k_n \notin \K_{n-1}$.
The corresponding Choi vector of $\tilde A_n$ is 
\begin{align}
\dket{\tilde A_n} = & \sum_{\K_{n-1},k_n} \dket{\tilde A_{k_n}} \otimes \ket{\K_{n-1},k_n}^{C_n} \otimes \ket{\K_{n-1},k_n}^{C_n'} \notag \\[-1mm]
& \hspace{20mm} \in \HS^{\tilde A_n^I C_n\tilde A_n^O C_n'}. \label{eq:Choi_tilde_An_QCQC}
\end{align}

\medskip

In place of the CP maps $\M_{(k_1,\ldots,k_n)}^{\to k_{n+1}}$ for QC-CCs, the internal circuit operations now control (coherently) the application of ``pure'' operators $V_{\K_{n-1},k_n}^{\to k_{n+1}}: \HS^{A_{k_n}^O \alpha_n} \to \HS^{A_{k_{n+1}}^I \alpha_{n+1}}$ (for $1 \le n \le N-1$, with $k_n \notin \K_{n-1}, k_{n+1}\notin \K_{n-1} \cup k_n$).
These operators depend on both $\K_{n-1}$ and $k_n$, and take the output of $A_{k_n}$ (along with the ancillary system in $\HS^{\alpha_n}$) to the input of $A_{k_{n+1}}$ (and the ancillary system in $\HS^{\alpha_{n+1}}$).
Similarly, the first and last internal operations control the operators $V_{\emptyset,\emptyset}^{\to k_1}: \HS^P \to \HS^{A_{k_1}^I \alpha_1}$ and $V_{\K_{N-1},k_N}^{\to F}: \HS^{A_{k_N}^O \alpha_N} \to \HS^{F\alpha_F}$ (with $k_N = \N \backslash \K_{N-1}$).
As with the external operations, we will work with the translation of these operators into the generic input and output spaces $\HS^{\tilde A_n^I}$ and $\HS^{\tilde A_n^O}$, denoted $\tilde V_{\emptyset,\emptyset}^{\to k_1}: \HS^P \to \HS^{\tilde A_1^I \alpha_1}$, $\tilde V_{\K_{n-1},k_n}^{\to k_{n+1}}: \HS^{\tilde A_n^O \alpha_n} \to \HS^{\tilde A_{n+1}^I \alpha_{n+1}}$ (for $1 \le n \le N-1$) and $\tilde V_{\K_{N-1},k_N}^{\to F}: \HS^{\tilde A_N^O \alpha_N} \to \HS^{F\alpha_F}$, and with respective Choi vectors  $\dket{\tilde V_{\emptyset,\emptyset}^{\to k_1}} \in \HS^{P\tilde A_1^I \alpha_1}$, $\dket{\tilde V_{\K_{n-1},k_n}^{\to k_{n+1}}} \in \HS^{\tilde A_n^O \alpha_n\tilde A_{n+1}^I \alpha_{n+1}}$ and $\dket{\tilde V_{\K_{N-1},k_N}^{\to F}} \in \HS^{\tilde A_N^O \alpha_NF\alpha_F}$.

\medskip

The circuit, as shown in Fig.~\ref{fig:QCQC}, is then obtained by embedding these operations into larger operations that involve the control system.
More precisely, before the time slot $t_1$, the circuit transforms the input state into a state that is sent coherently to all operations $A_{k_1}$ and, possibly, also to some ancillary system in $\HS^{\alpha_1}$, while accordingly attaching the control state $\ket{\emptyset,k_1}^{C_1}$ to each component of the superposition.
That is, instead of the operation $\tilde \M_1$ in the QC-CC case, the circuit now applies a (pure) operation of the form
\begin{align}
\tilde V_1 & \coloneqq \sum_{k_1} \tilde V_{\emptyset,\emptyset}^{\to k_1} \otimes \ket{\emptyset,k_1}^{C_1} : \quad \HS^P \to \HS^{\tilde A_1^I \alpha_1C_1}. \label{eq:def_V1_QCQC}
\end{align}

Between the time slots $t_n$ and $t_{n+1}$, for $1 \le n \le N-1$, the circuit acts coherently on the target, ancillary, and control systems.
It coherently controls the operation $V_{\K_{n-1},k_n}^{\to k_{n+1}}$ to apply depending on the state $\ket{\K_{n-1},k_n}^{C_n'}$ of the control system, before coherently sending the target system to all remaining $A_{k_{n+1}}$'s (with $k_{n+1} \notin \K_{n-1}\cup k_n$) and, possibly, an ancillary system in $\HS^{\alpha_{n+1}}$, while updating the control system to $\ket{\K_{n-1}\cup k_n,k_{n+1}}^{C_{n+1}}$, thereby encoding the next operation to apply, $k_{n+1}$ and erasing the information about the specific previous operation $k_n$ (among all the previously applied operations) by just recording the whole set of previously applied operations $\K_n \coloneqq \K_{n-1}\cup k_n$.
Formally, the circuit applies the operation 
\begin{align}
\tilde V_{n+1} \coloneqq & \sum_{\substack{\K_{n-1}, \\ k_n,k_{n+1}}} \tilde V_{\K_{n-1},k_n}^{\to k_{n+1}} \otimes \ket{\K_{n-1}\cup k_n,k_{n+1}}^{C_{n+1}} \notag \\[-8mm]
& \hspace{50mm} \bra{\K_{n-1},k_n}^{C_n'}: \notag \\[1mm]
& \hspace{15mm} \HS^{\tilde A_n^O \alpha_nC_n'} \to \HS^{\tilde A_{n+1}^I \alpha_{n+1}C_{n+1}}, \label{eq:tilde_Vn1_QCQC}
\end{align}
where the sum assumes, extending our established convention, that $k_n, k_{n+1}\in \N \backslash \K_{n-1}$ with $k_n \neq k_{n+1}$.

Finally, after time slot $t_N$, the application of the operations $V_{\K_{N-1},k_N}^{\to F}$ (with $\K_{N-1} = \N\backslash k_N$) is coherently controlled on the control system, taking the output of $A_{k_N}$, together with the ancillary state in $\HS^{\alpha_N}$, to the global output of the circuit in $\HS^F$ and, possibly, an ancillary system in $\HS^{\alpha_F}$.
The circuit thus applies the operation
\begin{align}
\tilde V_{N+1} & \coloneqq \sum_{k_N} \tilde V_{\N\backslash k_N,k_N}^{\to F} \otimes \bra{\N\backslash k_N,k_N}^{C_N'}: \notag \\[-2mm]
& \hspace{20mm} \HS^{\tilde A_N^O \alpha_NC_N'} \to \HS^{F\alpha_F}. \label{eq:def_VN1_QCQC}
\end{align}
The final ancillary system in $\HS^{\alpha_F}$ is subsequently discarded by the circuit.
Note that, in this final operation $\tilde V_{N+1}$ the control system does not need to be updated as, with $F$ replacing $A_{k_{N+1}}^I$, it would always be in the state $\ket{\N,F}^{C_{N+1}}$.
Indeed, this is crucial to allowing different causal histories to interfere within the QC-QC.
Moreover, this highlights the fact that, at the end of the circuit, each external operation has been applied exactly once, as required if the circuit is to give us a valid quantum supermap.

The Choi vectors of the operators, for completeness, are
\begin{align}
& \dket{\tilde V_1} = \sum_{k_1} \dket{\tilde V_{\emptyset,\emptyset}^{\to k_1}} \otimes \ket{\emptyset,k_1}^{C_1} \quad \in \HS^{P\tilde A_1^I \alpha_1C_1}, \label{eq:Choi_tilde_V1_QCQC} \\[1mm]
& \dket{\tilde V_{n+1}} = \sum_{\substack{\K_{n-1}, \\ k_n,k_{n+1}}} \dket{\tilde V_{\K_{n-1},k_n}^{\to k_{n+1}}} \otimes \ket{\K_{n-1},k_n}^{C_n'} \notag \\[-7mm]
& \hspace{40mm} \otimes \ket{\K_{n-1}\cup k_n,k_{n+1}}^{C_{n+1}} \notag \\[2mm]
& \hspace{30mm} \in \HS^{\tilde A_n^O \alpha_nC_n'\tilde A_{n+1}^I \alpha_{n+1}C_{n+1}}, \label{eq:Choi_tilde_Vn1_QCQC} \\[2mm]
& \dket{\tilde V_{N+1}} = \sum_{k_N} \dket{\tilde V_{\N\backslash k_N,k_N}^{\to F}} \otimes \ket{\N\backslash k_N,k_N}^{C_N'} \notag \\[-1mm]
& \hspace{35mm} \in \HS^{\tilde A_N^O \alpha_NC_N'F \alpha_F}. \label{eq:Choi_tilde_VN1_QCQC}
\end{align}
From these, the Choi matrices for the internal operations as CPTP maps (as considered in the previous sections), can be recovered as $\tilde M_n = \dketbra{\tilde V_n}{\tilde V_n}$ for $n \le N$, and $\tilde M_{N+1} = \Tr_{\alpha_{F}} \dketbra{\tilde V_{N+1}}{\tilde V_{N+1}}$.

\subsubsection{Trace-preserving conditions}

The TP conditions on the internal operations arise from the requirement that the operators $\tilde V_n : \varrho \mapsto \tilde V_n \varrho \tilde V_n^\dagger$ must act as isometries on their effective input spaces.
As for the previously considered classes of circuits, we will simply state the TP conditions here, while their full derivation is given in Appendix~\ref{app:proof_charact_QCQCs}.

To express the conditions in a compact form, let us first define, for all $1 \le n \le N$ and all $(k_1,\ldots,k_n)$, 
\begin{align}
\ket{w_{(k_1,\ldots,k_n)}} & \coloneqq \dket{V_{\emptyset,\emptyset}^{\to k_1}} \!*\! \dket{V_{\emptyset,k_1}^{\to k_2}} \!*\! \cdots \!*\! \dket{V_{\!\{k_1,\ldots,k_{n-2}\},k_{n-1}}^{\to k_n}} \notag \\
& \hspace{10mm} \in \HS^{PA_{\{k_1,\ldots,k_{n-1}\}}^{IO} A_{k_n}^I\alpha_n}, \label{eq:def_w_k1_kn_QCQC}
\end{align}
in terms of the Choi vectors of the operators $V_{\K_{n-1},k_n}^{\to k_{n+1}}$ (i.e., in the original, non-generic, Hilbert spaces) and, for all strict subsets $\K_{n-1}$ of $\N$ with $|\K_{n-1}|=n-1$ and all $k_n \in \N\backslash\K_{n-1}$,
\begin{align}
& \hspace{-3mm} \ket{w_{(\K_{n-1},k_n)}} \notag \\
& \coloneqq \!\!\!\!\sum_{\substack{(k_1,\ldots,k_{n-1}): \\ \{k_1,\ldots,k_{n-1}\} = \K_{n-1}}} \!\!\!\!\! \ket{w_{(k_1,\ldots,k_n)}} \quad \in \HS^{PA_{\K_{n-1}}^{IO} A_{k_n}^I\alpha_n}, \label{eq:def_w_Knm1_kn}
\end{align}
where the sum is taken over all ordered sequences $(k_1,\ldots,k_{n-1})$ of $\K_{n-1}$.
For the case of $n = N+1$, replacing $k_{N+1}$ by $F$, we similarly obtain, for all $(k_1,\ldots,k_N)$, the vectors $\ket{w_{(k_1,\ldots,k_N,F)}}$ and $\ket{w_{(\N,F)}} \in \HS^{PA_\N^{IO} F\alpha_F}$ (see Eqs.~\eqref{eq:W_QCQC}--\eqref{eq:def_w_k1_kN_F_QCQC} in Proposition~\ref{prop:descr_W_QCQC} below for explicit definitions).
Note that by construction we have $\ket{w_{(k_1,\ldots,k_n,k_{n+1})}} = \ket{w_{(k_1,\ldots,k_n)}} * \dket{V_{\{k_1,\ldots,k_{n-1}\},k_n}^{\to k_{n+1}}}$ and
\begin{align}
\ket{w_{(\K_n,k_{n+1})}} = \sum_{k_n \in \K_n} \ket{w_{(\K_n\backslash k_n,k_n)}} * \dket{V_{\K_n\backslash k_n,k_n}^{\to k_{n+1}}}. \label{eq:w_Kn_kn1_sum}
\end{align}

In terms of these vectors the TP conditions can then be written as
\begin{align}
& \hspace{2mm}\sum_{k_1} \Tr_{A_{k_1}^I\alpha_1} \ketbra{w_{(\emptyset,k_1)}}{w_{(\emptyset,k_1)}} = \id^P, \label{eq:TP_constr_QCQC_1}\\[1mm]
& \forall \, n=1,\dots,N-1,\ \forall \, \K_n,\notag \\
& \hspace{2mm}\sum_{k_{n+1} \notin \K_n} \!\! \Tr_{A_{k_{n+1}}^I\alpha_{n+1}} \ketbra{w_{(\K_n,k_{n+1})}}{w_{(\K_n,k_{n+1})}} \notag \\
& \hspace{4mm} = \!\! \sum_{k_n \in \K_n} \!\!\! \Tr_{\alpha_n} \! \ketbra{w_{(\K_n\backslash k_n,k_n)}}{w_{(\K_n\backslash k_n,k_n)}} \otimes \id^{A_{k_n}^O}, \label{eq:TP_constr_QCQC_n} \\[1mm]
& \text{and}\notag \\
& \hspace{2mm}\Tr_{F\alpha_F} \ketbra{w_{(\N,F)}}{w_{(\N,F)}} \notag \\
& \hspace{4mm} = \!\! \sum_{k_N \in \N} \!\!\! \Tr_{\alpha_N} \! \ketbra{w_{(\N\backslash k_N,k_N)}}{w_{(\N\backslash k_N,k_N)}} \otimes \id^{A_{k_N}^O}, \label{eq:TP_constr_QCQC_N}
\end{align}
where we note that, in the first condition, $\ket{w_{(\emptyset,k_1)}} = \ket{w_{(k_1)}} = \dket{V_{\emptyset,\emptyset}^{\to k_1}}$.

\medskip

The previous description of the process under consideration, as represented in Fig.~\ref{fig:QCQC} and with internal circuit operations $V_{\emptyset,\emptyset}^{\to k_1}, V_{\K_{n-1},k_n}^{\to k_{n+1}}, V_{\N\backslash k_N,k_N}^{\to F}$ giving vectors $\ket{w_{(\emptyset,k_1)}},\ket{w_{(\K_n\backslash k_n,k_n)}},\ket{w_{(\N\backslash k_N,k_N)}}$ satisfying the TP constraints of Eqs.~\eqref{eq:TP_constr_QCQC_1}--\eqref{eq:TP_constr_QCQC_N}, formally defines what we call a \emph{Quantum Circuit with Quantum Control of causal order} (QC-QC).

\subsubsection{Process matrix description}

To obtain the description of a QC-QC as a process matrix, we proceed analogously to the previous sections.
Indeed, note that as in the previous cases, the operations $\tilde V_n$ and $\tilde A_n$ are applied in a well-defined order.
The global operation $V: \HS^P \to \HS^{F\alpha_F}$ induced by the circuit (prior to tracing out $\HS^{\alpha_F}$) when the external operations $A_k$ are applied is obtained by composing all these operations $\tilde V_n$ and $\tilde A_n$ in that well-defined order.
Correspondingly, and similarly to the previous cases (see, e.g., Eqs.~\eqref{eq:induced_M_QCFO} and~\eqref{eq:Choi_M_QCCC_v2}), its Choi vector $\dket{V} \in \HS^{PF\alpha_F}$ is obtained by link-multiplying the Choi vectors of all these operations.
With the Choi vectors given by Eq.~\eqref{eq:Choi_tilde_An_QCQC} and Eqs.~\eqref{eq:Choi_tilde_V1_QCQC}--\eqref{eq:Choi_tilde_VN1_QCQC}, we obtain
\begin{align}
\dket{V} & = \dket{\tilde V_1} * \dket{\tilde A_1} * \dket{\tilde V_2} * \cdots * \dket{\tilde V_N} * \dket{\tilde A_N} * \dket{\tilde V_{N+1}} \notag \\
& = \sum_{(k_1,\ldots,k_N)} \dket{\tilde V_{\emptyset,\emptyset}^{\to k_1}} * \dket{\tilde A_{k_1}} * \dket{\tilde V_{\emptyset,k_1}^{\to k_2}} * \cdots \notag \\[-4mm]
& \hspace{25mm} \cdots * \dket{\tilde V_{\{k_1,\ldots,k_{N-2}\},k_{N-1}}^{\to k_N}} * \dket{\tilde A_{k_N}} \notag \\
& \hspace{48mm} * \dket{\tilde V_{\{k_1,\ldots,k_{N-1}\},k_N}^{\to F}} \notag \\[1mm]
& = \sum_{(k_1,\ldots,k_N)} \big( \dket{A_1} \otimes \cdots \otimes \dket{A_N} \big) * \ket{w_{(k_1,\ldots,k_N,F)}} \notag \\[1mm]
& = \big( \dket{A_1} \otimes \cdots \otimes \dket{A_N} \big) * \ket{w_{(\N,F)}} \ \in \HS^{PF\alpha_F} \label{eq:Choi_V_QCQC}
\end{align}
where the second equality is obtained by contracting the control systems (similarly to Eq.~\eqref{eq:Choi_M_QCCC_v2}), and the final two follow by identifying the external operations' Hilbert spaces with the corresponding generic ones (via the appropriate isomorphism, see Footnote~\ref{fn:QCQCtilde_translation}), reordering the terms in the link product (as in Eq.~\eqref{eq:induced_Mk1_kN_QCCC}), and rewriting the link product of internal operators in terms of the vectors $\ket{w_{(k_1,\ldots,k_N,F)}}$ and $\ket{w_{(\N,F)}} \in \HS^{PA_\N^{IO} F\alpha_F}$ defined in Eqs.~\eqref{eq:def_w_k1_kn_QCQC} and~\eqref{eq:def_w_Knm1_kn} (with $k_{N+1}$ replaced by $F$; cf.\ Eqs.~\eqref{eq:W_QCQC}--\eqref{eq:def_w_k1_kN_F_QCQC} below).

Analogous to the identification of the process matrix in the previous sections, we can identify $\ket{w_{(\N,F)}}$ as a ``process vector'' describing the QC-QC in the pure Choi representation prior to $\HS^{\alpha_F}$ being discarded.
In order to obtain the process matrix, we write the corresponding Choi matrix and trace out $\HS^{\alpha_F}$.
We thus obtain the following process matrix description for general QC-QCs:

\begin{proposition}[Process matrix description of QC-QCs] \label{prop:descr_W_QCQC}
The process matrix corresponding to the quantum circuit with quantum control of causal order depicted on Fig.~\ref{fig:QCQC} is
\begin{align}
W = & \Tr_{\alpha_{F}} \ketbra{w_{(\N,F)}}{w_{(\N,F)}} \notag \\
\text{with} \quad \ket{w_{(\N,F)}} \coloneqq & \sum_{(k_1,\ldots,k_N)} \ket{w_{(k_1,\ldots,k_N,F)}} \label{eq:W_QCQC}
\end{align}
and with
\begin{align}
& \ket{w_{(k_1,\ldots,k_N,F)}} \notag \\
& \quad \coloneqq \dket{V_{\emptyset,\emptyset}^{\to k_1}} * \dket{V_{\emptyset,k_1}^{\to k_2}} * \dket{V_{\{k_1\},k_2}^{\to k_3}} * \cdots \notag \\
& \qquad \quad \cdots * \dket{V_{\{k_1,\ldots,k_{N-2}\},k_{N-1}}^{\to k_N}} * \dket{V_{\{k_1,\ldots,k_{N-1}\},k_N}^{\to F}} \notag \\
& \hspace{40mm} \in \HS^{PA_\N^{IO} F\alpha_F}. \label{eq:def_w_k1_kN_F_QCQC}
\end{align}
\end{proposition}

Note that the process vector $\ket{w_{(\N,F)}}$ is a superposition of terms in each of which the target system is passed to each external operation exactly once (possibly in different orders).
This ensures that $W$ is linear in the operations, allowing us to obtain a valid quantum supermap, and reiterates the sense in which each operation is applied once and only once, even in the case where the order of application is placed in a superposition.
This moreover excludes, for instance, situations where a coherent control is used to control \emph{which} operations are applied (rather than their order), as considered, for example, in Refs.~\cite{oi03,abbott20,chiribella19,kristjansson19}---scenarios that indeed do not correspond to quantum supermaps.

\subsection{Characterisation}
\label{sec:QC-QC_characterisation}

The description of QC-QCs above allows us now to obtain the following characterisation of their process matrices.

%\begin{proposition}[Characterisation of QC-QCs] \label{prop:charact_W_QCQC}
%The process matrix $W \in \L(\HS^{PA_\N^{IO} F})$ of a quantum circuit with quantum control of causal order is such that there exist positive semidefinite matrices $W_{(\K_{n-1},k_n)} \in \L(\HS^{PA_{\K_{n-1}}^{IO} A_{k_n}^I})$, for all strict subsets $\K_{n-1}$ of $\N$ and all $k_n \in \N\backslash\K_{n-1}$, satisfying
%\begin{align}
%& \sum_{k_1 \in \N} \Tr_{A_{k_1}^I} W_{(\emptyset,k_1)} = \id^P, \notag \\
%& \forall \, \emptyset \subsetneq \K_n \subsetneq \N, \!\! \sum_{k_{n+1} \in \N \backslash \K_n} \!\!\! \Tr_{A_{k_{n+1}}^I} \!W_{(\K_n,k_{n+1})} \notag \\[-2mm]
%& \hspace{40mm} = \sum_{k_n \in \K_n} W_{(\K_n \backslash k_n,k_n)}\otimes \id^{A_{k_n}^O}, \notag \\[1mm]
%& \textup{and} \quad \Tr_F W = \sum_{k_N \in \N} W_{(\N \backslash k_N,k_N)}\otimes \id^{A_{k_N}^O}. \label{eq:charact_W_QCQC_decomp}
%\end{align}

%Conversely, any Hermitian matrix $W \in \L(\HS^{PA_\N^{IO} F})$ such that there exist positive semidefinite matrices $W_{(\K_{n-1},k_n)} \in \L(\HS^{PA_{\K_{n-1}}^{IO} A_{k_n}^I})$ for all $\K_{n-1} \subsetneq \N$ and $k_n \in \N\backslash\K_{n-1}$ satisfying Eq.~\eqref{eq:charact_W_QCQC_decomp} is the process matrix of a quantum circuit with quantum control of causal order.
%\end{proposition}

\begin{proposition}[Characterisation of QC-QCs] \label{prop:charact_W_QCQC}
The process matrix of a quantum circuit with quantum control of causal order is a positive semidefinite matrix $W \in \L(\HS^{PA_\N^{IO} F})$ such that there exist positive semidefinite matrices $W_{(\K_{n-1},k_n)} \in \L(\HS^{PA_{\K_{n-1}}^{IO} A_{k_n}^I})$, for all strict subsets $\K_{n-1}$ of $\N$ and all $k_n \in \N\backslash\K_{n-1}$, satisfying
\begin{align}
& \sum_{k_1 \in \N} \Tr_{A_{k_1}^I} W_{(\emptyset,k_1)} = \id^P, \notag \\
& \forall \, \emptyset \subsetneq \K_n \subsetneq \N, \!\! \sum_{k_{n+1} \in \N \backslash \K_n} \!\!\! \Tr_{A_{k_{n+1}}^I} \!W_{(\K_n,k_{n+1})} \notag \\[-2mm]
& \hspace{40mm} = \sum_{k_n \in \K_n} W_{(\K_n \backslash k_n,k_n)}\otimes \id^{A_{k_n}^O}, \notag \\[1mm]
& \textup{and} \quad \Tr_F W = \sum_{k_N \in \N} W_{(\N \backslash k_N,k_N)}\otimes \id^{A_{k_N}^O}. \label{eq:charact_W_QCQC_decomp}
\end{align}

Conversely, any positive semidefinite matrix $W \in \L(\HS^{PA_\N^{IO} F})$ such that there exist positive semidefinite matrices $W_{(\K_{n-1},k_n)} \in \L(\HS^{PA_{\K_{n-1}}^{IO} A_{k_n}^I})$ for all $\K_{n-1} \subsetneq \N$ and $k_n \in \N\backslash\K_{n-1}$ satisfying Eq.~\eqref{eq:charact_W_QCQC_decomp} is the process matrix of a quantum circuit with quantum control of causal order.
\end{proposition}

The full proof is given in Appendix~\ref{app:proof_charact_QCQCs}, and below we simply outline the proof approach.

The necessary condition is obtained by taking $W_{(\K_{n-1},k_n)} \coloneqq \Tr_{\alpha_n} \ketbra{w_{(\K_{n-1},k_n)}}{w_{(\K_{n-1},k_n)}}$, with $\ket{w_{(\K_{n-1},k_n)}}$ defined in Eq.~\eqref{eq:def_w_Knm1_kn}.
The constraints then follow readily from the form of Eq.~\eqref{eq:W_QCQC} and the TP conditions Eqs.~\eqref{eq:TP_constr_QCQC_1}--\eqref{eq:TP_constr_QCQC_N}.

To prove the sufficient condition, we once more show how to obtain an explicit construction of a QC-QC for any $W$ with such a decomposition; i.e., we show how to obtain the operators $V_{\emptyset,\emptyset}^{\to k_1}$, $V_{\K_{n-1},k_n}^{\to k_{n+1}}$ and $V_{\N \backslash k_N,k_N}^{\to F}$ satisfying the TP conditions needed to construct the circuit.
As for the other classes of circuits we have considered, this construction is not unique and different QC-QCs may be described by the same process matrix.

Finally, one can again verify that the constraints of Eqs.~\eqref{eq:charact_W_QCQC_decomp} indeed imply that $W$ satisfies the validity constraints for a process matrix (cf.\ Appendix~\ref{app:W_valid_matrices}).

\medskip

As one may expect, QC-CCs are a (strict) subset of QC-QCs.
One way to see this is from the characterisations of the corresponding classes:
given a process matrix $W$ for a QC-CC with a decomposition in terms of positive semidefinite matrices $W_{(k_1,\ldots,k_n)}$ as in Proposition~\ref{prop:charact_W_QCCC}, it is easily checked that the matrices $W_{(\K_{n-1},k_n)} \coloneqq \sum_{(k_1,\ldots,k_{n-1}):\{k_1,\ldots,k_{n-1}\}=\K_{n-1}} W_{(k_1,\ldots,k_{n-1},k_n)}$ satisfy the constraints of Proposition~\ref{prop:charact_W_QCQC}, and thus $W$ is also the process matrix for a QC-QC.
The fact QC-QCs are a strictly larger class of circuits (for $N\ge 2$) follows from the examples presented in the following subsection.

One way to explicitly recover QC-CCs from QC-QCs is by projecting the control system onto the ``classical'' basis $[\![ \K_{n-1},k_n ]\!]^{C_n}\coloneqq \ketbra{\K_{n-1},k_n}{\K_{n-1},k_n}^{C_n}$ prior to each time slot (cf.\ Eq.~\eqref{eq:tilde_Mn1_QCCC}).
This corresponds to the case where the control system of a QC-QC is decohered, and any coherence between the different causal orders is destroyed.
Although this leads to circuits with an effectively-classical control system recording only $(\K_{n-1},k_n)$ rather than the full order, such a control is already sufficient to describe fully the class of QC-CCs.
Indeed, the full order of operations can always be recorded in an ancillary system and used to further control the internal operations.
We could thus have also defined QC-CCs as using the classical control systems $[\![ \K_{n-1},k_n ]\!]^{C_n^{(\prime)}}$ and would have obtained the same class of process matrices.

One can similarly use ancillary systems in QC-QCs to coherently control the order of operations based on the full order of previous operations (as is indeed done in the ``quantum $N$-switch'' described in Appendix~\ref{app:Nswitch}).
However, in the case of a quantum control, storing only the pair ($\K_{n-1},k_n)$ allows the order of the operations in $\K_{n-1}$ to be forgotten, creating interference between the different causal orders and leading (in general, for $N \ge 3$) to a larger class of circuits than if control systems of the form $\ket{(k_1,\dots,k_n)}^{C_n^{(\prime)}}$ were used.

\medskip

Finally, we note that QC-QCs can also be defined for the case of process matrices with trivial global past and future Hilbert spaces.
For this case, corresponding to the original formulation of process matrices, we give a simplified formulation of the constraints of Proposition~\ref{prop:charact_W_QCQC} in Appendix~\ref{app:charact_dP_dF_1}.

\subsection{Examples}
\label{subsec:QCQC_examples}

\subsubsection{The ``quantum switch''}
\label{subsubsec:QSwitch}

The canonical example of a causally indefinite QC-QC is the ``quantum switch''~\cite{chiribella13}. 
It can be seen as a generalisation of the classical switch, which we presented as a QC-CC in Sec.~\ref{subsec:QCCCs_example}, to the case where the qubit system provided in $\HS^{P_\textup{c}}$ in the global past is used to control coherently, rather than classically, the order in which $N=2$ external operations $A_1$ and $A_2$ are applied to the $d_\textup{t}$-dimensional target system, initially provided in $\HS^{P_\textup{t}}$.
Adopting the same notation employed in Sec.~\ref{subsec:QCCCs_example} (notably, for the global past $\HS^P=\HS^{P_\textup{t}}\otimes \HS^{P_\textup{c}}$ and future $\HS^F=\HS^{F_\textup{t}}\otimes \HS^{F_\textup{c}}$), the circuit begins by coherently sending (via identity channels) the target system to $A_1$ then $A_2$ when the ``control qubit'' provided in $\HS^{P_\textup{c}}$ is in the state $\ket{1}$, and vice versa when it is $\ket{2}$ (cf.\ the possible implementation shown in Fig.~\ref{fig:Qswitch_impl}).

It is important to note here that, while the system in $\HS^{P_\textup{c}}$ (and subsequently recovered in $\HS^{F_\textup{c}}$) is generally referred to in the literature as the ``control qubit'', it is distinct from what we call the control system in the Hilbert spaces $\HS^{C_n^{(\prime)}}$ in the description of a QC-QC.
Instead, the information in $\HS^{P_\textup{c}}$ is propagated through the circuit in the QC-QC's control system and used to control the internal and external operations.

To see that the quantum switch can be described as a QC-QC, we can take (cf.\ Eq.~\eqref{eq:Ms_CSwitch} for the classical switch)
\begin{align}
\dket{V_{\emptyset,\emptyset}^{\to k_1}} & = \dket{\id}^{P_\textup{t}A_{k_1}^I} \otimes \ket{k_1}^{P_\textup{c}}, \notag \\
\dket{V_{\emptyset,k_1}^{\to k_2}} & = \dket{\id}^{A_{k_1}^OA_{k_2}^I}, \notag \\
\dket{V_{\{k_1\},k_2}^{\to F}} & = \dket{\id}^{A_{k_2}^OF_\textup{t}} \otimes \ket{k_1}^{F_\textup{c}}. \label{eq:Vs_QSwitch}
\end{align}
The corresponding operations can be interpreted intuitively:
$V_{\emptyset,\emptyset}^{\to k_1}$ is an identity channel sending the initial target system in $\HS^{P_\textup{t}}$ to the input space of $A_{k_1}$ when the state in $\HS^{P_\textup{c}}$ is $\ket{k_1}$;
$V_{\emptyset,k_1}^{\to k_2}$ is an identity channel sending the output of $A_{k_1}$ to the input of $A_{k_2}$;
and $V_{\{k_1\},k_2}^{\to F}$ sends the output of $A_{k_2}$ to the global future, while recording coherently $\ket{k_1}$ in $\HS^{F_\textup{c}}$, thereby completing the transmission of the control qubit initially provided in $\HS^{P_\textup{c}}$ (and whose state is transferred via the enlarged operations $\tilde V_1$ and $\tilde V_2$, as these update the control systems to $\ket{\emptyset,k_1}^{C_1}$ and $\ket{\{k_1\},k_2}^{C_2}$).
It is easy to verify that these operators indeed satisfy the TP constraints of Eqs.~\eqref{eq:TP_constr_QCQC_1}--\eqref{eq:TP_constr_QCQC_N}.

The process matrix describing the quantum switch defined by the operations \eqref{eq:Vs_QSwitch}, according to Proposition~\ref{prop:descr_W_QCQC}, is thus
\begin{align}
& W_\textup{QS} = \ketbra{w_\textup{QS}}{w_\textup{QS}}  \quad \text{with} \notag \\[2mm]
& \ket{w_\textup{QS}} \coloneqq \ket{1}^{P_\textup{c}} \dket{\id}^{P_\textup{t}A_1^I} \dket{\id}^{A_1^OA_2^I} \dket{\id}^{A_2^OF_\textup{t}} \ket{1}^{F_\textup{c}} \notag \\
& \hspace{12mm} + \ket{2}^{P_\textup{c}} \dket{\id}^{P_\textup{t}A_2^I} \dket{\id}^{A_2^OA_1^I} \dket{\id}^{A_1^OF_\textup{t}} \ket{2}^{F_\textup{c}} \notag \\
& \hspace{30mm} \in \HS^{P_\textup{c}P_\textup{t}A_1^{IO}A_2^{IO}F_\textup{t}F_\textup{c}}, \label{eq:W_QS}
\end{align}
where the tensor products are implicit.
We see clearly that we have a coherent superposition of terms corresponding to different causal orders, in contrast to the incoherent mixture in the process matrix $W_\textup{CS}$ of the classical switch in Eq.~\eqref{eq:W_CS}.
Indeed, one recovers $W_\textup{CS}$ by projecting the systems in $\HS^{P_\textup{c}}$ and/or $\HS^{F_\textup{c}}$ onto the basis $\{\ket{1},\ket{2}\}$ (or, similarly, by decohering the control system on the QC-QC; cf.\ the discussion at the end of Sec.~\ref{sec:QC-QC_characterisation}).
Note also that one can readily check that $W_\textup{QS}$ indeed satisfies the characterisation of Proposition~\ref{prop:charact_W_QCQC}, with $W_{(\emptyset,k_1)} = \dketbra{V_{\emptyset,\emptyset}^{\to k_1}}{V_{\emptyset,\emptyset}^{\to k_1}}$ and $W_{(\{k_1\},k_2)} = \dketbra{V_{\emptyset,\emptyset}^{\to k_1}}{V_{\emptyset,\emptyset}^{\to k_1}}\otimes \dketbra{V_{\emptyset,k_1}^{\to k_2}}{V_{\emptyset,k_1}^{\to k_2}}$.

The form of Eq.~\eqref{eq:W_QS} can be interpreted intuitively in a similar way to the individual operations in Eq.~\eqref{eq:Vs_QSwitch} discussed above.
When the control qubit in the global past is prepared in the state $\ket{1}^{P_\textup{c}}$, the induced ``conditional'' process vector for the remaining systems is precisely $\ket{w_\textup{QS}^1} = \ket{1}^{P_\textup{c}} * \ket{w_\textup{QS}} = \dket{\id}^{P_\textup{t}A_1^I} \dket{\id}^{A_1^OA_2^I} \dket{\id}^{A_2^OF_\textup{t}} \ket{1}^{F_\textup{c}}$, corresponding to identity channels taking the target system from $P_\textup{t}$ to $A_1$, then to $A_2$, and finally to $F_\textup{t}$ (while $F_\textup{c}$ receives the untouched ``control'' qubit $\ket{1}^{F_\textup{c}}$). 
Likewise, for the initial preparation of $\ket{2}^{P_\textup{c}}$, one obtains the conditional process vector $\ket{w_\textup{QS}^2} = \ket{2}^{P_\textup{c}} * \ket{w_\textup{QS}}$ describing identity channels first to $A_2$, then $A_1$.
More interestingly, when one prepares a superposition $\ket{\varphi_\textup{c}}^{P_\textup{c}} = \alpha \ket{1}^{P_\textup{c}} + \beta \ket{2}^{P_\textup{c}}$, the conditional process vector is $\ket{w_\textup{QS}^{\varphi_\textup{c}}} = \alpha \ket{w_\textup{QS}^1} + \beta \ket{w_\textup{QS}^2}$, corresponding to a superposition of the two causal orders. 
Of particular interest is the case of an equal superposition with $\alpha = \beta = \frac{1}{\sqrt{2}}$, and when the target system in $\HS^{P_\textup{t}}$ is a qubit prepared in some state $\ket{\psi_\textup{t}}^{P_\textup{t}}$.
Then, the conditional process vector (now with a trivial global past) is given by 
\begin{align}
& \ket{w_\textup{QS}^{+,\psi_\textup{t}}} \coloneqq \frac{1}{\sqrt{2}} \Big( \ket{\psi_\textup{t}}^{A_1^I} \dket{\id}^{A_1^OA_2^I} \dket{\id}^{A_2^OF_\textup{t}} \ket{1}^{F_\textup{c}} \notag \\[-2mm]
& \hspace{23mm} + \ket{\psi_\textup{t}}^{A_2^I} \dket{\id}^{A_2^OA_1^I} \dket{\id}^{A_1^OF_\textup{t}} \ket{2}^{F_\textup{c}} \Big) \notag \\
& \hspace{30mm} \in \HS^{A_1^{IO}A_2^{IO}F_\textup{t}F_\textup{c}}, \label{eq:W_QS_plus_psi}
\end{align}
with corresponding conditional process matrix $W_\textup{QS}^{+,\psi_\textup{t}} = \ketbra{w_\textup{QS}^{+,\psi_\textup{t}}}{w_\textup{QS}^{+,\psi_\textup{t}}}$, which is precisely the process matrix for the quantum switch given originally in  Refs.~\cite{araujo15,oreshkov16}.

Since the process matrix of the quantum switch, as written above, is rank-1 (hence it cannot be further decomposed as a nontrivial mixture of different process matrices) and since $\ket{w_\textup{QS}^{+,\psi_\textup{t}}}$ is clearly not compatible with a given, fixed causal order between $A_1$ and $A_2$, it follows that the quantum switch is causally nonseparable~\cite{araujo15,oreshkov16}. 
Note, however, that if we trace out the system in $\HS^{F_\textup{c}}$ one is left only with an incoherent mixture of terms corresponding to the two different causal orders.
Indeed, one has $\Tr_{F_\textup{c}} W_\textup{QS} = \Tr_{F_\textup{c}} W_\textup{CS}$, so one essentially recovers the classical switch and all coherent control is lost.
On the other hand, if one traces out only the target system sent to the global future in $\HS^{F_\textup{t}}$, one still has a non-classical switch with the coherence between the different causal orders maintained by the system in $\HS^{F_\textup{c}}$~\cite{araujo15}.

\medskip

The natural generalisation of the quantum switch to a superposition of the $N!$ possible orders of $N$ external operations~\cite{colnaghi12,araujo14,facchini15} can also be described as a QC-QC, as we show in Appendix~\ref{app:Nswitch}.
One key difference worth mentioning in the general case is that one needs to use the ancillary systems $\HS^{\alpha_n}$ to record the full causal order, as the pair $(\K_{n-1},k_n)$ stored by the QC-QCs control systems does not keep track of the full permutation to be applied.
(For the $N=2$ case described above, this was not an issue as $\K_{n-1}$ never contained more than one element.)

\subsubsection{A new type of quantum circuit with both dynamical and coherently-controlled causal order}
\label{subsubsec:new_QCQC}

The quantum switch (and rather straightforward generalisations with more operations) has, thus far, been the only causally nonseparable process for which a physical implementation is known. 
Our general description provides a framework allowing us to find new QC-QCs beyond this example. 
As an illustration, we present here a novel 3-operation QC-QC which differs qualitatively from the quantum switch in several key ways.
Firstly, unlike the standard ``3-switch'' (cf.\ Appendix~\ref{app:Nswitch}), it allows for the causal order to really be established ``dynamically'', depending (coherently) on the output of external operations (and not only a subsystem of $\HS^P$).
Secondly, it exploits the fact the control system only stores the unordered set of already applied operations in order to create interference between terms corresponding to different causal histories.
And lastly, its process matrix remains causally nonseparable, with no well-defined ``final'' operation, despite having only a trivial global future $\HS^F$.

\medskip

This new type of circuit fits our general description of QC-QCs as follows (cf.\ also the possible implementation discussed in the following subsection and Fig.~\ref{fig:new_QCQC}).
Consider a QC-QC with $N=3$ external operations, 2-dimensional input and output spaces $\HS^{A_k^I}$, $\HS^{A_k^O}$ (i.e., a qubit ``target'' system, with computational basis $\{\ket{0},\ket{1}\}$), and with trivial global past and future $\HS^P$, $\HS^F$ (i.e., $d_P = d_F = 1$),%
\footnote{Note that we could have similarly defined the process with nontrivial global past and future so that, as for the quantum switch, the initial target state and first operation are specified in the global past, and the output of the final operation and systems in $\HS^{\alpha_F}$ are sent to the global future (see Appendix~\ref{app:new_QCQC}). However, the key features of this process can be observed without needing to do so.}
defined by the operators
\begin{align}
& V_{\emptyset,\emptyset}^{\to k_1} \,=\, {\textstyle \frac{1}{\sqrt{3}}} \ket{\psi}^{A_{k_1}^I}, \notag \\[2mm]
& V_{\emptyset,k_1}^{\to k_2} \,=\, \left\{
\begin{array}{ll}
\ket{0}^{A_{k_2}^I}\bra{0}^{A_{k_1}^O} & \text{if } k_2 = k_1+1 \pmod{3} \\[1mm]
\ket{1}^{A_{k_2}^I}\bra{1}^{A_{k_1}^O} & \text{if } k_2 = k_1+2 \pmod{3}
\end{array}
\right. \!\!, \notag \\[2mm]
& V_{\{k_1\},k_2}^{\to k_3} \,=\, \left\{
\begin{array}{l}
\ket{0}^{A_{k_3}^I}\ket{0}^{\alpha_3}\bra{0}^{A_{k_2}^O} + \ket{1}^{A_{k_3}^I}\ket{1}^{\alpha_3}\bra{1}^{A_{k_2}^O} \\[1mm]
\hspace{25mm} \text{if } k_2 = k_1+1 \pmod{3} \\[2mm]
\ket{0}^{A_{k_3}^I}\ket{1}^{\alpha_3}\bra{0}^{A_{k_2}^O} + \ket{1}^{A_{k_3}^I}\ket{0}^{\alpha_3}\bra{1}^{A_{k_2}^O} \\[1mm]
\hspace{25mm} \text{if } k_2 = k_1+2 \pmod{3}
\end{array}
\right. \!\!, \notag \\[2mm]
& V_{\{k_1,k_2\},k_3}^{\to F} \,=\, \id^{{A_{k_3}^O} \alpha_3 \to \alpha_F^{(1)}} \otimes \ket{k_3}^{\alpha_F^{(2)}}, \label{eq:def_Vs_new_QCQC}
\end{align}
where we introduced an ancillary 2-dimensional system $\alpha_3$ (but no $\alpha_1, \alpha_2$), a 4-dimensional system $\alpha_F^{(1)}$ and a 3-dimensional system $\alpha_F^{(2)}$, defining $\alpha_F \coloneqq \alpha_F^{(1)} \alpha_F^{(2)}$ (with corresponding Hilbert spaces $\HS^{\alpha_F} \coloneqq \HS^{\alpha_F^{(1)} \alpha_F^{(2)}}$), and $\ket{\psi}$ is an arbitrary qubit state.
One can verify that the Choi vectors of these operators indeed satisfy the TP constraints of Eqs.~\eqref{eq:TP_constr_QCQC_1}--\eqref{eq:TP_constr_QCQC_N}, as required.

These operations can be interpreted as follows.
$V_{\emptyset,\emptyset}^{\to k_1}$ sends the state $\ket{\psi}$ to $A_{k_1}$ (and to each choice of $k_1$ with equal weight, in a superposition).
$V_{\emptyset,k_1}^{\to k_2}$ sends the output of $A_{k_1}$ to one of the remaining operations $A_{k_2}$ (for $k_2\neq k_1$) dynamically and coherently depending on the state of said output: the component in the state $\ket{0}^{A_{k_1}^O}$ is sent to $A_{k_1+1\!\pmod{3}}$, while the component in the state $\ket{1}^{A_{k_1}^O}$ is sent to $A_{k_1+2\!\pmod{3}}$.
$V_{\{k_1\},k_2}^{\to k_3}$ then sends the output of $A_{k_2}$ to the remaining operation $A_{k_3}$ and attaches an ancillary state $\ket{0}^{\alpha_3}$ if $k_2 = k_1+1 \pmod{3}$ or $\ket{1}^{\alpha_3}$ if $k_2 = k_1+2 \pmod{3}$, that is then flipped if $A_{k_3}^I$ is in the state $\ket{1}^{A_{k_3}^I}$ (i.e., a controlled \textsc{Not} gate is applied).%
\footnote{Note that introducing the ancillary system $\alpha_3$ in such a nontrivial way---in particular, with its state depending on whether $k_2 = k_1+1 \pmod{3}$ or $k_2 = k_1+2 \pmod{3}$---is indeed necessary (despite the fact that $\alpha_3$ is ultimately discarded) to ensure the internal operation $\tilde V_3$ acts as an isometry on its input spaces---i.e., to satisfy the TP constraint of Eq.~\eqref{eq:TP_constr_QCQC_n}, for $n=2$.}
Finally, $V_{\{k_1,k_2\},k_3}^{\to F}$ sends the output of $A_{k_3}$ along with the system in $\HS^{\alpha_3}$ to $\alpha_F^{(1)}$, while $\ket{k_3}$ is sent to $\alpha_F^{(2)}$.

The (tripartite) process matrix of this QC-QC, according to Proposition~\ref{prop:descr_W_QCQC}, is
\begin{align}
W & = \Tr_{\alpha_F} \ket{w}\!\bra{w} \notag \\
\text{with} \ \ket{w} & = \sum_{(k_1,k_2,k_3)} \dket{V_{\emptyset,\emptyset}^{\to k_{1}}} * \dket{V_{\emptyset,k_1}^{\to k_{2}}} \notag \\[-3mm]
& \hspace{20mm} * \dket{V_{\{k_1\},k_2}^{\to k_3}} * \dket{V_{\{k_1,k_2\},k_3}^{\to F}}.
\label{eq:processvec_newQCQC}
\end{align}
Using the technique of causal witnesses~\cite{araujo15,branciard16a,wechs19}, one can check, for any fixed but arbitrary state $\ket{\psi}$, that this process matrix is causally nonseparable (see Appendix~\ref{app:new_QCQC}). 
It is interesting also to note that although tracing out $\alpha_F$ (or even just $\alpha_F^{(2)}$) turns $W$ into an (incoherent) sum of 3 matrices (one for each value of $k_3$), these 3 matrices are not themselves valid process matrices: 
$W$ is not simply a convex mixture of 3 tripartite process matrices, each compatible with one operation $A_{k_3}$ being applied last%
\footnote{This can be checked using semidefinite programming methods, similar to causal witnesses; cf.\ Ref.~\cite{wechs19}.
In particular, $W$ is not just a convex combination of quantum switches; which operation is applied last is not predetermined (even probabilistically), but depends on the operations $A_1, A_2, A_3$.}
(as is, for instance, the ``3-switch'', cf.\ Appendix~\ref{app:Nswitch}, after tracing out $F$).
This is due to the fact that the causal order here is established dynamically. 

\medskip

Our general description of the QC-QC class thus allowed us to present here a new type of example, which combines both a coherent and dynamical control of causal order in a way not done by the quantum switch or its direct generalisations.
In Appendix~\ref{app:new_QCQC}, we present a slightly more general family of such processes that may be of further interest and provides further insight into the form of the particular example presented here.
We hence see that QC-QCs provide an interesting new class of quantum supermaps with concrete interpretations that go beyond the well-studied quantum switch and its generalisations.
Further study of this class, and of other new types of QC-QCs, may uncover further interesting examples, and we believe this to be an important direction for future research.

\subsection{Possible implementations}
\label{subsec:implementations}

The theoretical description of QC-QCs above raises of course the question of their practical realisation. 
Here we shall present some basic ideas which show that such implementations are indeed possible.

The problem at hand is to find some physical systems on which the desired operations can be implemented.
In particular, one needs to find suitable systems to encode the control states of the form%
\footnote{The primed versions $\ket{\K_{n-1},k_n}^{C_n'}$ will simply be taken here to be the same as the $\ket{\K_{n-1},k_n}^{C_n}$'s, and will not be discussed any further.}
$\ket{\K_{n-1},k_n}^{C_n}$ that can be used to control both the external operations $\tilde A_{k_n}$ (acting on some other, target systems) as in Eq.~\eqref{eq:tilde_An_QCQC}, as well as the internal circuit operations $\tilde V_{\K_{n-1},k_n}^{\to k_{n+1}}$, and be updated by the $\tilde V_n$'s as in Eqs.~\eqref{eq:def_V1_QCQC}--\eqref{eq:def_VN1_QCQC}.

One natural choice is to let the external operations be implemented at $N$ different spatial locations, and to let the target systems be carried by a physical entity (e.g., a photon) that passes through them; which operation is actually realised on the target carrier then depends on its path. 
This idea leads one to consider control states of the refined form $\ket{{\cal K}_{n-1},k_n}^{C_n} = \ket{{\cal K}_{n-1}{\scriptstyle |k_n}}^{C_n^\textup{past ops.}} \otimes \ket{k_n}^{C_n^\textup{path}}$, where $\ket{k_n}^{C_n^\textup{path}}$ denotes the path $k_n$ of the carrier that undergoes operation $A_k$ at time $t_n$, and $\ket{{\cal K}_{n-1}{\scriptstyle |k_n}}^{C_n^\textup{past ops.}}$ is the state of some complementary control system that records the required information about $\K_{n-1}$ (which, in general, may be encoded differently for different $k_n$).%
\footnote{Nontrivial states $\ket{{\cal K}_{n-1}{\scriptstyle |k_n}}^{C_n^\textup{past ops.}}$ are required as soon as $k_n$ does not identify $\K_{n-1}$ uniquely, i.e., when there are two non-vanishing components in the process vector that involve some $\ket{{\cal K}_{n-1},k_n}^{C_n}$ and $\ket{{\cal K}_{n-1}',k_n}^{C_n}$, with the same $k_n$ and with ${\cal K}_{n-1} \neq {\cal K}_{n-1}'$. 
The dimension of the corresponding Hilbert space $\HS^{C_n^\textup{past ops.}}$ may be required to be up to $\binom{N-1}{n-1}$ (the number of possible subsets ${\cal K}_{n-1} \subseteq {\cal N}\backslash k_n$ with $|{\cal K}_{n-1}| = n-1$).}
This complementary system, just like potential ancillary systems, could be encoded, e.g., on some different degrees of freedom of the physical carrier, other than the path and the target system.

In such an implementation, the internal circuit operations $\tilde V_n$ need to route the physical carrier while performing the operations $V_{\K_{n-1},k_n}^{\to k_{n+1}}$, i.e., to act jointly on the path and internal degrees of freedom of the carrier, so as to recover Eqs.~\eqref{eq:def_V1_QCQC}--\eqref{eq:def_VN1_QCQC}.
As these internal circuit operations are, in general, different for each value of $n = 1,\dots,N+1$, one possibility for their implementation is to have a circuit with fast-switching elements. 
Another possibility is to introduce yet an additional system that acts as a ``timer'' (of dimension at most $N+1$), to be ``incremented'' at every time slot $t_n$, and which also controls (in an essentially classical manner) the application of the correct internal circuit operation.

In what follows, we outline more concretely how such generic approaches to implementing QC-QCs can be applied to the two examples discussed above (i.e., the quantum switch, and the new QC-QC defined by Eq.~\eqref{eq:def_Vs_new_QCQC}) using photons as the physical carriers.

Let us first note, however, that the same ideas can of course be used for implementing QC-CCs, as particular cases of QC-QCs.
Indeed, this would actually be simpler experimentally as the control systems need not be kept coherent.
One then has refined classical control states of the form $[\![(k_1,\ldots,k_{n-1},k_n)]\!]^{C_n} = [\![(k_1,\ldots,k_{n-1})]\!]^{C_n^\textup{past ops.}} \otimes [\![k_n]\!]^{C_n^\textup{path}}$, and the physical carrier can be routed to the correct external operation at each time slot using only classical routers.
Note, nevertheless, that although all QC-CCs are causally separable it remains an open problem whether all causally separable process matrices correspond to a QC-CC or, more generally, are physically realisable~\cite{oreshkov16,wechs19}.

\subsubsection{The quantum switch}
\label{subsubsec:impl_switch}

\begin{figure}[t]
\begin{center}
\includegraphics[scale=1]{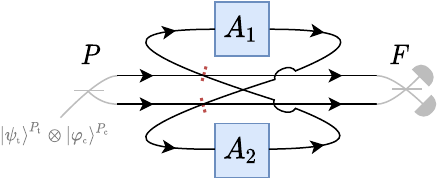}
\caption{A new possible photonic implementation of the quantum switch, in which the control qubit gets encoded in the path degree of freedom, and the target system in an internal degree of freedom of the photon.
The dashed optical elements (\protect\raisebox{-1mm}{\protect\includegraphics[scale=1]{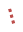}}) are reflecting mirrors, which are momentarily removed between the time slots $t_1$ and $t_2$ (i.e., between the applications of the operations $A_1$ and $A_2$, in either order).
Example operations in the global past $P$ (the preparation of an initial target state $\ket{\psi_\textup{t}}^{P_\textup{t}}$ and the control qubit in a superposition state $\ket{\varphi_\textup{c}}^{P_\textup{c}}$, see Sec.~\ref{subsubsec:QSwitch}) and future $F$ (the measurement of the final control system in $F_\textup{c}$ in a superposition basis) are shown in grey for clarity.}
\label{fig:Qswitch_impl}
\end{center}
\end{figure}

The generic implementation procedure described above can be applied to the example of the quantum switch presented in Sec.~\ref{subsubsec:QSwitch}.
One possible such implementation using photonic particle carriers is shown in Fig.~\ref{fig:Qswitch_impl}, where fast-switching removable mirrors are used to implement the different internal operations.

Interestingly, this proposal differs from previous photonic implementations of the quantum switch~\cite{procopio15,rubino17,rubino17a,goswami18,goswami20,wei19,guo20}, and highlights previously overlooked redundancies in some such implementations.
Indeed, compared to the implementation initially proposed in Ref.~\cite{araujo14} and realised experimentally in Refs.~\cite{goswami18,goswami20}, this new proposal exploits fewer degrees of freedom of the photons (at the price of using such fast-switching optical elements in the circuit): just the path as the control, and some internal degree of freedom of the photon (e.g., polarisation or orbital angular momentum) as the target systems. 
In contrast, in Refs.~\cite{araujo14,goswami18,goswami20}, the control system is copied (coherently) from the polarisation to the path degrees of freedom, inducing a redundancy in the implementation.%
\footnote{With the control systems encoded in the path of the photons (such that $\ket{\emptyset,k_1}^{C_1} = \ket{k_1}^{C_1^\textup{path}}$ and $\ket{\{k_1\},k_2}^{C_2} = \ket{k_2}^{C_2^\textup{path}}$), the implementations of Refs.~\cite{araujo14,goswami18,goswami20} could be written as QC-QCs by taking (instead of Eq.~\eqref{eq:Vs_QSwitch} above)
\begin{align}
\dket{V_{\emptyset,\emptyset}^{\to k_1}} & = \dket{\id}^{P_\textup{t}A_{k_1}^I} \otimes \ket{k_1}^{P_\textup{c}} \otimes \ket{k_1}^{\alpha_1}, \notag \\
\dket{V_{\emptyset,k_1}^{\to k_2}} & = \dket{\id}^{A_{k_1}^OA_{k_2}^I} \otimes \ket{k_1}^{\alpha_1} \otimes \ket{k_1}^{\alpha_2}, \notag \\
\dket{V_{\{k_1\},k_2}^{\to F}} & = \dket{\id}^{A_{k_2}^OF_\textup{t}} \otimes \ket{k_1}^{F_\textup{c}} \otimes \ket{k_1}^{\alpha_2}
\end{align}
with $\ket{k_1}^{P_\textup{c} / \alpha_n / F_\textup{c}} = \ket{H}$ for $k_1=1$, or $\ket{V}$ for $k_1=2$. As one can see the control system gets (redundantly) copied onto and transferred through the ancillary systems $\alpha_n$ (i.e., through the polarisation of the photons). This also illustrates the fact that the same process can be given different descriptions in terms of a QC-QC.}
Similarly, in the implementations of Refs.~\cite{procopio15,rubino17,guo20}, four spatial degrees of freedom are exploited, rather than the two in the implementation we propose here.
As a result, the internal operations can be ensured to be applied to photons in the same spatial modes (although at different times), ensuring that the applications of each $A_k$ at different time slots are truly indistinguishable.

As suggested above, one could avoid using fast-switching elements by introducing an explicit ``timer'' system. 
Here for instance, if the target system is encoded in some other internal degree of freedom of the photon, then polarisation could be used as such a ``timer'' by initially preparing it in the state $\ket{V}$, replacing the removable mirrors in the setup of Fig.~\ref{fig:Qswitch_impl} by fixed polarising beam-splitters (which reflect $\ket{V}$ and transmit $\ket{H}$), and adding wave plates, e.g., at the exit ports of $A_1$ and $A_2$ that switch the polarisation, $\ket{V} \leftrightarrow \ket{H}$ (so as to ``increment'' the timer). 
We then simply have a passive optical circuit, which uses the path, polarisation and some other degree of freedom of the photon as the target system---as in Refs.~\cite{araujo14,goswami18,goswami20}, although in a structurally different manner.

\subsubsection{Novel QC-QC with dynamical and coherently-controlled causal order}
\label{subsubsec:impl_new_QCQC}

\begin{figure}%[t]
\begin{center}
\includegraphics[width=\columnwidth]{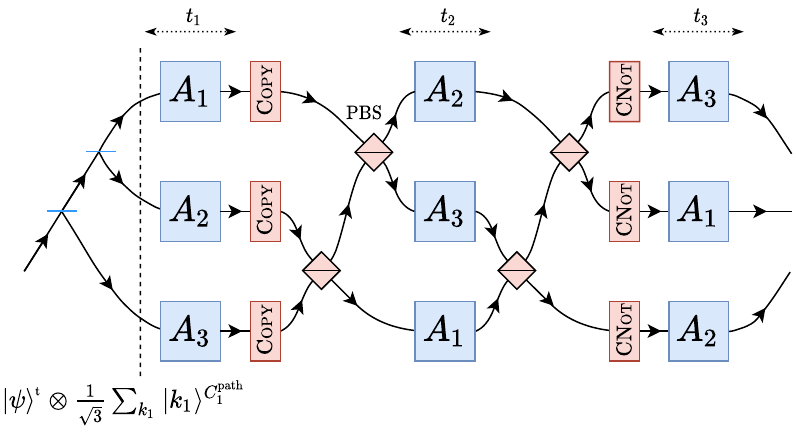}
\caption{
A possible proposal for the realisation of our new QC-QC defined by Eq.~\eqref{eq:def_Vs_new_QCQC} in a photonic circuit.
The target system on which the external operations $A_k$ act is an internal degree of freedom of a photon, while the control and ancillary systems are encoded in the path and polarisation degrees of freedom, as described in the main text.
The $\textsc{Copy}$ gate, as part of $\tilde V_2$, copies the internal state onto the polarisation state, which contributes to the control system $C_2$. 
The polarising beam-splitters (PBS) route the photon based on its polarisation.
The $\textsc{CNot}$ gates finally act jointly on the internal and polarisation degrees of freedom, so as to realise the desired operation $\tilde V_3$, as prescribed by Eqs.~\eqref{eq:tilde_Vn1_QCQC} and~\eqref{eq:def_Vs_new_QCQC}.
Note that the circuit is depicted here in an ``unfolded'' form, for clarity. 
In fact, each operation $A_k$ should be realised once and only once, and the boxes depicted here thus identified as a single operation delocalised in time~\cite{oreshkov19}.
Correspondingly, each physical ``box'' $A_k$ should be built only once: the circuit connections should ``loop back'' between time slots $t_n$ and $t_{n+1}$, from the output of one box to the input of some other boxes---e.g., using fast switching elements (or a ``timer'' system) to realise the different operations $\tilde V_n$ as in the implementation of the quantum switch in Fig.~\ref{fig:Qswitch_impl}.
} 
\label{fig:new_QCQC}
\end{center}
\end{figure}

The implementation procedure we outlined can also be used to give a concrete proposal for the implementation of the new QC-QC presented in Sec.~\ref{subsubsec:new_QCQC} and defined by Eq.~\eqref{eq:def_Vs_new_QCQC}.
In Fig.~\ref{fig:new_QCQC} we depict a possible such photonic implementation, in which a 2-dimensional target system (initially in the state $\ket{\psi}^\textup{t}$, in the generic target space $\HS^\textup{t}$) is encoded in some internal degree of freedom of a photon (e.g., its orbital angular momentum).
The control systems $C_1$ and $C_3$ are simply the path of the photon, such that $\ket{\emptyset,k_1}^{C_1} = \ket{k_1}^{C_1^\textup{path}}$ and $\ket{\{k_1,k_2\},k_3}^{C_3} = \ket{k_3}^{C_3^\textup{path}}$; 
to define the control system $C_2$ on the other hand, we need to make use of some further 2-dimensional degree of freedom $\alpha$ of the photon, which we take to be the polarisation (with basis states $\ket{0}^\alpha = \ket{V},\ket{1}^\alpha = \ket{H}$), such that $\ket{\{k_1\},k_2}^{C_2} = \ket{0}^\alpha \otimes \ket{k_2}^{C_2^\textup{path}}$ if $k_2 = k_1+1 \ (\text{mod }3)$, $\ket{\{k_1\},k_2}^{C_2} = \ket{1}^\alpha \otimes \ket{k_2}^{C_2^\textup{path}}$ if $k_2 = k_1+2 \ (\text{mod }3)$.
The ancillary system $\alpha_3$ is also taken to be the polarisation.

The ``\textsc{Copy}'' and ``\textsc{CNot}'' gates implement the operations $V_\textsc{Copy} = \sum_{i=0,1} \ket{i}^\textup{t}\ket{i}^\alpha\bra{i}^\textup{t}$ and $V_\textsc{CNot} = \sum_{i,j=0,1} \ket{i}^\textup{t}\ket{i \oplus j}^\alpha\bra{i}^\textup{t}\bra{j}^\alpha$ (with $\oplus$ denoting addition modulo 2), respectively. 
(Note that $V_\textsc{Copy}$ could be realised by preparing the polarisation in the state $\ket{0}^\alpha$ and applying $V_\textsc{CNot}$.)
As can be checked, with the choice of encoding above, the circuit shown in Fig.~\ref{fig:new_QCQC} indeed realises the internal operations $\tilde V_n$ obtained from Eq.~\eqref{eq:def_Vs_new_QCQC}, via Eqs.~\eqref{eq:def_V1_QCQC}--\eqref{eq:def_VN1_QCQC}.
It can clearly be seen in this circuit how the causal order is established dynamically: which operation $A_{k_2}$ the photon is routed to after undergoing the first operation $A_{k_1}$ depends on its polarisation (with $\ket{V}$ being reflected, $\ket{H}$ being transmitted at the polarising beam-splitters)---which the output of $A_{k_1}$ is ``copied'' onto (in the $\{\ket{0},\ket{1}\}$ basis) by $V_\textsc{Copy}$. 
It can also be seen (by following the trajectories taken by each of the $\ket{V}$ and $\ket{H}$ components, which are untouched between the \textsc{Copy} and \textsc{CNot} gates) how the proposed configuration of the beam-splitters guarantees that each operation is applied once and only once on each path.

While the realisation of this QC-QC in the lab would undoubtedly be a major challenge, it would represent a major step towards showing that more general QC-QCs exploiting dynamical, coherent control of causal order can be realised and, eventually, exploited in the laboratory, and we challenge experimental groups to the task.

\subsection{Correlations generated by quantum circuits with quantum control of causal order}
\label{subsec:causal_correlations}

One of the initial motivations in the study of processes with indefinite causal order was the possibility that they might allow one to generate correlations incompatible with any well-defined causal structure, thereby providing a particularly strong, model independent proof of the nonclassical causal structure of the world~\cite{oreshkov12}.
This, in particular, was a question the process matrix framework was originally conceived to tackle~\cite{oreshkov12}.
In this approach, the external operations $\A_k$ are interpreted as being applied by ``parties'' $\mathrm{A}_k$ in closed, locally causal, laboratories, and the process matrix $W$ as describing the physical process by which these parties interact, possibly in a causally indefinite manner.
Since we are interested in the relation between the parties---and in particular, the correlations they can observe---we take both the global past and future to be trivial (i.e., $d_P = d_F = 1$).
Each party $\mathrm{A}_k$ can then apply an instrument, potentially conditioned on some setting (a classical ``input'') $x_k$, producing outcomes (the ``outputs'') $a_k$ and hence with Choi matrices denoted $\{A_{a_k|x_k}\}_{a_k}$.

The correlation between all $N$ parties' inputs and outputs is represented by the conditional probability distribution $P(a_1,\ldots,a_N|x_1,\ldots,x_N)$, which can be calculated, within the process matrix formalism (see in particular Eq.~\eqref{eq:choi_map_M} and Footnote~\ref{footnote:pm_formalisms} in Sec.~\ref{subsec:PM_supermaps}), from the process matrix $W$ and the Choi matrices of the parties' instruments by the so-called ``generalised Born rule''~\cite{oreshkov12,araujo15}:
\begin{align}
& P(a_1,\ldots,a_N|x_1,\ldots,x_N) \notag \\
& \qquad = (A_{a_1|x_1} \otimes \cdots \otimes A_{a_N|x_N}) * W \notag \\
& \qquad = \Tr \big[ (A_{a_1|x_1} \otimes \cdots \otimes A_{a_N|x_N})^T W \big]. \label{eq:P_a1aN_x1xN}
\end{align}
Of particular interest is whether correlations obtained from a process matrix are ``causal''---i.e., can be explained by referring to a well-defined causal structure (allowing for probabilistic causal structures and for dynamical causal orders)---or not~\cite{oreshkov12,oreshkov16,abbott16}.
More specifically, in the multipartite case, causal correlations can be characterised~\cite{oreshkov16} or directly defined~\cite{abbott16} in a recursive manner, as convex combinations of correlations compatible with a given party acting first, and such that whatever that party does, the conditional correlation shared by the remaining ones is again causal (with a single-partite correlation being trivially causal). 
It was shown that the set of causal correlations (for a given scenario, i.e., a given number of parties, each with a given number of possible inputs, and a given number of possible outputs for each input) forms a convex polytope~\cite{branciard16,oreshkov16,abbott16}, delimited by so-called ``causal inequalities''~\cite{oreshkov12}.

The correlations generated by causally separable processes are necessarily causal~\cite{oreshkov12,branciard16,oreshkov16,abbott16}.
However, by only imposing a quantum description for the parties' local operations, and without making any assumption on the global causal structure, the general process matrix formalism allows in principle for (causally nonseparable) processes that generate noncausal correlations (and thus violate causal inequalities)~\cite{oreshkov12}.
Nevertheless, not all causally nonseparable processes can generate noncausal correlations~\cite{araujo15,oreshkov16,feix16}.
In fact, it is an open question of considerable interest whether any physically conceivable process, that could be built in the lab, can indeed violate a causal inequality.

The class of quantum circuits considered here does not, unfortunately, allow us to answer this open question.
Indeed, even though QC-QCs may define causally nonseparable processes, in Appendix~\ref{app:only_causal_correl} we prove that:
\begin{proposition}[Causality of QC-QC correlations] \label{prop:causal_correlations}
Quantum circuits with quantum control of order can only generate causal correlations.
\end{proposition}

Hence, QC-QCs cannot violate causal inequalities.
(This implies, \emph{a fortiori}, that QC-CCs, or QC-FOs, can also not violate causal inequalities, although this was already known; indeed, as we showed in the previous sections, those classes contain only causally separable processes.)
This generalises the previous results of Refs.~\cite{branciard16,oreshkov16} for the quantum switch.

\section{Probabilistic quantum circuits}
\label{sec:probQCs}

So far, we have studied quantum circuits that, although taking probabilistic external operations as inputs, are by themselves deterministic; that is, they arise from the composition of deterministic internal operations and can be realised without post-selection. 
In general, however, one can also consider circuits consisting of probabilistic operations that can produce several classical outcomes. 
In this section, we will characterise the probabilistic quantum supermaps obtained when allowing for probabilistic circuit operations in the classes that we introduced above. 
To that end, we replace each internal CPTP map in the above descriptions by a set of (trace non-increasing) CP maps (each corresponding to a given outcome) that sum up to a CPTP map\footnote{As before, when dealing with the trace-preserving conditions, we will restrict to the effective input spaces of the corresponding maps. This will be considered implicit in the following.}---i.e., by a quantum instrument~\cite{davies70}.

Such combinations of CP maps define ``probabilistic quantum circuits'' that can be represented by a set of ``probabilistic process matrices'', and that can be realised by post-selecting on the corresponding classical outcomes (where the probability of post-selection may depend on the external operations plugged into the circuit).
Technically speaking, such probabilistic quantum circuits define so-called ``quantum superinstruments''~\cite{chiribella13a}.

In what follows, we shall characterise the probabilistic quantum circuits thus obtained---and their elements, i.e., the probabilistic process matrices---for fixed causal order (Propositions~\ref{prop:descr_Wr_probQCFO} and~\ref{prop:charact_Wr_probQCFO}), for operations used in parallel (Proposition~\ref{prop:charact_Wr_probQCPAR}), for classical (Propositions~\ref{prop:descr_Wr_probQCCC} and~\ref{prop:charact_Wr_probQCCC}) and quantum (Propositions~\ref{prop:descr_Wr_probQCQC} and~\ref{prop:charact_Wr_probQCQC}) control of causal order, and even for general quantum superinstruments (Proposition~\ref{prop:charact_Wr_probIND}).

We note that results similar to those developed in this section have recently been presented in Ref.~\cite{bavaresco20}, where Bavaresco \emph{et al.}\ define so-called \emph{``two-copy parallel''}, \emph{``two-copy sequential''}, \emph{``two-copy general''} and \emph{``two-copy separable'' testers}. 
These correspond to the particular case with $N = 2$ and trivial $\HS^P$ and $\HS^F$ of the classes characterised in our Propositions~\ref{prop:charact_Wr_probQCPAR},~\ref{prop:charact_Wr_probQCFO},~\ref{prop:charact_Wr_probIND} and~\ref{prop:charact_Wr_probQCQC} respectively. 
Here, we derive these characterisations (as well as that of an additional class in Proposition~\ref{prop:charact_Wr_probQCCC}) for the general $N$-operation case using our constructive approach.  
Our results notably imply that the ``two-copy separable'' testers (in the terminology of Ref.~\cite{bavaresco20}) can be realised as probabilistic circuits with quantum control of causal order, providing a physical interpretation for that class.

\subsection{Probabilistic quantum circuits with fixed causal~order}
\label{sec:probQCFOs}

Let us start with the probabilistic counterpart of QC-FOs. 
This case has previously been studied in the literature, and equivalent or closely related concepts have been introduced under the names of~\emph{probabilistic quantum network}~\cite{chiribella09,bisio11}, \emph{generalised instrument}~\cite{chiribella09,bisio11}, \emph{measuring strategy}~\cite{gutoski06}, \emph{quantum tester}~\cite{chiribella09,bisio11,chiribella16} and \emph{process POVM}~\cite{ziman08}.

A given realisation (for a given set of classical outcomes) of a probabilistic quantum circuit with the fixed causal order $(\A_1,\ldots,\A_N)$ consists of internal CP maps $\M_1^{[r_1]}: \L(\HS^P) \to \L(\HS^{A_1^I\alpha_1})$, $\M_{n+1}^{[r_{n+1}]}: \L(\HS^{A_n^O\alpha_n}) \to \L(\HS^{A_{n+1}^I\alpha_{n+1}})$ for $1 \le n \le N-1$, and $\M_{N+1}^{[r_{N+1}]}: \L(\HS^{A_N^O\alpha_N}) \to \L(\HS^F)$, which are composed as in Fig.~\ref{fig:QCFO}, with $r_1,\dots,r_{N+1}$ denoting the classical outcomes, and with each of the CP maps being part of a quantum instrument---that is, with their sum over the classical outcomes yielding a CPTP map. 

To simplify the description, we note that the classical outcomes can always be encoded onto suitable orthogonal states of the ancillary systems, and the post-selection can be performed at the end as part of the last internal operation (before $F$). 
This allows us to describe any such probabilistic circuit without loss of generality as a circuit in which all internal operations are still deterministic, except for the last one, which is a CP map $\M_{N+1}^{[r]}$, belonging to an instrument $\{\M_{N+1}^{[r]}\}_r$ with classical outcomes $r$~\cite{chiribella09}. 

The TP conditions satisfied by the internal operations are thus given by Eqs.~\eqref{eq:TP_constr_QCFO_1}--\eqref{eq:TP_constr_QCFO_n} and by the TP condition for the final instrument, which is simply obtained by replacing $M_{N+1}$ by $\sum_r M_{N+1}^{[r]}$ in Eq.~\eqref{eq:TP_constr_QCFO_N} (see Appendix~\ref{app:proof_charact_QCCCs}). 
We formally call any such process with circuit operations $\M_1, \ldots, \M_N$ and $\{\M_{N+1}^{[r]}\}_r$ that are composed as in Fig.~\ref{fig:QCFO} and that satisfy these trace preserving conditions, a \emph{probabilistic quantum circuit with fixed causal order} (pQC-FO).  
Similarly to Proposition~\ref{prop:descr_W_QCFO}, we obtain the following process matrix description of pQC-FOs.

\begin{proposition}[Process matrix description of pQC-FOs] \label{prop:descr_Wr_probQCFO}
The probabilistic process matrix describing the specific realisation of such a pQC-FO, corresponding to the classical outcome $r$, is
\begin{align}
W^{[r]} = M_1 * M_2 * \cdots * M_N * M_{N+1}^{[r]} \quad \in \ \L(\HS^{PA_\N^{IO} F}). \label{eq:Wr_probQCFO}
\end{align}
The entire pQC-FO is described by the set $\{W^{[r]}\}_r$ of all such probabilistic process matrices, for all classical outcomes $r$. 
\end{proposition}

The corresponding characterisation follows directly from that of QC-FOs given by Proposition~\ref{prop:charact_W_QCFO}, and is proven in Appendix~\ref{app:proof_charact_QCFOs}. 
An equivalent result has also been proven in Ref.~\cite{chiribella09} (Theorem 4).

\begin{proposition}[Characterisation of pQC-FOs] \label{prop:charact_Wr_probQCFO}
A probabilistic quantum circuit with a fixed causal order is represented by a set of positive semidefinite matrices $\{W^{[r]} \in \L(\HS^{PA_\N^{IO} F})\}_r$, whose sum $W \coloneqq \sum_r W^{[r]}$ is the process matrix of a quantum circuit with the same fixed causal order (as characterised in Proposition~\ref{prop:charact_W_QCFO}).

Conversely, any set of positive semidefinite matrices $\{W^{[r]} \in \L(\HS^{PA_\N^{IO} F})\}_r$ whose sum is the process matrix of a QC-FO represents a probabilistic quantum circuit with the same fixed causal order.
\end{proposition}

That is, for a probabilistic quantum circuit with the fixed causal order $(\A_1,\ldots,\A_N)$, the reduced matrices $W_{(n)}$ of $W \coloneqq \sum_r W^{[r]}$, as defined in Proposition~\ref{prop:charact_W_QCFO}, satisfy the constraints of Eq.~\eqref{eq:charact_W_QCFO}.
A given matrix $W^{[r]} \in \L(\HS^{PA_\N^{IO} F})$ is then a probabilistic process matrix describing a particular realisation of a pQC-FO if and only if it is an element of a pQC-FO $\{W^{[r]}\}_r$, characterised as in Proposition~\ref{prop:charact_Wr_probQCFO} above.

\medskip

As a simple example of a pQC-FO, we revisit the first example from Sec.~\ref{subsec:QCFO_examples}, where the two external operations are applied one after the other. 
Let us add a measurement in the computational basis after the second external operation, with the post-measurement state being sent to the global future. 
That is, the identity channel from $\HS^{A_2^O}$ to $\HS^F$ gets replaced by a quantum instrument $\{\M_{3}^{[i]}\}_i$, with $\M_{3}^{[i]}: \L(\HS^{A^O_2}) \to \L(\HS^F)$ given by $\M_{3}^{[i]}(\rho) = \Tr\big[\rho \cdot \ketbra{i}{i}^{A_2^O}\big] \ketbra{i}{i}^F$, with Choi matrices $M_{3}^{[i]} = \ketbra{i}{i}^{A^O_2} \otimes \ketbra{i}{i}^{F}$ and where $\{\ket{i}^{A_2^O}\}_i$ and $\{\ket{i}^F\}_i$ are the computational bases (in one-to-one correspondence) of $\HS^{A_2^O}$ and $\HS^F$. 
According to Proposition~\ref{prop:descr_Wr_probQCFO}, the probabilistic process matrix that describes the specific realisation of such a pQC-FO, corresponding to a particular measurement outcome $i$, is 
\begin{equation}
    W_{P \text{\scalebox{.4}[1]{$\to$}} A_1 \text{\scalebox{.4}[1]{$\to$}} A_2 \text{\scalebox{.4}[1]{$\to$}} F}^{[i]} = \dketbra{\id}{\id}^{P A^I_1} \otimes \dketbra{\id}{\id}^{A^O_1 A^I_2} \otimes \ketbra{i}{i}^{A^O_2} \otimes \ketbra{i}{i}^{F}.
\end{equation}
The entire pQC-FO is described by the set $\{W_{P \text{\scalebox{.4}[1]{$\to$}} A_1 \text{\scalebox{.4}[1]{$\to$}} A_2 \text{\scalebox{.4}[1]{$\to$}} F}^{[i]}\}_i$,  and it is straightforward to check that it satisfies the characterisation of Proposition~\ref{prop:charact_Wr_probQCFO}.

\medskip

The particular case with operations used in parallel is discussed in Appendix~\ref{app:QCPARs}. 
There, we outline a proof of the following characterisation of probabilistic quantum circuits with operations used in parallel (pQC-PARs):

\begin{proposition}[Characterisation of pQC-PARs] \label{prop:charact_Wr_probQCPAR}
A probabilistic quantum circuit with operations used in parallel is represented by a set of positive semidefinite matrices $\{W^{[r]} \in \L(\HS^{PA_\N^{IO} F})\}_r$, whose sum $W \coloneqq \sum_r W^{[r]}$ is the process matrix of a quantum circuit with operations used in parallel (as characterised in Proposition~\ref{prop:charact_W_QCPAR}).

Conversely, any set of positive semidefinite matrices $\{W^{[r]} \in \L(\HS^{PA_\N^{IO} F})\}_r$ whose sum is the process matrix of a QC-PAR represents a probabilistic quantum circuit with operations used in parallel.
\end{proposition}

\subsection{Probabilistic quantum circuits with classical control of causal order}
\label{sec:probQCCCs}

To obtain the probabilistic counterpart of QC-CCs, the deterministic objects that need to be replaced by probabilistic ones are the CPTP maps obtained by summing up the elements of the circuit instruments $\{\M_\emptyset^{\to k_1}\}_{k_1 \in \N}$ etc. 
Equivalently, one replaces these circuit instruments by more ``fine-grained'' instruments that admit additional classical outcomes. 
A particular realisation of such a circuit thus starts with a CP map $\M_\emptyset^{\to k_1[r_1]}: \L(\HS^{P}) \to \L(\HS^{A_{k_1}^I\alpha_1})$, which is part of a quantum instrument $\{\M_\emptyset^{\to k_1[r_1]}\}_{k_1 \in \N;r_1}$ with a classical output value $k_1$ that determines the first external operation to be applied, as well as an additional classical output $r_1$. 
Similarly, the subsequent circuit maps are given by $\M_{(k_1,\ldots,k_n)}^{\to k_{n+1}[r_{n+1}]}: \L(\HS^{A_{k_n}^O \alpha_n}) \to \L(\HS^{A_{k_{n+1}}^I \alpha_{n+1}})$, and belong to instruments $\{\M_{(k_1,\ldots,k_n)}^{\to k_{n+1}[r_{n+1}]}\}_{k_{n+1} \in \N \backslash \{k_1, \ldots, k_n\};r_{n+1}}$.%
\footnote{The set of additional, fine-grained classical outcomes could in principle be conditioned on $(k_1,\ldots,k_n,k_{n+1})$. However, we can always take these sets to be the same by taking their union and appending null elements to the instruments where necessary.}

Similarly to the QC-FO case above, the fine-grained outcomes can be encoded in the ancillary systems of the circuit, and the post-selection can be deferred to the last map. 
This gives rise to a circuit which is deterministic except for the last internal operation before $F$---this gets replaced by CP maps $\M_{(k_1,\ldots,k_N)}^{\to F\,[r]}$, which belong to instruments $\{\M_{(k_1,\ldots,k_N)}^{\to F\,[r]}\}_r$. 
The TP conditions satisfied by the internal operations are thus given by Eqs.~\eqref{eq:TP_constr_QCCC_1}--\eqref{eq:TP_constr_QCCC_n}, together with the TP conditions for the last internal operation, which are obtained by replacing $M_{(k_1,\ldots,k_N)}^{\to F}$ by $\sum_r M_{(k_1,\ldots,k_N)}^{\to F\,[r]}$ in Eq.~\eqref{eq:TP_constr_QCCC_N} (see Appendix~\ref{app:proof_charact_QCCCs}). 

We formally call any process of the kind described, with internal circuit operations $\{\M_\emptyset^{\to k_1}\}_{k_1}$, $\{\M_{(k_1,\ldots,k_n)}^{\to k_{n+1}}\}_{k_{n+1}}$ and $\{\M_{(k_1,\ldots,k_N)}^{\to F\,[r]}\}_r$ that are composed as in Fig.~\ref{fig:QCCCv1} and that satisfy these trace preserving conditions, a \emph{probabilistic quantum circuit with classical control of causal order} (pQC-CC). 
The process matrix description of pQC-CCs is obtained similarly to Proposition~\ref{prop:descr_W_QCCC}.
\begin{proposition}[Process matrix description of pQC-CCs] \label{prop:descr_Wr_probQCCC}
The probabilistic process matrix describing the specific realisation of such a pQC-CC, corresponding to the classical outcome $r$, is given by
\begin{align}
W^{[r]} = \sum_{(k_1,\ldots,k_N)} W_{(k_1,\ldots,k_N,F)}^{[r]} \label{eq:Wr_probQCCC}
\end{align}
with
\begin{align}
& W_{(k_1,\ldots,k_N,F)}^{[r]} \coloneqq M_\emptyset^{\to k_1} * M_{(k_1)}^{\to k_2} * M_{(k_1,k_2)}^{\to k_3} * \cdots \notag \\
& \hspace{27mm} \cdots * M_{(k_1,\ldots,k_{N-1})}^{\to k_N} * M_{(k_1,\ldots,k_N)}^{\to F\,[r]}. \label{eq:Wr_k1_kN_probQCCC}
\end{align}
The entire pQC-CC is described by the set $\{W^{[r]}\}_r$ of all such probabilistic process matrices, for all classical outcomes $r$. 
\end{proposition}

Probabilistic quantum circuits with classical control of causal order can then be characterised as follows:

\begin{proposition}[Characterisation of pQC-CCs] \label{prop:charact_Wr_probQCCC}
A probabilistic quantum circuit with classical control of causal order is represented by a set of positive semidefinite matrices $\{W^{[r]}\}_r$, where the matrices $W^{[r]} \in \L(\HS^{PA_\N^{IO} F})$ can be decomposed in terms of positive semidefinite matrices $W_{(k_1,\ldots,k_n)} \in \L(\HS^{PA_{\{k_1,\ldots,k_{n-1}\}}^{IO} A_{k_n}^I})$ and $W_{(k_1,\ldots,k_N,F)}^{[r]} \in \L(\HS^{PA_\N^{IO} F})$, in such a way that
\begin{align}
\forall \, r, \quad W^{[r]} = \sum_{(k_1,\ldots,k_N)} W_{(k_1,\ldots,k_{N},F)}^{[r]}, \label{eq:charact_Wr_probQCCC}
\end{align}
and such that the matrices $W_{(k_1,\ldots,k_n)}$ and $W_{(k_1,\ldots,k_N,F)} \coloneqq \sum_r W_{(k_1,\ldots,k_N,F)}^{[r]}$ satisfy Eq.~\eqref{eq:charact_W_QCCC_decomp_constr} of Proposition~\ref{prop:charact_W_QCCC}.

Conversely, any set $\{W^{[r]}\}_r$ of positive semidefinite matrices with the properties above represents a probabilistic quantum circuit with classical control of causal order.  
\end{proposition}

The proof extends directly from that of Proposition~\ref{prop:charact_W_QCCC}; see Appendix~\ref{app:proof_charact_QCCCs}. 
We then have that a given matrix $W^{[r]} \in \L(\HS^{PA_\N^{IO} F})$ is a probabilistic process matrix describing a particular realisation of a probabilistic QC-CC if and only if it is an element of a pQC-CC $\{W^{[r]}\}_r$, characterised as in Proposition~\ref{prop:charact_Wr_probQCCC} above.

Note that, contrary to the case of pQC-FOs and pQC-PARs above (Propositions~\ref{prop:charact_Wr_probQCFO} and~\ref{prop:charact_Wr_probQCPAR}) and to the cases of pQC-QCs and pGENs discussed below (Propositions~\ref{prop:charact_Wr_probQCQC} and~\ref{prop:charact_Wr_probIND}), simply requiring that the sum over the classical outcomes should yield the process matrix of a QC-CC is not sufficient for a set $\{W^{[r]}\}_r$ to have a realisation as a pQC-CC. 
A counterexample that satisfies this weaker condition, but not the stronger constraints of Proposition~\ref{prop:charact_Wr_probQCCC}, is discussed at the end of Sec.~\ref{sec:probQCQCs} below (see also  Footnote~\ref{footnote:specific_probQCCC} in the proof in Appendix~\ref{app:proof_charact_QCCCs}).

\medskip

We already encountered an example of a pQC-CC earlier in this paper. 
The matrices $W_{(k_1,\ldots,k_N,F)}$ introduced in Sec.~\ref{sec:QCCCs} that describe the particular realisation of a QC-CC where the order ends up being $(k_1,k_2,\ldots,k_N)$ are probabilistic process matrices, and the set of all such  matrices, for all possible orders, constitutes a pQC-CC.
It describes the situation where the additional, fine-grained outcomes on which one post-selects coincide with the outcomes $k_1,\ldots,k_N$ that determine the order of the external operations. 
Formally, the classical outcomes $r$ are taken to be the ordered sequences $(k_1,\ldots,k_N)$ of elements in $\N$, and the last internal CP maps are given by $\M_{(k_1,\ldots,k_N)}^{\to F\,[r]} = \delta_{r,(k_1,\ldots,k_N)} \, \M_{(k_1,\ldots,k_N)}^{\to F}$. 
According to Proposition~\ref{prop:descr_Wr_probQCCC}, the process matrix description of the pQC-CC thus obtained is $\{ W^{[r=(k_1,\ldots,k_N)]}\}_{(k_1,\ldots,k_N)}$, with $W^{[r=(k_1,\ldots,k_N)]} = W_{(k_1,\ldots,k_N,F)}$.

\subsection{Probabilistic quantum circuits with quantum control of causal order}
\label{sec:probQCQCs}

As we did for the previous classes of circuits, we can once again use the circuit's ancillary systems to encode the classical outcomes (and defer the post-selection to the last operation only), and to purify all internal operations (so that all circuit operations, except for the last one, have a single Kraus operator). Without loss of generality, we can thus describe a probabilistic quantum circuit with quantum control of causal order as we did for a (deterministic) QC-QC in Sec.~\ref{subsec:QCQC_descr}, with circuit operations $\tilde V_1$, $\tilde V_{n+1}$ as in Eqs.~\eqref{eq:def_V1_QCQC}--\eqref{eq:tilde_Vn1_QCQC}, and simply replacing $\tilde V_{N+1}$ in Eq.~\eqref{eq:def_VN1_QCQC} by a set of operators
\begin{align}
\tilde V^{[r]}_{N+1} & \coloneqq \sum_{k_N} \tilde V_{\N\backslash k_N,k_N}^{\to F\,[r]} \otimes \bra{\N\backslash k_N,k_N}^{C_N'}, \label{eq:def_VN1r_probQCQC}
\end{align}
each corresponding to the classical outcome $r$ of the circuit.

The first $N$ operations $\tilde V_n$ are required to satisfy the TP conditions of Eqs.~\eqref{eq:TP_constr_QCQC_1}--\eqref{eq:TP_constr_QCQC_n}, as before. For the final internal circuit instrument, it is now the map $\varrho \mapsto \sum_r \tilde{V}^{[r]}_{N+1} \varrho \tilde{V}^{[r]}_{N+1}$ (rather than $\varrho \mapsto \tilde{V}_{N+1} \varrho \tilde{V}_{N+1}$) that must be TP. The corresponding TP condition is simply obtained (see Appendix~\ref{app:proof_charact_QCQCs}) by replacing $\Tr_{F\alpha_F} \ketbra{w_{(\N,F)}}{w_{(\N,F)}}$ by $\sum_r \Tr_{F\alpha_F} \ketbra{w_{(\N,F)}^{[r]}}{w_{(\N,F)}^{[r]}}$ in Eq.~\eqref{eq:TP_constr_QCQC_N}, with $\ket{w_{(\N,F)}^{[r]}}$ given below in Eq.~\eqref{eq:Wr_probQCQC}. 
That is, it reads
\begin{align}
& \sum_r \Tr_{F\alpha_F} \ketbra{w_{(\N,F)}^{[r]}}{w_{(\N,F)}^{[r]}} \notag \\[-2mm]
& \hspace{5mm} = \!\! \sum_{k_N \in \N} \!\!\! \Tr_{\alpha_N} \! \ketbra{w_{(\N\backslash k_N,k_N)}}{w_{(\N\backslash k_N,k_N)}} \otimes \id^{A_{k_N}^O}. \label{eq:TP_constr_probQCQC_N}
\end{align}

We formally call any process abiding by the above description, with internal circuit operations $\tilde V_1$, $\tilde V_{n+1}$ and $\tilde V_{N+1}^{[r]}$ given by Eqs.~\eqref{eq:def_V1_QCQC}, Eq.~\eqref{eq:tilde_Vn1_QCQC} and Eq.~\eqref{eq:def_VN1r_probQCQC}, respectively, which are composed as in Fig.~\ref{fig:QCQC} and which satisfy the TP conditions of Eqs.~\eqref{eq:TP_constr_QCQC_1},~\eqref{eq:TP_constr_QCQC_n}, and~\eqref{eq:TP_constr_probQCQC_N}, a \emph{probabilistic quantum circuit with quantum control of causal order} (pQC-QC).
Similarly to Proposition~\ref{prop:descr_W_QCQC}, we obtain the process matrix description of pQC-QCs as follows.

\begin{proposition}[Process matrix description of pQC-QCs] \label{prop:descr_Wr_probQCQC}
The probabilistic process matrix describing the particular realisation of such a pQC-QC, corresponding to the measurement outcome $r$, is given by
\begin{align}
\label{eq:Wr_probQCQC}
& W^{[r]} = \Tr_{\alpha_F}\ketbra{w_{(\N,F)}^{[r]}}{w_{(\N,F)}^{[r]}} \ \text{with} \notag \\[2mm]
& \ket{w_{(\N,F)}^{[r]}} \notag \\
& \quad \coloneqq  \!\! \sum_{(k_1,\ldots,k_N)} \!\! \dket{V_{\emptyset,\emptyset}^{\to k_1}} * \dket{V_{\emptyset,k_1}^{\to k_2}} * \dket{V_{\{k_1\},k_2}^{\to k_3}} * \cdots \notag \\[-3mm]
& \hspace{20mm} \ \cdots \!*\! \dket{V_{\!\{k_1\!,\ldots,k_{N-2}\}\!,k_{N-1}\!}^{\to k_N}} \!*\! \dket{V_{\!\{k_1\!,\ldots,k_{N-1}\}\!,k_N\!}^{\to F\,[r]}}\!.
\end{align}
The entire pQC-QC is described by the set $\{W^{[r]}\}_r$ of all such probabilistic process matrices, for all classical outcomes $r$. 
\end{proposition}

The following proposition then characterises probabilistic quantum circuits with quantum control of causal orders: 

\begin{proposition}[Characterisation of pQC-QCs] \label{prop:charact_Wr_probQCQC}
A probabilistic quantum circuit with quantum control of causal order is represented by a set of positive semidefinite matrices $\{W^{[r]} \in \L(\HS^{PA_\N^{IO} F})\}_r$, whose sum $W \coloneqq \sum_r W^{[r]}$ is the process matrix of a quantum circuit with quantum control of causal order (as characterised in Proposition~\ref{prop:charact_W_QCQC}).

Conversely, any set of positive semidefinite matrices $\{W^{[r]} \in \L(\HS^{PA_\N^{IO} F})\}_r$ whose sum is the process matrix of a QC-QC represents a probabilistic quantum circuit with quantum control of causal order.
\end{proposition}

The proof extends directly from that of Proposition~\ref{prop:charact_W_QCQC}; see Appendix~\ref{app:proof_charact_QCQCs}.
We then have that a given matrix $W^{[r]} \in \L(\HS^{PA_\N^{IO} F})$ is the probabilistic process matrix describing a particular realisation of a pQC-QC if and only it is an element of a pQC-QC $\{W^{[r]}\}_r$, characterised as in Proposition~\ref{prop:charact_Wr_probQCQC} above.

\medskip

Similarly to the example that we discussed above for the QC-CC case, the matrices $W_{(\N \backslash k_N,k_N)}$ in Proposition~\ref{prop:charact_W_QCQC} are probabilistic process matrices that constitute a pQC-QC. 
They are realised when the last internal operation $\tilde V_{N+1}$ is replaced by an operation that measures the control system $C_N'$ (while also transforming the output state of the last operation and the ancilla). 
Formally, one takes $r \in \N$ and $V_{\N\backslash k_N,k_N}^{\to F\,[r]} =  \delta_{r,k_N} V_{\N\backslash k_N,k_N}^{\to F}$ in Eq.~\eqref{eq:def_VN1r_probQCQC}. 
The pQC-QC thus obtained is $\{W^{[r=k_N]}\}_{k_N}$ with $W^{[r=k_N]} = W_{(\N \backslash k_N,k_N)}$.

As another example, let us once again consider the quantum switch, with its process matrix description Eq.~\eqref{eq:W_QS}, and where the control qubit is measured at the end of the circuit in the basis $\{\ket{+}^{F_\textup{c}},\ket{-}^{F_\textup{c}}\}$, with $\ket{\pm}^{F_\textup{c}} \coloneqq \frac{\ket{1}^{F_\textup{c}}\pm\ket{2}^{F_\textup{c}}}{\sqrt{2}}$. 
The internal operations that constitute the corresponding probabilistic quantum circuit are (in their Choi representation) $\dket{V_{\emptyset,\emptyset}^{\to k_1}} = \ket{k_1}^{P_\text{c}} \otimes \dket{\id}^{P_\text{t}A_{k_1}^I}$ and $\dket{V_{\emptyset,k_1}^{\to k_2}} = \dket{\id}^{A_{k_1}^OA_{k_2}^I}$ as in Eq.~\eqref{eq:Vs_QSwitch}, and now $\dket{V_{\{k_1\},k_2}^{\to F\,[\pm]}} = \frac{(\pm 1)^{k_2}}{\sqrt{2}}\dket{\id}^{A_{k_2}^OF_\text{t}}$.

The corresponding probabilistic process matrix description is therefore, according to Proposition~\ref{prop:descr_Wr_probQCQC}, $\{W_\textup{QS}^{[+]},W_\textup{QS}^{[-]}\}$ with $W_\textup{QS}^{[\pm]} = \ketbra{w_\textup{QS}^{[\pm]}}{w_\textup{QS}^{[\pm]}}$, and 
\begin{align}
    \ket{w_\textup{QS}^{[\pm]}} = & \frac{1}{\sqrt{2}}\Big(\ket{1}^{P_\text{c}} \dket{\id}^{P_\text{t}A_1^I} \dket{\id}^{A_1^OA_2^I} \dket{\id}^{A_2^OF_\text{t}} \notag \\[-1mm]
    & \qquad \pm \ket{2}^{P_\text{c}} \dket{\id}^{P_\text{t}A_2^I} \dket{\id}^{A_2^OA_1^I} \dket{\id}^{A_1^OF_\text{t}}\Big) .
\end{align}
We then have that $W_\textup{QS}^{[+]}+W_\textup{QS}^{[-]} = \Tr_{F_\text{c}} W_\textup{QS}$ is indeed the process matrix of a QC-QC. 
In fact, it is even the process matrix of a QC-CC,%
\footnote{One may more generally define a subclass of pQC-QCs, strictly larger than that of pQC-CCs, containing the sets of positive semidefinite matrices $\{W^{[r]}\}_r$ that sum up to the process matrix of a QC-CC~\cite{bavaresco20}. The physical interpretation of this subclass remains in general to be clarified.}
since $\Tr_{F_\text{c}} W_\textup{QS} = \Tr_{F_\text{c}} W_\textup{CS}$, with $W_\textup{CS}$ the process matrix of the classical switch (Eq.~\eqref{eq:W_CS}). 
Nevertheless, $\{W_\textup{QS}^{[+]},W_\textup{QS}^{[-]}\}$ is not a probabilistic QC-CC. In order to realise it, the two causal orders need to be coherently superposed in the switch before the control qubit is measured.

To see this, note that the matrices $W_\textup{QS}^{[\pm]}$ do not satisfy the additional constraints in Proposition~\ref{prop:charact_Wr_probQCCC}. 
That is, $W_\textup{QS}^{[\pm]}$ cannot be decomposed as $W_\textup{QS}^{[\pm]} = W^{[\pm]}_{(1,2,F)}+W^{[\pm]}_{(2,1,F)}$, such that, with $W_{(1,2,F)} \coloneqq W^{[+]}_{(1,2,F)} + W^{[-]}_{(1,2,F)}$ and $W_{(2,1,F)} \coloneqq W^{[+]}_{(2,1,F)} + W^{[-]}_{(2,1,F)}$, we obtain a decomposition of $W_\textup{QS}^{[+]}+W_\textup{QS}^{[-]} = W_{(1,2,F)}+W_{(2,1,F)}$ as in Proposition~\ref{prop:charact_W_QCCC}.
This follows from the fact that $W^{[\pm]}$ are rank-one projectors, and can therefore not be further decomposed into a (nontrivial) sum of positive semidefinite matrices, and neither $W_\textup{QS}^{[+]}$ nor $W_\textup{QS}^{[-]}$ satisfies individually the constraints on either $W_{(1,2,F)}$ or $W_{(2,1,F)}$ in Eq.~\eqref{eq:charact_W_QCCC_decomp_constr}.

Note that the example we described here is precisely the probabilistic quantum circuit that one uses in the canonical application of the quantum switch, a task where one has two unitaries that either commute or anticommute, and the aim of which is to determine which of the two properties holds true~\cite{chiribella12}. 
The pQC-QC described here allows one to discriminate between the two cases with certainty, while this is not possible with a pQC-CC~\cite{chiribella12,araujo15}.

We thus recover straightforwardly this known advantage of the quantum switch over causally separable processes.
However, the characterisation of the full class of QC-QCs and of their probabilistic versions now also allows us to go beyond that simple, canonical example, and to search for new applications of physically realisable, causally nonseparable processes in a more systematic way.
In Section~\ref{sec:applications}, we will illustrate this through a specific example.

\subsection{General quantum superinstruments}

As already mentioned in Sec.~\ref{subsec:PM_supermaps}, one can also consider probabilistic supermaps in the most general situation, where it is not specified \emph{a priori} how the external operations are to be connected. 
A particular such probabilistic supermap takes the $N$ external operations $\A_k$ to a CP map $\M^{[r]}: \L(\HS^P) \to \L(\HS^F)$, associated to a classical output $r$, such that summing over all $r$ yields a deterministic supermap (i.e., such that the induced map $\sum_r \M^{[r]} (\cdot)$ is TP whenever all external operations are TP). 
We call the set of such probabilistic supermaps for all classical outputs $r$ a \emph{(general) quantum superinstrument} (pGEN). Such general quantum superinstruments have previously been characterised in Refs.~\cite{quintino19,quintino19a}. 

It follows from Eq.~\eqref{eq:choi_map_M} that the process matrix description of a particular realisation of a pGEN is given by a positive semidefinite matrix $W^{[r]}$, with the sum over all $W^{[r]}$ being a valid (deterministic) process matrix.
We therefore have the following characterisation: 

\begin{proposition}[Characterisation of pGENs] \label{prop:charact_Wr_probIND}
A general quantum superinstrument is represented by a set of positive semidefinite matrices $\{W^{[r]} \in \L(\HS^{PA_\N^{IO} F})\}_r$, whose sum $W \coloneqq \sum_r W^{[r]}$ is a valid process matrix. 

Conversely, any set of positive semidefinite matrices $\{W^{[r]} \in \L(\HS^{PA_\N^{IO} F})\}_r$ whose sum is a valid process matrix represents a general quantum superinstrument.
\end{proposition}

\section{Applications}
\label{sec:applications}

One of the motivations for the investigation of quantum causal structures is the prospect that indefinite causal orders could enable new quantum information processing tasks and protocols, and that causal nonseparability could be used as an information processing resource~\cite{chiribella13}. 
Indeed, some advantages in this respect have recently been identified, for instance in regard to quantum query complexity~\cite{chiribella12,colnaghi12,araujo14,facchini15,taddei20}, quantum communication complexity~\cite{feix15,guerin16} and other information processing tasks~\cite{ebler18,salek18,chiribella18,mukhopadhyay18,
mukhopadhyay20,procopio19,frey19,loizeau20,caleffi20,gupta19,procopio20,
zhao20,felce20,guha20,sazim20,wilson20a}. 
These studies have focused particularly on the quantum switch and its straightforward $N$-operation generalisation, since these were so far the only known examples of causally nonseparable processes with a physical interpretation. 

The process matrix descriptions of the different classes of circuits we introduce here, as well as their probabilistic versions, allow us to more systematically search for advantages in quantum information processing arising from causal nonseparability. 
By finding tasks for which QC-QCs provide an advantage over circuits with definite causal order, we can thereby identify new applications of causally nonseparable processes that are more general than the quantum switch and for which a physical implementation scheme exists. 

One natural type of information processing tasks in the context of higher-order maps are ``higher-order quantum computation'' problems, such as the cloning~\cite{chiribella08}, the storage and retrieval~\cite{chiribella08}, or the replication of the inverse or transpose~\cite{chiribella16,quintino19,quintino19a} of some undisclosed, black-box operation of which one or multiple copies are available. 
Another natural type of tasks are generalised channel discrimination problems, in which one is given some black-box operations which, collectively, belong to one of a finite number of classes (or, equivalently, are promised to obey one of several properties). 
Examples are the discrimination of phase relations between unitary operations~\cite{chiribella12,araujo14,taddei20} or the discrimination between different cause and effect structures of unitaries~\cite{chiribella19b}. 
To quantify how a given class of circuits performs for some task, one needs to optimise over the corresponding higher order transformations in order to maximise some figure of merit, such as the channel fidelity between the desired ``target'' channel and the output of the supermap, or the success probability of the task. 

The characterisations provided in this article allow us to optimise the performance of different classes of circuits in both these types of problems by exploiting semidefinite programming (SDP) techniques.
In this section, we present a concrete example of one such problem and show, in particular, a gap between the performance of probabilistic QC-CCs and QC-QCs.
The task is a natural generalisation of the discrimination task studied in Ref.~\cite{shimbo18} which we call the \emph{$K$-unitary equivalence determination problem}.
One is given $K$ \emph{reference boxes} which implement black-box unitary operations $U_1,\dots,U_K$, and a further \emph{target box} that implements one of the $U_k$ ($1\le k \le K$) with probability $1/K$.
The aim is to determine which of the reference boxes is implemented by the target box, while using each of the $K+1$ boxes exactly once.
For simplicity (as in Ref.~\cite{shimbo18}) we shall consider the case where the boxes all implement qubit unitaries, with the reference boxes chosen randomly according to the Haar measure on $\mathrm{SU}(2)$.%
\footnote{Note that the problem considered in Ref.~\cite{shimbo18} corresponds to the case of $K=2$, although they also considered the possibility of receiving $n_k\ge 1$ copies of each reference box $U_k$. 
While we could also consider such a generalisation, we focus on the case of $n_k=1$ for simplicity, and because the SDP methods we employ would quickly become intractable when more copies are considered.}

Let us denote the input and output spaces of the reference boxes as $\HS^{A_k^I}$ and $\HS^{A^O_k}$ (with $k = 1,\dots,K$), those of the target box as $\HS^{A_{K+1}^I}$ and $\HS^{A_{K+1}^O}$, and the probabilistic quantum circuit that we have at our disposal by $\{W^{[r]}\}_{r=1,\dots,K}$, with $W^{[r]} \in  \L(\HS^{A^{IO}_{\{1,\ldots,K+1\}}})$, where $P$ and $F$ are trivial, and the outcome $r$ of the probabilistic circuit corresponds to the guess of which reference box is implemented by the target box. 
The success probability (for a specific choice of $U_1,\dots,U_K$) is then
\begin{align}
\label{eq:success_probability_fixedU}
    p_{(U_1,\dots,U_K)} =& \frac{1}{K} \sum_{r = 1}^K S(U_1,\dots,U_K,U_r) * W^{[r]} ,
\end{align}
where $S(U_1,\dots,U_K,U_r) \coloneqq \dketbra{U_1}{U_1} \otimes \cdots \otimes \dketbra{U_K}{U_K} \otimes \dketbra{U_r}{U_r}$ (which lives in the same space $\L(\HS^{A^{IO}_{\{1,\ldots,K+1\}}})$ as $W^{[r]}$, so that the link product above returns a scalar value, as required).

For the case we are considering of Haar random $U_k$'s, the probability of success is obtained by averaging Eq.~\eqref{eq:success_probability_fixedU} over the (normalised) Haar measure $\mu$, giving
\begin{align}
\label{eq:success_probability}
    p_\text{succ} =& \int\!\cdots\!\int \mathrm{d}\mu(U_1)\cdots \mathrm{d}\mu(U_K)\, p_{(U_1,\dots,U_K)} \notag\\
	=& \frac{1}{K} \sum_{r = 1}^K \tilde{S}_r * W^{[r]},
\end{align}
where $\tilde{S}_r \coloneqq \int\!\cdots\!\int \mathrm{d}\mu(U_1)\cdots \mathrm{d}\mu(U_K)\, S(U_1,\dots,U_K,U_r)$.
For qubits, the $\tilde{S}_r$'s can be calculated analytically by using the fact that, for $U\in\mathrm{SU}(2)$, one has $\int  \mathrm{d}\mu(U) \dketbra{U}{U}  = \frac{1}{2} \id$ and $\int \mathrm{d}\mu(U) \dketbra{U}{U} \otimes \dketbra{U}{U} = \frac{1}{4} (\id + \frac{1}{3} \sum_{i,j} \sigma_i \otimes \sigma_j \otimes \sigma_i \otimes \sigma_j)$ (where the $\sigma_i$'s are the 3 Pauli matrices $\sigma_x,\sigma_y,\sigma_z$).

For a given class $\text{pX}_K$ of $K$-outcome probabilistic quantum circuits, with $\text{pX}\in\{\text{pQC-PAR},\text{pQC-FO},\text{pQC-CC},\text{pQC-QC},\text{pGEN}\}$, the problem is thus to find
\begin{align}\label{eq:psucc_SDP}
	p^\text{X}_\text{succ} = \max & \quad p_\text{succ}\notag\\
	\text{s.t.} & \quad \{W^{[r]}\}_{r=1,\ldots,K} \in \text{pX}_K.
\end{align}
For each of the classes $\text{pX}_K$ specified above, characterised by one of the Propositions~\ref{prop:charact_Wr_probQCFO}, \ref{prop:charact_Wr_probQCPAR}, \ref{prop:charact_Wr_probQCCC}, \ref{prop:charact_Wr_probQCQC} or~\ref{prop:charact_Wr_probIND}, this optimisation task is an SDP problem and is thus tractable for small enough $K$.

For $K=2$, Ref.~\cite{shimbo18} found $p_\text{succ}^\text{QC-PAR}=p_\text{succ}^\text{QC-FO} = 0.875$. 
Using the SDP solver SCS~\cite{donoghue2016,scs_code}, we found that no improvement over this was possible even with general quantum superinstruments (and thus also for the classes of probabilistic QC-CCs and QC-QCs since $p_\text{succ}^\text{QC-PAR} \le p_\text{succ}^\text{QC-FO} \le p_\text{succ}^\text{QC-CC} \le p_\text{succ}^\text{QC-QC} \le p_\text{succ}^\text{GEN}$).%
\footnote{We briefly note, however, that if one chooses a different distribution from the Haar measure for the reference boxes, one can observe advantages even for $K=2$.
For example, if each $U_i$ is drawn uniformly at random from the finite set $\mathcal{U}=\{\sigma_y,R_y,T\}$ (where $R_y = \frac{1}{\sqrt{2}} \begin{psmallmatrix}1 & -1\\ 1 & 1\end{psmallmatrix}$ is a Bloch sphere rotation of $\pi/2$ around the $y$ axis and $T$ is the phase shift gate $\begin{psmallmatrix}1 & 0\\ 0 & e^{i \pi/4}\end{psmallmatrix}$) one obtains a strict separation between all the classes considered.
}
For $K=3$, however, we found a (admittedly small, but still) strict separation between all the classes of probabilistic quantum circuits except pQC-FOs and pQC-CCs (indicating that dynamical definite causal order provides no advantage in the $3$-unitary equivalence determination problem).
The results are summarised in Table~\ref{tab:res_1}.
Finally, we note that, for $K>3$, the SDP problem \eqref{eq:psucc_SDP} became too large for us to solve.%
\footnote{Note that each $W^{[r]}$ is a $2^{2(K+1)} \times 2^{2(K+1)}$ matrix. Even for $K=3$, using a start-of-the-art, memory efficient SDP solver~\cite{donoghue2016,scs_code} the SDP required up to 25GB of RAM (depending on the class of circuits) and several hours to solve.}

\begin{table}[h]
    \centering
	\setlength{\tabcolsep}{5pt}
\begin{tabular}{ c|ccccc  }
 % \hline
$K$ & pQC-PAR & pQC-FO & pQC-CC & pQC-QC & pGEN\\
 \hline
 2  & $0.875^*$ & $0.875^*$ &  0.875 & 0.875 & 0.875\\
 3 & 0.6919 & 0.6998 & 0.6998 & 0.7080 & 0.7093\\
\end{tabular}
    \caption{Maximal success probability $p_\text{succ}$ for the $K$-unitary equivalence determination problem for $K=2,3$ for probabilistic quantum circuits from the classes: pQC-PAR (probabilistic quantum circuits with operations used in parallel), pQC-FO (probabilistic quantum circuits compatible with one of the $(K{+}1)!$ fixed causal orders), pQC-CC (probabilistic quantum circuits with classical control of causal order), pQC-QC (probabilistic quantum circuits with quantum control of causal order), pGEN (general, possibly causally indefinite, quantum superinstruments). Starred figures were already given in Ref.~\cite{shimbo18}.}
    \label{tab:res_1}
\end{table}

This shows that causal indefiniteness is indeed a resource for the $3$-unitary equivalence determination problem, and moreover that an advantage can be obtained using probabilistic QC-QCs, i.e., in a way that is physically realisable, at least in principle.
This is in contrast with other problems such as the exact probabilistic reversal of an unknown unitary operation. For the particular instances of that problem studied in Ref.~\cite{quintino19}, we found no advantage using QC-QCs, although (as shown in Ref.~\cite{quintino19}) general quantum superinstruments can provide an advantage over QC-FO ones. 
This means that the results presented in the present paper do not provide a physical interpretation of the advantage identified in Ref.~\cite{quintino19}.
We expect further study to unveil new quantum information tasks for which (probabilistic) QC-QCs provide advantages over all circuits with a definite, possibly dynamical, causal structure.

\section{Discussion}

The central question of our paper was which completely CP-preserving (CCP) quantum supermaps beyond those that correspond to standard, fixed-order quantum circuits have a physical interpretation.
A major motivation for this study was that general CCP quantum supermaps can exhibit indefinite causal order, a phenomenon which has recently attracted substantial interest and whose physical realisability is a crucial open question. 
Similarly to previous investigations that focused on the fixed-order case~\cite{chiribella08,chiribella09}, we adopted a constructive, bottom-up approach in order to find concrete realisations of more general types of quantum supermaps in terms of generalised quantum circuits. 
This first led us to introduce \emph{quantum circuits with classical control of causal order} (QC-CCs), in which the order of operations is established dynamically in a classically-controlled manner.
A crucial point in our construction was to keep track of which ``external'' input operations had already been applied, in order to ensure that each external operation is applied once and only once throughout the circuit.  
We then moved on to \emph{quantum circuits with quantum control of causal order} (QC-QCs) by including explicit control systems that encode the relevant information and by introducing coherences between the target and ancillary systems and the control.
Importantly, in the QC-QC case, we let the control system record the unordered set of previously applied operations rather than their full order, allowing different orders to ``interfere'' while still ensuring that each external operation appears once and only once in each coherent ``branch'' of the circuit.

Although we have thus-far overlooked this point, in the case of coherent control, it is no longer obvious that the latter can be understood as each external operation being applied once and only once in the overall circuit.
For the quantum switch in particular, this has led to some controversy, and it has been argued by some authors that its standard quantum-mechanical realisations should be considered \emph{simulations} rather than genuine \emph{realisations} of the corresponding supermap with indefinite causal order, given that each external operation is associated with two spacetime events~\cite{maclean17,paunkovic20}. 
In Ref.~\cite{oreshkov19}, it has been shown that the external operations in the quantum switch are indeed applied once and only once on some well-defined input and output systems. 
These systems are \emph{time-delocalised subsystems}, that is, they are nontrivial subsystems of composite systems whose constituents are associated with different times.
This argument applies also to general QC-QCs, where one can similarly identify time-delocalised input and output subsystems for all external operations, and which can therefore be seen as genuine realisations of quantum supermaps with indefinite causal order in that same sense.

All the types of generalised circuits we described correspond to distinct classes of quantum supermaps, which we fully characterised in the process matrix framework (Propositions~\ref{prop:charact_W_QCFO}, \ref{prop:charact_W_QCPAR}, \ref{prop:charact_W_QCCC} and~\ref{prop:charact_W_QCQC}). 
These characterisations in terms of convex semidefinite constraints notably allow one to verify whether a given process matrix is in a given class or not. 
Using similar techniques as for witnesses of causal nonseparability~\cite{araujo15,branciard16a,wechs19} one can, for instance, show that the classical switch does not have a fixed order, that the quantum switch cannot be described by a classical control (in that case the problem reduces to a witness of causal nonseparability), or that the process matrix $W_\text{OCB}$ originally introduced by Oreshkov \emph{et al.}~\cite{oreshkov12} or the tripartite ``classical'' example of Baumeler \emph{et al.}~\cite{baumeler14a} are not realisable as QC-QCs.\footnote{Note that this also follows from the fact that these can violate causal inequalities, while we showed that QC-QCs cannot. 
In the case of $W_\text{OCB}$, the problem again reduces to a witness of causal nonseparability, since for $N = 2$ and trivial global past and future systems, QC-QCs reduce to simple classical mixtures of the two possible orders (cf.\ Appendix~\ref{app:charact_dP_dF_1}).}  

Let us elaborate further on how the classes of quantum supermaps we identified here relate to other classes that have been studied before. 
As noted in Sec.~\ref{sec:QCCCs}, the process matrices describing QC-CCs are causally separable. 
Whether the converse holds---i.e., whether any causally separable process matrix satisfies the constraints of Proposition~\ref{prop:charact_W_QCCC} and can therefore be realised as a QC-CC---is an open problem in the general $N$-operation case~\cite{wechs19}.
A similar open question is whether the process matrices in the QC-QC class are the only process matrices that cannot violate causal inequalities, i.e., whether any \emph{extensibly causal} process matrix~\cite{oreshkov16,feix16} can be realised as a QC-QC.
Another important class is that of \emph{unitary} or \emph{pure} supermaps, which map unitary input operations to a unitary output operation. 
This class was introduced in Ref.~\cite{araujo17}, where it was argued that physically realisable supermaps should be \emph{unitarily extensible}, that is, recoverable from a unitary supermap by preparing a fixed state in some subsystem of the global past, and tracing out some subsystem of the global future.
One can check that by introducing suitable additional Hilbert spaces and suitably extending the internal circuit operations, one can find such a unitary extension for any QC-QC process matrix (and therefore also for any QC-CC, QC-FO and QC-PAR process matrix). 
For the case of two input operations, the converse also holds, that is, any unitarily extensible supermap with two input operations can be realised as a QC-QC. 
This follows from~\cite{barrett20,yokojima20}, where it was shown that all unitary supermaps with two input operations are ``variations of the quantum switch'', which can straightforwardly be verified to satisfy the characterisation of two-operation QC-QCs. 
In the general case, however, the set of unitarily extensible process matrices is strictly larger than the QC-QC class, since there exist unitarily extensible process matrices with three input operations that violate causal inequalities~\cite{araujo17}. 
This finding in fact motivated the authors of Ref.~\cite{araujo17} to suggest a bottom-up approach of the kind taken in our paper.

The fact that there remains a gap between the class of QC-QCs obtained from our bottom-up approach and the class of general quantum supermaps, which was obtained from a top-down approach by just imposing some consistency constraints, stands in contrast to the fixed-order case, where the form of QC-FOs obtained constructively and with an axiomatic approach matched~\cite{chiribella09}.
Another central question for future research is therefore whether and how quantum supermaps outside the QC-QC class can be given a physical interpretation.
In an upcoming work~\cite{wechs20}, it is shown that certain supermaps that go beyond the QC-QC class have realisations on time-delocalised subsystems as introduced in Ref.~\cite{oreshkov19}.
Note also that while we relaxed the assumption of a well-defined causal order for the external operations, there remains some well-defined causal order ``inside the circuit'', for the internal circuit operations.
One may wonder whether there could be a way to also relax this definite causal order of the internal operations, and whether it could allow to realise more general CCP supermaps. 

More generally, another direction is also to study new types of circuits beyond quantum supermaps in which the requirement that each operation should be applied once and only once is relaxed~\cite{abbott20,chiribella19,kristjansson19,vanrietvelde20}, or where the trace-preserving constraints are not required to hold for \emph{all} possible external operations, but only for some limited subsets that are allowed to be plugged in. We note in this regard that our negative result on the impossible violation of causal inequalities would still hold in the latter case (with a similar proof). 

Our approach allowed us to find qualitatively new examples of physically realisable processes with indefinite causal structure, one of which we discussed in detail in Sec.~\ref{subsubsec:new_QCQC}.
On that basis, an interesting future research direction is to devise laboratory experiments that implement such processes in practice. 
A suitable experimental platform could be photonic setups, similarly to those used in laboratory implementations of the quantum switch, with spatially separate ``boxes'' realising the operations $\A_k$, and with the control system including the path, as outlined in Sec.~\ref{subsec:implementations}. 
Other types of implementations could also be conceivable, for instance based on superconducting qubits~\cite{friis14} or trapped ions~\cite{friis15}. 

Indefinite causal order has also been speculated to arise at the interface of quantum theory and gravity, and a gravitational realisation of the quantum switch which involves a massive object in a quantum superposition of locations has been proposed as a thought experiment~\cite{zych19}. 
A natural question is whether other QC-QCs could have realisations in similar gravitational settings. 

Finally, in Sec.~\ref{sec:probQCs} we extended our characterisations to probabilistic quantum circuits (Propositions~\ref{prop:charact_Wr_probQCFO}, \ref{prop:charact_Wr_probQCPAR}, \ref{prop:charact_Wr_probQCCC},   \ref{prop:charact_Wr_probQCQC} and \ref{prop:charact_Wr_probIND}). These results open the door to a more systematic search for applications of quantum circuits beyond causally ordered ones. 
We illustrated this by an example in Sec.~\ref{sec:applications}, a discrimination problem where probabilistic QC-QCs yield a higher success probability than probabilistic QC-CCs. 
Identifying further such tasks for which QC-QCs perform better than circuits with well-defined causal order will shed more light on the usefulness of indefinite causal order for quantum information processing. 

\medskip

\emph{Note added in proof.}---%
By coincidence, on the same day as the first pre-print version of the present paper was posted on the arXiv server~\cite{wechs21}, another paper appeared, which also showed that causal inequalities cannot be violated in some similar circuit-like quantum models~\cite{purves21}.

\section*{Acknowledgements}

We thank Jessica Bavaresco, Fabio Costa, Mehdi Mhalla, Mio Murao, Ognyan Oreshkov and Marco T{\'u}lio Quintino for fruitful discussions, and acknowledge financial support from the \emph{``Investissements d'avenir''} (ANR-15-IDEX-02) program of the French National Research Agency, from the Swiss National Science Foundation (NCCR SwissMAP and Starting Grant DIAQ) and from the Program of Concerted Research Actions (ARC) of the Universit\'e Libre de Bruxelles.

\appendix

\section{Process matrices with or without ``global past'' and ``global future'' systems $P, F$}

\subsection{Equivalence between the two process matrix frameworks}
\label{app:W_equiv_descr}

Process matrices were initially introduced as the most general way to map quantum operations to probabilities in a consistent manner (so as to only output nonnegative and normalised probabilities), without assuming any \emph{a priori} global causal structure~\cite{oreshkov12}. 
Here, as in Ref.~\cite{araujo17}, we consider a slightly different version of process matrices that take the $N$ CP maps $\A_k: \L(\HS^{A_k^I}) \to \L(\HS^{A_k^O})$, with Choi representation $A_k \in \L(\HS^{A_k^{IO}})$, to a new CP map $\M: \L(\HS^P) \to \L(\HS^F)$ (rather than to some probabilities), from some ``global past'' Hilbert space $\HS^P$ to some ``global future'' Hilbert space $\HS^F$, with Choi representation $M \in \L(\HS^{PF})$; see Fig.~\ref{fig:general_W}.

Any mapping from quantum operations to probabilities, as described by a process matrix $W$ in the original version of the framework~\cite{oreshkov12}, can equivalently be seen as a deterministic CCP quantum supermap conforming to the definition in Sec.~\ref{subsec:PM_supermaps}, which acts as in Eq.~\eqref{eq:choi_map_M}, and where the ``global past'' and ``global future'' Hilbert spaces are trivial, i.e., one-dimensional ($d_P = d_F = 1$).
As mentioned in Footnote~\ref{footnote:pm_formalisms}, the ``generalised Born rule'', which yields the probabilities in the original formalism, is formally recovered by identifying in that case the (scalar) output of the induced map $\M: 1 \mapsto (A_1 \otimes \cdots \otimes A_N) * W$ with the probability distribution $P(\A_1,\ldots,\A_N)$. 
The non-negativity and normalisation of these probabilities, as imposed in the original framework, implies that the corresponding supermap must indeed be CCP and deterministic. 

Conversely, the process matrix $W$ that specifies the action of a deterministic supermap (cf.\ Eq.~\eqref{eq:choi_map_M}) can be seen as a process matrix in the original framework where one has two additional operations, one of which corresponds to a state preparation in the ``global past'' Hilbert space (i.e., its output Hilbert space is $\HS^P$ and its input Hilbert space is trivial), and the other to a measurement of the output system in the ``global future'' Hilbert space (i.e., its input Hilbert space is $\HS^F$ and its output Hilbert space is trivial). 
The constraints on a CCP and deterministic quantum supermap imply that one indeed obtains a mapping to valid (nonnegative and normalised) probabilities, as required in the original version of the process matrix framework, when these two additional operations are included.

\subsection{Validity conditions for process matrices}
\label{app:W_valid_matrices}

The requirement that process matrices must yield nonnegative and normalised probabilities can be expressed more directly in terms of some simple conditions that these matrices must satisfy. 
These validity constraints were first derived for the case of two operations in Ref.~\cite{oreshkov12} and generalised to more complex scenarios (including the general, $N$-operation case) in Refs.~\cite{araujo15,oreshkov16,araujo17,wechs19}. 
With the equivalence of the two frameworks established above, we can use these previous characterisations in order to formulate the validity constraints for a matrix $W$ to describe a completely CP-preserving and deterministic quantum supermap as per Eq.~\eqref{eq:choi_map_M}. 

\medskip

We will use a somewhat different notation here. For any $W \in \L(\HS^{PA_\N^{IO}F})$ and any nonempty subset $\K$ of $\N$, we define the partial traces%
\footnote{As it appears from Eq.~\eqref{eq:validity_cstr}, if $W$ is a valid process matrix in $ \L(\HS^{PA_\N^{IO}F})$, then the matrices $W^{[P\K F]}$ and $W^{[P\K]}$ are also, up to normalisation, valid process matrices in the corresponding spaces.}
$W^{[P\K F]} \coloneqq \Tr_{A_{\N\backslash\K}^{IO}} W \in \L(\HS^{PA_\K^{IO}F})$, $W^{[P\K]} \coloneqq \Tr_{A_{\N\backslash\K}^{IO}F} W \in \L(\HS^{PA_\K^{IO}})$ and $W^{[P]} \coloneqq \Tr_{A_\N^{IO}F} W \in \L(\HS^P)$.
We furthermore use the ``trace-out-and-replace'' notation of Ref.~\cite{araujo15}, defined as
\begin{align}\label{eq:trade_and_pad}
		{}_{X}W \coloneqq (\Tr_X W) \otimes \frac{\id^X}{d_X}\,, \quad {}_{[1-X]}W \coloneqq W - {}_{X}W, 
\end{align}
where $d_X \coloneqq \dim \HS^X$ (and where the second part of the definition above can be applied recursively, in a commutative manner, so that, e.g., ${}_{\Pi_{k\in\{k_1,k_2,\ldots,k_n\}}[1-X_k]}W = {}_{[1-X_{k_1}]}\big({}_{\Pi_{k\in\{k_2,\ldots,k_n\}}[1-X_k]}W\big)$).

We then have that $W \in \L(\HS^{PA_\N^{IO}F})$ is a valid process matrix if and only if $W \succeq 0$, $\Tr W = d_P \, \Pi_{k\in\N} d_k^O$, and $W$ is in some subspace $\L^{P\N F}$ of $\L(\HS^{PA_\N^{IO}F})$, characterised as
\begin{align}\label{eq:validity_cstr}
W \in \L^{P\N F} & \Leftrightarrow \ \forall \ \emptyset \subsetneq \K \subseteq \N, \ {}_{ \prod_{k \in \K}[1-A_k^O]} W^{[P\K]} = 0 \notag \\
& \hspace{10mm} \text{ and } {}_{[1-P]} W^{[P]} = 0 \notag \\[2mm]
& \hspace{-7mm} \Leftrightarrow \ \forall \ \emptyset \subsetneq \K \subsetneq \N, \ W^{[P\K F]} \in \L^{P\K F}, \notag \\
& \hspace{3mm} {}_{ \prod_{k \in \N}[1-A_k^O]} W^{[P\N]} = 0 \text{ and } {}_{[1-P]} W^{[P]} = 0.
\end{align}
For trivial spaces $\HS^P$ and $\HS^F$, one directly recovers Eqs.~(A5) and~(A6) from Ref.~\cite{wechs19}.

\medskip

One can check that all (classes of) deterministic process matrices characterised in the paper (cf.\ Propositions~\ref{prop:charact_W_QCFO}, \ref{prop:charact_W_QCPAR}, \ref{prop:charact_W_QCCC} and \ref{prop:charact_W_QCQC}) satisfy these validity constraints.
To see this directly, e.g., for the QC-QC class (which contains the other classes under consideration), one can show recursively, from $|\K| = N$ down to $|\K| = 1$,%
\footnote{To prove Eq.~\eqref{eq:WPK_QCQC}, one may first show (also recursively) that for any $k\in\K$ and any $n = 0, \ldots, |\N\backslash\K|$,
\begin{align}
& \hspace{-5mm} \sum_{\emptyset \subseteq \K' \subseteq \N\backslash\K} d_{\K'k}^O \Tr_{A_{\N\backslash\K\K'}^{IO}} \Tr_{A_k^I} W_{(\N\backslash\K'k,k)} \notag \\
= & \sum_{|\K'|<n} d_{\K'k}^O \Tr_{A_{\N\backslash\K\K'}^{IO}} \sum_{k'\in\K\backslash k} W_{(\N\backslash\K'k k',k')} \otimes \id^{A_{k'}^O} \notag \\
& + \sum_{|\K'|=n} d_{\K'k}^O \Tr_{A_{\N\backslash\K\K'}^{IO}} \sum_{k''\in\K'k} \Tr_{A_{k''}^I} W_{(\N\backslash\K'k,k'')} \notag \\
& + \sum_{|\K'|>n} d_{\K'k}^O \Tr_{A_{\N\backslash\K\K'}^{IO}} \Tr_{A_k^I} W_{(\N\backslash\K'k,k)}.
\end{align}
For $n = |\N\backslash\K|$, this is then equal to
\begin{align}
\sum_{\emptyset \subseteq \K' \subseteq \N\backslash\K} d_{\K'k}^O \Tr_{A_{\N\backslash\K\K'}^{IO}} \sum_{k'\in\K\backslash k} W_{(\N\backslash\K'k k',k')} \otimes \id^{A_{k'}^O}
\end{align}
if $|\K|>1$, and to $d_\N^O \id^P$ if $|\K|=1$.}
that the constraints of Eq.~\eqref{eq:charact_W_QCQC_decomp} imply that
\begin{align}
%W^{[P\K]} & = \Tr_{A_{\N\backslash\K}^{IO}F} W \notag \\
W^{[P\K]} = \sum_{\emptyset \subseteq \K' \subseteq \N\backslash\K} d_{\K'}^O \Tr_{A_{\N\backslash\K\K'}^{IO}} \sum_{k\in\K} W_{(\N\backslash\K' k,k)} \otimes \id^{A_k^O} \label{eq:WPK_QCQC}
\end{align}
and $W^{[P]} = d_\N^O\, \id^P$ (using the short-hand notations $d_{\K'}^O \coloneqq \Pi_{k\in\K'}d_k^O$, $\K\K' \coloneqq \K\cup\K'$ and $\K' k \coloneqq \K'\cup\{k\}$), from which the first set of constraints in Eq.~\eqref{eq:validity_cstr} above are easily verified.

\subsection{Characterisation of quantum circuits with trivial ``global past'' and ``global future'' systems}
\label{app:charact_dP_dF_1}

For ease of reference, we give here explicit versions of our characterisations for trivial ``global past'' and ``global future'' systems ($d_P = d_F = 1$)---i.e., for the original version of process matrices which map quantum operations to probabilities.

\medskip

For QC-FOs, Proposition~\ref{prop:charact_W_QCFO} becomes:
\begin{customprop}{\ref*{prop:charact_W_QCFO}'}[Characterisation of QC-FOs with trivial $\HS^P, \HS^F$] \label{prop:charact_W_QCFO_trivial_PF}
The process matrix $W \in \L(\HS^{A_\N^{IO}})$ of a quantum circuit with the fixed causal order $(\A_1, \A_2, \ldots, \A_N)$ is a positive semidefinite matrix such that its reduced matrices $W_{(n)} \coloneqq \frac{1}{d_n^O d_{n+1}^O \cdots d_N^O}\Tr_{A_n^OA_{\{n+1,\ldots,N\}}^{IO}} W \in \L(\HS^{A_{\{1,\ldots,n-1\}}^{IO}A_n^I})$ (defined for $1 \le n \le N$, relative to the fixed order just specified) satisfy
\begin{align}
& \Tr W_{(1)} = 1, \notag \\
& \forall \, n = 1, \ldots, N-1, \quad \Tr_{A_{n+1}^I} W_{(n+1)} = W_{(n)} \otimes \id^{A_n^O}, \notag \\
& \textup{and} \quad W = W_{(N)} \otimes \id^{A_N^O}. \label{eq:charact_W_QCFO_trivial_PF}
\end{align}

Conversely, any positive semidefinite matrix $W \in \L(\HS^{A_\N^{IO}})$ whose reduced matrices $W_{(n)}$ satisfy the constraints of Eq.~\eqref{eq:charact_W_QCFO_trivial_PF} is the process matrix of a quantum circuit with the fixed causal order $(\A_1, \A_2, \ldots, \A_N)$.
\end{customprop}

\medskip

For QC-PARs, Proposition~\ref{prop:charact_W_QCPAR} becomes:
\begin{customprop}{\ref*{prop:charact_W_QCPAR}'}[Characterisation of QC-PARs with trivial $\HS^P, \HS^F$] \label{prop:charact_W_QCPAR_trivial_PF}
The process matrix $W \in \L(\HS^{A_\N^{IO}})$ of a quantum circuit with operations used in parallel is of the form
\begin{align}
W = W_{(I)} \otimes \id^{A_\N^O} \quad \textup{with} \quad \Tr W_{(I)} = 1, \label{eq:charact_W_QCPAR_trivial_PF}
\end{align}
for some positive semidefinite matrix $W_{(I)} \in \L(\HS^{A_\N^I})$ (which is nothing but a density matrix describing a quantum state sent to all $N$ operations).

Conversely, any positive semidefinite matrix $W \in \L(\HS^{A_\N^{IO}})$ satisfying Eq.~\eqref{eq:charact_W_QCPAR_trivial_PF} above is the process matrix of a quantum circuit with operations used in parallel.
\end{customprop}

\medskip

For QC-CCs, Proposition~\ref{prop:charact_W_QCCC} becomes:
\begin{customprop}{\ref*{prop:charact_W_QCCC}'}[Characterisation of QC-CCs with trivial $\HS^P, \HS^F$] \label{prop:charact_W_QCCC_trivial_PF}
The process matrix $W \in \L(\HS^{A_\N^{IO}})$ of a quantum circuit with classical control of causal order can be decomposed in terms of positive semidefinite matrices $W_{(k_1,\ldots,k_n)} \in \L(\HS^{A_{\{k_1,\ldots,k_{n-1}\}}^{IO} A_{k_n}^I})$, for all nonempty ordered subsets $(k_1,\ldots,k_n)$ of $\N$ (with $1 \le n \le N$, $k_i \neq k_j$ for $i \neq j$), in such a way that
\begin{align}
W = \sum_{(k_1,\ldots,k_N)} W_{(k_1,\ldots,k_N)}\otimes \id^{A_{k_N}^O} \label{eq:charact_W_QCCC_decomp_sum_trivial_PF}
\end{align}
and
\begin{align}
& \sum_{k_1} \Tr W_{(k_1)} = 1, \notag \\[1mm]
& \forall \, n = 1, \ldots, N{-}1, \ \forall \, (k_1, \ldots, k_n), \notag \\
& \hspace{3mm} \sum_{k_{n+1}} \Tr_{A_{k_{n+1}}^I} \! W_{(k_1,\ldots,k_n,k_{n+1})} = W_{(k_1,\ldots,k_n)} \otimes \id^{A_{k_n}^O}. \label{eq:charact_W_QCCC_decomp_constr_trivial_PF}
\end{align}

Conversely, any Hermitian matrix $W \in \L(\HS^{A_\N^{IO}})$ that admits a decomposition in terms of positive semidefinite matrices $W_{(k_1,\ldots,k_n)} \in \L(\HS^{A_{\{k_1,\ldots,k_{n-1}\}}^{IO} A_{k_n}^I})$ satisfying Eqs.~\eqref{eq:charact_W_QCCC_decomp_sum_trivial_PF}--\eqref{eq:charact_W_QCCC_decomp_constr_trivial_PF} above is the process matrix of a quantum circuit with classical control of causal order.
\end{customprop}

As mentioned in the main text, this characterisation is equivalent to the sufficient condition for causal separability presented in Ref.~\cite{wechs19}---albeit using different notation: what is denoted $W_{(k_1, \ldots, k_n)}$ here corresponds to $\frac{1}{d_{k_n}^O \cdots d_{k_N}^O} \Tr_{A_{k_n}^OA_{\{k_{n+1},\ldots,k_N\}}^{IO}} W_{(k_1, \ldots, k_n)}$ in Ref.~\cite{wechs19}; 
Eq.~\eqref{eq:charact_W_QCCC_decomp_sum_trivial_PF} corresponds to Eq.~(30) in Ref.~\cite{wechs19}, and the last line in~\eqref{eq:charact_W_QCCC_decomp_constr_trivial_PF} is equivalent (for $n = 1, \ldots, N-1$) to Eq.~(32) in Ref.~\cite{wechs19}.

Note that for $N = 2$, in the case of a trivial $\HS^P$ (and in fact, whether $\HS^F$ is trivial or not)  the characterisation of Proposition~\ref{prop:charact_W_QCCC} just reduces to a probabilistic mixture of the two possible fixed causal orders $(\A_1,\A_2)$ and $(\A_2,\A_1)$. 
Indeed, the constraints in this case read $W = W_{(1,2,F)} + W_{(2,1,F)}$, $\Tr W_{(1)} + \Tr W_{(2)} = 1$, $\Tr_{A_{k_2}^I} \! W_{(k_1,k_2)} = W_{(k_1)} \otimes \id^{A_{k_1}^O}$ and $\Tr_{F} \! W_{(k_1,k_2,F)} = W_{(k_1,k_2)} \otimes \id^{A_{k_2}^O}$. One thus sees that $W$ is the convex mixture, with weights $\Tr W_{(k_1)}$, of the process matrices $\frac{1}{\Tr W_{(k_1)}} W_{(k_1,k_2,F)}$ (or 0 if $\Tr W_{(k_1)} = 0$), each compatible with the corresponding fixed order $(\A_{k_1},\A_{k_2})$.
In order to have an order between the $N$ operations $\A_k$ that is not predefined (even probabilistically), in that case with trivial $\HS^P$, we therefore need $N \ge 3$. 
In contrast, for $N = 2$ and a nontrivial $\HS^P$, a non-predefined order is possible: an example is the classical switch considered in Sec.~\ref{subsec:QCCCs_example} (even with $\HS^F$ traced out).

\medskip

Finally, for QC-QCs, Proposition~\ref{prop:charact_W_QCQC} becomes:
\begin{customprop}{\ref*{prop:charact_W_QCQC}'}[Characterisation of QC-QCs with trivial $\HS^P, \HS^F$] \label{prop:charact_W_QCQC_trivial_PF}
The process matrix $W \in \L(\HS^{A_\N^{IO}})$ of a quantum circuit with quantum control of causal order is such that there exist positive semidefinite matrices $W_{(\K_{n-1},k_n)} \in \L(\HS^{A_{\K_{n-1}}^{IO} A_{k_n}^I})$, for all strict subsets $\K_{n-1}$ of $\N$ and all $k_n \in \N\backslash\K_{n-1}$, satisfying
\begin{align}
& \sum_{k_1 \in \N} \Tr W_{(\emptyset,k_1)} = 1, \notag \\
& \forall \, \emptyset \subsetneq \K_n \subsetneq \N, \!\! \sum_{k_{n+1} \in \N \backslash \K_n} \!\!\! \Tr_{A_{k_{n+1}}^I} \!W_{(\K_n,k_{n+1})} \notag \\[-2mm]
& \hspace{40mm} = \sum_{k_n \in \K_n} W_{(\K_n \backslash k_n,k_n)}\otimes \id^{A_{k_n}^O}, \notag \\[1mm]
& \textup{and} \quad W = \sum_{k_N \in \N} W_{(\N \backslash k_N,k_N)}\otimes \id^{A_{k_N}^O}. \label{eq:charact_W_QCQC_decomp_trivial_PF}
\end{align}

Conversely, any Hermitian matrix $W \in \L(\HS^{A_\N^{IO}})$ such that there exist positive semidefinite matrices $W_{(\K_{n-1},k_n)} \in \L(\HS^{A_{\K_{n-1}}^{IO} A_{k_n}^I})$ for all $\K_{n-1} \subsetneq \N$ and $k_n \in \N\backslash\K_{n-1}$ satisfying Eq.~\eqref{eq:charact_W_QCQC_decomp_trivial_PF} is the process matrix of a quantum circuit with quantum control of causal order.
\end{customprop}

Note that for $N=2$, with a trivial $\HS^F$ (and now, whether $\HS^P$ is trivial or not), the characterisation of Proposition~\ref{prop:charact_W_QCQC} coincides with that of Proposition~\ref{prop:charact_W_QCCC}, i.e., QC-QCs reduce to QC-CCs. 
In that case, the last line of Eq.~\eqref{eq:charact_W_QCQC_decomp} becomes $W = W_{(\{1\},2)} \otimes \id^{A_{2}^O}+W_{(\{2\},1)} \otimes \id^{A_{1}^O}$, and the constraints are identical to those in Proposition~\ref{prop:charact_W_QCCC}, with $W_{(k_1)} = W_{(\emptyset,k_1)}$, $W_{(k_1,k_2)} = W_{(\{k_1\},k_2)}$ and $W_{(k_1,k_2,F)} = W_{(\{k_1\},k_2)} \otimes \id^{A_{k_2}^O}$. 
For $N = 2$ with a nontrivial $\HS^F$ on the other hand, the two classes do not coincide. 
A counterexample is given by the quantum switch (even when taking its ``global past'' space $\HS^P$ to be trivial, by fixing the input state as in Eq.~\eqref{eq:W_QS_plus_psi}), which is causally nonseparable. 

Together with the observation made after Proposition~\ref{prop:charact_W_QCCC_trivial_PF}, it follows that for $N=2$ with both $d_P = d_F = 1$, the classes of QC-QCs and QC-CCs both collapse to a probabilistic mixture of QC-FOs.

\medskip

The characterisation of the various classes of probabilistic quantum circuits for $d_P=d_F=1$ (Propositions~\ref{prop:charact_Wr_probQCFO}, \ref{prop:charact_Wr_probQCPAR}, \ref{prop:charact_Wr_probQCCC}, \ref{prop:charact_Wr_probQCQC} and \ref{prop:charact_Wr_probIND}) can be obtained similarly to what we have done here for the deterministic case.

\section{Process matrix characterisation of quantum circuits}
\label{app:proofs_charact}

In this appendix we prove the process matrix characterisations of the different classes of quantum circuits considered in this paper. For each of these classes, we first derive the TP conditions that the respective internal circuit operations must satisfy so that they act trace-preservingly on all input states they can receive, i.e., on their ``effective input spaces''. We then prove the necessary and sufficient conditions for the deterministic case separately. In order to prove the sufficient condition in particular, we provide a method to construct an explicit circuit from a given process matrix in the class under consideration. We then extend the proofs to the respective probabilistic circuits.

\medskip

In the proofs (for the sufficient conditions) below we will use the following lemma to ``invert'' the link product for vectors:

\begin{lemma}[Link product inversion] \label{lemma:invert_link_product}
Let $\ket{a} \in \HS^{XY}$ and $\ket{c} \in \HS^{XZ}$, and define $A_X \coloneqq \Tr_Y \ketbra{a}{a} \in \L(\HS^X)$.

A necessary condition for the existence of $\ket{b} \in \HS^{YZ}$ such that $\ket{a} * \ket{b} = \ket{c}$ is that $\ket{c} \in \range(A_X) \otimes \HS^Z$.
Under this condition, a solution is given by
\begin{align}
\ket{b} & \coloneqq \ket{a^+} * \ket{c} \notag \\[1mm]
\text{with} \quad \ket{a^+} & \coloneqq \big(\bra{a} A_X^+ \otimes \id^Y \big)^T \in \HS^{XY}, \label{eq:invert_link_product}
\end{align}
where $A_X^+$ is the Moore-Penrose pseudoinverse of $A_X$.
\end{lemma}

\begin{proof}
Let us denote by $\{\ket{i}^Y\}_i$ the computational basis of $\HS^Y$ and define, for any $\ket{a} \in \HS^{XY}$ and $\ket{b} \in \HS^{YZ}$, $\ket{a_i} \coloneqq (\id^X \otimes \bra{i}^Y) \ket{a} \in \HS^X$ and $\ket{b_i} \coloneqq (\bra{i}^Y \otimes \id^Z) \ket{b} \in \HS^Z$.

By noting that $\ket{a} * \ket{b} = \sum_i \ket{a_i} \otimes \ket{b_i}$ (as in Eq.~\eqref{eq:def_pure_link_product}) and $A_X = \sum_i \ketbra{a_i}{a_i}$, it appears clearly that the link product $\ket{a} * \ket{b}$ is in $\range(A_X) \otimes \HS^Z$, which proves the necessary condition stated in the proposition.

Suppose that this condition is indeed satisfied. The vectors $\ket{a^+}$ and $\ket{b}$ above can be written more explicitly, in terms of the $\ket{a_i}$'s, as%
\footnote{If all the vectors $\ket{a_i}$ are orthogonal, as it will be the case in the proofs of Appendices~\ref{app:proof_charact_QCFOs} and~\ref{app:proof_charact_QCCCs} below (for the characterisation of QC-FOs and QC-CCs), then one has $A_X^+ = \sum_{i:\ket{a_i}\neq 0} \ketbra{a_i}{a_i} / \braket{a_i}{a_i}^2$, so that Eq.~\eqref{eq:invert_link_product_explicit} can be written even more directly as $\ket{a^+} = \sum_{i:\ket{a_i}\neq 0} \frac{\bra{a_i}^T}{\braket{a_i}{a_i}} \otimes \ket{i}^Y$, $\ket{b} = \sum_{i:\ket{a_i}\neq 0} \ket{i}^Y \otimes \big( \frac{\bra{a_i}}{\braket{a_i}{a_i}} \otimes \id^Z \ket{c} \big)$. \label{footnote:invert_link_prod_ortho_decomp}}
\begin{align}
\ket{a^+} & = {\textstyle \sum_i} \big(\!\bra{a_i} A_X^+ \big)^T \otimes \ket{i}^Y, \notag \\[1mm]
\ket{b} & = {\textstyle \sum_i} \ket{i}^Y \otimes \big( \bra{a_i} A_X^+ \otimes \id^Z \ket{c} \big). \label{eq:invert_link_product_explicit}
\end{align}
We then have $\ket{a} * \ket{b} = \sum_i \ket{a_i} \otimes \ket{b_i} = \sum_i \ket{a_i} \otimes \bra{a_i} A_X^+ \otimes \id^Z \ket{c} = A_X A_X^+ \otimes \id^Z \ket{c} = \ket{c}$, where we used the fact that $A_XA_X^+$ is the projector onto (and therefore acts as the identity within) the range of $A_X$.
This proves that $\ket{b}$ defined in Eq.~\eqref{eq:invert_link_product} is indeed a solution to $\ket{a} * \ket{b} = \ket{c}$.
\end{proof}

To verify the necessary condition $\ket{c} \in \range(A_X) \otimes \HS^Z$ when using Lemma~\ref{lemma:invert_link_product} in the proofs below, let us also make the following observation:

\begin{observation} \label{observation:nec_cond_range}
Let $\{\ket{c_k}\}_k$ be a family of vectors, with each $\ket{c_k} \in \HS^{XZ_k} = \HS^X \otimes \HS^{Z_k}$ for some (possibly different) Hilbert spaces $\HS^X, \HS^{Z_k}$, and define $C_X \coloneqq \sum_k \Tr_{Z_k} \ketbra{c_k}{c_k} \in \L(\HS^X)$.

One has that for each $k$, $\ket{c_k} \in \range(C_X) \otimes \HS^{Z_k}$.
\end{observation}

\begin{proof}
Denoting by $\Pi_X \coloneqq C_XC_X^+$ the projector onto the range of $C_X$ and by $\Pi_X^\perp \coloneqq \id^X - \Pi_X$ its orthogonal projector in $\HS^X$, one has
\begin{align}
\sum_k \Tr\big[(\Pi_X^\perp \otimes \id^{Z_k} ) \ketbra{c_k}{c_k} \big] & = \sum_k \Tr\big[\Pi_X^\perp (\Tr_{Z_k} \ketbra{c_k}{c_k}) \big] & \notag \\
& = \Tr\big[\Pi_X^\perp C_X \big] = 0.
\end{align}
Since the individual summands in the sum above cannot be negative, we conclude that each of them (and hence, $(\Pi_X^\perp \otimes \id^{Z_k} ) \ket{c_k}$) must be zero, and therefore that $(\Pi_X \otimes \id^{Z_k} ) \ket{c_k} = \ket{c_k}$---i.e., $\ket{c_k} \in \range(C_X) \otimes \HS^{Z_k}$.
\end{proof}

\subsection{QC-FOs: Proofs of Propositions~\ref{prop:charact_W_QCFO} and~\ref{prop:charact_Wr_probQCFO}}
\label{app:proof_charact_QCFOs}

Here, we will prove the characterisations of QC-FOs (Proposition~\ref{prop:charact_W_QCFO}) and pQC-FOs (Proposition~\ref{prop:charact_Wr_probQCFO}). Equivalent results were already proven in Refs.~\cite{chiribella09,gutoski06}. Below we make the constructive proofs for the sufficient conditions somewhat more explicit. Also, the proofs for QC-CCs and QC-QCs will follow very similar paths, so it is useful to first present the simpler case.

For ease of notations and to avoid repetitions, it will be convenient in this section to define $\HS^{A^I_{N+1}} \coloneqq \HS^F$.

\subsubsection{Trace-preserving conditions}
\label{app:subsubsec_QCFO_TP}

Let us first derive the TP conditions of Eqs.~\eqref{eq:TP_constr_QCFO_1}--\eqref{eq:TP_constr_QCFO_N} that the internal circuit operations of a QC-FO must satisfy.
Consider for that a QC-FO as depicted in Fig.~\ref{fig:QCFO}, and suppose one inputs some state $\rho \in \L(\HS^P)$ into the circuit.

We first require the state $\M_1(\rho) = \rho * M_1$ after applying the first internal circuit operation $\M_1$ to have the same trace as $\rho$. That is, we want
\begin{align}
& \Tr[\rho * M_1] = \Tr[(\rho^T \otimes \id^{A_1^I\alpha_1}) M_1] \notag \\
& \quad = \Tr[\rho^T (\Tr_{A_1^I\alpha_1} M_1)] = \Tr[\rho] \ \ ( = \Tr[\rho^T]).
\end{align}
As this must hold for all $\rho \in \L(\HS^P)$, this constraint is equivalent to
\begin{align}
\Tr_{A_1^I\alpha_1} M_1 = \id^P,
\end{align}
as in Eq.~\eqref{eq:TP_constr_QCFO_1}, which is indeed the standard trace-preserving condition for the Choi representation of a quantum map.

For $n = 1, \ldots, N$, the states of the global system going through the circuit right before and right after the application of $\M_{n+1}$ are obtained (in terms of the Choi representations and link products) as $\rho * M_1 * A_1 * M_2 * \cdots * M_n * A_n = (\rho \otimes A_1 \otimes \cdots \otimes A_n) * (M_1 * M_2 * \cdots * M_n) \in \L(\HS^{A_n^O\alpha_n})$ and $\rho * M_1 * A_1 * M_2 * \cdots * M_n * A_n * M_{n+1} = (\rho \otimes A_1 \otimes \cdots \otimes A_n) * (M_1 * M_2 * \cdots * M_n * M_{n+1}) \in \L(\HS^{A_{n+1}^I\alpha_{n+1}})$, respectively (with $A_{N+1}^I = F$ and a trivial ancillary space $\HS^{\alpha_{N+1}}$ for $n=N$). Their traces are
\begin{align}
& \Tr[(\rho \otimes A_1 \otimes \cdots \otimes A_n) * (M_1 * \cdots * M_n)] \notag \\
& \ = \Tr\big[\big( (\rho \otimes A_1 \otimes \cdots \otimes A_n)^T \otimes \id^{\alpha_n} \big) \notag \\[-2mm]
& \hspace{30mm} \big( (M_1 * \cdots * M_n) \otimes \id^{A_n^O} \big) \big] \notag \\
& \ = \Tr\big[ (\rho \otimes A_1 \otimes \cdots \otimes A_n)^T \notag \\[-1mm]
& \hspace{25mm} \big( \Tr_{\alpha_n} (M_1 * \cdots * M_n) \otimes \id^{A_n^O} \big) \big] \label{eq:Tr_before_Mn_QCFO}
\end{align}
and
\begin{align}
& \Tr[(\rho \otimes A_1 \otimes \cdots \otimes A_n) * (M_1 * \cdots * M_n * M_{n+1})] \notag \\
& \ = \Tr\big[\big( (\rho \otimes A_1 \otimes \cdots \otimes A_n)^T \otimes \id^{A_{n+1}^I\alpha_{n+1}} \big) \notag \\[-1mm]
& \hspace{35mm} (M_1 * \cdots * M_n * M_{n+1}) \big] \notag \\
& \ = \Tr\big[ (\rho \otimes A_1 \otimes \cdots \otimes A_n)^T \notag \\[-1mm]
& \hspace{18mm} \Tr_{A_{n+1}^I\alpha_{n+1}} (M_1 * \cdots * M_n * M_{n+1}) \big]. \label{eq:Tr_after_Mn_QCFO}
\end{align}
We require these to be equal, for all possible initial states $\rho \in \L(\HS^P)$ and all possible external CP maps with Choi matrices $A_k \in \L(\HS^{A_k^{IO}})$. As $\rho \otimes A_1 \otimes \cdots \otimes A_n$ spans the whole space $\L(\HS^{PA_{\{1,\ldots,n\}}^{IO}})$, this is indeed equivalent to
\begin{align}
& \Tr_{A_{n+1}^I\alpha_{n+1}} (M_1 * \cdots * M_n * M_{n+1}) \notag \\
& \hspace{20mm} = \Tr_{\alpha_n} (M_1 * \cdots * M_n) \otimes \id^{A_n^O},
\end{align}
as in Eqs.~\eqref{eq:TP_constr_QCFO_n} (for $1\le n < N$) and~\eqref{eq:TP_constr_QCFO_N} (for $n=N$).

\medskip

The TP conditions for probabilistic QC-FOs follow from the exact same reasoning, with the last internal circuit operation $\M_{N+1}$ replaced by $\sum_r \M_{N+1}^{[r]}$, the CPTP map obtained by summing over the classical outcomes.

\subsubsection{Proof of Proposition~\ref{prop:charact_W_QCFO}: Necessary condition}

Consider the process matrix $W = M_1 * M_2 * \cdots * M_{N+1}$ of a QC-FO, as per Proposition~\ref{prop:descr_W_QCFO}, with the Choi matrices $M_n$ satisfying the TP conditions of Eqs.~\eqref{eq:TP_constr_QCFO_1}--\eqref{eq:TP_constr_QCFO_N}.

Note first that as all $M_n \succeq 0$, it directly follows that $W$ is positive semidefinite.

Defining $W_{(N+1)} \coloneqq W$, the reduced matrices $W_{(n)}$ defined in Proposition~\ref{prop:charact_W_QCFO} can be obtained recursively (from $n=N$, down to $n=1$) as $W_{(n)} = \frac{1}{d_n^O}\Tr_{A_n^OA_{n+1}^I} W_{(n+1)}$.
Similarly, Eqs.~\eqref{eq:TP_constr_QCFO_n}--\eqref{eq:TP_constr_QCFO_N} imply that $\Tr_{\alpha_n} (M_1 * \cdots * M_n)= \frac{1}{d_n^O}\Tr_{A_n^OA_{n+1}^I} [ \Tr_{\alpha_{n+1}} (M_1 * \cdots * M_n * M_{n+1})]$.
Since $W_{(n)}$ and $\Tr_{\alpha_n} (M_1 * \cdots * M_n)$ are equal for $n = N+1$ (with a trivial $\HS^{\alpha_{N+1}}$) and satisfy the same recursive property, it follows that they are the same for all $n = 1, \ldots, N+1$:
\begin{align}
W_{(n)} = \Tr_{\alpha_n} (M_1 * \cdots * M_n).
\end{align}
The constraints of Eq.~\eqref{eq:charact_W_QCFO} are then simply equivalent to (and therefore readily implied by) the TP conditions of Eqs.~\eqref{eq:TP_constr_QCFO_1}--\eqref{eq:TP_constr_QCFO_N}.

\subsubsection{Proof of Proposition~\ref{prop:charact_W_QCFO}: Sufficient condition}
\label{app:proof_charact_QCFOs_SC}

Consider a positive semidefinite matrix $W \in \L(\HS^{PA_\N^{IO} F})$ whose reduced matrices $W_{(n)} \coloneqq \frac{1}{d_n^O d_{n+1}^O \cdots d_N^O}\Tr_{A_n^OA_{\{n+1,\ldots,N\}}^{IO}F} W \in \L(\HS^{PA_{\{1,\ldots,n-1\}}^{IO}A_n^I})$ satisfy the constraints of Eq.~\eqref{eq:charact_W_QCFO}.
We will show that $W$ is the process matrix of a QC-FO with the causal order $(\A_1,\ldots,\A_N)$, by constructing some internal circuit operations $\M_n$ (CPTP maps, in the sense of the constraints of Eqs.~\eqref{eq:TP_constr_QCFO_1}--\eqref{eq:TP_constr_QCFO_N}) explicitly.

\medskip

Since $W \succeq 0$, then all $W_{(n)} \succeq 0$ as well (for $1 \le n \le N+1$, with again $W_{(N+1)} \coloneqq W$), which admit a spectral decomposition of the form
\begin{align}
W_{(n)} = \sum_i \ketbra{w_{(n)}^i}{w_{(n)}^i},
\end{align}
for some eigenbasis consisting of $r_n \coloneqq \rank W_{(n)}$ (non-normalised and nonzero) orthogonal vectors $\ket{w_{(n)}^i} \in \HS^{PA_{\{1,\ldots,n-1\}}^{IO}A_n^I}$.
Let us then introduce, for each $n = 1, \ldots, N+1$, some $r_n$-dimensional ancillary Hilbert space $\HS^{\alpha_n}$ with its computational basis $\{\ket{i}^{\alpha_n}\}_{i=1}^{r_n}$, and define
\begin{align}
\ket{w_{(n)}} \coloneqq \sum_i \ket{w_{(n)}^i} \otimes \ket{i}^{\alpha_n} \ \in \HS^{PA_{\{1,\ldots,n-1\}}^{IO}A_n^I\alpha_n}, \label{eq:def_wn}
\end{align}
such that $W_{(n)} = \Tr_{\alpha_n} \ketbra{w_{(n)}}{w_{(n)}}$.

For $n = 1, \ldots, N$, the assumption that $\Tr_{A_{n+1}^I} W_{(n+1)} = \Tr_{A_{n+1}^I\alpha_{n+1}} \ketbra{w_{(n+1)}}{w_{(n+1)}} = W_{(n)} \otimes \id^{A_n^O}$ implies, after further tracing out over $A_n^O$ and via Observation~\ref{observation:nec_cond_range} (here with a single vector $\ket{c_k}$), that $\ket{w_{(n+1)}} \in \range(W_{(n)}) \otimes \HS^{A_n^OA_{n+1}^I\alpha_{n+1}}$.
Using Lemma~\ref{lemma:invert_link_product} above, this property ensures that one can relate $\ket{w_{(n)}}$ and $\ket{w_{(n+1)}}$ by defining, for $1 \le n \le N$,%
\footnote{As in Lemma~\ref{lemma:invert_link_product}, in Eq.~\eqref{eq:def_wn_p} $W_{(n)}^+$ denotes the Moore-Penrose pseudoinverse of $W_{(n)}$ (and similarly for $W_{(k_1,\ldots,k_n)}^+$ and $\Omega_{\K_n}^+$ in Eqs.~\eqref{eq:def_wk1_kn_p} and~\eqref{eq:def_omega_Kn_p} further below). More explicitly, from the spectral decomposition of $W_{(n)}$ in terms of (nonzero orthogonal) vectors $\ket{w_{(n)}^i}$, we can write $\dket{V_{n+1}} = \sum_i \ket{i}^{\alpha_n} \otimes \Big( \frac{\bra{w_{(n)}^i}}{\braket{w_{(n)}^i}{w_{(n)}^i}} \otimes \id^{A_n^O\!A_{n+1}^I\alpha_{n+1}} \ket{w_{(n+1)}} \Big)$ for all $1 \le n \le N$: cf.\ Eq.~\eqref{eq:invert_link_product_explicit} and Footnote~\ref{footnote:invert_link_prod_ortho_decomp}. \label{footnote:explicit_V_n1_QCFO}}
\begin{align}
& \ket{w_{(n)}^+} \!\coloneqq\! (\bra{w_{(n)}} \!W_{\!(n)}^+ \!\otimes\! \id^{\alpha_n}\!)^T \in \HS^{PA_{\{1,\ldots,n-1\}}^{IO}A_n^I\alpha_n}, \label{eq:def_wn_p} \\
& \dket{V_{n+1}} \coloneqq \ket{w_{(n)}^+} * \ket{w_{(n+1)}} \in \HS^{A_n^O\!\alpha_nA_{n+1}^I\alpha_{n+1}}%,
\end{align}
so that
\begin{align}
\ket{w_{(n)}} * \dket{V_{n+1}} = \ket{w_{(n+1)}}. \label{eq:QCFO_wn_wn1} 
\end{align}
By further defining $\dket{V_1} \coloneqq \ket{w_{(1)}} \in \HS^{PA_1^I\alpha_1}$, one recursively obtains that
\begin{align}
\dket{V_1} * \dket{V_2} * \cdots * \dket{V_n} = \ket{w_{(n)}} \label{eq:QCFO_rec_Vs}
\end{align}
for all $1 \le n \le N+1$.

From the double-ket vectors $\dket{V_n}$ just introduced we can then define the operators
\begin{align}
M_n & \coloneqq \dketbra{V_n}{V_n} \quad \text{for } 1 \le n \le N \notag \\[1mm]
\text{and } M_{N+1} & \coloneqq \Tr_{\alpha_{N+1}} \dketbra{V_{N+1}}{V_{N+1}},
\end{align}
such that 
\begin{align}
& \Tr_{\alpha_n} (M_1 * \cdots * M_n) = \Tr_{\alpha_n} \ketbra{w_{(n)}}{w_{(n)}} = W_{(n)} \notag \\[1mm]
& \text{and } \ M_1 * \cdots * M_{N+1} = W .\label{eq:QCFO_rec_Ms}
\end{align}
As the $W_{(n)}$'s are assumed to satisfy the constraints of Eq.~\eqref{eq:charact_W_QCFO}, then by construction the (positive semidefinite) $M_n$'s satisfy the (equivalent, once Eq.~\eqref{eq:QCFO_rec_Ms} is established) TP conditions of Eqs.~\eqref{eq:TP_constr_QCFO_1}--\eqref{eq:TP_constr_QCFO_N}.
This proves that the operators $M_n$ define CPTP maps $\M_1: \L(\HS^P) \to \L(\HS^{A_1^I\alpha_1})$, $\M_{n+1}: \L(\HS^{A_n^O\alpha_n}) \to \L(\HS^{A_{n+1}^I\alpha_{n+1}})$ for $n=1,\ldots,N-1$, and $\M_{N+1}: \L(\HS^{A_N^O\alpha_N}) \to \L(\HS^F)$, as required for the internal circuit operations of a QC-FO.
The second line of Eq.~\eqref{eq:QCFO_rec_Ms} above shows, according to Proposition~\ref{prop:descr_W_QCFO}, that $W$ is indeed the process matrix of the QC-FO thus constructed.

\subsubsection{Proof of Proposition~\ref{prop:charact_Wr_probQCFO}}

The proofs of both the necessary and the sufficient conditions above extend easily to the characterisation of probabilistic QC-FOs, as given by Proposition~\ref{prop:charact_Wr_probQCFO} in Sec.~\ref{sec:probQCFOs}.

For the necessary condition, recall that a pQC-FO with the order $(\A_1,\ldots,\A_N)$ is described by the matrices $W^{[r]}$ given in Proposition~\ref{prop:descr_Wr_probQCFO}. Since all $M_n, M_{N+1}^{[r]} \succeq 0$, then clearly each $W^{[r]}$ is positive semidefinite; furthermore, since all $M_n$ for $1 \le n \le N$, as well as $\sum_r M_{N+1}^{[r]}$, satisfy the TP conditions of Eqs.~\eqref{eq:TP_constr_QCFO_1}--\eqref{eq:TP_constr_QCFO_N}, then the sum $\sum_r W^{[r]} = M_1 * M_2 * \cdots * M_N * \big( \sum_r M_{N+1}^{[r]} \big)$ is indeed the process matrix of a QC-FO (with the same fixed causal order) as per Proposition~\ref{prop:descr_W_QCFO}.

Conversely for the sufficient condition, let $\{W^{[r]}\}_r$ be a set of positive semidefinite matrices whose sum is the process matrix of a QC-FO with the order $(\A_1,\ldots,\A_N)$. Introducing a Hilbert space $\HS^{F'}$ with computational basis states $\{\ket{r}^{F'}\}_r$, let us define the (positive semidefinite) ``extended'' matrix $W' \coloneqq \sum_r W^{[r]} \otimes \ketbra{r}{r}^{F'} \in \L(\HS^{PA_\N^{IO}FF'})$, such that $W' * \ketbra{r}{r}^{F'} = W^{[r]}$ and $\Tr_{F'} W' = \sum_r W^{[r]}$, so that $W'$ and $\sum_r W^{[r]}$ have the same reduced matrices $W_{(n)}$ (as defined in Proposition~\ref{prop:charact_W_QCFO}). As $\sum_r W^{[r]}$ is assumed to satisfy the constraints of Eq.~\eqref{eq:charact_W_QCFO}, then so does $W'$ (with $F$ replaced by $FF'$), so that $W'$ is the process matrix of a (deterministic) QC-FO with the same order and with the global future space $\HS^{FF'}$ as per Proposition~\ref{prop:charact_W_QCFO}.
We can thus decompose it as $W' = M_1 * M_2 * \cdots * M_N * M_{N+1}'$, where all $M_1 \in \L(\HS^{P A^I_1 \alpha_1})$, $M_{n+1} \in \L(\HS^{A_n^O\alpha_nA_{n+1}^I\alpha_{n+1}})$, and $M_{N+1}' \in \L(\HS^{A^O_N \alpha_N F F'})$ are CP maps satisfying the TP conditions of Eqs.~\eqref{eq:TP_constr_QCFO_1}--\eqref{eq:TP_constr_QCFO_N}.
Defining the CP maps $M_{N+1}^{[r]} \coloneqq M_{N+1}' * \ketbra{r}{r}^{F'}$, whose sum $\sum_r M_{N+1}^{[r]} = \Tr_{F'} M_{N+1}'$ satisfies the TP condition of Eq.~\eqref{eq:TP_constr_QCFO_N}, we then obtain $W^{[r]} = W' * \ketbra{r}{r}^{F'} = M_1 * M_2 * \cdots * M_N * M_{N+1}^{[r]}$, which is of the form of Eq.~\eqref{eq:Wr_probQCFO} and thus proves according to Proposition~\ref{prop:descr_Wr_probQCFO} that $\{W^{[r]}\}_r$ indeed has a realisation as a pQC-FO with the fixed causal order $(\A_1,\ldots,\A_N)$.

\subsection{QC-CCs: Proofs of Propositions~\ref{prop:charact_W_QCCC} and~\ref{prop:charact_Wr_probQCCC}}
\label{app:proof_charact_QCCCs}

In this section, we will derive the TP conditions for QC-CCs, and prove Propositions~\ref{prop:charact_W_QCCC} and~\ref{prop:charact_Wr_probQCCC}.

To avoid repetitions we define here $k_{N+1} \coloneqq F$, $\HS^{A^I_{k_{N+1}}} \coloneqq \HS^F$ and $\HS^{\tilde A^I_{N+1}} \coloneqq \HS^F$, as for instance in $M_{(k_1,\ldots,k_N)}^{\to k_{N+1}} = M_{(k_1,\ldots,k_N)}^{\to F} (\in \L(\HS^{A_{k_N}^O \alpha_NA_{k_{N+1}}^I}) = \L(\HS^{A_{k_N}^O \alpha_NF}))$ and $W_{(k_1,\ldots,k_N,k_{N+1})} = W_{(k_1,\ldots,k_N,F)} (\in \L(\HS^{PA_\N^{IO}A_{k_{N+1}}^I}) = \L(\HS^{PA_\N^{IO}F}))$.

\subsubsection{Trace-preserving conditions}
\label{app:subsubsec_QCCC_TP}

The TP conditions of Eqs.~\eqref{eq:TP_constr_QCCC_1}--\eqref{eq:TP_constr_QCCC_N} for QC-CCs can be obtained in the very same way as those for the QC-FO case (see Sec.~\ref{app:subsubsec_QCFO_TP}) after noting that, according to the description of QC-CCs given in Sec.~\ref{subsec:QCCC_description}, it is now (for each $(k_1,\ldots,k_n)$) the sums%
\footnote{Here we indulge some slight abuse of notation, as the different maps $\M_{(k_1,\ldots,k_n)}^{\to k_{n+1}}$ in the sums have different output spaces (cf.\ Footnote~\ref{footnote:different_output_spaces}). Note however that after tracing out over $A_{k_{n+1}}^I$, all summands in the sums of Eqs.~\eqref{eq:TP_constr_QCCC_1} and~\eqref{eq:TP_constr_QCCC_n} belong to the same spaces.}
$\sum_{k_{n+1}} \M_{(k_1,\ldots,k_n)}^{\to k_{n+1}}$ that must preserve the trace of all possible input states (rather than the internal operations $\M_{n+1}$ in the QC-FO case).

\medskip

Note that these constraints can also similarly be obtained from the alternative description of QC-CCs given in Sec.~\ref{subsec:QCCC_revisit}, by requiring that the global internal circuit operations $\tilde \M_n$ preserve the trace of their global input states (including the control systems), as well as the probabilities for a given order $(k_1,\ldots,k_n)$ (for the thus far applied external operations) to be realised. Let us indeed check that, for consistency.

Suppose that one inputs some state $\rho \in \L(\HS^P)$ into the circuit. Requiring first that $\tilde \M_1$ (with Choi matrix given in Eq.~\eqref{eq:def_tilde_M1_QCCC}) preserves the trace of $\rho$, we must have
\begin{align}
\label{eq:TP_QCCC_calculation1}
& \Tr[\tilde \M_1(\rho)] = \Tr[\rho * \tilde M_1] \notag \\
& \quad = \Tr\big[\big(\rho^T \otimes \id^{\tilde A_1^I\alpha_1C_1}\big) \big( {\textstyle \sum_{k_1}} \tilde M_\emptyset^{\to k_1} \otimes [\![(k_1)]\!]^{C_1} \big)\big] \notag \\
& \quad = \Tr\big[\rho^T \big( {\textstyle \sum_{k_1}} \Tr_{A_{k_1}^I\alpha_1} M_\emptyset^{\to k_1} \big)\big] = \Tr[\rho]
\end{align}
(where the isomorphism between $\HS^{\tilde A_1^I}$ and each $\HS^{A_{k_1}^I}$ allowed us to remove the tildes in the last line, that is, we use that $\Tr_{\tilde A_1^I\alpha_1} \tilde M_\emptyset^{\to k_1} = \Tr_{\tilde A_1^I\alpha_1} [M_\emptyset^{\to k_1} * \dketbra{\id}{\id}^{A^I_{k_1}\tilde A_1^I}] = \Tr_{A_{k_1}^I\alpha_1} M_\emptyset^{\to k_1}$, see Footnote~\ref{footnote:generic_IO_spaces}).
As this must hold for all $\rho \in \L(\HS^P)$, this constraint is indeed equivalent to Eq.~\eqref{eq:TP_constr_QCCC_1}.

For $n = 1, \ldots, N$, the states $\varrho_{(n)}' \in \L(\HS^{\tilde A_n^O\alpha_n C_n'})$ and $\varrho_{(n+1)} \in \L(\HS^{\tilde A_{n+1}^I\alpha_{n+1} C_{n+1}})$ of the global system going through the circuit right before and right after the application of $\tilde \M_{n+1}$,%
\footnote{For $n=N$, one should replace $\tilde \M_{N+1}$ in the definition of $\varrho_{(N+1)}$ by an operation $\tilde \M_{N+1}^*$ that records the final order $(k_1,\ldots,k_N)$ in some control system $C_{N+1}$, i.e., with a Choi matrix
$\tilde M_{N+1}^* = \sum_{(k_1,\ldots,k_N)} \tilde M_{(k_1,\ldots,k_N)}^{\to F} \otimes [\![(k_1,\ldots,k_N)]\!]^{C_N'} \otimes [\![(k_1,\ldots,k_N,F)]\!]^{C_{N+1}} \in \L(\HS^{\tilde A_N^O \alpha_NC_N'FC_{N+1}})$ (with a trivial $\HS^{\alpha_{N+1}}$), instead of $\tilde M_{N+1}$ given by Eq.~\eqref{eq:def_tilde_MN1_QCCC}.} 
respectively, are easily obtained recursively as
\begin{align}
\varrho_{(n)}' & = \rho * \tilde M_1 * \tilde A_1 * \tilde M_2 * \tilde A_2 * \cdots * \tilde M_n * \tilde A_n \notag \\
& = \!\! \sum_{(k_1,\ldots,k_n)} \!\!\! \big( \rho * M_\emptyset^{\to k_1} * A_{k_1} * M_{(k_1)}^{\to k_2} * A_{k_2} * \cdots \notag \\[-3mm]
& \hspace{17mm} \cdots *\! M_{(k_1,\ldots,k_{n-1})}^{\to k_n} \!*\! A_{k_n} * \dketbra{\id}{\id}^{A^O_{k_n}\tilde A_n^O}\big) \notag \\%[1mm]
& \hspace{32mm} \!\otimes\! [\![(k_1,\ldots,k_n)]\!]^{C_n'}
 \label{eq:link_prod_rho_tilde_M1_tilde_An}
\end{align}
and
\begin{align}
\varrho_{(n+1)} & = \rho * \tilde M_1 * \tilde A_1 * \tilde M_2 * \tilde A_2 * \cdots * \tilde M_n * \tilde A_n * \tilde M_{n+1} \notag \\
& = \!\! \sum_{(k_1,\ldots,k_n,k_{n+1})} \!\!\! \big( \rho \!*\! M_\emptyset^{\to k_1} \!*\! A_{k_1} \!*\! M_{(k_1)}^{\to k_2} \!*\! A_{k_2} \!*\!\cdots \notag \\[-3mm]
& \hspace{22mm} \cdots \!*\! M_{(k_1,\ldots,k_{n-1})}^{\to k_n} \!*\! A_{k_n} \!*\! M_{(k_1,\ldots,k_n)}^{\to k_{n+1}} \notag \\[1mm]
& \hspace{45mm}  * \dketbra{\id}{\id}^{A^I_{k_{n+1}}\tilde A_{n+1}^I} \big)  \notag \\[1mm]
& \hspace{25mm} \otimes [\![(k_1,\ldots,k_n,k_{n+1})]\!]^{C_{n+1}}, \label{eq:link_prod_rho_tilde_M1_tilde_Mn1}
\end{align}
where the isomorphism from Footnote~\ref{footnote:generic_IO_spaces} allows us again to use the ``non-tilde'' versions of the internal and external circuit operations, as we did also in Eq.~\eqref{eq:Choi_M_QCCC_v2} of the main text; only the identifications of $\HS^{A^O_{k_n}}$ with $\HS^{\tilde A^O_n}$, and that of $\HS^{A^I_{k_{n+1}}}$ with $\HS^{\tilde A^I_{n+1}}$ need to be maintained here.\footnote{More formally, we use that $\dketbra{\id}{\id}^{A^I_{k_n}\tilde A_n^I} * \tilde A_{k_n} = A_{k_n} * \dketbra{\id}{\id}^{A^O_{k_{n}}\tilde A^O_{n}}$ and $\dketbra{\id}{\id}^{A^O_{k_n}\tilde A_n^O} * \tilde M_{(k_1,\ldots,k_n)}^{\to k_{n+1}} = M_{(k_1,\ldots,k_n)}^{\to k_{n+1}} * \dketbra{\id}{\id}^{A^I_{k_{n+1}}\tilde A^I_{n+1}}$ when recursively establishing Eqs.~\eqref{eq:link_prod_rho_tilde_M1_tilde_An}--\eqref{eq:link_prod_rho_tilde_M1_tilde_Mn1}.}
The probability that a given order $(k_1,\ldots,k_n)$ is realised can be obtained in both cases as $\Tr[(\id^{\tilde A_O^n\alpha_n} \otimes [\![ (k_1,\ldots,k_n) ]\!]^{C_n'})\varrho_{(n)}']$ and $\sum_{k_{n+1}}\Tr[(\id^{\tilde A_{n+1}^I \alpha_{n+1}} \otimes [\![ (k_1,\ldots,k_n,k_{n+1}) ]\!]^{C_{n+1}})\varrho_{(n+1)}]$, with (similarly to Eqs.~\eqref{eq:Tr_before_Mn_QCFO}--\eqref{eq:Tr_after_Mn_QCFO})
\begin{align}
& \Tr\big[(\id^{\tilde A_O^n\alpha_n} \otimes [\![ (k_1,\ldots,k_n) ]\!]^{C_n'})\varrho_{(n)}'\big] \notag \\
& = \Tr\big[\rho * M_\emptyset^{\to k_1} * A_{k_1} * \cdots * M_{(k_1,\ldots,k_{n-1})}^{\to k_n} * A_{k_n} \big] \notag \\
& = \Tr\!\big[\big(\rho \otimes A_{k_1} \otimes \cdots \otimes A_{k_n} \big)^T \notag \\[-1mm]
& \hspace{15mm} \big( \Tr_{\alpha_n} ( M_\emptyset^{\to k_1} * \cdots * M_{(k_1,\ldots,k_{n-1})}^{\to k_n} ) \otimes \id^{A_{k_n}^O} \big) \big] \label{eq:tr_varrho_n_QCCC}
\end{align}
and
\begin{align}
& \sum_{k_{n+1}} \Tr[(\id^{\tilde A_{n+1}^I \alpha_{n+1}} \otimes [\![ (k_1,\ldots,k_n,k_{n+1}) ]\!]^{C_{n+1}})\varrho_{(n+1)}] \notag \\
& = \sum_{k_{n+1}} \Tr\big[\rho * M_\emptyset^{\to k_1} * A_{k_1} * \cdots * A_{k_n} * M_{(k_1,\ldots,k_n)}^{\to k_{n+1}} \big] \notag \\
& = \Tr\!\big[\big(\rho \otimes A_{k_1} \otimes \cdots \otimes A_{k_n} \big)^T \notag \\
& \hspace{12mm} \sum_{k_{n+1}} \Tr_{A_{k_{n+1}}^I\alpha_{n+1}} ( M_\emptyset^{\to k_1} * \cdots * M_{(k_1,\ldots,k_n)}^{\to k_{n+1}} ) \big]. \label{eq:tr_varrho_n_1_QCCC}
\end{align}
For a classical control these probabilities must be preserved---i.e., the expressions in Eqs.~\eqref{eq:tr_varrho_n_QCCC} and~\eqref{eq:tr_varrho_n_1_QCCC} above must be equal---for each (well-defined) order $(k_1,\ldots,k_n)$. (Note that imposing that all these probabilities are preserved also implies that the whole trace of $\varrho_{(n)}'$ and $\varrho_{(n+1)}$ is preserved.) As this must hold for all $\rho \in \L(\HS^P)$ and all $A_k \in \L(\HS^{A_k^{IO}})$, we then obtain the TP conditions of Eqs.~\eqref{eq:TP_constr_QCCC_n} (for $1\le n < N$) and~\eqref{eq:TP_constr_QCCC_N} (for $n=N$, with $k_{N+1} = F$, $\HS^{A_{k_{N+1}}^I} = \HS^F$ and a trivial $\HS^{\alpha_{N+1}}$), as claimed above.

\medskip

To derive the TP conditions for probabilistic QC-CCs, one again follows the exact same reasoning as for the deterministic case, with $M_{(k_1,\ldots,k_N)}^{\to F}$ replaced by $\sum_r M_{(k_1,\ldots,k_N)}^{\to F\,[r]}$.

\subsubsection{Proof of Proposition~\ref{prop:charact_W_QCCC}: Necessary condition}
\label{app:proof_charact_QCCCs_NC}

Consider the process matrix $W = \sum_{(k_1,\ldots,k_N)} W_{(k_1,\ldots,k_{N},F)}$ of a QC-CC, as per Proposition~\ref{prop:descr_W_QCCC}, with the $W_{(k_1,\ldots,k_{N},F)}$'s of the form of Eq.~\eqref{eq:W_k1_kN_QCCC}, and with the Choi matrices $M_{(k_1,\ldots,k_n)}^{\to k_{n+1}} (\succeq 0)$ of the internal circuit operations satisfying the TP conditions of Eqs.~\eqref{eq:TP_constr_QCCC_1}--\eqref{eq:TP_constr_QCCC_N}.

Let us then define, for all $1 \le n \le N$ and all $(k_1,\ldots,k_n)$, the matrices
\begin{align}
W_{(k_1,\ldots,k_n)} \coloneqq & \Tr_{\alpha_n} \big( M_\emptyset^{\to k_1} * M_{(k_1)}^{\to k_2} * \cdots * M_{(k_1,\ldots,k_{n-1})}^{\to k_n} \big) \notag \\[1mm]
& \hspace{10mm} \in \L(\HS^{PA_{\{k_1,\ldots,k_{n-1}\}}^{IO}A_{k_n}^I}).
\end{align}
As all $M_{(k_1,\ldots,k_n)}^{\to k_{n+1}} \succeq 0$, it directly follows that all $W_{(k_1,\ldots,k_n)}$'s (including the $W_{(k_1,\ldots,k_N,k_{N+1})} = W_{(k_1,\ldots,k_N,F)}$'s for $n= N+1$) are also positive semidefinite. Furthermore, the constraints of Eq.~\eqref{eq:charact_W_QCCC_decomp_constr} are simply equivalent to the TP conditions of Eqs.~\eqref{eq:TP_constr_QCCC_1}--\eqref{eq:TP_constr_QCCC_N}, and are thus readily satisfied by assumption.

\subsubsection{Proof of Proposition~\ref{prop:charact_W_QCCC}: Sufficient condition}

Consider a Hermitian matrix $W \in \L(\HS^{PA_\N^{IO} F})$ that admits a decomposition in terms of positive semidefinite matrices $W_{(k_1,\ldots,k_n)} \in \L(\HS^{PA_{\{k_1,\ldots,k_{n-1}\}}^{IO} A_{k_n}^I})$ and $W_{(k_1,\ldots,k_N,F)} \in \L(\HS^{PA_\N^{IO} F})$ satisfying Eqs.~\eqref{eq:charact_W_QCCC_decomp_sum}--\eqref{eq:charact_W_QCCC_decomp_constr}.
We will show that $W$ is the process matrix of a QC-CC by explicitly constructing some internal circuit operations $\M_{(k_1,\ldots,k_n)}^{\to k_{n+1}}$ satisfying the TP conditions of Eqs.~\eqref{eq:TP_constr_QCCC_1}--\eqref{eq:TP_constr_QCCC_N}, as required.

\medskip

The positive semidefinite matrices $W_{(k_1,\ldots,k_n)}$ (for $1 \le n \le N+1$, recalling that $k_{N+1} = F$) admit a spectral decomposition of the form
\begin{align}
W_{(k_1,\ldots,k_n)} = \sum_i \ketbra{w_{(k_1,\ldots,k_n)}^i}{w_{(k_1,\ldots,k_n)}^i},
\end{align}
for some eigenbasis consisting of $r_{(k_1,\ldots,k_n)} \coloneqq \rank W_{(k_1,\ldots,k_n)}$ (nonnormalised and nonzero) orthogonal vectors $\ket{w_{(k_1,\ldots,k_n)}^i} \in \HS^{PA_{\{k_1,\ldots,k_{n-1}\}}^{IO} A_{k_n}^I}$.
Similarly to what we did for the QC-FO case, let us introduce, for each $n = 1, \ldots, N+1$, some ancillary Hilbert space $\HS^{\alpha_n}$ of dimension $r_n \ge \max_{(k_1,\ldots,k_n)} r_{(k_1,\ldots,k_n)}$ with computational basis $\{\ket{i}^{\alpha_n}\}_{i=1}^{r_n}$, and define%
\footnote{Formally, if $r_{(k_1,\ldots,k_n)} < r_n$, in Eq.~\eqref{eq:def_wk1_kn} we complete the set of $r_{(k_1,\ldots,k_n)}$ nonzero vectors $\ket{w_{(k_1,\ldots,k_n)}^i} \neq 0$ by $(r_n - r_{(k_1,\ldots,k_n)})$ null vectors $\ket{w_{(k_1,\ldots,k_n)}^i} = 0$. \label{footnote:complete_null_vectors}}
\begin{align}
\ket{w_{(k_1,\ldots,k_n)}} \coloneqq & \sum_i \ket{w_{(k_1,\ldots,k_n)}^i} \otimes \ket{i}^{\alpha_n} \notag \\[-2mm]
& \hspace{10mm} \in \HS^{PA_{\{k_1,\ldots,k_{n-1}\}}^{IO}A_{k_n}^I\alpha_n}, \label{eq:def_wk1_kn}
\end{align}
such that $W_{(k_1,\ldots,k_n)} = \Tr_{\alpha_n} \ketbra{w_{(k_1,\ldots,k_n)}}{w_{(k_1,\ldots,k_n)}}$.

Similarly again to the QC-FO case, it can be seen here, via Observation~\ref{observation:nec_cond_range}, that the assumption (from Eq.~\eqref{eq:charact_W_QCCC_decomp_constr}) that $\sum_{k_{n+1}} \Tr_{A_{k_{n+1}}^I} \! W_{(k_1,\ldots,k_n,k_{n+1})} = \sum_{k_{n+1}} \Tr_{A_{k_{n+1}}^I\alpha_{n+1}} \!  \ketbra{w_{(k_1,\ldots,k_{n+1})}}{w_{(k_1,\ldots,k_{n+1})}} = W_{(k_1,\ldots,k_n)} \otimes \id^{A_{k_n}^O}$ implies that $\ket{w_{(k_1,\ldots,k_{n+1})}} \in \range(W_{(k_1,\ldots,k_n)}) \otimes \HS^{A_{k_n}^OA_{k_{n+1}}^I\alpha_{n+1}}$.
Recalling again Lemma~\ref{lemma:invert_link_product}, this property ensures that one can relate $\ket{w_{(k_1,\ldots,k_n)}}$ and $\ket{w_{(k_1,\ldots,k_{n+1})}}$ by defining, for $1 \le n \le N$ and for each $(k_1,\ldots,k_n,k_{n+1})$,%
\footnote{As in Footnote~\ref{footnote:explicit_V_n1_QCFO} for the QC-FO case (see also Eq.~\eqref{eq:invert_link_product_explicit} and Footnote~\ref{footnote:invert_link_prod_ortho_decomp}), from the spectral decomposition of $W_{(k_1,\ldots,k_n)}$ in terms of (orthogonal) vectors $\ket{w_{(k_1,\ldots,k_n)}^i}$, we can write more explicitly $\dket{V_{(k_1,\ldots,k_n)}^{\to k_{n+1}}} = \sum_{i: \ket{w_{(k_1,\ldots,k_n)}^i} \neq 0} \ket{i}^{\alpha_n} \otimes \Big( \frac{\bra{w_{(k_1,\ldots,k_n)}^i}}{\braket{w_{(k_1,\ldots,k_n)}^i}{w_{(k_1,\ldots,k_n)}^i}} \otimes \id^{A_{k_n}^O\!A_{k_{n+1}}^I\alpha_{n+1}} \ket{w_{(k_1,\ldots,k_{n+1})}} \Big)$ for all $1 \le n \le N$ and all $(k_1,\ldots,k_n,k_{n+1})$. \label{footnote:explicit_V_n1_QCCC}}
\begin{align}
& \ket{w_{(k_1,\ldots,k_n)}^+} \coloneqq (\bra{w_{(k_1,\ldots,k_n)}} W_{(k_1,\ldots,k_n)}^+ \otimes \id^{\alpha_n})^T \notag \\
& \hspace{30mm} \in \HS^{PA_{\{k_1,\ldots,k_{n-1}\}}^{IO}A_{k_n}^I\alpha_n}, \label{eq:def_wk1_kn_p} \\[1mm]
& \dket{V_{(k_1,\ldots,k_n)}^{\to k_{n+1}}} \coloneqq \ket{w_{(k_1,\ldots,k_n)}^+} * \ket{w_{(k_1,\ldots,k_n,k_{n+1})}} \notag \\
& \hspace{30mm} \in \HS^{A_{k_n}^O\alpha_nA_{k_{n+1}}^I\alpha_{n+1}},
\end{align}
so that
\begin{align}
\ket{w_{(k_1,\ldots,k_n)}} * \dket{V_{(k_1,\ldots,k_n)}^{\to k_{n+1}}} = \ket{w_{(k_1,\ldots,k_n,k_{n+1})}}. \label{eq:QCCC_wk1_kn_wk1_kn1} 
\end{align}
By further defining $\dket{V_\emptyset^{\to k_1}} \coloneqq \ket{w_{(k_1)}} \in \HS^{PA_{k_1}^I\alpha_1}$, one recursively obtains
\begin{align}
\dket{V_\emptyset^{\to k_1}} * \dket{V_{(k_1)}^{\to k_{2}}} * \cdots * \dket{V_{(k_1,\ldots,k_{n-1})}^{\to k_n}} = \ket{w_{(k_1,\ldots,k_n)}} \label{eq:QCCC_rec_Vs}
\end{align}
for all $1 \le n \le N+1$ and all $(k_1,\ldots,k_n)$.

From the double-ket vectors $\dket{V_{(k_1,\ldots,k_{n-1})}^{\to k_n}}$ just introduced we can then define the operators
\begin{align}
& M_{\!(k_1,\ldots,k_{n-1})}^{\to k_n} \!\coloneqq\! \dketbra{V_{\!(k_1,\ldots,k_{n-1})}^{\to k_n}}{V_{\!(k_1,\ldots,k_{n-1})}^{\to k_n}} \ \ \text{for } 1 \!\le\! n \!\le\! N \notag \\[1mm]
& \text{and } M_{(k_1,\ldots,k_N)}^{\to F} \coloneqq \Tr_{\alpha_{N+1}} \dketbra{V_{(k_1,\ldots,k_N)}^{\to F}}{V_{(k_1,\ldots,k_N)}^{\to F}},
\end{align}
such that 
\begin{align}
& \Tr_{\alpha_n} (M_\emptyset^{\to k_1} * M_{(k_1)}^{\to k_2} * \cdots * M_{(k_1,\ldots,k_{n-1})}^{\to k_n}) \notag \\
& \quad = \Tr_{\alpha_n} \ketbra{w_{(k_1,\ldots,k_n)}}{w_{(k_1,\ldots,k_n)}} = W_{(k_1,\ldots,k_n)} \notag \\[2mm]
& \text{and} \quad M_\emptyset^{\to k_1} * M_{(k_1)}^{\to k_2} * \cdots * M_{(k_1,\ldots,k_N)}^{\to F} = W_{(k_1,\ldots,k_N,F)}. \label{eq:QCCC_rec_Ms}
\end{align}
As the $W_{(k_1,\ldots,k_n)}$'s are assumed to satisfy the constraints of Eq.~\eqref{eq:charact_W_QCCC_decomp_constr}, then by construction the (positive semidefinite) $M_{(k_1,\ldots,k_{n-1})}^{\to k_n}$'s satisfy the (equivalent, once Eq.~\eqref{eq:QCCC_rec_Ms} is established) TP conditions of Eqs.~\eqref{eq:TP_constr_QCCC_1}--\eqref{eq:TP_constr_QCCC_N}.
This proves that the operators $M_{(k_1,\ldots,k_{n-1})}^{\to k_n}$ define valid QC-CC internal circuit operations $\M_\emptyset^{\to k_1}: \L(\HS^P) \to \L(\HS^{A_{k_1}^I\alpha_1})$, $\M_{(k_1,\ldots,k_n)}^{\to k_{n+1}}: \L(\HS^{A_{k_n}^O\alpha_n}) \to \L(\HS^{A_{k_{n+1}}^I\alpha_{n+1}})$ for $n=1,\ldots,N-1$, and $\M_{(k_1,\ldots,k_N)}^{\to F}: \L(\HS^{A_{k_N}^O\alpha_N}) \to \L(\HS^F)$.
The last line of Eq.~\eqref{eq:QCCC_rec_Ms} above shows, according to Proposition~\ref{prop:descr_W_QCCC}, that $W = \sum_{(k_1,\ldots,k_N)} W_{(k_1,\ldots,k_N,F)}$ is indeed the process matrix of the QC-CC thus constructed.

\subsubsection{Proof of Proposition~\ref{prop:charact_Wr_probQCCC}}

The proofs above extend again easily to the characterisation of probabilistic QC-CCs, as given by Proposition~\ref{prop:charact_Wr_probQCCC} in Sec.~\ref{sec:probQCCCs}.

For the necessary condition, recall that according to Proposition~\ref{prop:descr_Wr_probQCCC}, a pQC-CC is described by a set of positive semidefinite matrices $W^{[r]}$ obtained indeed as in Eq.~\eqref{eq:charact_Wr_probQCCC}, with the matrices $W_{(k_1,\ldots,k_{N},F)}^{[r]}$ obtained as in Eq.~\eqref{eq:Wr_k1_kN_probQCCC}.
As all matrices $M_{(k_1,\ldots,k_{n-1})}^{\to k_n}$ and $\sum_r M_{(k_1,\ldots,k_N)}^{\to F\,[r]}$ must satisfy the TP conditions of Eqs.~\eqref{eq:TP_constr_QCCC_1}--\eqref{eq:TP_constr_QCCC_N}, then the (positive semidefinite) matrices $W_{(k_1,\ldots,k_n)} \coloneqq \Tr_{\alpha_n} (M_\emptyset^{\to k_1} * M_{(k_1)}^{\to k_2} * \cdots * M_{(k_1,\ldots,k_{n-1})}^{\to k_n})$ and $W_{(k_1,\ldots,k_N,F)} \coloneqq \sum_r W_{(k_1,\ldots,k_N,F)}^{[r]} = M_\emptyset^{\to k_1} * M_{(k_1)}^{\to k_2} * \cdots * M_{(k_1,\ldots,k_{N-1})}^{\to k_N} * \big( \sum_r M_{(k_1,\ldots,k_N)}^{\to F\,[r]} \big)$ are those that enter the decomposition of the process matrix of a QC-CC (see Sec.~\ref{app:proof_charact_QCCCs_NC} above), and therefore must satisfy Eq.~\eqref{eq:charact_W_QCCC_decomp_constr} of Proposition~\ref{prop:charact_W_QCCC}.

Conversely for the sufficient condition, consider a set of positive semidefinite matrices $W^{[r]}$ that can be decomposed in terms of positive semidefinite matrices $W_{(k_1,\ldots,k_n)}$ and $W_{(k_1,\ldots,k_N,F)}^{[r]}$ as per Proposition~\ref{prop:charact_Wr_probQCCC}, and define again the ``extended'' matrix $W' \coloneqq \sum_r W^{[r]} \otimes \ketbra{r}{r}^{F'} \in \L(\HS^{PA_\N^{IO}FF'})$ by introducing an additional Hilbert space $\HS^{F'}$ with computational basis states $\{\ket{r}^{F'}\}_r$.
We note now that $W'$ is the process matrix of a (deterministic) QC-CC with global future space $\HS^{FF'}$,%
\footnote{Note that this would not be true in general if we just assumed that the probabilistic process matrices $W^{[r]}$ sum up to a QC-CC, without imposing the decomposition of Eq.~\eqref{eq:charact_Wr_probQCCC} for each $r$ individually. See also the discussion after Proposition~\ref{prop:charact_Wr_probQCCC}, and the quantum switch example at the end of Sec.~\ref{sec:probQCQCs}. \label{footnote:specific_probQCCC}}
since it has a decomposition as in Proposition~\ref{prop:charact_W_QCCC} (with $W_{(k_1,\ldots,k_N,FF')}' \coloneqq \sum_r W^{[r]}_{(k_1,\ldots,k_N,F)} \otimes \ketbra{r}{r}^{F'}$ and $\Tr_{FF'} W_{(k_1,\ldots,k_N,FF')}' = \Tr_F W_{(k_1,\ldots,k_N,F)}$, which satisfies the corresponding constraints by assumption).
One can therefore construct internal circuit operations $M_{(k_1,\ldots,k_n)}^{\to k_{n+1}}$ and $M_{(k_1,\ldots,k_N)}^{\prime \ \to FF'}$ as in the proof for the sufficient condition of Proposition~\ref{prop:charact_W_QCCC} above, satisfying the TP conditions of Eqs.~\eqref{eq:TP_constr_QCCC_1}--\eqref{eq:TP_constr_QCCC_N}, such that in particular $W_{(k_1,\ldots,k_N,FF')}' = M_\emptyset^{\to k_1} * M_{(k_1)}^{\to k_2} * \cdots * M_{(k_1,\ldots,k_{N-1})}^{\to k_N} * M_{(k_1,\ldots,k_N)}^{\prime \ \to FF'}$. Defining the CP maps $M_{(k_1,\ldots,k_N)}^{\to F\,[r]} \coloneqq M_{(k_1,\ldots,k_N)}^{\prime \ \to FF'} * \ketbra{r}{r}^{F'}$ (whose sum $\sum_r M_{(k_1,\ldots,k_N)}^{\to F\,[r]} = \Tr_{F'} M_{(k_1,\ldots,k_N)}^{\prime \ \to FF'}$ satisfies the required TP condition), we obtain $W^{[r]}_{(k_1,\ldots,k_N,F)} = W_{(k_1,\ldots,k_N,FF')}' * \ketbra{r}{r}^{F'} = M_\emptyset^{\to k_1} * M_{(k_1)}^{\to k_2} * \cdots * M_{(k_1,\ldots,k_{N-1})}^{\to k_N} * M_{(k_1,\ldots,k_N)}^{\to F\,[r]}$, so that each $W^{[r]}$ is indeed of the form of Eqs.~\eqref{eq:Wr_probQCCC}--\eqref{eq:Wr_k1_kN_probQCCC}, which proves, according to Proposition~\ref{prop:descr_Wr_probQCCC}, that $\{W^{[r]}\}_r$ is a pQC-CC.

\subsection{QC-QCs: Proofs of Propositions~\ref{prop:charact_W_QCQC} and~\ref{prop:charact_Wr_probQCQC}}
\label{app:proof_charact_QCQCs}

In this section, we will derive the TP conditions for QC-QCs, and prove Propositions~\ref{prop:charact_W_QCQC} and~\ref{prop:charact_Wr_probQCQC}.

As in the previous section, we define $k_{N+1} \coloneqq F$, $\HS^{A^I_{k_{N+1}}} \coloneqq \HS^F$ and $\HS^{\tilde A^I_{N+1}} \coloneqq \HS^F$, and here also $\HS^{\alpha_{N+1}} \coloneqq \HS^{\alpha_F}$, as for instance in $\dket{V_{\K_N \backslash k_N,k_N}^{\to k_{N+1}}} = \dket{V_{\N \backslash k_N,k_N}^{\to F}} (\in \HS^{A_{k_N}^O \alpha_NA_{k_{N+1}}^I \alpha_{N+1}} = \HS^{A_{k_N}^O \alpha_NF \alpha_F})$ and $W_{(\K_N,k_{N+1})} = W_{(\N,F)} \coloneqq W (\in \L(\HS^{PA_{\K_N}^{IO} A_{k_{N+1}}^I}) = \L(\HS^{PA_\N^{IO} F}))$ (with $\K_N = \N$).

\subsubsection{Trace-preserving conditions}
\label{app:subsubsec_QCQC_TP}

The TP conditions of Eqs.~\eqref{eq:TP_constr_QCQC_1}--\eqref{eq:TP_constr_QCQC_N} are obtained by imposing that each internal circuit operation $\tilde V_n$ in a QC-QC preserves the norm (as we consider pure states and pure operations) of their global input state, involving the target, the ancillary and the control systems.

Suppose that one inputs some state $\ket{\psi} \in \HS^P$ into the circuit. The global state $\ket{\varphi_{(1)}} \in \HS^{\tilde A_1^I \alpha_1C_1}$ right after the first internal circuit operation $\tilde V_1$ is
\begin{align}
\ket{\varphi_{(1)}} & = \ket{\psi} * \dket{\tilde V_1} = \sum_{k_1} \big( \ket{\psi} * \dket{\tilde V_{\emptyset,\emptyset}^{\to k_1}} \big) \otimes \ket{\emptyset,k_1}^{C_1}.
\end{align}
We want its norm to be equal to that of $\ket{\psi}$---i.e., we want
\begin{align}
\braket{\varphi_{(1)}}{\varphi_{(1)}} & = \sum_{k_1} \Tr \big[ \ketbra{\psi}{\psi} * \dketbra{\tilde V_{\emptyset,\emptyset}^{\to k_1}}{\tilde V_{\emptyset,\emptyset}^{\to k_1}} \big] \notag \\[-1mm]
& =\Tr \!\big[ \big(\!\ketbra{\psi}{\psi}\!\big)^{\!T} \!\big( \sum_{k_1} \Tr_{A_{k_1}^I\!\alpha_1}\! \dketbra{V_{\emptyset,\emptyset}^{\to k_1}}{V_{\emptyset,\emptyset}^{\to k_1}} \! \big) \big] \notag \\[-1mm]
& \qquad = \braket{\psi}{\psi} \quad \big( \, = \Tr \big[ \big(\ketbra{\psi}{\psi} \big)^T \big] \, \big)
\end{align}
where, similarly to Eq.~\eqref{eq:TP_QCCC_calculation1}, we removed the tildes by using the appropriate isomorphism (see Footnote~\ref{fn:QCQCtilde_translation}).

As this must hold for all $\ket{\psi} \in \HS^P$, this constraint is indeed equivalent to Eq.~\eqref{eq:TP_constr_QCQC_1} (with $\ket{w_{(\emptyset,k_1)}} = \ket{w_{(k_1)}} = \dket{V_{\emptyset,\emptyset}^{\to k_1}}$).

For $n = 1, \ldots, N$, the states $\ket{\varphi_{(n)}'} \in \HS^{\tilde A_n^O\alpha_n C_n'}$ and $\ket{\varphi_{(n+1)}} \in \HS^{\tilde A_{n+1}^I\alpha_{n+1} C_{n+1}}$ of the global system going through the circuit right before and right after the application of $\tilde V_{n+1}$, respectively, are easily obtained recursively, and can be expressed---by rearranging the sums, introducing the vectors $\ket{\psi,A_{\K_n}} \coloneqq \ket{\psi} \bigotimes_{k \in \K_n} \dket{A_k} \in \HS^{PA_{\K_n}^{IO}}$ and using the vectors $\ket{w_{(\K_{n-1},k_n)}} \in \HS^{PA_{\K_{n-1}}^{IO} A_{k_n}^I\alpha_n}$ defined in Eqs.~\eqref{eq:def_w_Knm1_kn} and~\eqref{eq:W_QCQC}---as
\begin{align}
& \ket{\varphi_{(n)}'} = \ket{\psi} \!*\! \dket{\tilde V_1} \!*\! \dket{\tilde A_1} \!*\! \dket{\tilde V_2} \!*\! \dket{\tilde A_2} \!*\! \cdots \!*\! \dket{\tilde V_n} \!*\! \dket{\tilde A_n} \notag \\[2mm]
& \quad = \sum_{(k_1,\ldots,k_n)} \!\! \big(\ket{\psi} \!*\! \dket{V_{\emptyset,\emptyset}^{\to k_1}} \!*\! \dket{A_{k_1}} \!*\! \dket{V_{\emptyset,k_1}^{\to k_2}} \!*\! \dket{A_{k_2}} \!*\! \cdots \notag \\[-3mm]
& \hspace{20mm} \cdots \!*\! \dket{V_{\{k_1,\ldots,k_{n-2}\},k_{n-1}}^{\to k_n}} \!*\! \dket{A_{k_n}} * \dket{\id}^{A^O_{k_n}\tilde A^O_n} \big) \notag \\
& \hspace{43mm} \otimes \ket{\{k_1,\ldots,k_{n-1}\},k_n}^{C_n'} \notag \\[2mm]
& \quad = \!\!\! \sum_{\substack{\K_n, \\ (k_1,\ldots,k_n) \in \K_n}} \!\!\!\!\!\!\!\!\! \ket{\psi,A_{\K_n}} * \big( \dket{V_{\emptyset,\emptyset}^{\to k_1}} * \dket{V_{\emptyset,k_1}^{\to k_2}} * \cdots \notag \\[-6mm]
& \hspace{30mm} \cdots * \dket{V_{\{k_1,\ldots,k_{n-2}\},k_{n-1}}^{\to k_n}} \big) * \dket{\id}^{A^O_{k_n}\tilde A^O_n}  \notag \\
& \hspace{43mm} \otimes \ket{\{k_1,\ldots,k_{n-1}\},k_n}^{C_n'} \notag \\[2mm]
& \quad = \sum_{\K_n, k_n \in \K_n} \!\!\! \big( \ket{\psi,A_{\K_n}} \!*\! \ket{w_{(\K_n\backslash k_n,k_n)}} * \dket{\id}^{A^O_{k_n}\tilde A^O_n} \big)\notag \\[-3mm]
&\hspace{43mm} \otimes \ket{\K_n\backslash k_n,k_n}^{C_n'} \label{eq:link_prod_psi_tilde_V1_tilde_An}
\end{align}
and
\begin{align}
& \ket{\varphi_{(n+1)}} = \ket{\psi} \!*\! \dket{\tilde V_1} \!*\! \dket{\tilde A_1} \!*\! \cdots \!*\! \dket{\tilde V_n} \!*\! \dket{\tilde A_n} \!*\! \dket{\tilde V_{n+1}} \notag \\[1mm]
& \quad = \sum_{(k_1,\ldots,k_n,k_{n+1})} \!\! \big(\ket{\psi} \!*\! \dket{V_{\emptyset,\emptyset}^{\to k_1}} \!*\! \dket{A_{k_1}} \!*\! \cdots \notag \\[-1mm]
& \hspace{16mm} \cdots \!*\! \dket{A_{k_n}} \!*\! \dket{V_{\{k_1,\ldots,k_{n-1}\},k_n}^{\to k_{n+1}}} * \dket{\id}^{A^I_{k_{n+1}}\tilde A^I_{n+1}} \big) \notag \\[1mm]
& \hspace{43mm} \otimes \ket{\{k_1,\ldots,k_n\},k_{n+1}}^{C_{n+1}} \notag \\[2mm]
& \quad = \! \sum_{\substack{\K_n, (k_1,\ldots,k_n) \in \K_n, \\ k_{n+1} \notin \K_n}} \!\!\! \ket{\psi,A_{\K_n}} * \big( \dket{V_{\emptyset,\emptyset}^{\to k_1}} * \dket{V_{\emptyset,k_1}^{\to k_2}} * \cdots \notag \\[-3mm]
& \hspace{26mm} \cdots * \dket{V_{\{k_1,\ldots,k_{n-1}\},k_n}^{\to k_{n+1}}} \big) * \dket{\id}^{A^I_{k_{n+1}}\tilde A^I_{n+1}} \notag \\[1mm]
& \hspace{43mm} \otimes \ket{\{k_1,\ldots,k_n\},k_{n+1}}^{C_{n+1}} \notag \\[2mm]
& \quad = \!\! \sum_{\K_n, k_{n+1} \notin \K_n} \!\!\! \big( \ket{\psi,A_{\K_n}} \!*\! \ket{w_{(\K_n,k_{n+1})}} * \dket{\id}^{A^I_{k_{n+1}}\tilde A^I_{n+1}} \big) \notag \\[-2mm]
& \hspace{43mm} \otimes \ket{\K_n,k_{n+1}}^{C_{n+1}} \label{eq:link_prod_psi_tilde_V1_tilde_Vn1}
\end{align}
(where the sums $\sum_{\K_n}$ are over all subsets $\K_n$ of $\N$ such that $|\K_n| = n$). 
From the second lines in Eqs.~\eqref{eq:link_prod_psi_tilde_V1_tilde_An}--\eqref{eq:link_prod_psi_tilde_V1_tilde_Vn1}, we again removed the tildes using the appropriate isomorphism (similarly to Eqs.~\eqref{eq:link_prod_rho_tilde_M1_tilde_An}--\eqref{eq:link_prod_rho_tilde_M1_tilde_Mn1} above, and as in Eq.~\eqref{eq:Choi_V_QCQC} in the main text).

The squared norms of $\ket{\varphi_{(n)}'}$ and $\ket{\varphi_{(n+1)}}$ are then
\begin{align}
& \braket{\varphi_{(n)}'}{\varphi_{(n)}'} \notag \\
& = \!\!\!\!\! \sum_{\K_n, k_n \in \K_n} \!\!\!\!\!\! \Tr\big[ \ketbra{\psi,\!A_{\K_n}}{\psi,\!A_{\K_n}} \!*\! \ketbra{w_{(\K_n\backslash k_n,k_n)}}{w_{(\K_n\backslash k_n,k_n)}} \notag \\[-4mm]
&\hspace{62mm} * \dketbra{\id}{\id}^{A^O_{k_n}\tilde A^O_n} \!\big] \notag \\[1mm]
& = \sum_{\K_n} \Tr\big[ \big( \ketbra{\psi,A_{\K_n}}{\psi,A_{\K_n}} \big)^T \notag \\[-3mm]
& \hspace{17mm} \big( \!\! \sum_{k_n \in \K_n} \!\!\! \Tr_{\alpha_n}\! \ketbra{w_{(\K_n\backslash k_n,k_n)}}{w_{(\K_n\backslash k_n,k_n)}} \!\otimes\! \id^{A_{k_n}^O} \big) \big] \label{eq:varphi_n_varphi_n}
\end{align}
and
\begin{align}
& \braket{\varphi_{(n+1)}}{\varphi_{(n+1)}} \notag \\
& = \!\!\!\!\! \sum_{\K_n, k_{n+1} \notin \K_n} \!\!\!\!\!\! \Tr\big[ \ketbra{\psi,\!A_{\K_n}}{\psi,\!A_{\K_n}} \!*\! \ketbra{w_{(\K_n,k_{n+1})}}{w_{(\K_n,k_{n+1})}} \notag \\[-4mm]
&\hspace{57mm}* \dketbra{\id}{\id}^{A^I_{k_{n+1}}\tilde A^I_{n+1}} \big] \notag \\[1mm]
& = \sum_{\K_n} \Tr\big[ \big( \ketbra{\psi,A_{\K_n}}{\psi,A_{\K_n}} \big)^T \notag \\[-3mm]
& \hspace{16mm} \big( \!\! \sum_{k_{n+1} \notin \K_n} \!\!\!\! \Tr_{A_{k_{n+1}}^I\alpha_{n+1}}\! \ketbra{w_{(\K_n,k_{n+1})}}{w_{(\K_n,k_{n+1})}} \big) \big]. \label{eq:varphi_n1_varphi_n1}
\end{align}
We require these norms to be the same, for all possible $\ket{\psi}$ and all $A_k$'s.
Let us take, for a given $\K_n$ with $|\K_n|=n$, all $A_{k'} = 0$ for all $k' \notin \K_n$. The sums $\sum_{\K_n}$ in Eqs.~\eqref{eq:varphi_n_varphi_n} and~\eqref{eq:varphi_n1_varphi_n1} above then reduce to just the single term corresponding to that particular $\K_n$. Hence, the equality of Eqs.~\eqref{eq:varphi_n_varphi_n} and~\eqref{eq:varphi_n1_varphi_n1} must in fact hold for each $\K_n$ individually (and not just for their sums):
\begin{align}
& \Tr\big[ \big(\ketbra{\psi,A_{\K_n}}{\psi,A_{\K_n}}\big)^T \notag \\
& \hspace{10mm} \big( \sum_{k_{n+1} \notin \K_n} \!\! \Tr_{A_{k_{n+1}}^I\alpha_{n+1}} \ketbra{w_{(\K_n,k_{n+1})}}{w_{(\K_n,k_{n+1})}} \big) \big] \notag \\
& = \Tr\big[ \big(\ketbra{\psi,A_{\K_n}}{\psi,A_{\K_n}}\big)^T \notag \\
& \hspace{10mm} \big( \sum_{k_n \in \K_n} \!\! \Tr_{\alpha_n} \ketbra{w_{(\K_n\backslash k_n,k_n)}}{w_{(\K_n\backslash k_n,k_n)}} \otimes \id^{A_{k_n}^O} \big) \big]. 
\end{align}
As this must hold for the $\ketbra{\psi,A_{\K_n}}{\psi,A_{\K_n}}$'s spanning the whole spaces $\L(\HS^{PA_{\K_n}^{IO}})$, this implies that one must have, for all $\K_n$,
\begin{align}
& \sum_{k_{n+1} \notin \K_n} \!\! \Tr_{A_{k_{n+1}}^I\alpha_{n+1}} \ketbra{w_{(\K_n,k_{n+1})}}{w_{(\K_n,k_{n+1})}} \notag \\
& \hspace{3mm} = \!\! \sum_{k_n \in \K_n} \!\!\! \Tr_{\alpha_n} \! \ketbra{w_{(\K_n\backslash k_n,k_n)}}{w_{(\K_n\backslash k_n,k_n)}} \otimes \id^{A_{k_n}^O},
\end{align}
which indeed gives the TP conditions of Eqs.~\eqref{eq:TP_constr_QCQC_n} (for $1 \le n < N$) and~\eqref{eq:TP_constr_QCQC_N} (for $n = N$, with $k_{N+1} = F$, $\HS^{A_{k_{N+1}}^I} = \HS^F$, $\alpha_{N+1} = \alpha_F$ and $\ket{w_{(\K_N,k_{N+1})}} = \ket{w_{(\N,F)}}$).

\medskip

For the case of probabilistic QC-QCs, the only difference is that we impose $\sum_r \braket{\varphi_{(N+1)}^{[r]}}{\varphi_{(N+1)}^{[r]}} = \braket{\varphi_{(N)}'}{\varphi_{(N)}'}$, where $\ket{\varphi_{(N+1)}^{[r]}} = \ket{\psi} * \dket{\tilde V_1} * \dket{\tilde A_1} * \cdots * \dket{\tilde V_N} * \dket{\tilde A_N} * \dket{\tilde V_{N+1}^{[r]}}$ is the (unnormalised) state of the global system after the last internal operation $\tilde V_{N+1}^{[r]}$, corresponding to the classical outcome $r$ of the probabilistic circuit. The same reasoning as above leads to the same constraint as Eq.~\eqref{eq:TP_constr_QCQC_N}, with $\Tr_{F\alpha_F} \ketbra{w_{(\N,F)}}{w_{(\N,F)}}$ replaced by $\sum_r \Tr_{F\alpha_F} \ketbra{w_{(\N,F)}^{[r]}}{w_{(\N,F)}^{[r]}}$.

\subsubsection{Proof of Proposition~\ref{prop:charact_W_QCQC}: Necessary condition}
\label{app:proof_charact_QCQCs_NC}

Consider the process matrix $W \!=\! \Tr_{\alpha_F} \!\ketbra{w_{(\N\!,F)}}{w_{(\N\!,F)}}$ with $\ket{w_{(\N,F)}} = \sum_{(k_1,\ldots,k_N)} \ket{w_{(k_1,\ldots,k_N,F)}}$ of a QC-QC, as per Proposition~\ref{prop:descr_W_QCQC}, with $\ket{w_{(k_1,\ldots,k_N,F)}} \in \HS^{PA_\N^{IO} F\alpha_{F}}$ of the form of Eq.~\eqref{eq:def_w_k1_kN_F_QCQC}, and with the internal circuit operations $V_{\K_{n-1},k_n}^{\to k_{n+1}}$ satisfying the TP conditions of Eqs.~\eqref{eq:TP_constr_QCQC_1}--\eqref{eq:TP_constr_QCQC_N}, written in terms of the vectors $\ket{w_{(\K_{n-1},k_n)}}$ defined in Eq.~\eqref{eq:def_w_Knm1_kn}.

Let us then define, for all $1 \le n \le N$, all subsets $\K_{n-1}$ of $\N$ with $|\K_{n-1}|=n-1$ and all $k_n \in \N \backslash \K_{n-1}$, the matrices
\begin{align}
W_{(\K_{n-1},k_n)} \coloneqq & \Tr_{\alpha_n} \ketbra{w_{(\K_{n-1},k_n)}}{w_{(\K_{n-1},k_n)}} \notag \\[1mm]
& \hspace{20mm} \in \L(\HS^{PA_{\K_{n-1}}^{IO}A_{k_n}^I}).
\end{align}
It is clear, from their definition, that all $W_{(\K_{n-1},k_n)}$'s are positive semidefinite.
Furthermore, the constraints of Eq.~\eqref{eq:charact_W_QCQC_decomp} are simply equivalent to the TP conditions of Eqs.~\eqref{eq:TP_constr_QCQC_1}--\eqref{eq:TP_constr_QCQC_N}, and are thus readily satisfied by assumption.

\subsubsection{Proof of Proposition~\ref{prop:charact_W_QCQC}: Sufficient condition}

Consider a Hermitian matrix $W \in \L(\HS^{PA_\N^{IO} F})$ such that there exist positive semidefinite matrices $W_{(\K_{n-1},k_n)} \in \L(\HS^{PA_{\K_{n-1}}^{IO} A_{k_n}^I})$ for all $\K_{n-1} \subsetneq \N$ and $k_n \in \N\backslash\K_{n-1}$ satisfying Eq.~\eqref{eq:charact_W_QCQC_decomp}.
As in the cases above, we will show that $W$ is the process matrix of a QC-QC by explicitly constructing now some internal circuit operations $V_{\emptyset,\emptyset}^{\to k_1}$, $V_{\K_{n-1},k_n}^{\to k_{n+1}}$ and $V_{\N \backslash k_N,k_N}^{\to F}$ satisfying the TP conditions of Eqs.~\eqref{eq:TP_constr_QCQC_1}--\eqref{eq:TP_constr_QCQC_N}, as required.

\medskip

The positive semidefinite matrices $W_{(\K_{n-1},k_n)}$ (for $1 \le n \le N+1$, with $k_{N+1} \coloneqq F$ and $W_{(\N,F)} \coloneqq W$) admit a spectral decomposition of the form
\begin{align}
W_{(\K_{n-1},k_n)} = \sum_i \ketbra{w_{(\K_{n-1},k_n)}^i}{w_{(\K_{n-1},k_n)}^i}, \label{eq:spectral_decomp_W_Kn1_kn}
\end{align}
for some eigenbasis consisting of $r_{(\K_{n-1},k_n)} \coloneqq \rank W_{(\K_{n-1},k_n)}$ (nonnormalised and nonzero) orthogonal vectors $\ket{w_{(\K_{n-1},k_n)}^i} \in \HS^{PA_{\K_{n-1}}^{IO} A_{k_n}^I}$.
Similarly to what we did for the previous cases, let us introduce, for each $n = 1, \ldots, N+1$, some ancillary Hilbert space $\HS^{\alpha_n}$ of dimension $r_n \ge \max_{(\K_{n-1},k_n)} r_{(\K_{n-1},k_n)}$ with computational basis $\{\ket{i}^{\alpha_n}\}_{i=1}^{r_n}$, and define%
\footnote{As previously (see  Footnote~\ref{footnote:complete_null_vectors}), if $r_{(\K_{n-1},k_n)} < r_n$, in Eq.~\eqref{eq:def_wKn1_kn} we complete the set of $r_{(\K_{n-1},k_n)}$ nonzero vectors $\ket{w_{(\K_{n-1},k_n)}^i} \neq 0$ by $(r_n - r_{(\K_{n-1},k_n)})$ null vectors.}
\begin{align}
\ket{w_{(\K_{n-1},k_n)}} \coloneqq & \sum_i \ket{w_{(\K_{n-1},k_n)}^i} \otimes \ket{i}^{\alpha_n} \notag \\[-1mm]
& \hspace{15mm} \in \HS^{PA_{\K_{n-1}}^{IO}A_{k_n}^I\alpha_n}, \label{eq:def_wKn1_kn}
\end{align}
such that $W_{(\K_{n-1},k_n)} = \Tr_{\alpha_n} \ketbra{w_{(\K_{n-1},k_n)}}{w_{(\K_{n-1},k_n)}}$.

Contrary to the previous two cases, we do not want to relate $\ket{w_{(\K_{n-1},k_n)}}$ and $\ket{w_{(\K_n,k_{n+1})}}$ directly by a link product (as in Eqs.~\eqref{eq:QCFO_wn_wn1} and~\eqref{eq:QCCC_wk1_kn_wk1_kn1}), but via a sum as in Eq.~\eqref{eq:w_Kn_kn1_sum}. 
Some more preliminary work is therefore required.
To this aim, let us introduce some Hilbert space $\HS^{O'}$ isomorphic to any $\HS^{A_k^O}$ (which we assumed to all be isomorphic) and some $N$-dimensional Hilbert space $\Gamma$ with computational basis $\{\ket{k}^\Gamma\}_{k\in\N}$, and let us define, for each nonempty subset $\K_n$ of $\N$,
\begin{align}
& \ket{\omega_{\K_n}} \coloneqq \sum_{k_n\in\K_n} \ket{w_{(\K_n\backslash k_n,k_n)}} \otimes \dket{\id}^{A_{k_n}^OO'} \otimes \ket{k_n}^\Gamma \notag \\[-1mm]
& \hspace{30mm} \in \HS^{PA_{\K_n}^{IO}\alpha_nO'\Gamma}, \label{eq:def_omega_Kn}
\end{align}
such that $\Omega_{\K_n} \coloneqq \Tr_{\alpha_nO'\Gamma} \ketbra{\omega_{\K_n}}{\omega_{\K_n}} = \sum_{k_n} \Tr_{\alpha_n} \ketbra{w_{(\K_n\backslash k_n,k_n)}}{w_{(\K_n\backslash k_n,k_n)}} \otimes \id^{A_{k_n}^O} = \sum_{k_n} W_{(\K_n\backslash k_n,k_n)} \otimes \id^{A_{k_n}^O} \in \L(\HS^{PA_{\K_n}^{IO}})$.

As before it can be seen here, via Observation~\ref{observation:nec_cond_range}, that the assumption (from Eq.~\eqref{eq:charact_W_QCQC_decomp}) that $\sum_{k_{n+1}} \Tr_{A_{k_{n+1}}^I} W_{(\K_n,k_{n+1})} = \sum_{k_{n+1}} \Tr_{A_{k_{n+1}}^I\alpha_{n+1}}  \ketbra{w_{(\K_n,k_{n+1})}}{w_{(\K_n,k_{n+1})}} = \sum_{k_n} W_{(\K_n \backslash k_n,k_n)}\otimes \id^{A_{k_n}^O} = \Omega_{\K_n}$ implies that $\ket{w_{(\K_n,k_{n+1})}} \in \range(\Omega_{\K_n}) \otimes \HS^{A_{k_{n+1}}^I\alpha_{n+1}}$.
Using once again Lemma~\ref{lemma:invert_link_product}, this ensures that one can relate $\ket{\omega_{\K_n}}$ and $\ket{w_{(\K_n,k_{n+1})}}$ by defining, for each $\K_n$ and $k_{n+1} \notin \K_n$,
\begin{align}
\ket{\omega_{\K_n}^+} & \coloneqq (\bra{\omega_{\K_n}} \Omega_{\K_n}^+ \otimes \id^{\alpha_nO'\Gamma})^T \in \HS^{PA_{\K_n}^{IO}\alpha_nO'\Gamma}, \label{eq:def_omega_Kn_p} \\[1mm]
\ket{V_{\K_n}^{k_{n+1}}} & \coloneqq \ket{\omega_{\K_n}^+} * \ket{w_{(\K_n,k_{n+1})}} \in \HS^{\alpha_nO'\Gamma A_{k_{n+1}}^I\alpha_{n+1}},
\end{align}
so that
\begin{align}
\ket{\omega_{\K_n}} * \ket{V_{\K_n}^{k_{n+1}}} = \ket{w_{(\K_n,k_{n+1})}}. \label{eq:inv_omega_Kn}
\end{align}
From $\ket{V_{\K_n}^{k_{n+1}}}$ thus obtained, let us then define%
\footnote{More explicitly, using the form of Eq.~\eqref{eq:invert_link_product_explicit} for $\ket{V_{\K_n}^{k_{n+1}}}$, one obtains $\dket{V_{\K_n\backslash k_n,k_n}^{\to k_{n+1}}} = \big[ \sum_i \ket{i}^{\alpha_n} \otimes \big( \bra{w_{(\K_n\backslash k_n,k_n)}^i} \otimes \id^{A_{k_n}^O} \big) \big] \Omega_{\K_n}^+ \otimes \id^{A_{k_{n+1}}^I\alpha_{n+1}} \ket{w_{(\K_n,k_{n+1})}}$.
(Note that one cannot in general further simplify this expression as in Footnotes~\ref{footnote:invert_link_prod_ortho_decomp}, \ref{footnote:explicit_V_n1_QCFO} and~\ref{footnote:explicit_V_n1_QCCC}, because for a given $\K_n$, $\Omega_{\K_n}^+$ is generally not simply expressed in terms of the $\ket{w_{(\K_n\backslash k_n,k_n)}^i}$'s, and the vectors $\ket{\omega_{\K_n}^{i,j,k_n}} \coloneqq (\id^{PA_{\K_n}^{IO}} \otimes \bra{i}^{\alpha_n} \bra{j}^{O'} \bra{k_n}^{\Gamma}) \ket{\omega_{\K_n}}=\ket{w_{(\K_n\backslash k_n,k_n)}^i}\ket{j}^{A_{k_n}^O}$---which play the role of the $\ket{a_i}$'s, e.g., in Eq.~\eqref{eq:invert_link_product_explicit}---are not necessarily orthogonal.)}
\begin{align}
& \dket{V_{\K_n\backslash k_n,k_n}^{\to k_{n+1}}} \coloneqq \big(\dket{\id}^{A_{k_n}^OO'} \otimes \ket{k_n}^\Gamma \big) * \ket{V_{\K_n}^{k_{n+1}}} \notag \\
% & \hspace{7mm} = \dket{\id}^{A_{k_n}^OO'} \!* \big( \ket{k_n}^\Gamma \!* \ket{\omega_{\K_n}^+} * \ket{w_{(\K_n,k_{n+1})}} \big) \notag \\
& \hspace{25mm} \in \HS^{A_{k_n}^O\alpha_nA_{k_{n+1}}^I\alpha_{n+1}}. \label{eq:V_Knkn_kn_kn1}
\end{align}
With this, and using the definition of Eq.~\eqref{eq:def_omega_Kn}, Eq.~\eqref{eq:inv_omega_Kn} gives
\begin{align}
\sum_{k_n\in\K_n} \ket{w_{(\K_n\backslash k_n,k_n)}} * \dket{V_{\K_n\backslash k_n,k_n}^{\to k_{n+1}}} = \ket{w_{(\K_n,k_{n+1})}}. \label{eq:sum_kn_link_prod_wV}
\end{align}
By further defining $\dket{V_{\emptyset,\emptyset}^{\to k_1}} \coloneqq \ket{w_{(\emptyset,k_1)}} \in \HS^{PA_{k_1}^I\alpha_1}$ and
\begin{align}
& \ket{w_{(k_1,\ldots,k_n)}} \coloneqq \dket{V_{\emptyset,\emptyset}^{\to k_1}} * \dket{V_{\emptyset,k_1}^{\to k_2}} * \dket{V_{\{k_1\},k_2}^{\to k_3}} * \cdots \notag \\
& \hspace{37mm} \cdots * \dket{V_{\{k_1,\ldots,k_{n-2}\},k_{n-1}}^{\to k_n}} \label{eq:def_wknm1_wkn}
\end{align}
for all $(k_1,\ldots,k_n)$, one recursively obtains from Eq.~\eqref{eq:sum_kn_link_prod_wV}
\begin{align}
& \ket{w_{(\K_n,k_{n+1})}} = \sum_{\substack{(k_1,\ldots,k_n): \\ \{k_1,\ldots,k_n\} = \K_n}} \!\!\!\!\! \ket{w_{(k_1,\ldots,k_n,k_{n+1})}} \label{eq:QCQC_rec_Vs}
\end{align}
for all $0 \le n \le N$ and all $\K_n$, $k_{n+1}$, as desired (see Eq.~\eqref{eq:def_w_Knm1_kn}).

Recall that the matrices $W_{(\K_{n-1},k_n)} = \Tr_{\alpha_n} \ketbra{w_{(\K_{n-1},k_n)}}{w_{(\K_{n-1},k_n)}}$ are assumed to satisfy Eq.~\eqref{eq:charact_W_QCQC_decomp}. This implies that the operations $V_{\emptyset,\emptyset}^{\to k_1}$, $V_{\K_n \backslash k_n,k_n}^{\to k_{n+1}}$ and $V_{\N \backslash k_N,k_N}^{\to F}$ constructed above (via their Choi representations) satisfy the (equivalent, once Eq.~\eqref{eq:QCQC_rec_Vs} is established, with $\ket{w_{(k_1,\ldots,k_n)}}$ defined in Eq.~\eqref{eq:def_wknm1_wkn}) TP conditions of Eqs.~\eqref{eq:TP_constr_QCQC_1}--\eqref{eq:TP_constr_QCQC_N}, and thus define valid internal circuit operations for a QC-QC.

All in all, we thus find that
\begin{align}
W = W_{(\N,F)} = \Tr_{\alpha_F} \ketbra{w_{(\N,F)}}{w_{(\N,F)}}
\end{align}
with $\alpha_F \!\coloneqq\! \alpha_{N+1}$, $\ket{w_{(\N\!,F)}} \!=\! \sum_{(k_1\!,\ldots,k_N\!)} \! \ket{w_{(k_1\!,\ldots,k_N\!,F)}}$ according to Eq.~\eqref{eq:QCQC_rec_Vs} (for $n = N+1$, with again $k_{N+1} \coloneqq F$), and with the $\ket{w_{(k_1,\ldots,k_N,F)}}$'s defined by Eq.~\eqref{eq:def_wknm1_wkn}---as in Proposition~\ref{prop:descr_W_QCQC}.
This proves that $W$ is indeed the process matrix of the QC-QC thus constructed.

\subsubsection{Proof of Proposition~\ref{prop:charact_Wr_probQCQC}}

Once again, the proofs above extend easily to the characterisation of probabilistic QC-QCs, as given by Proposition~\ref{prop:charact_Wr_probQCQC} in Sec.~\ref{sec:probQCQCs}, in an analogous way to the proofs for pQC-FOs and pQC-CCs.

For the necessary condition, consider a pQC-QC $\{W^{[r]}\}_r$---i.e., according to Proposition~\ref{prop:descr_Wr_probQCCC}, a set of matrices $W^{[r]}$ of the form of Eq.~\eqref{eq:Wr_probQCQC}, with the operations $V_{\emptyset,\emptyset}^{\to k_1}$, $V_{\K_{n-1},k_n}^{\to k_{n+1}}$ and $V_{\N \backslash k_N,k_N}^{\to F \, [r]}$ satisfying the TP conditions of Eqs.~\eqref{eq:TP_constr_QCQC_1}--\eqref{eq:TP_constr_QCQC_n} and~\eqref{eq:TP_constr_probQCQC_N}.
Introducing some additional ancillary system $\alpha_F'$ with computational basis states $\ket{r}^{\alpha_F'}$, and defining $V_{\N \backslash k_N,k_N}^{\to F} \coloneqq \sum_r V_{\N \backslash k_N,k_N}^{\to F \, [r]} \otimes \ket{r}^{\alpha_F'}$ and $\ket{w_{(\N,F)}} \coloneqq \sum_r \ket{w_{(\N,F)}^{[r]}} \otimes \ket{r}^{\alpha_F'}$, one can then write $\sum_r W^{[r]} = \sum_r \Tr_{\alpha_F}\ketbra{w_{(\N,F)}^{[r]}}{w_{(\N,F)}^{[r]}} = \Tr_{\alpha_F\alpha_F'}\ketbra{w_{(\N,F)}}{w_{(\N,F)}}$, with $\ket{w_{(\N,F)}}$ of the form of Eq.~\eqref{eq:W_QCQC}, and with the operations $V_{\emptyset,\emptyset}^{\to k_1}$, $V_{\K_{n-1},k_n}^{\to k_{n+1}}$ and $V_{\N \backslash k_N,k_N}^{\to F}$ now satisfying the appropriate TP conditions~\eqref{eq:TP_constr_QCQC_1}--\eqref{eq:TP_constr_QCQC_N} for a QC-QC (with $\alpha_F$ replaced by $\alpha_F\alpha_F'$). This shows, according to Proposition~\ref{prop:descr_W_QCQC}, that $\sum_r W^{[r]}$ is indeed the process matrix of a QC-QC.

Conversely for the sufficient condition, consider a set of positive semidefinite matrices $\{W^{[r]} \in \L(\HS^{PA_\N^{IO}F}\}_r$ whose sum is the process matrix of a QC-QC. As we did before, let us define the ``extended'' matrix $W' \coloneqq \sum_r W^{[r]} \otimes \ketbra{r}{r}^{F'} \in \L(\HS^{PA_\N^{IO}FF'})$ by introducing an additional Hilbert space $\HS^{F'}$ with computational basis states $\ket{r}^{F'}$.
The decomposition of $\sum_r W^{[r]}$ in terms of positive semidefinite matrices $W_{(\K_{n-1},k_n)}$ obtained from Proposition~\ref{prop:charact_W_QCQC} readily gives essentially the same decomposition for $W'$, which also satisfies the required constraints (with just $F$ replaced by $FF'$)--which proves that $W'$ is also the process matrix of a QC-QC. According to the proof in the previous subsection, we can thus construct internal circuit operations $V_{\emptyset,\emptyset}^{\to k_1}$, $V_{\K_{n-1},k_n}^{\to k_{n+1}}$ (satisfying the TP conditions of Eqs.~\eqref{eq:TP_constr_QCQC_1}--\eqref{eq:TP_constr_QCQC_n}) and $V_{\N \backslash k_N,k_N}^{\to FF'}$, such that $W' = \Tr_{\alpha_F}\ketbra{w_{(\N,FF')}}{w_{(\N,FF')}}$ with $\ket{w_{(\N,FF')}} \coloneqq \sum_{(k_1,\ldots,k_N)} \dket{V_{\emptyset,\emptyset}^{\to k_1}} * \dket{V_{\emptyset,k_1}^{\to k_2}} \cdots * \dket{V_{\{k_1,\ldots,k_{N-2}\},k_{N-1}}^{\to k_N}} * \dket{V_{\{k_1,\ldots,k_{N-1}\},k_N}^{\to FF'}}$ satisfying the TP condition of Eq.~\eqref{eq:TP_constr_QCQC_N} (with $F$ replaced by $FF'$).
Defining (via their Choi representation) the operations $\dket{V_{\{k_1,\ldots,k_{N-1}\},k_N}^{\to F \, [r]}} \coloneqq \dket{V_{\{k_1,\ldots,k_{N-1}\},k_N}^{\to FF'}} * \ket{r}^{F'}$ and from these the vectors $\ket{w_{(\N,F)}^{[r]}} \coloneqq \sum_{(k_1,\ldots,k_N)} \dket{V_{\emptyset,\emptyset}^{\to k_1}} * \dket{V_{\emptyset,k_1}^{\to k_2}} \cdots * \dket{V_{\{k_1,\ldots,k_{N-2}\},k_{N-1}}^{\to k_N}} * \dket{V_{\{k_1,\ldots,k_{N-1}\},k_N}^{\to F \, [r]}} = \ket{w_{(\N,FF')}} * \ket{r}^{F'}$ (which now satisfy the TP condition of Eq.~\eqref{eq:TP_constr_probQCQC_N}), we find that each $W^{[r]} = W' * \ketbra{r}{r}^{F'} = \Tr_{\alpha_F}\ketbra{w_{(\N,F)}^{[r]}}{w_{(\N,F)}^{[r]}}$ is of the form of Eq.~\eqref{eq:Wr_probQCQC}, which proves, according to Proposition~\ref{prop:descr_Wr_probQCQC}, that $\{W^{[r]}\}_r$ is indeed a pQC-QC.

\section{Quantum circuits with operations used in parallel}
\label{app:QCPARs}

In this appendix we describe quantum circuits with operations used in parallel (QC-PARs), and show explicitly how these can be obtained as particular cases of QC-FOs.

\begin{figure*}[t]
	\begin{center}
	\includegraphics[scale=.57]{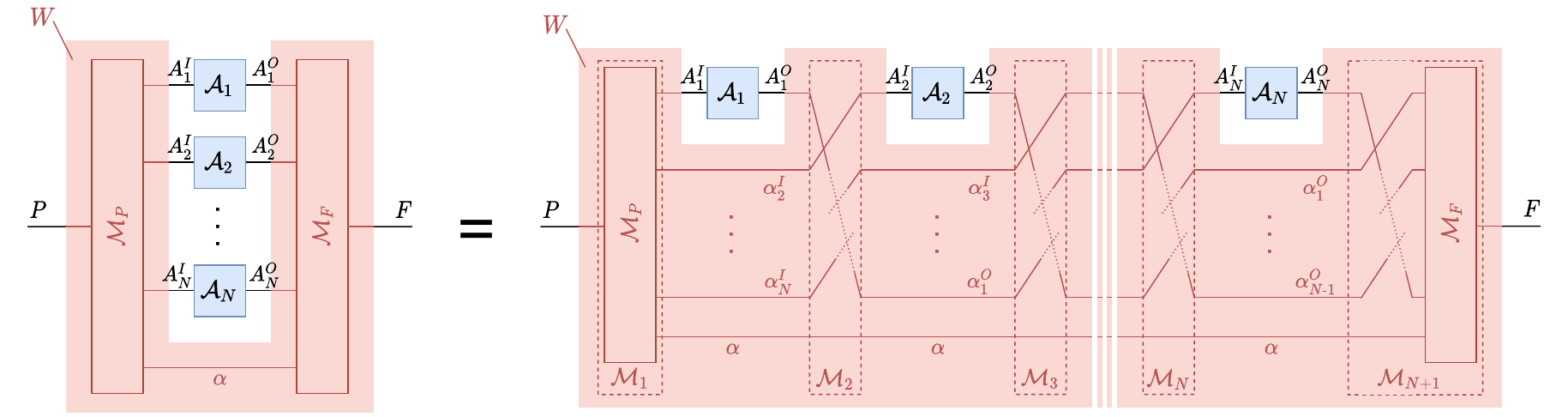}
	\end{center}
	\caption{The left hand side shows a quantum circuit with operations used in parallel (QC-PAR), with its process matrix representation given by $W = M_P * M_F$. The right hand side shows an equivalent circuit that conforms to our description of QC-FOs, with internal circuit operations $\M_2, \ldots, \M_N$ that simply transfer the inputs of subsequent operations and the outputs of previous ones via ancillary systems, and with process matrix representation $W = M_1 * M_2 * \cdots * M_{N+1} = M_P * M_F$. This demonstrates that QC-PARs are a particular case of QC-FOs.}
	\label{fig:QCPAR}
\end{figure*}

We consider a circuit as on the left hand side in Fig.~\ref{fig:QCPAR}, with just a first internal circuit operation $\M_P: \L(\HS^P) \to \L(\HS^{A_\N^I\alpha})$ (with Choi matrix $M_P \in \L(\HS^{PA_\N^I\alpha})$) and a final internal circuit operation $\M_F: \L(\HS^{A_\N^O\alpha}) \to \L(\HS^F)$ (with Choi matrix $M_F \in \L(\HS^{A_\N^O\alpha F})$), which satisfy the TP conditions (easily obtained in the same way as those of Eqs.~\eqref{eq:TP_constr_QCFO_1}--\eqref{eq:TP_constr_QCFO_N} for the more general QC-FO case)
\begin{align}
& \Tr_{A_\N^I\alpha} M_P = \id^P \label{eq:TP_constr_QCPAR_P} \\[1mm]
& \textup{and} \quad \Tr_F (M_P * M_F) = \Tr_\alpha M_P \otimes \id^{A_\N^O}. \label{eq:TP_constr_QCPAR_F}
\end{align}
The corresponding process matrix is easily obtained as
\begin{align}
W = M_P * M_F \quad \in \ \L(\HS^{PA_\N^{IO}F}).
\end{align}

The fact that such a (positive semidefinite) process matrix satisfies Eq.~\eqref{eq:charact_W_QCPAR} in Proposition~\ref{prop:charact_W_QCPAR} follows directly from the TP conditions above, with $W_{(I)} = \Tr_\alpha M_P$.
(Note that it also clearly satisfies Eq.~\eqref{eq:charact_W_QCFO} in Proposition~\ref{prop:charact_W_QCFO}, with $W_{(n)} = \Tr_{A_{\{n+1,\ldots,N\}}^I\alpha} M_P \otimes \id^{A_{\{1,\ldots,n-1\}}^O}$ for each $n$.)
Conversely, consider a positive semidefinite matrix $W \in \L(\HS^{PA_\N^{IO} F})$ satisfying Eq.~\eqref{eq:charact_W_QCPAR}. Following the proofs of the previous appendix (see Sec.~\ref{app:proof_charact_QCFOs_SC}, and Footnote~\ref{footnote:explicit_V_n1_QCFO} in particular), one can diagonalise $W_{(I)}$ in the form $W_{(I)} = \sum_i \ketbra{w_{(I)}^i}{w_{(I)}^i}$ with orthogonal nonzero vectors $\ket{w_{(I)}^i} \in \HS^{PA_\N^I}$, introduce a ($\rank W_{(I)}$)-dimensional ancillary Hilbert space $\HS^\alpha$ with computational basis $\{\ket{i}^\alpha\}_i$, and define $\ket{w_{(I)}} \coloneqq \sum_i \ket{w_{(I)}^i} \ket{i}^\alpha$, $M_P \coloneqq \ketbra{w_{(I)}}{w_{(I)}}$ (such that $\Tr_\alpha M_P = W_{(I)}$) and $M_F \coloneqq \sum_{i,i'} \ketbra{i}{i'}^\alpha \otimes \Big( \big( \frac{\bra{w_{(I)}^i}}{\braket{w_{(I)}^i}{w_{(I)}^i}} \otimes \id^{A_\N^OF}\big) W \big( \frac{\ket{w_{(I)}^{i'}}}{\braket{w_{(I)}^{i'}}{w_{(I)}^{i'}}} \otimes \id^{A_\N^OF}\big) \Big)$. One can check that the maps $M_P$ and $M_F$ thus defined satisfy the TP constraints~\eqref{eq:TP_constr_QCPAR_P}--\eqref{eq:TP_constr_QCPAR_F} above, and allow one to recover $W = M_P * M_F$ as the process matrix of the corresponding QC-PAR.

\medskip

In order to see that our presentation of QC-PARs here fits in the more general description of QC-FOs, let us introduce, for each $n = 1, \ldots, N$, some ancillary systems $\HS^{\alpha_n} \coloneqq \HS^{\alpha_1^O\dots\alpha_{n-1}^O\alpha_{n+1}^I\dots\alpha_N^I\alpha}$, with each $\HS^{\alpha_k^I}$ isomorphic to $\HS^{A_k^I}$ and each $\HS^{\alpha_k^O}$ isomorphic to $\HS^{A_k^O}$ so as to transfer the inputs of subsequent operations and the outputs of previous ones via these ancillary systems (as hinted at in the main text); cf.\ Fig.~\ref{fig:QCPAR}.

The internal circuit operations can then be defined as follows.
The first operation $\M_1: \L(\HS^P) \to [ \L(\HS^{A_1^I\alpha_1}) = \L(\HS^{A_1^I\alpha_2^I\ldots\alpha_N^I\alpha}) ]$ is taken to be formally the same as $\M_P: \L(\HS^P) \to \L(\HS^{A_\N^I\alpha})$, up to the identification $\alpha_k^I \equiv A_k^I$ (via the isomorphism that relates each $\HS^{\alpha_k^I}$ to $\HS^{A_k^I}$).%
\footnote{More rigorously: $\M_1$ is defined as the composition of $\M_P$ with the ``identity'' operations $\ket{i}^{A_k^I} \mapsto \ket{i}^{\alpha_k^I}$ (where $\{\ket{i}^{A_k^I}\}_i$ and $\{\ket{i}^{\alpha_k^I}\}_i$ are the computational bases of $\HS^{A_k^I}$ and $\HS^{\alpha_k^I}$, in one-to-one correspondence), for $k = 2, \ldots, N$. In terms of their Choi matrices: $M_1 = M_P * \dketbra{\id}{\id}^{A_2^I\alpha_2^I} * \cdots * \dketbra{\id}{\id}^{A_N^I\alpha_N^I}$. (The following operations $\M_n$ can similarly be defined more rigorously.)}
The subsequent internal circuit operations $\M_{n+1}: [ \L(\HS^{A_n^O\alpha_n}) = \L(\HS^{A_n^O\alpha_1^O\ldots\alpha_{n-1}^O\alpha_{n+1}^I\ldots\alpha_N^I\alpha})] \to [\L(\HS^{A_{n+1}^I\alpha_{n+1}}) = \L(\HS^{A_{n+1}^I\alpha_1^O\ldots\alpha_n^O\alpha_{n+2}^I\ldots\alpha_N^I\alpha})]$ are taken to be the identity, up to the identifications $A_n^O \equiv \alpha_n^O$ and $\alpha_{n+1}^I \equiv A_{n+1}^I$. Finally, the last internal operation $\M_{N+1}: [\L(\HS^{A_N^O\alpha_N}) = \L(\HS^{A_N^O\alpha_1^O\ldots\alpha_{N-1}^O\alpha})] \to \L(\HS^F)$ is taken to be formally the same as $\L(\HS^{A_\N^O\alpha}) \to \L(\HS^F)$, up to the identification $\alpha_k^O \equiv A_k^O$.

One can easily verify that one thus recovers the process matrix as $W = M_1 * M_2 * \cdots * M_{N+1} = M_P * M_F$.

\medskip

The probabilistic counterpart of a QC-PAR is obtained by replacing the last internal circuit operation $\M_F$ by an instrument $\{\M_F^{[r]}\}_r$, which satisfies the TP condition of Eq.~\eqref{eq:TP_constr_QCPAR_F}, with $M_F$ replaced by $\sum_r M_F^{[r]}$. 
It is immediate to check that the set of probabilistic process matrices $\{W^{[r]}\}_r$, with $W^{[r]} = M_P * M_F^{[r]}$ satisfies the characterisation of Proposition~\ref{prop:charact_Wr_probQCPAR}. 
Conversely, for a set of positive semidefinite matrices $\{W^{[r]} \in \L(\HS^{PA_\N^{IO} F})\}_r$ whose sum is the process matrix of a QC-PAR, we define (as we did for pQC-FOs, pQC-CCs and pQC-QCs in the previous appendix) the ``extended'' matrix $W' \coloneqq \sum_r W^{[r]} \otimes \ketbra{r}{r}^{F'} \in \L(\HS^{PA_\N^{IO}FF'})$, which is the process matrix of a (deterministic) QC-PAR with the global future space $\HS^{FF'}$ as per Proposition~\ref{prop:charact_W_QCPAR}, and can thus be decomposed as $W' = M_P * M_{F}'$, where $M_P \in \L(\HS^{P A^I_\N \alpha})$ and $M_{F}' \in \L(\HS^{A^O_\N \alpha F F'})$ satisfy the TP conditions of Eqs.~\eqref{eq:TP_constr_QCPAR_P}--\eqref{eq:TP_constr_QCPAR_F}.
Defining the CP maps $M_{F}^{[r]} \coloneqq M_{F}' * \ketbra{r}{r}^{F'}$, we obtain $W^{[r]} = W' * \ketbra{r}{r}^{F'} = M_P * M_{F}^{[r]}$, which provides a realisation of the matrices $W^{[r]}$ as probabilistic process matrices of a pQC-PAR.

\section{Further examples of QC-QCs}
\label{app:furtherExamplesQCQCs}

Here we present some generalisations of the examples of QC-QCs presented in Sec.~\ref{subsec:QCQC_examples}.

\subsection{The ``quantum $N$-switch'' and generalisations}
\label{app:Nswitch}

The quantum switch can easily be generalised to a setup involving $N$ operations $A_k$ (all with isomorphic $d_\textup{t}$-dimensional input and output Hilbert spaces, for simplicity) and a ``control system'' used to coherently control between applying the operations in the $N!$ possible permutations of orders (or some subset thereof)~\cite{colnaghi12,araujo14,facchini15,procopio19,procopio20,taddei20,wilson20,sazim20,wilson20a,pinzani20}.
Such an $N$-operation quantum switch (or ``quantum $N$-switch'') requires, in general, a control system of dimension $N!$ so as to encode each of the possible permutations $\pi\coloneqq(k_1,\ldots,k_N)$ of the $N$ operations.
It can be obtained as a QC-QC, for instance by introducing $d_\textup{t}$-dimensional Hilbert spaces $\HS^{P_\textup{t}}$ and $\HS^{F_\textup{t}}$ (for the ``target'' systems in the global past and future) and $N!$-dimensional isomorphic Hilbert spaces $\HS^{P_\textup{c}}$, $\HS^{F_\textup{c}}$ and $\HS^{\alpha_n}$ (for the global past/future ``control'' and for the ancillary systems) with orthonormal bases $\{\ket{\pi}\}_{\pi\in\Pi_\N}$, with $\Pi_\N$ denoting the set of all permutations $\pi = (\pi(1),\ldots,\pi(N))$ of $\N$, and by taking%
\footnote{Note that with the choice of Eq.~\eqref{eq:Vs_N_switch}, there is some redundancy in the information encoded in the control systems $C_n$ and in the ancillary systems $\alpha_n$ (when considering the ``complete'' internal circuit operation $\tilde V_n$); indeed, given $(\K_{n-1},k_n)$ one would only need to record $(k_1,\dots,k_{n-2})$ in an ancillary system to uniquely determine the full history $(k_1,\dots,k_n)$.
For practical purposes, the dimension of the Hilbert spaces used to encode this information could thus be reduced. 
See also the discussion in Sec.~\ref{subsubsec:impl_switch} about redundant encodings in certain implementations of the quantum switch.}
\begin{align}
& \forall \, k_1, \quad \dket{V_{\emptyset,\emptyset}^{\to k_1}} = \dket{\id}^{P_\textup{t}A_{k_1}^I} \otimes \sum_{\substack{\pi\in \Pi_\N:\\ \pi(1) = k_1}} \ket{\pi}^{P_\textup{c}} \otimes \ket{\pi}^{\alpha_1}, \notag \\[1mm]
& \forall \, \K_{n-1}, k_n, k_{n+1}, \notag \\[-1mm]
& \hspace{3mm} \dket{V_{\K_{n-1},k_n}^{\to k_{n+1}}} = \dket{\id}^{A_{k_n}^O\!A_{k_{n+1}}^I} \otimes \hspace{-8mm} \sum_{\substack{\pi\in \Pi_\N: \\ \{\pi(1),\ldots,\pi(n-1)\} = \K_{n-1},\\ \pi(n) = k_n, \pi(n+1) = k_{n+1}}} \hspace{-10mm} \ket{\pi}^{\alpha_n} \!\otimes\! \ket{\pi}^{\alpha_{n+1}}\!, \notag \\[1mm]
& \forall \, k_N, \quad \dket{V_{\N\!\backslash k_N\!,k_N}^{\to F}} = \dket{\id}^{A_{k_N}^O\!F_\textup{t}} \otimes \!\sum_{\substack{\pi\in \Pi_\N:\\ \pi(N) = k_N}}\!\!\! \ket{\pi}^{\alpha_N} \!\otimes\! \ket{\pi}^{F_\textup{c}}. \notag \\[-3mm] \label{eq:Vs_N_switch}
\end{align}
These indeed satisfy the TP constraints of Eqs.~\eqref{eq:TP_constr_QCQC_1}--\eqref{eq:TP_constr_QCQC_N}, and give $W_\textup{QS}^{(N)} = \ketbra{w_\textup{QS}^{(N)}}{w_\textup{QS}^{(N)}}$ with
\begin{align}
& \ket{w_\textup{QS}^{(N)}} \coloneqq \!\sum_{(k_1,\ldots,k_N)=:\pi}\! \ket{\pi}^{P_\textup{c}} \dket{\id}^{P_\textup{t}A_{k_1}^I} \dket{\id}^{A_{k_1}^OA_{k_2}^I} \cdots \notag \\[-4mm]
& \hspace{37mm} \cdots \dket{\id}^{A_{k_{N-1}}^O\!A_{k_N}^I} \!\dket{\id}^{A_{k_N}^O\!F_\textup{t}} \!\ket{\pi}^{F_\textup{c}} \notag \\
& \hspace{45mm} \in \HS^{P_\textup{c}P_\textup{t}A_\N^{IO}F_\textup{t}F_\textup{c}} \label{eq:def_N_switch}
\end{align}
according to Proposition~\ref{prop:descr_W_QCQC}.
Note that in the quantum $N$-switch, while the control of causal order is indeed quantum, there is no real notion of ``dynamical'' causal order, as the full order $\pi\coloneqq(k_1,\ldots,k_N)$ corresponding to each component $\ket{\pi}^{P_\textup{c}} \dket{\id}^{P_\textup{t}A_{k_1}^I} \cdots \dket{\id}^{A_{k_N}^O\!F_\textup{t}} \ket{\pi}^{F_\textup{c}}$ of $\ket{w_\textup{QS}^{(N)}}$ is encoded from the start in the state of the control system (and, with the choice of Eq.~\eqref{eq:Vs_N_switch}, is transmitted, untouched, throughout the circuit by the ancillary states $\ket{\pi}^{\alpha_n}$).

As was the case for the quantum switch (i.e., when $N=2$), the process matrix $\Tr_{F_\textup{c}} W_\textup{QS}^{(N)}$ obtained by tracing out the system in $\HS^{F_\textup{c}}$ from the quantum $N$-switch is simply an incoherent mixture of terms corresponding to the $N!$ different orders.
Indeed, one obtains a ``classical $N$-switch'' generalising the classical switch described in Sec.~\ref{subsec:QCCCs_example}.

\medskip

While the quantum $N$-switch is the most straightforward and extensively studied generalisation of the quantum switch, further generalisations are possible.
The simplest such possibility would be to replace all the identity channels in the quantum $N$-switch (i.e., the $\dket{\id}$'s in Eq.~\eqref{eq:Vs_N_switch} or~\eqref{eq:def_N_switch}) applied to the target system by any, potentially different, arbitrary unitaries (or even, taking the external operations to have non-isomorphic input and output Hilbert spaces, isometries), as was considered, for instance, for the case of $N=2$ in Refs.~\cite{taddei19,yokojima20}.
Such a choice would indeed still give QC-QCs satisfying the TP constraints of Eqs.~\eqref{eq:TP_constr_QCQC_1}--\eqref{eq:TP_constr_QCQC_N}, as is easily verified. 
Taking this one step further, one could introduce further ancillary systems $\alpha_n'$ to act as ``memory channels'' across the different time slots.

Like the quantum $N$-switch, none of these generalisations exhibit any form of really dynamical causal order, and instead exploit only coherent control conditioned on some quantum system that remains fixed throughout the process.
Indeed, this is also true of other previous attempts to define coherent control of causal order (see, e.g., Ref.~\cite{pinzani20}).

\subsection{A family of new QC-QCs with dynamical and coherently-controlled causal order}
\label{app:new_QCQC}

Here we present a more general family of new QC-QCs, of which the example given in Sec.~\ref{subsubsec:new_QCQC} of the main text is a specific case.

We consider, as in the main text, QC-QCs with $N=3$ operations $A_1, A_2, A_3$.
For simplicity, all input and output Hilbert spaces $\HS^{A_k^I}$ and $\HS^{A_k^O}$ are taken to be isomorphic (of the same dimension $d_\textup{t}$).
In contrast to the example of Sec.~\ref{subsubsec:new_QCQC}, we will consider here a nontrivial global past $P\coloneqq P_\textup{t} P_\textup{c}$ (with corresponding Hilbert space $\HS^P\coloneqq \HS^{P_\textup{t}P_\textup{c}}$, $d_{P_\textup{t}}=d_\textup{t}$ and $d_{P_\textup{c}}=3$) and global future $F\coloneqq F_\textup{t} F_\alpha F_\textup{c}$ (with corresponding Hilbert space $\HS^F\coloneqq \HS^{F_\textup{t} F_\alpha F_\textup{c}}$, $d_{F_\textup{t}}=d_\textup{t}$, $d_{F_\alpha} \ge 2$ and $d_{F_\textup{c}}=3$).

A relatively simple way to satisfy the TP constraints of Eqs.~\eqref{eq:TP_constr_QCQC_1}--\eqref{eq:TP_constr_QCQC_N} is to consider operators $V_{\emptyset,\emptyset}^{\to k_1}$, $V_{\K_{n-1},k_n}^{\to k_{n+1}}$ and $V_{\N\backslash k_N,k_N}^{\to F}$ of the form
\begin{align}
V_{\emptyset,\emptyset}^{\to k_1} & \coloneqq \id^{P_\textup{t}\to A_{k_1}^I} \otimes \bra{k_1}^{P_\textup{c}}, \label{eq:new_QCQC_Vk1} \\[2mm]
V_{\emptyset,k_1}^{\to k_2} & \coloneqq (\id^{A_{k_2}^I} \otimes \bra{\sigma_{(k_1,k_2)}}^\alpha) \mathcal{V}_{k_1}, \label{eq:new_QCQC_Vk2}
\end{align}
for some isometries%
\footnote{$\mathcal{V}_{k_1} : \HS^{A_{k_1}^O} \to \HS^{A_{k_2}^I\alpha}$ may depend, as indicated by its subscript, on $k_1$, but it must have the same form for both values of $k_2 \neq k_1$ (recall that all $\HS^{A_{k_2}^I}$'s are isomorphic). Similarly below, $\mathcal{V}_{k_3}' : \HS^{A_{k_2}^O\alpha'} \to \HS^{A_{k_3}^I\alpha_3}$ may depend on $k_3$, but for a given $k_3$ it must have the same form for both initial orders $(k_1,k_2)$ (with all $\HS^{A_{k_2}^O}$'s being isomorphic).}
$\mathcal{V}_{k_1} : \HS^{A_{k_1}^O} \to \HS^{A_{k_2}^I\alpha}$, where we introduced a 2-dimensional auxiliary Hilbert space $\HS^\alpha$ with orthonormal basis $\{\ket{0}^\alpha,\ket{1}^\alpha\}$, which encodes the ``signature'' of the order $(k_1,k_2)$ in such a way that $\sigma_{(k_1,k_2)} \coloneqq 0$ if $k_2 = k_1+1 \pmod{3}$, and $\sigma_{(k_1,k_2)} \coloneqq 1$ if $k_2 = k_1+2 \pmod{3}$ (and such that $\forall \, k_1, \sum_{k_2} \ketbra{\sigma_{(k_1,k_2)}}{\sigma_{(k_1,k_2)}}^\alpha = \id^\alpha$);
\begin{align}
V_{\{k_1\},k_2}^{\to k_3} & \coloneqq \mathcal{V}_{k_3}' (\id^{A_{k_2}^O} \otimes \ket{\sigma_{(k_1,k_2)}}^{\alpha'}), \label{eq:new_QCQC_Vk3}
\end{align}
for some isometries $\mathcal{V}_{k_3}' : \HS^{A_{k_2}^O\alpha'} \to \HS^{A_{k_3}^I\alpha_3}$, where we similarly introduced a 2-dimensional auxiliary Hilbert space $\HS^{\alpha'}$, as well as an ancillary ($d_{F_\alpha}$-dimensional) system $\alpha_3$; and
\begin{align}
V_{\{k_1,k_2\},k_3}^{\to F} & \coloneqq \id^{A_{k_3}^O \to F_\textup{t}} \otimes \id^{\alpha_3 \to F_\alpha} \otimes \ket{k_3}^{F_\textup{c}}. \label{eq:new_QCQC_VF}
\end{align}
According to Proposition~\ref{prop:descr_W_QCQC}, the process matrix corresponding to the choice of operators above is then $W = \ketbra{w}{w}$ with $\ket{w} \coloneqq \sum_{(k_1,k_2,k_3)} \ket{w_{(k_1,k_2,k_3,F)}}$ and
\begin{align}
\ket{w_{(k_1,k_2,k_3,F)}} \hspace{-15mm} \notag \\
= & \dket{V_{\emptyset,\emptyset}^{\to k_1}} * \dket{V_{\emptyset,k_1}^{\to k_2}} * \dket{V_{\{k_1\},k_2}^{\to k_3}} * \dket{V_{\{k_1,k_2\},k_3}^{\to F}} \notag \\
= & \ket{k_1}^{P_\textup{c}} \otimes \dket{\id}^{P_\textup{t}A_{k_1}^I} \otimes \big( \dket{{\cal V}_{k_1}}^{A_{k_1}^O A_{k_2}^I\alpha} * \ket{\sigma_{(k_1,k_2)}}^\alpha \big) \notag \\
& \qquad \otimes \big( \dket{{\cal V}_{k_3}'}^{A_{k_2}^O \alpha' A_{k_3}^I \alpha_3} * \ket{\sigma_{(k_1,k_2)}}^{\alpha'} * \dket{\id}^{\alpha_3 F_\alpha} \big) \notag \\
& \qquad \qquad \otimes \dket{\id}^{A_{k_3}^O F_\textup{t}} \otimes \ket{k_3}^{F_\textup{c}}.
\end{align}
Since $W = \ketbra{w}{w}$ is a rank-1 process matrix, and since there exists some preparation of states in the global past such that the induced process is not compatible with any given operation being applied first, then it follows (referring to the same argument as for the ``pure'' quantum switch~\cite{araujo15,oreshkov16}) that $W$ is causally nonseparable.

As in the example of Sec.~\ref{subsubsec:new_QCQC}, one may now choose to fix the preparation of some particular global past state, and/or (perhaps partially) trace out some systems in the global future. 
Whether the resulting process matrix remains causally nonseparable or not may then depend on the choice of initial state and of isometries ${\cal V}_{k_1}$ and ${\cal V}_{k_3}'$. 
The specific example of the main text corresponds to inputting the initial state $\ket{\psi}^{P_\textup{t}}\otimes\frac{1}{\sqrt{3}}\sum_{k_1}\ket{k_1}^{P_\textup{c}}$, choosing ${\cal V}_{k_1} = V_\textsc{Copy}$ and ${\cal V}_{k_3}' = V_\textsc{CNot}$ (see Sec.~\ref{subsubsec:impl_new_QCQC}) and tracing out $F$ completely, which indeed results in a causally nonseparable process.%
\footnote{More specifically, we computed the random robustness~\cite{araujo15,branciard16a,wechs19} for $1000$ random qubit states $\ket{\psi}$ via SDP, and always found values in the interval $[0.51, 0.53]$, which indeed certifies causal nonseparability.}
Had we chosen, for instance, ${\cal V}_{k_3}' = \id^{A_{k_2}^O \to A_{k_3}^I} \otimes \id^{\alpha' \to \alpha_3}$ instead (with the same initial state preparation and the same ${\cal V}_{k_1}$, corresponding to removing the \textsc{CNot} gates in Fig.~\ref{fig:new_QCQC}), the resulting process matrix after tracing out $F$ would have become causally separable.
Indeed, one would have $\dket{{\cal V}_{k_3}'}^{A_{k_2}^O \alpha' A_{k_3}^I \alpha_3} * \ket{\sigma_{(k_1,k_2)}}^{\alpha'} * \dket{\id}^{\alpha_3 F_\alpha} = \dket{\id}^{A_{k_2}^O A_{k_3}^I}\otimes\ket{\sigma_{(k_1,k_2)}}^{F_\alpha}$, and thus
\begin{equation}
	\Tr_F W = \sum_{(k_1,k_2,k_3)} \Tr_F \ketbra{w_{(k_1,k_2,k_3,F)}}{w_{(k_1,k_2,k_3,F)}}.
\end{equation}
The sum above gives a decomposition of the process $\Tr_F W$ (now with a trivial $F$) into terms $W_{(k_1,\dots,k_3,F)}$ satisfying the conditions of Proposition~\ref{prop:charact_W_QCCC}, thereby showing that this process is a QC-CC and thus causally separable.

\medskip

The construction above thus provides a family of novel QC-QCs that can exhibit a range of different behaviours.
One can imagine yet further generalisations, for example by introducing further ancillary systems in a nontrivial way.
The exploration of such possibilities, or of completely new families of causally nonseparable QC-QCs, provides a key direction for future research.

\section{Quantum circuits with quantum control of causal order cannot violate causal inequalities}
\label{app:only_causal_correl}

In this final appendix we prove Proposition~\ref{prop:causal_correlations}, that QC-QCs (and \emph{a fortiori}, QC-CCs or QC-FOs) can only generate causal correlations.

\medskip

For any subset $\K = \{k_1,\ldots,k_n\}$ of $\N$, we shall denote by $\vec x_\K \coloneqq (x_{k_1},\ldots,x_{k_n})$ and $\vec a_\K \coloneqq (a_{k_1},\ldots,a_{k_n})$ the list of inputs and outputs for the parties in $\K$, and by $A_{\vec a_\K|\vec x_\K} \coloneqq \bigotimes_{k \in \K} A_{a_k|x_k} \in \L(\HS^{A_\K^{IO}})$ the corresponding joint operations (in their Choi representation). With these notations, the correlations $P(\vec a_\N|\vec x_\N)$ obtained from a QC-QC are given, as in Eq.~\eqref{eq:P_a1aN_x1xN} (and for all $\vec x_\N, \vec a_\N$), by%
\footnote{Recall from Sec.~\ref{subsec:causal_correlations} that we consider here a scenario with no (or trivial, 1-dimensional) ``global past'' and ``global future'' Hilbert spaces $\HS^P, \HS^F$. 
For all link products written in this appendix, the Hilbert spaces of their two factors are the same: the link products therefore correspond here to full traces, and give scalar values.}
\begin{align}
P(\vec a_\N|\vec x_\N) = A_{\vec a_\N|\vec x_\N} * W,
\end{align}
with $W$ satisfying the constraints of Proposition~\ref{prop:charact_W_QCQC_trivial_PF}, i.e., such that there exist positive semidefinite matrices $W_{(\K_{n-1},k_n)}$, for all strict subsets $\K_{n-1}$ of $\N$ and all $k_n \in \N\backslash\K_{n-1}$, satisfying Eq.~\eqref{eq:charact_W_QCQC_decomp_trivial_PF}. 
In order to lighten the notations, we will replace here the dummy labels $\K_n$, $k_n$ and $k_{n+1}$ used in Eq.~\eqref{eq:charact_W_QCQC_decomp_trivial_PF} by just $\K$, $k$ and $\ell$, respectively.

Let us define, for any nonempty strict subset $\K$ of $\N$, any $k\in\K$ and any $\ell \in \N \backslash \K$,
\begin{align}
r_{(\K,\ell)}(\vec a_\K | \vec x_\K) & \coloneqq A_{\vec a_\K|\vec x_\K} * \big( \Tr_{A_\ell^I} W_{(\K,\ell)} \big) , \notag \\
s_{(\K\backslash k,k)}(\vec a_\K|\vec x_\K) & \coloneqq A_{\vec a_\K|\vec x_\K} * \big( W_{(\K\backslash k,k)} \otimes \id^{A_k^O} \big) , \label{eq:def_r_s}
\end{align}
with the first definition extending to $r_{(\emptyset,\ell)}(\vec a_{\emptyset} | \vec x_{\emptyset}) \coloneqq \Tr W_{(\emptyset,\ell)}$ for $\K = \emptyset$, and the second one also applying to $\K=\N$.
We note that $r_{(\K,\ell)}(\vec a_\K|\vec x_\K)$ and $s_{(\K\backslash k,k)}(\vec a_\K|\vec x_\K)$ are nonnegative functions %
 of the inputs and outputs of the parties in $\K$, which inherit the following properties from Eq.~\eqref{eq:charact_W_QCQC_decomp_trivial_PF}:
\begin{align}
& \sum_{\ell \in \N} r_{(\emptyset,\ell)}(\vec a_{\emptyset} | \vec x_{\emptyset}) = 1, \notag \\
& \forall \, \emptyset \!\subsetneq\! \K \!\subsetneq\! \N, \sum_{\ell \in \N \backslash \K} \!\!r_{(\K,\ell)}(\vec a_\K | \vec x_\K) = \sum_{k \in \K} \!s_{(\K \backslash k,k)}(\vec a_\K|\vec x_\K), \notag \\[1mm]
& \textup{and} \quad P(\vec a_\N|\vec x_\N) = \sum_{k \in \N} s_{(\N \backslash k,k)}(\vec a_\N|\vec x_\N). \label{eq:consrt_r_s}
\end{align}
This incites us to further define the functions
\begin{align}
f_\K(\vec a_\K | \vec x_\K) & \coloneqq \sum_{\ell \in \N \backslash \K} r_{(\K,\ell)}(\vec a_\K | \vec x_\K) \notag \\[-2mm]
& \hspace{18mm} = \sum_{k \in \K} s_{(\K \backslash k,k)}(\vec a_\K|\vec x_\K), \label{eq:def_fK}
\end{align}
with $f_{\emptyset}(\vec a_\emptyset | \vec x_\emptyset) \coloneqq \sum_\ell r_{(\emptyset,\ell)}(\vec a_{\emptyset} | \vec x_{\emptyset}) = 1$ and $f_\N(\vec a_\N | \vec x_\N) \coloneqq \sum_k s_{(\N \backslash k,k)} (\vec a_\N|\vec x_\N) =  P(\vec a_\N|\vec x_\N)$ for $\K = \emptyset$ and $\K = \N$, resp.

As the $r_{(\K,\ell)}$'s are nonnegative, it is clear from the definition above that for each $\K \subsetneq\N, \vec x_\K, \vec a_\K$ there must exist some weights $q_{\K, \vec x_\K, \vec a_\K}^{\ell} \ge 0$ such that $\sum_{\ell \in \N \backslash \K} q_{\K, \vec x_\K, \vec a_\K}^{\ell} = 1$ and
\begin{align}
r_{(\K,\ell)}(\vec a_\K | \vec x_\K) = f_\K(\vec a_\K | \vec x_\K) \, q_{\K, \vec x_\K, \vec a_\K}^{\ell} \label{eq:r_f_qs}
\end{align}
for all $\ell \in \N \backslash \K$.
Furthermore, using the fact that (for each $x_\ell$) the sum $\sum_{a_\ell} A_{a_\ell|x_\ell}$ is a CPTP map, i.e., that $\Tr_{A_\ell^O} \sum_{a_\ell} A_{a_\ell|x_\ell} = \id^{A_\ell^I}$, one finds (replacing $\K\backslash k$ by $\K$ and $k$ by $\ell (\notin \K)$ in the definition of Eq.~\eqref{eq:def_r_s}) that
\begin{align}
& \sum_{a_\ell} s_{(\K,\ell)}(\vec a_{\K\cup\ell}|\vec x_{\K\cup\ell}) \notag \\[-2mm]
& \qquad = \big( A_{\vec a_\K|\vec x_\K} \otimes \sum_{a_\ell} A_{a_\ell|x_\ell} \big) * \big( W_{(\K,\ell)} \otimes \id^{A_\ell^O} \big) \notag \\[-1mm]
& \qquad = \big( A_{\vec a_\K|\vec x_\K} \otimes \Tr_{A_\ell^O} \sum_{a_\ell} A_{a_\ell|x_\ell} \big) * W_{(\K,\ell)} \notag \\[-1mm]
& \qquad = \big( A_{\vec a_\K|\vec x_\K} \otimes \id^{A_\ell^I} \big) * W_{(\K,\ell)} \notag \\[1mm]
& \qquad = A_{\vec a_\K|\vec x_\K} * \big( \Tr_{A_\ell^I} W_{(\K,\ell)} \big) = r_{(\K,\ell)}(\vec a_\K | \vec x_\K).
\end{align}
It follows, as above, that one can define (for all $\K, \vec x_\K, \vec a_\K$ and all $\ell \in \N \backslash \K$) a valid conditional probability distribution $P^{(\ell)}_{\K,\vec x_\K, \vec a_\K}(a_\ell|x_\ell)$ for party $A_\ell$---which, as indicated by the subscript, depends on $\K$, $\vec x_\K$ and $\vec a_\K$---such that
\begin{align}
s_{(\K,\ell)}(\vec a_{\K\cup\ell}|\vec x_{\K\cup\ell}) = r_{(\K,\ell)}(\vec a_\K | \vec x_\K) \ P^{(\ell)}_{\K,\vec x_\K, \vec a_\K}(a_\ell|x_\ell). \label{eq:s_r_Ps}
\end{align}

\medskip

With this in place, we are now in a position to prove the following:
\begin{proposition}\label{prop:correl_QCQC_rec}
For any $n = 0, \ldots, N$, one can decompose the correlations $P(\vec a_\N|\vec x_\N)$ as
\begin{align}
& P(\vec a_\N|\vec x_\N) \notag \\
& \quad = \sum_{\K: \, |\K| = n} f_\K(\vec a_\K | \vec x_\K) \, P^\textup{causal}_{\K,\vec x_\K,\vec a_\K}(\vec a_{\N \backslash \K}|\vec x_{\N \backslash \K}), \label{eq:correl_QCQC_rec}
\end{align}
where the sum runs over all $n$-partite subsets $\K$ of $\N$, with $f_\K(\vec a_\K | \vec x_\K)$ defined in Eq.~\eqref{eq:def_fK}, and where (for all $\K, \vec x_\K, \vec a_\K$) the $P^\textup{causal}_{\K,\vec a_\K,\vec x_\K}(\vec a_{\N\backslash\K}|\vec x_{\N\backslash\K})$'s are causal correlations for the parties in $\N\backslash\K$ (and with $P^\textup{causal}_{\N,\vec x_\N,\vec a_\N}(\vec a_\emptyset|\vec x_\emptyset) = 1$ for the $n = N$ case).
\end{proposition}

\begin{proof}
We prove the above claim recursively, starting from $n = N$, down to $n=0$.

For $n = N$, the result follows directly from the fact that $f_\N(\vec a_\N | \vec x_\N) =  P(\vec a_\N|\vec x_\N)$ (as noted above, and which follows from the third line of Eq.~\eqref{eq:consrt_r_s}).

Suppose that a decomposition of the form of Eq.~\eqref{eq:correl_QCQC_rec} exists, with a sum over subsets $\K'$ of cardinality $n+1$ (with $n+1 \ge 1$).
Then using the definition of $f_{\K'}(\vec a_{\K'} | \vec x_{\K'})$, rewriting the sums $\sum_{\K': \, |\K'| = n+1} \sum_{\ell \in \K'}$ in the equivalent form $\sum_{\K: \, |\K| = n} \sum_{\ell \in \N \backslash \K}$ (with $\K = \K'\backslash \ell$, so that $\K' = \K\cup\ell$) and using Eqs.~\eqref{eq:s_r_Ps} and~\eqref{eq:r_f_qs}, we obtain:
\begin{align}
& P(\vec a_\N|\vec x_\N) \notag \\
& \quad = \sum_{\K': \, |\K'| = n+1} \, \sum_{\ell \in \K'} s_{(\K' \backslash \ell,\ell)}(\vec a_{\K'}|\vec x_{\K'}) \notag \\[-3mm]
& \hspace{40mm} P^\textup{causal}_{\K',\vec x_{\K'},\vec a_{\K'}}(\vec a_{\N \backslash \K'}|\vec x_{\N \backslash \K'}) \notag \\[1mm]
& \quad = \sum_{\K: \, |\K| = n} \, \sum_{\ell \in \N \backslash \K} s_{(\K,\ell)}(\vec a_{\K\cup\ell}|\vec x_{\K\cup\ell}) \notag \\[-3mm]
& \hspace{35mm} P^\textup{causal}_{\K\cup\ell,\vec x_{\K\cup\ell},\vec a_{\K\cup\ell}}(\vec a_{\N \backslash \K \backslash \ell}|\vec x_{\N \backslash \K \backslash \ell}) \notag \\[1mm]
& \quad = \sum_{\K: \, |\K| = n} f_\K(\vec a_\K | \vec x_\K) \sum_{\ell \in \N \backslash \K} q_{\K, \vec x_\K, \vec a_\K}^{\ell} \ P^{(\ell)}_{\K,\vec x_\K, \vec a_\K}(a_\ell|x_\ell) \notag \\[-1mm]
& \hspace{35mm} P^\textup{causal}_{\K\cup\ell,\vec x_{\K\cup\ell},\vec a_{\K\cup\ell}}(\vec a_{\N \backslash \K \backslash \ell}|\vec x_{\N \backslash \K \backslash \ell}). \label{eq:correl_QCQC_rec_proof}
\end{align}
One can see that $\sum_{\ell \in \N \backslash \K} q_{\K, \vec x_\K, \vec a_\K}^{\ell} \ P^{(\ell)}_{\K,\vec x_\K, \vec a_\K}(a_\ell|x_\ell)$ $P^\textup{causal}_{\K\cup\ell,\vec x_{\K\cup\ell},\vec a_{\K\cup\ell}}(\vec a_{\N \backslash \K \backslash \ell}|\vec x_{\N \backslash \K \backslash \ell})$ in the above expression defines (for each $\K, \vec x_\K, \vec a_\K$) a causal probability distribution for the parties in $\N\backslash\K$~\cite{oreshkov16,abbott16}: indeed it is written as a convex mixture (with weights $q_{\K, \vec x_\K, \vec a_\K}^{\ell}$) of correlations compatible with a given party $\ell \in \N\backslash\K$ acting first (with a response function $P^{(\ell)}_{\K,\vec x_\K, \vec a_\K}(a_\ell|x_\ell)$ which does not depend on the inputs of the other parties in $\N\backslash\K\backslash\ell$) and such that whatever that party's input and output, the conditional correlations $P^\textup{causal}_{\K\cup\ell,\vec x_{\K\cup\ell},\vec a_{\K\cup\ell}}(\vec a_{\N \backslash \K \backslash \ell}|\vec x_{\N \backslash \K \backslash \ell})$ shared by the other parties in $\N\backslash\K\backslash\ell$ are causal.

This shows that Eq.~\eqref{eq:correl_QCQC_rec_proof} provides a decomposition of $P(\vec a_\N|\vec x_\N)$ in the form of Eq.~\eqref{eq:correl_QCQC_rec}, and thereby proves, by recursion, that such a decomposition indeed exists for all $n = 0, \ldots, N$.
\end{proof}

To conclude the proof of Proposition~\ref{prop:causal_correlations}, it then suffices to notice (remembering that $f_{\emptyset}(\vec a_\emptyset | \vec x_\emptyset) = 1$) that the latter simply corresponds to the case $n=0$ of Proposition~\ref{prop:correl_QCQC_rec} above.

\bibliography{bib_classical_quantum_control_of_orders}

\end{document}